\documentclass{article}
\usepackage{amssymb,amsmath,graphicx,ucs,hyperref}
\usepackage[utf8x]{inputenc}
\usepackage[french]{babel}

\newtheorem{proposition}{Proposition}

\newtheorem{theorem}{Théorème}
\newtheorem{lemma}{Lemme}

\newenvironment{proof}{\noindent \emph{Preuve. }}{\hfill \hbox{\rlap{$\sqcap$}$\sqcup$}\\}

\title{Empilements compacts avec trois tailles de disque}
\author{
  Thomas Fernique
  \footnote{Universit\'e Paris 13, CNRS, Sorbonne Paris Cit\'e, UMR 7030, 93430 Villetaneuse, France.}
  \and
  Amir Hashemi
  \footnote{Isfahan University of Technology, 84156-83111 Isfahan, Iran.}
  \and
  Olga Sizova
  \footnote{Faculty of Mathematics, Higher School of Economics, Moscow, Russia}
  }

\begin{document}
\maketitle

\section{Introduction}

Un {\em empilement de disques} est un ensemble de disques d'intérieurs disjoints.
Il est dit {\em compact} si le graphe qui relie les centres des disques tangents est triangulé, c'est-à-dire que les interstices entre les disques sont tous des triangles curvilignes (notion introduite par L\'aszl\'o Fejes T\'oth dans \cite{FT64}).\\

Il n'y a qu'un seul empilement compact avec des disques tous identiques, celui où les disques sont centrés sur les sommets du réseau triangulaire ({\em empilement hexagonal compact}).
Les empilements compacts avec deux tailles de disque ont été classifiés dans \cite{Ken06}.
Il y a neuf ratios de rayons possibles.
Dans tous les cas, l'empilement compact le plus dense est périodique (même si des empilements apériodiques peuvent être possibles).
De plus, dans $7$ des $9$ cas, il a également été montré que l'empilement le plus dense avec ces deux tailles de disque est un empilement compact (\cite{Hep00,Hep03,Ken04}).\\

On s'intéresse ici aux empilements compacts avec des disques de rayon $s<r<1$, où chaque type de disque apparaît.
Sauf précision contraire, un ``empilement'' désignera un tel empilement dans tout ce qui suit.
Ce problème a été examiné dans \cite{Mes17}, où il est montré qu'il y a au plus $11462$ couples $(r,s)$ possibles.
Plusieurs exemples d'empilements sont donnés, mais l'auteur suggère qu'une caractérisation complète est hors de portée des moyens de calcul informatiques actuels.
On montre ici :

\begin{theorem}
  \label{th:164}
  Il y a exactement $164$ couples $(r,s)$ permettant un empilement.
\end{theorem}

Les empilements possibles ne sont pas tous classifiés : on se contente d'en exhiber un, périodique, pour chaque valeur.
La question de dé\-ter\-miner quels sont les empilements les plus denses reste ouverte.
En particulier, en existe-t-il toujours un périodique ?
Reste également ouverte la question de savoir si, quand ils existent, ces empilements maximisent la densité parmi {\em tous} les empilements, {\em i.e.}, sans hypothèse de compacité.

\section{Notations}

Le grand disque est supposé de rayon $1$.
On note $r$ et $s$ les rayons du moyen et du petit disque, avec $0<s<r<1$.
On appelle {\em couronne} d'un disque dans un empilement compact saturé l'ensemble des disques auquel il est tangent.
Elle est dite {\em petite}, {\em moyenne} ou {\em grande} selon que le disque entouré est petit, moyen ou grand.
Un {\em codage} d'une couronne est un mot formé des rayons des disques qui la constituent, ordonnés de sorte à ce que l'un soit tangent au suivant et le dernier au premier.
Toute permutation circulaire ou renversement de ce mot correspond à la même couronne : on choisira généralement l'écriture minimale pour l'ordre lexicographique.
\'Etant donnés des disques de rayons $x$, $y$ et $z$ deux à deux tangents, on note $\widehat{xyz}$ l'angle non-orienté entre les segments reliant le centre du disque de rayon $y$ aux deux autres centres.
La figure~\ref{fig:notations} illustre cela.

\begin{figure}
  \centering
  \includegraphics[width=0.4\textwidth]{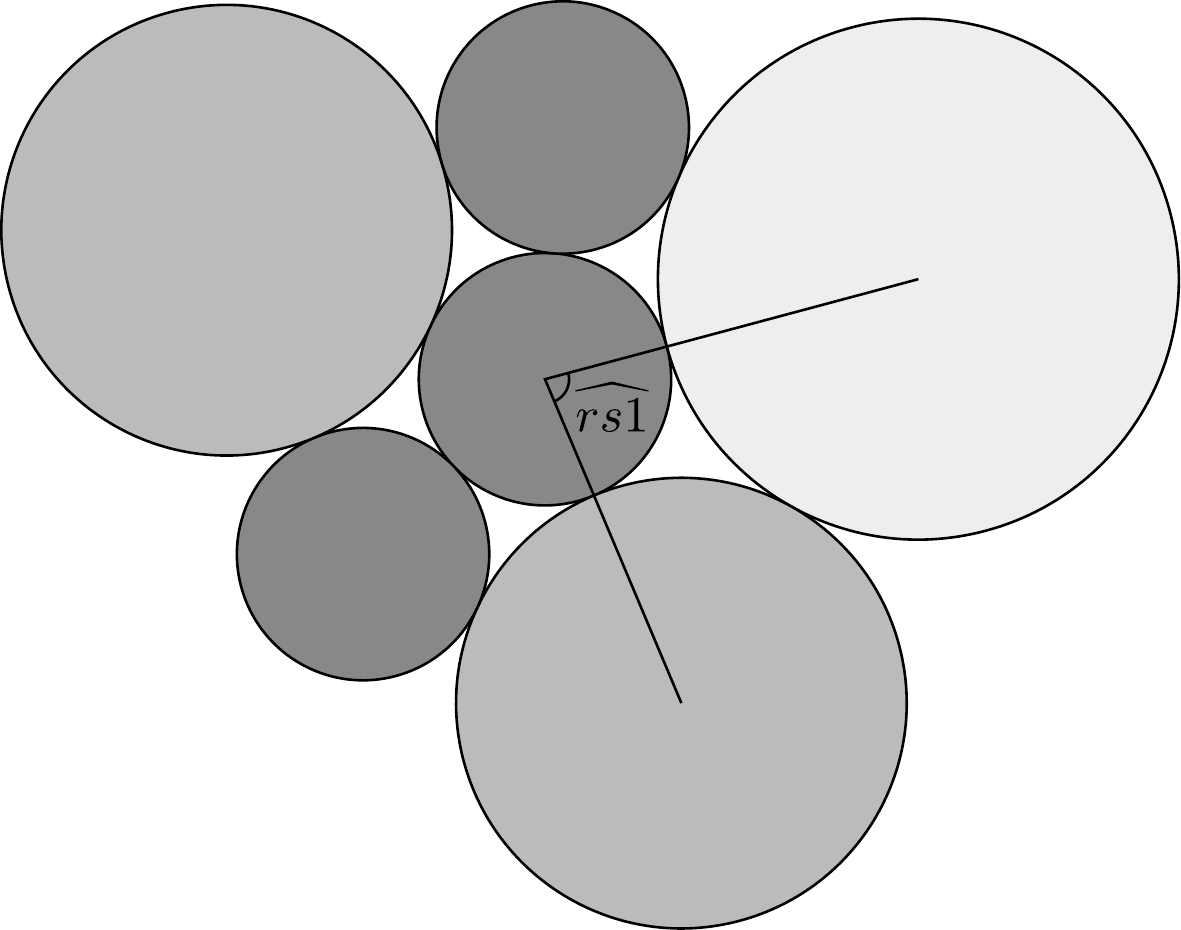}
  \caption{La petite couronne $1rsrs$ et l'angle $\widehat{rs1}$.}
  \label{fig:notations}
\end{figure}

\section{Petites couronnes}
\label{sec:petites}

Une petite couronne est formée d'au plus $6$ disques, et dans ce cas ce sont nécessairement tous des petits disques (c'est l'empilement hexagonal habituel).
Si elle comporte un disque moyen ou un grand, elle est donc formée d'au plus $5$ disques.
En particulier, il y a un nombre fini de petites couronnes différentes, quelque soient les valeurs de $r$ et $s$.
On va ici toutes les expliciter.\\

\noindent Soit $P_{\vec{k}}(r,s)$ la fonction définie pour $\vec{k}\in\mathbb{N}^6$ et $1>r>s>0$ par
$$
P_{\vec{k}}(r,s):=k_1\widehat{1s1}+k_2\widehat{1sr}+k_3\widehat{1ss}+k_4\widehat{rsr}+k_5\widehat{rss}+k_6\widehat{sss}.
$$
Les six angles impliqués sont ceux qui, dans un triangle rejoignant les centres de trois disques deux à deux tangents, peuvent apparaître au centre d'un petit disque.
Le vecteur $\vec{k}$ sera appelé {\em vecteur d'angles}.
Une petit couronne correspond alors une solution de l'équation
\begin{equation}\label{eq:petite_couronne}
P_{\vec{k}}(r,s)=2\pi.
\end{equation}
Chacun des six angles impliqués, vu comme une fonction de $r$ et $s$, est croissant en $r$ et décroissant en $s$.
La fonction $P_{\vec{k}}$ est donc majorée par sa limite en $(1,0)$ et minorée par sa plus petite limite sur la diagonale $r=s$.
On calcule
$$
\lim_{r\to 1\atop s\to 0}P_{\vec{k}}(r,s)=k_1\pi+k_2\pi+k_3\frac{\pi}{2}+k_4\pi+k_5\frac{\pi}{2}+k_6\frac{\pi}{3},
$$
$$
\inf_{r}\lim_{s\to r}P_{\vec{k}}(r,s)=\lim_{r\to 1\atop s\to 1}P_{\vec{k}}(r,s)=k_1\frac{\pi}{3}+k_2\frac{\pi}{3}+k_3\frac{\pi}{3}+k_4\frac{\pi}{3}+k_5\frac{\pi}{3}+k_6\frac{\pi}{3}.
$$
Comme tous ces angles, excepté $\widehat{sss}$, sont de plus {\em strictement} décroissants en $s$, $P_{\vec{k}}$ n'atteint pas son maximum sur l'ouvert $1>r>s>0$ et les bornes sont strictes, sauf si $k_1=k_2=k_3=k_4=k_5=0$, ce qui correspond à une petite couronne formée de $6$ petits disques.
Outre ce cas, l'équation $P_{\vec{k}}(r,s)=2\pi$ a donc une solution si et seulement si
\begin{equation}
k_1+k_2+k_3+k_4+k_5+k_6< 6<3k_1+3k_2+\frac{3}{2}k_3+3k_4+\frac{3}{2}k_5+k_6.
\end{equation}
Cette condition donne $383$ valeurs de $\vec{k}$ possibles.
Pour qu'une valeur de $\vec{k}$ corresponde réellement à une couronne, il faut aussi qu'il existe, dans le graphe représenté Fig.~\ref{fig:couronnable}, un cycle tel que $\vec{k}$ compte les passages de ce cycle dans chaque type d'arête.
Une boucle de ce graphe peut ne pas être empruntée par le cycle, mais si elle l'est il faut pouvoir y accéder, à moins que le cycle n'emprunte qu'elle.
Ceci se traduit par les trois conditions
\begin{eqnarray}
(k_1=0) \lor (k_2\neq 0 \lor k_3\neq 0) \lor (k_2=k_3=k_4=k_5=k_6=0)\\
(k_4=0) \lor (k_2\neq 0 \lor k_5\neq 0) \lor (k_1=k_2=k_3=k_5=k_6=0)\\
(k_6=0) \lor (k_3\neq 0 \lor k_5\neq 0) \lor (k_1=k_2=k_3=k_4=k_5=0)
\end{eqnarray}
Outre ces boucles, le cycle fait $k_0:=\min(k_2,k_3,k_5)$ tours entre les trois sommets, plus éventuellement quelques aller-retours qui empruntent un nombre pair de fois chaque arête par laquelle ils passent.
Ceci se traduit par les conditions
\begin{equation}
  k_2-k_0\in2\mathbb{N},\qquad
  k_3-k_0\in2\mathbb{N},\qquad
  k_5-k_0\in2\mathbb{N}.
\end{equation}

\begin{figure}
  \centering
  \includegraphics[width=0.3\textwidth]{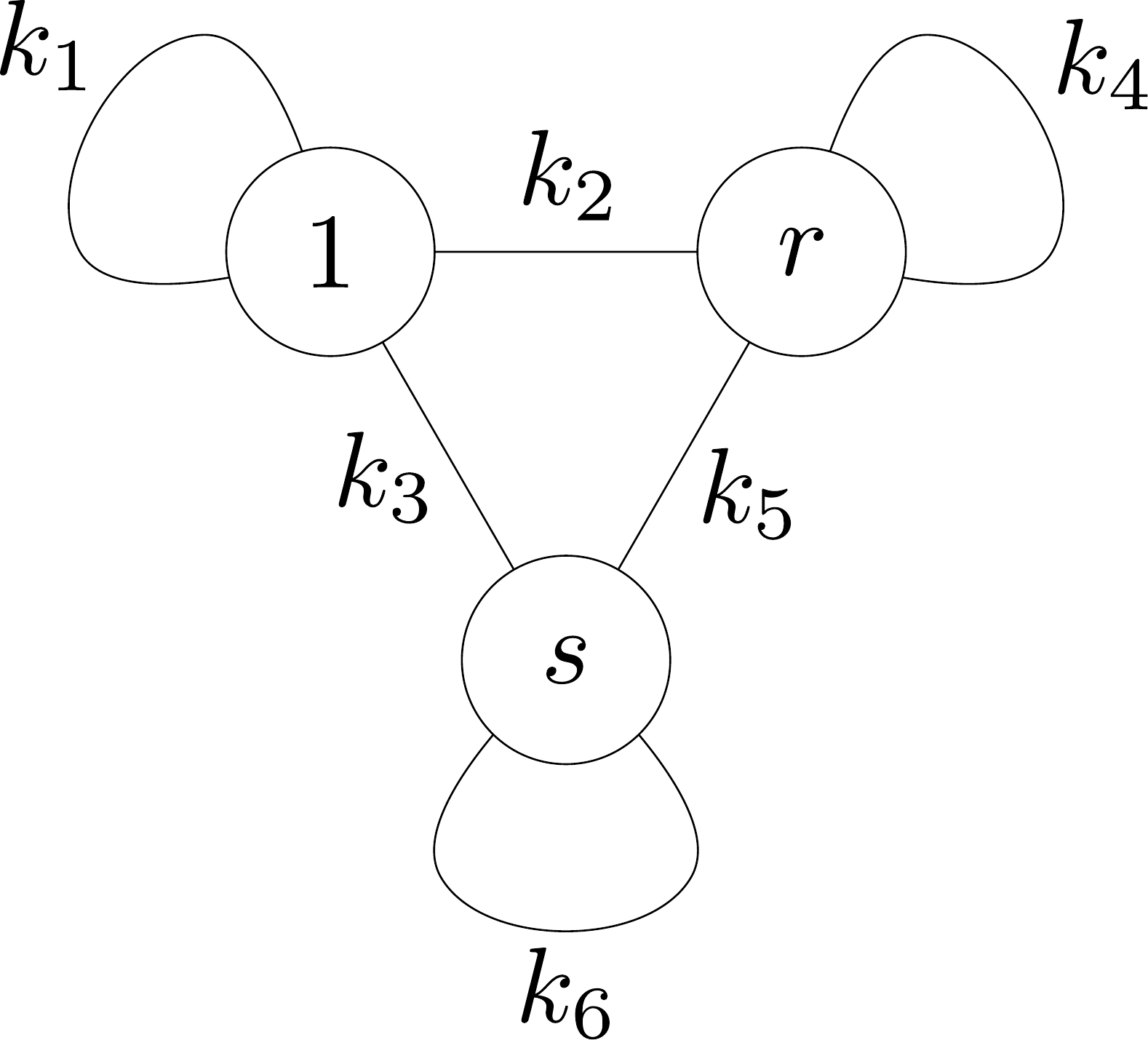}
  \caption{Les cycles de ce graphe codent les couronnes associées à $\vec{k}$.}
  \label{fig:couronnable}
\end{figure}

Toutes les conditions précédentes réduisent à $56$ le nombre de vecteurs d'angles  possibles.
La table~\ref{tab:petites_couronnes} donne un codage des couronnes correspondantes\footnote{Un vecteur d'angles pourrait correspondre à des couronnes différentes. Par exemple, $\vec{k}=(0,2,2,0,2,0)$ correspond à 1rsr1s, 1r1srs, 1rs1sr et 1rs1rs. Mais ce n'est jamais le cas ici.}.\\

\begin{table}[hbtp]
  \centering
  \begin{tabular}{cccccccccc}
    rrrrr & rrrrs & rrrss & rrsrs & rrrr & rrsss & rsrss & rrrs & rrr & rrss\\
    11111 & 1111s & 111ss & 11s1s & 1111 & 11sss & 1s1ss & 111s & 111 & 11ss\\
    1111r & 111rs & 11rss & 11srs & 111r & 1rsss & 1srss & 11rs & 11r & 1rss\\
    111rr & 11r1s & 1r1ss & 1rs1s & 11rr &  &  & 1r1s & 1rr & \\
    11r1r & 11rrs & 1rrss & 1rsrs & 1r1r &  &  & 1rrs &  & \\
    11rrr & 1r1rs & r1rss & rrs1s & 1rrr &  &  & r1rs &  & \\
    1r1rr & 1rr1s &  &  &  &  &  &  &  & \\
    1rrrr & 1rrrs &  &  &  &  &  &  &  & \\
    & r11rs &  &  &  &  &  &  &  & \\
    & r1rrs &  &  &  &  &  &  &  & \\
  \end{tabular}
  \caption{
    Les $55$ petites couronnes possibles, outre $ssssss$.
    Les couronnes de la première ligne n'ont pas de grand disque, celles de la deuxième n'ont pas de moyen disque.
    Les codages d'une même colonne ont tous la même image par l'application qui remplace chaque 1 par un r (utilisé dans le lemme~\ref{lem:ratio_minimal}).
  }
  \label{tab:petites_couronnes}
\end{table}

\section{Polynôme associé à une petite couronne}
\label{sec:equation}

En prenant le cosinus de chaque membre de l'équation \ref{eq:petite_couronne} puis en développant grâce aux formules d'addition et aux polynômes de Tchebychev, on obtient une équation polynomiale à coefficients entiers en les sinus et cosinus des six angles impliquées dans $P_{\vec{k}}(r,s)$.
La loi des cosinus permet d'exprimer les cosinus des angles $\widehat{1s1}$, $\widehat{1sr}$, $\widehat{1ss}$, $\widehat{rsr}$, $\widehat{rss}$ et $\widehat{sss}$ en fonction de $r$ et $s$ (dans cet ordre) :
$$
1-\frac{2}{(1+s)^2},\quad
1-\frac{2r}{(r+s)(1+s)},\quad
\frac{s}{1+s},\quad
1-\frac{2r^2}{(r+s)^2},\quad
\frac{s}{r+s},\quad
\frac{1}{2}.
$$
Le carré des sinus s'en déduit (dans le même ordre) :
$$
\frac{4s(s+2)}{(s + 1)^4},\quad
\frac{4rs(r + s + 1)}{(s + 1)^2(r + s)^2},\quad
\frac{2s + 1}{(s + 1)^2},\quad
\frac{4s(2r + s)r^2}{(r + s)^4},\quad
\frac{r(r + 2s)}{(r + s)^2},\quad
\frac{3}{4}.
$$
On exprime les sinus eux-mêmes avec des variables auxiliaires :
$$
\frac{2X_1}{(s + 1)^2},\quad
\frac{2X_2}{(s + 1)(r + s)},\quad
\frac{X_3}{s + 1},\quad
\frac{2rX_4}{(r + s)^2},\quad
\frac{X_5}{r + s},\quad
\frac{X_6}{2},
$$
où les carrés des $X_i$ valent respectivement
$$
s(s+2),\quad
rs(r+s+1),\quad
2s+1,\quad
s(2r+s),\quad
r(r + 2s),\quad
3.
$$
Ceci donne un système d'équations polynomiales à coefficients entiers en les variables $r$, $s$ et $X_1,\ldots,X_6$.
Quitte à remplacer $X_i^2$ par son expression en $r$ et $s$, on peut toujours supposer que $X_i$ n'apparaît jamais à une puissance $k\geq 2$.
On élimine ensuite successivement les $X_i$ qui restent en remarquant :
$$
AX_i+B=0~\Rightarrow~ A^2X_i^2-B^2=0.
$$
On obtient un polynôme entier bivarié dont les valeurs de $r$ et $s$ compatibles avec la petite couronne considérée sont racines.
Il peut y avoir d'autres racines (l'élimination des $X_i$ n'est pas bijective) qu'il faudra écarter {\em in fine} en vérifiant l'équation \ref{eq:petite_couronne}.
On peut aussi simplifier le polynôme en supprimant les multiplicités des facteurs, ainsi que les facteurs qui n'ont pas de racine $0<s<r<1$.\\

\begin{table}[hbtp]
  \centering
  \begin{tabular}{lll}
  11111 & $5s^4+20s^3+10s^2-20s+1$ & $0.701$ \\
  1111s & $s^4 - 10s^2 - 8s + 9$ & $0.637$ \\
  111ss & $s^8 - 8s^7 - 44s^6 - 232s^5 - 482s^4 - 24s^3 + 388s^2 - 120s + 9$ & $0.545$ \\
  11s1s & $8s^3 + 3s^2 - 2s - 1$ & $0.533$ \\
  1111 & $s^2 + 2s - 1$ & $0.414$ \\
  11sss & $9s^4 - 12s^3 - 26s^2 - 12s + 9$ & $0.386$ \\
  1s1ss & $s^4 - 28s^3 - 10s^2 + 4s + 1$ & $0.349$ \\
  111s & $2s^2 + 3s - 1$ & $0.280$ \\
  111 & $3s^2 + 6s - 1$ & $0.154$ \\
  11ss & $s^2 - 10s + 1$ & $0.101$ 
  \end{tabular}
  \caption{Couronne, polynôme minimal et valeur approchée de $s$.}
  \label{tab:petites_couronnes_deux_disques}
\end{table}

En procédant ainsi, on obtient un polynôme univarié en $s$ pour les $10$ petites couronnes qui n'ont pas de moyen disque (Tab.~\ref{tab:petites_couronnes_deux_disques}).
Chacune des $10$ petites couronnes sans grand disque donne le même polynôme que la couronne où $r$ a été remplacé par $1$, avec la variable $\tfrac{s}{r}$ au lieu de $s$.
Les $35$ autres petites couronnes donnent un polynôme explicite mais parfois assez complexe (Tab.~\ref{tab:degre_petites_couronnes}).
Par exemple, 11rs donne :
$$
r^2s^4 - 2r^2s^3 - 2rs^4 - 23r^2s^2 - 28rs^3 + s^4 - 24r^2s - 58rs^2 - 2s^3 + 16r^2 - 8rs + s^2.
$$

\begin{table}[hbtp]
  \centering
  \begin{tabular}{lr|lr|lr|lr|lr}
1r1r	&	$2$	&	11rr	&	$4$	&	1rrsr	&	$6$	&	1rsss	&	$8$	&	111rr	&	$12$	\\
1r1s	&	$2$	&	1rss	&	$4$	&	1111r	&	$7$	&	1srrs	&	$8$	&	11rrr	&	$12$	\\
1rsr	&	$2$	&	11r1s	&	$6$	&	11r1r	&	$7$	&	1srss	&	$8$	&	111rs	&	$18$	\\
111r	&	$3$	&	11rs	&	$6$	&	1rs1s	&	$7$	&	1r1rr	&	$10$	&	1rrrs	&	$18$	\\
11r	&	$3$	&	11rsr	&	$6$	&	11srs	&	$8$	&	1rrrr	&	$10$	&	11rss	&	$24$	\\
1rr	&	$3$	&	1rr1s	&	$6$	&	1r1ss	&	$8$	&	1rsrs	&	$10$	&	1rrss	&	$24$	\\
1rrr	&	$3$	&	1rrs	&	$6$	&	1rssr	&	$8$	&	1r1rs	&	$11$	&	11rrs	&	$28$	\\
  \end{tabular}
  \caption{Degré du polynôme en $r$ et $s$ imposé par chaque petite couronne.}
  \label{tab:degre_petites_couronnes}
\end{table}

\section{Moyennes couronnes et polynômes associés}
\label{sec:moyennes}

Un moyen disque peut être entouré d'un nombre arbitrairement grand de petits disques si ceux-ci sont suffisamment petits.
Il y a donc un nombre infini de moyennes couronnes.
Mais les empilements considérés contiennent toujours un petit disque, donc une petite couronne qui, elle, donne une contrainte sur la taille des disques :

\begin{lemma}
  \label{lem:ratio_minimal}
  Le ratio $\tfrac{s}{r}$ est uniformément minoré dans les empilements.
\end{lemma}

\begin{proof}
  Considérons un empilement.
  Par définition, il contient des disques de toutes les tailles.
  Il contient donc une petite couronne différente de ssssss.
  Remplacer tous les 1 par des r (sans changer les autres lettres) dans ce codage s'avère donner le codage d'une autre petite couronne (voir Table~\ref{tab:petites_couronnes}).
  De plus, dans cette nouvelle petite couronne, le ratio $\tfrac{s}{r}$ est plus petit que dans la couronne originale.
  En effet, les gros disques ayant été ``dégonflés'' en moyens, le périmètre de la couronne a diminué, donc la taille du disque qu'elle entoure aussi.
  Or il n'y a que $10$ ratios $\tfrac{s}{r}$ possibles pour une couronne de disques de tailles $s<r$ autour d'un disque de taille $s$ (Table~\ref{tab:petites_couronnes_deux_disques}).
\end{proof}

Soit $\alpha$ la minoration de $\tfrac{s}{r}$ donné par une petite couronne de l'empilement considéré.
On a alors un nombre fini de moyennes couronnes possibles.
Plus précisément, soit $M_{\vec{l}}(r,s)$ la fonction définie pour $\vec{l}\in\mathbb{N}^6$ et $1>r>s>\alpha r$ par
$$
M_{\vec{l}}(r,s):=l_1\widehat{1r1}+l_2\widehat{1rr}+l_3\widehat{1rs}+l_4\widehat{rrr}+l_5\widehat{rrs}+l_6\widehat{srs}.
$$
Les angles impliqués sont tous strictement décroissants en $r$, sauf $\widehat{rrr}$ (qui n'apparaît seul que dans la couronne rrrrrr qu'on ignore ici) et tous croissants en $s$.
La fonction $M_{\vec{l}}$ est donc strictement majorée par sa plus grande limite sur la diagonale $r=s$ et strictement minorée par sa plus petite limite sur la droite $\tfrac{s}{r}=\alpha$.
D'un côté on a
$$
\lim_{s\to r}M_{\vec{l}}(r,s)=l_1\widehat{1r1}+(l_2+l_3)\widehat{1rr}+(l_4+l_5+l_6)\frac{\pi}{3},
$$
d'où
$$
\sup_r\lim_{s\to r}M_{\vec{l}}(r,s)=l_1\pi+(l_2+l_3)\frac{\pi}{2}+(l_4+l_5+l_6)\frac{\pi}{3}.
$$
De l'autre côté, la loi des cosinus permet de montrer
$$
\lim_{\tfrac{s}{r}\to\alpha}\cos(\widehat{rrs})=\frac{1}{1+\alpha}
\quad\textrm{et}\quad
\lim_{\tfrac{s}{r}\to\alpha}\cos(\widehat{srs})=1-\frac{2\alpha^2}{(1+\alpha)^2}.
$$
Comme $\widehat{1rs}\geq \widehat{rrs}$, on en déduit
$$
\lim_{\tfrac{s}{r}\to\alpha}M_{\vec{l}}(r,s)\geq (l_1+l_2+l_4)\frac{\pi}{3}+(l_3+l_5)u_\alpha+l_6v_\alpha,
$$
où
$$
u_\alpha:=\arccos\left(\frac{1}{1+\alpha}\right)
\quad\textrm{et}\quad
v_\alpha:=\arccos\left(1-\frac{2\alpha^2}{(1+\alpha)^2}\right).
$$
Une moyenne couronne correspond à une solution de l'équation $M_{\vec{l}}=2\pi$.
L'existence d'une moyenne couronne impose donc
\begin{equation}
l_1+l_2+l_4+\frac{3}{\pi}(l_3u_\alpha+l_5u_\alpha+l_6v_\alpha)< 6< 3l_1+\frac{3}{2}l_2+\frac{3}{2}l_3+l_4+l_5+l_6.
\end{equation}
Comme une moyenne couronne a au plus $5$ disques moyens ou grands, on a aussi
\begin{equation}
  l_1+l_2+l_4+\frac{1}{2}(l_3+l_5)<6.
\end{equation}
Combinée avec des conditions de cycle similaires au cas des petites couronnes, on en déduit une majoration du nombre de moyennes couronnes possibles selon la minoration $\alpha$ de $\tfrac{s}{r}$ donnée par une petite couronne de l'empilement (Tab.~\ref{tab:nombre_moyennes_couronnes}).\\
\begin{table}[hbtp]
  \centering
  \begin{tabular}{c|c|c|c|c|c|c|c|c|c}
    $0.701$ & $0.637$ & $0.545$ & $0.533$ & $0.414$ & $0.386$ & $0.349$ & $0.280$ & $0.154$ & $0.101$\\
    $84$ & $94$ & $130$ & $143$ & $197$ & $241$ & $272$ & $386$ & $889$ & $1654$
  \end{tabular}
  \caption{Majoration du nombre de moyennes couronnes en fonction de $\alpha$.}
  \label{tab:nombre_moyennes_couronnes}
\end{table}

Pour associer ensuite un polynôme en $r$ et $s$ à chaque moyenne couronne, on procède comme pour les petites couronnes.
On développe le cosinus de chaque membre de l'équation $M_{\vec{l}}=2\pi$.
La loi des cosinus permet d'exprimer les cosinus des angles $\widehat{1r1}$, $\widehat{1rr}$, $\widehat{1rs}$, $\widehat{rrr}$, $\widehat{rrs}$ et $\widehat{srs}$ en fonction de $r$ et $s$ (dans cet ordre) :
$$
1-\frac{2}{(1+r)^2},\quad
\frac{r}{1+r},\quad
1-\frac{2s}{(r+s)(1+r)},\quad
\frac{1}{2},\quad
\frac{r}{r+s},\quad
1-\frac{2s^2}{(r+s)^2}.
$$
Le carré des sinus s'en déduit (dans le même ordre) :
$$
\frac{4r(r + 2)}{(r + 1)^4},\quad
\frac{2r + 1}{(r + 1)^2},\quad
\frac{4rs(r + s + 1)}{(r + 1)^2(r + s)^2},\quad
\frac{3}{4},\quad
\frac{s(2r + s)}{(r + s)^2},\quad
\frac{4r(r + 2s)s^2}{(r + s)^4}.
$$
On exprime les sinus eux-mêmes avec des variables auxiliaires :
$$
\frac{2X_7}{(r + 1)^2},\quad
\frac{X_8}{r + 1},\quad
\frac{2X_2}{(r + 1)(r + s)},\quad
\frac{X_6}{2},\quad
\frac{X_4}{r + s},\quad
\frac{2sX_5}{(r + s)^2},
$$
où $X_1,\ldots,X_6$ ont déjà été définies et $X_7$ et $X_8$ ont respectivement pour carré
$$
r(r+2),\quad
2r+1.
$$
L'élimination des $X_i$ donne un polynôme entier bivarié associé à la moyenne couronne considérée.
Comme il y a en généralement plus de disques dans une moyenne couronne que dans une petite (jusqu'à $33$ petits disques), ces polynômes sont plus compliqués que ceux associés aux petites couronnes\footnote{Le calcul de tous les $1654$ polynômes prend 2h 21min sur notre ordinateur et crée un fichier de $36$Mo. Il montre que le degré moyen est $57$, avec un maximum à $416$ pour $\textrm{11rrssssssssssss}=\textrm{11rrs}^{12}$.}.

\section{Rayons et couronnes}
\label{sec:calcul}

Pour déterminer les rayons possibles, il suffit en théorie de considérer chaque couple formée d'une petite couronne et d'une moyenne couronne, de calculer les deux polynômes associés et de trouver les racines de ce système d'équations polynomiales vérifiant $0<s<r<1$.
La résolution d'un tel système est cependant en pratique souvent délicate.\\

Par exemple, la petite couronne 111rr et la moyenne couronne 111rrs donnent deux polynômes de degré respectifs $12$ et $38$ (avec des coefficients valant jusqu'à $10^{14}$ pour le second), et leur résolution exacte avec le logiciel SageMath \cite{sage} nécessite une heure et $21$ minutes sur notre ordinateur de bureau.\\

La résolution d'un système polynomial par un logiciel de calcul formel utilise généralement les {\em bases de Gröbner} pour calculer la variété de l'idéal engendré par les polynômes (quand elle est de dimension $0$).
Mais le calcul de ces bases peut être très coûteux, même pour seulement deux équations quand celles-ci sont de degré élevé.
Les $X_i$ ont d'ailleurs été éliminés des équations (Parties \ref{sec:equation} et \ref{sec:moyennes}) sans recourir aux bases de Gröbner mais par une méthode plus ``manuelle'' (multiplier $AX_i+B$ par $AX_i-B$) qui utilise le fait que $X_i^2$ s'exprime facilement en fonction de $r$ et $s$ (ce que le logiciel de calcul formel, destiné à traiter le cas général, n'utilise sans doute que partiellement).
On s'inspire ici de \cite{sage2} (page 201).\\

Rappelons que le résultant de deux polynômes univariés est un scalaire qui s'annule si et seulement si les deux polynômes ont une racine commune.
Si $P$ et $Q$ sont deux polynômes de $\mathbb{Z}[r,s]$, on peut les voir comme des polynômes en $r$ à coefficients dans $\mathbb{Z}[s]$ : leur résultant est alors un polynôme $\textrm{Res}_r(s)$ qui s'annule en $s_0$ si et seulement si $P(r,s_0)$ et $Q(r,s_0)$ ont une racine $r$ commune.
Symétriquement, échanger $r$ et $s$ donne un polynôme $\textrm{Res}_s(r)$ qui s'annule en $r_0$ si et seulement si $P(r_0,s)$ et $Q(r_0,s)$ ont une racine $s$ commune.
Les couples annulant $P$ et $Q$ sont alors dans le produit cartésien des racines de ces résultants.\\

Le calcul des résultants est rapide : ce sont les déterminants des matrices de Sylvester des polynômes.
Le calcul exact de leurs racines aussi\footnote{Pour un logiciel de calcul formel, {\em calculer} une racine d'un polynôme revient juste à déterminer un intervalle qui ne contient que cette racine, à partir de quoi il peut fournir à l'utilisateur une approximation à la précision voulue de cette racine (via, par exemple, la méthode de Newton).}.
Le produit cartésien nous donne alors de nombreux couples $(r,s)$ {\em candidats} parmi lesquels il faut trouver les solutions des équations d'angles associées aux couronnes.\\

On filtre d'abord ces candidats en trois passes en utilisant l'{\em arithmétique d'intervalles}.
Chaque passe ne garde que les couples $(r,s)$ tels que :
\begin{enumerate}
\item
  $0<\overline{s}$, $\underline{s}<\overline{r}$ et $\underline{r}<1$, où $[\underline{x},\overline{x}]$ est l'intervalle représentant $x$ ;
\item
  $0$ est dans les intervalles représentant $2\pi-P_{\vec{k}}$ et $2\pi-M_{\vec{l}}$, où $P_{\vec{k}}$ et $M_{\vec{l}}$ correspondent à la petite et à la moyenne couronnes ayant donné $(r,s)$ ;
\item
  il existe $\vec{n}\in\mathbb{N}^6$ tel que $0$ soit dans l'intervalle représentant $2\pi-G_{\vec{n}}$, où
  $$
  G_{\vec{n}}(r,s):=n_1\widehat{111}+n_2\widehat{11r}+n_3\widehat{11s}+n_4\widehat{r1r}+n_5\widehat{r1s}+n_6\widehat{s1s}.
  $$
\end{enumerate}
En d'autres termes, la première passe élimine {\em des} couples hors du domaine, la seconde {\em des} couples qui n'admettent pas de petite ou moyenne couronne, la dernière {\em des} couples qui n'admettent pas de grande couronne\footnote{Les valeurs de $r$ et $s$ permettent de borner la norme des vecteurs $\vec{n}$ à rechercher}.
Plus la précision sur les intervalles est grande et plus le filtrage est efficace, mais plus il est lent.
Une quatrième et dernière passe vérifie alors que les candidats restants sont {\em réellement} des racines des équations de couronnes.
On procède pour chaque couronne comme pour calculer son polynôme associé (Partie~\ref{sec:equation}), avec deux ajouts :
\begin{enumerate}
\item
  Lors de chaque multiplication par $AX_i-B$, on vérifie que ce terme n'est pas nul.
  On essaie d'abord en arithmétique d'intervalle (plus rapide), et seulement à défaut\footnote{Cas qui ne s'est présenté que pour le couple petite/moyenne couronnes 1r1r/11r1s.} avec les valeurs exactes de $r$ et $s$.
\item
  Lorsque les $X_i$ ont tous été éliminés, on vérifie l'équation obtenue avec les valeurs exactes de $r$ et $s$.\\
\end{enumerate}

On peut alors, pour chaque couple $(r,s)$ retenu, procéder comme dans les troisième et dernière passes pour trouver {\em toutes} les couronnes (petites, moyennes et grandes) compatibles.
Ceci afin de pouvoir ensuite déterminer (de façon combinatoire) les empilements possibles.\\

Reprenons l'exemple de la petite couronne 111rr et de la moyenne 111rrs.
Les résultants sont deux polynômes de degré $336$, chacun ayant $85$ racines réelles.
Il y a donc $7225$ couples candidats.
La première passe en garde $45$.
La deuxième n'en laisse qu'un seul.
La troisième l'élimine.
Aucune vérification finale n'est donc nécessaire ! 
Le tout en environ $15$ secondes sur un ordinateur de bureau avec $53$ bits de précision pour les intervalles.\\

Certains cas restent néanmoins impraticable.
Par exemple, la petite couronne 11rrs et la moyenne couronne $\textrm{11rrssssssssssss}=\textrm{11rrs}^{12}$ donnent un polynôme de degré $28$ et un de degré $416$ avec des coefficients ayant jusqu'à $155$ chiffres.
Rien qu'obtenir ce polynôme de degré $416$ demande $6$ minutes de calcul, et le calcul du résultant excède la capacité mémoire de notre ordinateur.\\

Surtout, si l'on croise le nombre de petites couronnes selon la minoration du ratio $\tfrac{s}{r}$ associée (Tab.~\ref{tab:petites_couronnes}) et le nombre de moyennes couronnes pour cette minoration (Tab.~\ref{tab:nombre_moyennes_couronnes}), on obtient un total de $16805$ couples de couronnes !
L'approche purement calculatoire précédente semble donc vouée à l'échec.
Aussi l'avons nous complétée par une approche combinatoire.
Le principe est d'éliminer autant de cas que possible en montrant que la combinatoire des petites et moyennes couronnes exclue la possibilité d'un empilement, indépendamment des valeurs de $r$ et $s$ (et donc des grandes couronnes).
Nous avons pour cela partitionné les empilements en trois types, successivement étudiés dans les trois parties suivantes.

\section{Deux phases}

Un empilement est dit admettre {\em deux phases} s'il n'y a pas de contact entre petits et moyens disques.
Les grands disques font ``tampon'' entre les petits et les moyens.
Une petite couronne ne contenant pas de moyen disque, elle donne une des $10$ équations en $s$ données Tab.~\ref{tab:petites_couronnes_deux_disques}.
De même pour une moyenne couronne, qui ne contient pas de petit disque, en remplaçant $s$ par $r$ dans les équations Tab.~\ref{tab:petites_couronnes_deux_disques}.
Comme $s<r$, il y a $C_{10}^2=45$ couples $(r,s)$ candidats.\\

Parmi ces $45$ candidats, $18$ admettent un empilement périodique (voir Annexe \ref{sec:empilements}).
Ce sont exactement ceux qui ont une grande couronne qui contient à la fois un petit et un moyen disque et permet ainsi de lier les deux phases.
Montrons que l'absence d'une telle couronne dans les $27$ cas restants interdit l'existence d'un empilement (qui a, par définition, des disques des trois tailles).

\begin{lemma}
  \label{lem:biphase1}
  Si aucune grande couronne d'un candidat à deux phases
  \begin{enumerate}
  \item ne contient trois grands disques consécutifs et un moyen disque ;
  \item ne contient à la fois un petit et un moyen disque.
  \end{enumerate}
  Alors aucun empilement n'est possible.
\end{lemma}

\begin{proof}
  Supposons qu'un tel empilement existe et obtenons une contradiction.
  Appelons disque $1_r$ un grand disque en contact avec au moins un moyen disque.
  Considérons un grand disque $D$ dans la couronne d'un disque $1_r$.
  Comme cette couronne contient un moyen disque (par définition d'un disque $1_r$), la première hypothèse assure que $D$ a un voisin qui n'est pas un grand disque.
  La seconde hypothèse assure que ce voisin est forcément moyen, et donc que $D$ est lui-même un disque $1_r$.
  Ainsi, les seuls voisins d'un disque $1_r$ sont des disques moyens ou $1_r$.
  De même, les seuls voisins d'un moyen disque sont moyens ou $1_r$ (il ne peuvent pas être petits par définition d'un empilement à deux phases, et s'ils sont grands ils sont $1_r$ par définition d'un disque $1_r$).
  Considérons maintenant un moyen disque de l'empilement.
  Ses voisins sont des disques moyens ou $1_r$.
  Les voisins de ses voisins aussi d'après ce qu'on vient de voir, et ainsi de suite.
  L'empilement ne contient pas de petit disque : contradiction.
\end{proof}

\begin{figure}[hbtp]
  \includegraphics[width=\textwidth]{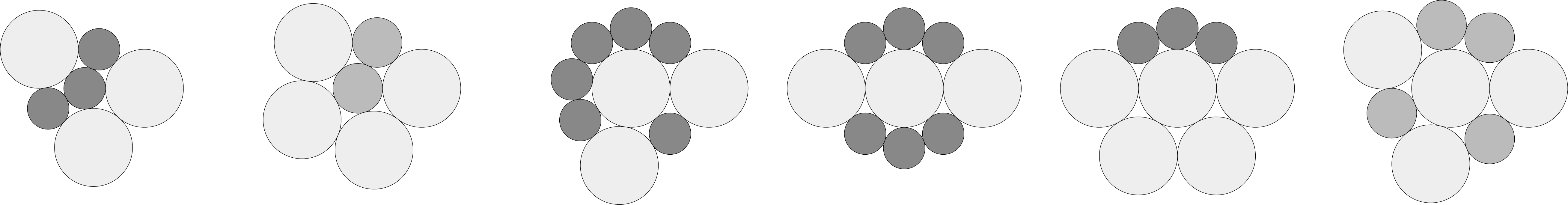}
  \caption{Un des $24$ candidats à deux phases, représenté par ses couronnes, qui n'admet pas d'empilement d'après le lemme~\ref{lem:biphase1}.}
  \label{fig:biphase_impossible_1}
\end{figure}

Le lemme~\ref{lem:biphase1} reste valable en échangeant partout ``moyen'' et ``petit'' (dans l'énoncé et dans la preuve).
Il élimine $24$ des $27$ candidats restants (Fig.~\ref{fig:biphase_impossible_1}).
Il en reste trois (Fig.~\ref{fig:biphase_impossible_2}), éliminés en allant un cran plus loin que dans le lemme~\ref{lem:biphase1} :\\

\begin{figure}[hbtp]
  \includegraphics[width=\textwidth]{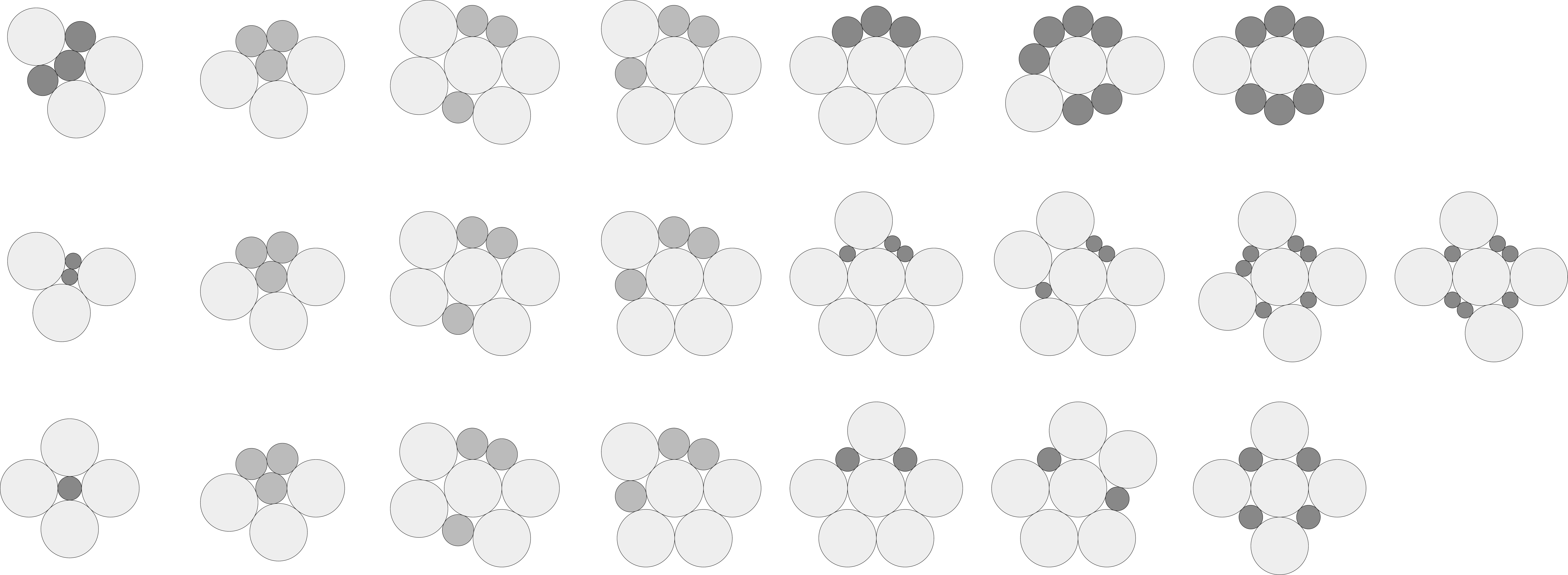}
  \caption{Trois candidats à deux phases, représentés ligne par ligne par leurs couronnes, qui n'admettent pas d'empilement d'après le lemme~\ref{lem:biphase2}.}
  \label{fig:biphase_impossible_2}
\end{figure}

\begin{lemma}
  \label{lem:biphase2}
  Aucun des candidats représentés Figure~\ref{fig:biphase_impossible_2} ne permet d'empilement.
\end{lemma}

\begin{proof}
  Supposons qu'un empilement existe.
  Considérons la seule grande couronne qui contienne un moyen disque et trois grands consécutifs (en quatrième position Fig.~\ref{fig:biphase_impossible_2}).
  Si elle n'apparaît pas dans l'empilement, alors le lemme~\ref{lem:biphase1} s'applique et donne une contradiction.
  Sinon, soit {\bf a} le grand disque de cette couronne qui a deux grands voisins (Fig.~\ref{fig:biphase_impossible_3}, premier dessin).
  Supposons qu'il n'ait aucun moyen voisin.
  Ses deux grands voisins dans la couronne lui imposent deux grands voisins {\bf b} et {\bf c} (Fig.~\ref{fig:biphase_impossible_3}, deuxième dessin).
  Comme ces mêmes voisins ayant au plus trois grands disques consécutifs dans leur couronne, ils imposent à {\bf b} et {\bf c} des moyens voisins {\bf d}  et {\bf e} (Fig.~\ref{fig:biphase_impossible_3}, troisième dessin).
  Par ailleurs, {\bf a} ayant maintenant $5$ grands voisins, il a forcément un sixième et dernier grand voisin {\bf f} (Fig.~\ref{fig:biphase_impossible_3}, troisième dessin).
  Enfin, {\bf b} et {\bf c} ayant eux aussi au plus trois grands disques consécutifs dans leur couronne, ils imposent à {\bf f} des moyens voisins {\bf g} et {\bf h} (Fig.~\ref{fig:biphase_impossible_3}, dernier dessin).
  Ainsi, les seuls grands disques sans moyen voisin de l'empilement sont encerclés par des grands disques avec un moyen voisin.
  L'empilement ne contient donc pas de petit disque : contradiction.
\end{proof}

\begin{figure}[hbtp]
  \centering
  \includegraphics[width=1\textwidth]{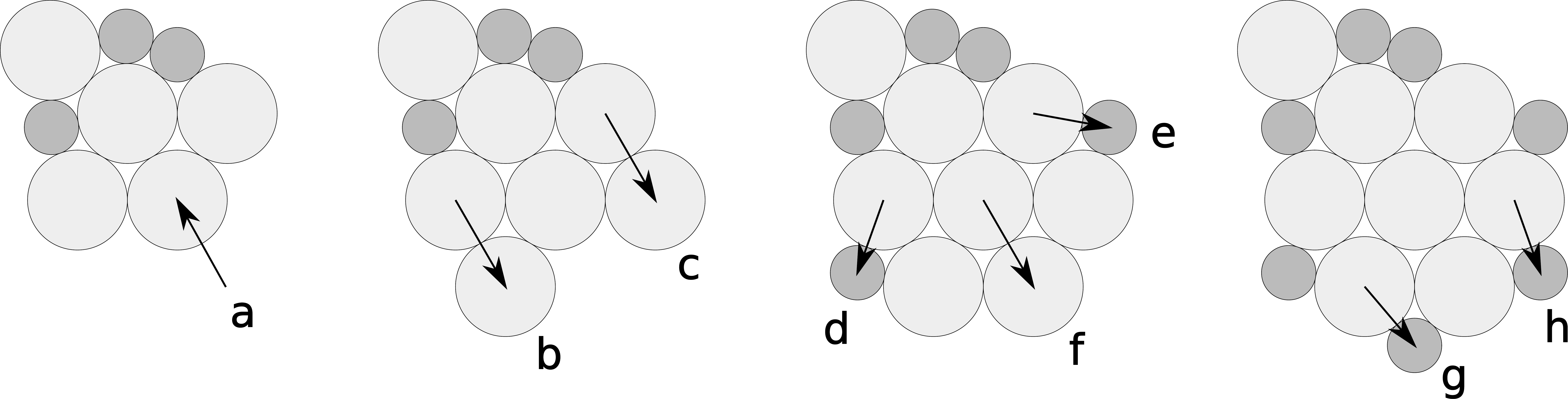}
  \caption{Dans un empilement avec la couronne de gauche, un grand disque touche un moyen ou est encerclé par $6$ grands disques qui touchent un moyen.}
  \label{fig:biphase_impossible_3}
\end{figure}

\section{Deux petites couronnes}

On cherche ici les empilements où peuvent apparaître (au moins) deux petites couronnes différentes (autres que ssssss).
L'intérêt est double.
D'une part, comme il n'y a que $55$ petites couronnes, il y a au plus $C_{55}^2=1485$ couples à considérer (au lieu de $16688$).
D'autre part, les petites couronnes correspondent à des équations généralement beaucoup plus simples que les moyennes.\\

Si les deux petites couronnes ne contiennent que des petits et grands disques, alors chacune caractérise une valeur différente de $s$.
On élimine donc ces paires, tout comme celles qui ne contiennent que des petits et moyens disques (valeur différente de $\tfrac{s}{r}$).
Il reste $1395$ systèmes de deux équations en $r$ et $s$.\\

On calcule les résultants et leurs racines comme détaillé Partie~\ref{sec:calcul}.
Il y a trois paires pour lesquelles le calcul des racines par SageMath est problématique : l'exécution s'arrête pour rrrs/11rss et rrss/1111r en renvoyant respectivement \verb+AssertionError+ et \verb+RuntimeError('maximum recursion depth exceeded')+, quant à 1rrss/1111r, nous avons interrompu l'exécution au bout de $4$ heures.
Dans ces trois cas, les racines sont (facilement) calculées en arithmétique d'intervalles : on obtient un ensemble $\mathcal{C}_{\textrm{err}}$ de $2276$ couples $(r,s)$ candidats.
Les $1392$ autres cas ne posent pas de problème et donnent un ensemble $\mathcal{C}$ de $185813$ couples $(r,s)$ candidats.\\

La première passe, sur le domaine, réduit $\mathcal{C}_{\textrm{err}}$ à $179$ couples et $\mathcal{C}$ à $1394$ couples.
La deuxième passe est ici effectuée sur les équations d'angles des deux petites couronnes, et non sur celles d'une petite et d'une moyenne couronne comme dans le cas général exposé Partie~\ref{sec:calcul}.
Elle réduit $\mathcal{C}_{\textrm{err}}$ à $0$ couple, ce qui règle le problème des trois erreurs ci-dessus, et $\mathcal{C}$ à $312$ couples.
La troisième passe réduit  $\mathcal{C}$ à $37$ couples, que la dernière passe valide tous (avec toutes leurs couronnes).
Le tout (résultants, racines et les quatre passes) en moins de $2$ minutes sur notre ordinateur de bureau.\\

Ces $37$ couples peuvent donc former des petites, moyennes et grandes couronnes, mais il reste à vérifier qu'ils admettent un empilement.
Ils s'avèrent tous partager une propriété bien particulière : il n'y a jamais de moyen disque dans leurs grandes couronnes.
Ceci permet de montrer :

\begin{lemma}
  \label{lem:2pc}
  Si aucune petite couronne d'un candidat à deux petites couronnes
  \begin{enumerate}
  \item ne contient trois petits disques consécutifs et un grand disque ;
  \item ne contient à la fois un moyen et un grand disque.
  \end{enumerate}
  Alors aucun empilement n'est possible.
\end{lemma}

\begin{proof}
  Supposons qu'un tel empilement existe et obtenons une contradiction.
  Appelons disque $s_1$ un petit disque en contact avec au moins un grand disque.
  Considérons un petit disque $d$ dans la couronne d'un disque $s_1$.
  Comme cette couronne contient un grand disque (par définition d'un disque $s_1$), la première hypothèse assure que $d$ a un voisin qui n'est pas un petit disque.
  La seconde hypothèse assure que ce voisin est forcément grand, et donc que $d$ est lui-même un disque $s_1$.
  Ainsi, les seuls voisins d'un disque $s_1$ sont des disques grands ou $s_1$.
  De même, les seuls voisins d'un grand disque sont grands ou $s_1$ (il ne peuvent pas être moyen car les grandes couronnes des $37$ candidats ne contiennent jamais de moyen disque, et s'ils sont petits ils sont $s_1$ par définition d'un disque $s_1$).
  Considérons maintenant un grand disque de l'empilement.
  Ses voisins sont des disques grands ou $s_1$.
  Les voisins de ses voisins de même, et ainsi de suite.
  L'empilement ne contient donc pas de moyen disque : contradiction.
\end{proof}

\begin{figure}[hbtp]
  \centering
  \includegraphics[width=\textwidth]{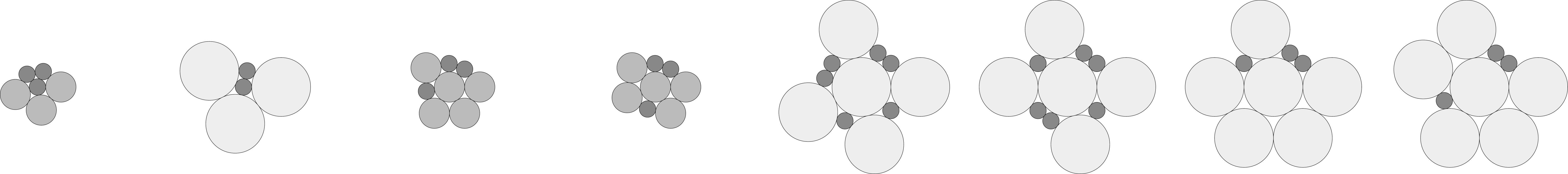}
  \caption{Un candidat à deux petites couronnes qui n'a pas d'empilement.}
  \label{fig:2pc_impossible}
\end{figure}

\begin{figure}[hbtp]
  \centering
  \includegraphics[width=0.83\textwidth]{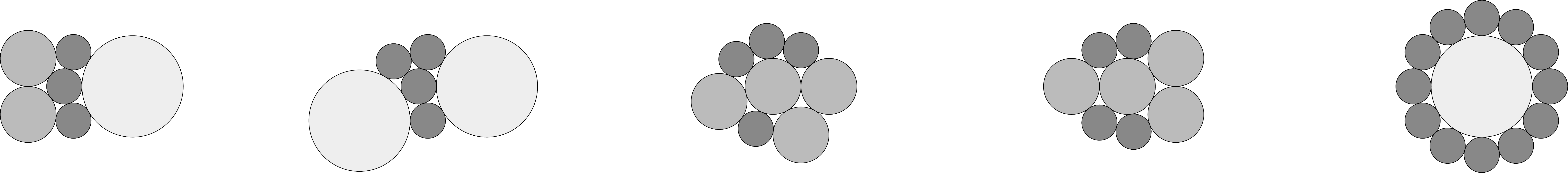}
  \caption{Le seul candidat à deux petites couronnes admettant un empilement.}
  \label{fig:2pc_possible}
\end{figure}

Le lemme~\ref{lem:2pc} reste valable en échangeant partout ``moyen'' et ``grand'' (dans l'énoncé et dans la preuve).
Il élimine tous les candidats (Fig.~\ref{fig:2pc_impossible}), sauf le cas des petites couronnes 1srrs et 1s1ss (Fig.~\ref{fig:2pc_possible}).
Ce cas admet en effet un empilement (périodique, voir Appendix~\ref{sec:empilements}), bien qu'une seule petite couronne soit utilisable :

\begin{proposition}
  Pour $r$ et $s$ compatibles avec les deux petites couronnes 1srrs et 1s1ss, il n'existe pas d'empilement contenant la petite couronne 1s1ss.
\end{proposition}

\begin{proof}
  Les deux petites couronnes caractérisent $r$ et $s$, donc toutes les couronnes.
  Il n'y a que ces deux petites couronnes, les moyennes couronnes rrsssrs et rrssrss (même vecteur d'angles), et la grande couronne $\textrm{s}^{12}$.
  Supposons qu'un empilement contienne la petite couronne 1s1ss et considérons le centre de cette couronne.
  On montre par induction sur $k$ que les disques à distance $k$ de ce centre sont soit des grands disques, soit des petits disques qui ont un facteur 1ss, 1s1 ou s1s dans leur couronne.
  C'est vrai pour $k=0$.
  Supposons que ce soit vrai pour $k>0$ et considérons un disque à distance $k+1$.
  Il est dans une couronne d'un disque à distance $k$, donc soit une grande couronne, soit une petite couronne qui, par hypothèse d'induction, contient 1ss, 1s1 ou s1s : ça ne peut être que 1s1ss.
  C'est donc un disque petit ou grand, et comme tout petit disque dans la grande couronne $\textrm{s}^{12}$ ou la petite couronne 1s1ss a un facteur 1ss, 1s1 ou s1s dans sa couronne, l'hypothèse d'induction est vérifiée pour $k+1$.
  On en déduit qu'il ne peut y avoir de moyen disque dans cet empilement, ce qui contredit la définition d'un empilement.
\end{proof}

\section{Une petite couronne}

Le dernier type d'empilement sont ceux qui n'ont qu'une phase et qu'une petite couronne (autre que ssssss).
Chacun contient donc (au moins) une moyenne couronne avec un petit disque (sinon il y a deux phases), et tous les petits disques de cette moyenne couronne ont la même couronne.
Aussi supposera-t-on toujours que, dans les couples petite/moyenne couronnes utilisés pour calculer $(r,s)$, la moyenne couronne contient un petit disque.
L'intérêt est que cela impose une contrainte combinatoire simple mais forte sur les couples considérés.\\

Considérons, par exemple, le cas de la petite couronne 11rrs et de la moyenne couronne $\textrm{11rrs}^{12}$ (où le simple calcul du résultant posait problème Partie~\ref{sec:calcul}).
Si les centres s et r de ces deux couronnes se touchent, alors il existe un facteur xry de la petite\footnote{C'est-à-dire un facteur d'une des codages de la couronne.} et un facteur xsy de la moyenne, où x et y sont les deux disques qui touchent à la fois s et r.
Ici, le facteur sss de $\textrm{11rrs}^{12}$ impose srs dans 11rrs : ce couple peut être éliminé sans se lancer dans les calculs décrits Partie~\ref{sec:calcul} !\\

Les couples considérés sont cependant des solutions (potentielles) des é\-qua\-tions $P_{\vec{k}}=2\pi$ et $M_{\vec{l}}=2\pi$, c'est-à-dire des vecteurs d'angles et non leur codage par un mot sur $\{1,r,s\}$.
Les vecteurs d'angles des petites couronnes n'admettent jamais qu'un seul codage, mais ceux des moyennes en ont généralement plusieurs.
Dans l'exemple précédent, $\textrm{11rrs}^{12}$ est l'unique codage de $(1,1,1,1,1,11)$, mais le vecteur $(0, 0, 4, 0, 6, 10)$, par exemple, admet $1022$ codages différents.
Il faut vérifier qu'au moins un des codages vérifie cette contrainte.\\

Formellement, on dit qu'un petit vecteur d'angles $\vec{k}$ {\em couvre} un moyen vecteur d'angles $\vec{l}$ s'il existe un codage de $\vec{l}$ tel que pour tout facteur xsy de ce codage, le codage de $\vec{k}$ contient $xry$.
On dit aussi que $\vec{k}$ {\em précouvre} $\vec{l}$ s'il existe un codage de $\vec{l}$ tel que pour tout facteur xs de ce codage, le codage de $\vec{k}$ contient $xr$.
Cette condition est plus faible, mais elle se vérifie directement sur les vecteurs d'angles (qui comptent justement les facteurs de taille $2$), donc plus rapidement.\\

Considérons les $16805$ couples petit/moyen vecteurs d'angles candidats.
Ne garder que ceux où la la petite couronne contient un moyen disque et la moyenne un petit réduit à $12265$ couples.
La vérification de précouverture réduit à $2889$ couples en quelques millisecondes, puis celles de couverture à $803$ couples en moins de $2$ minutes\footnote{Mais en plus de $15$ minutes si on ne vérifie pas la précouverture avant.}.
Ces $803$ couples correspondent à $192$ moyennes couronnes différentes, au lieu des $1654$ initiales, et leurs polynômes sont de plus souvent parmi les plus simples\footnote{Le calcul de ces $192$ polynômes prend 1min 51s sur notre ordinateur et crée un fichier de $256$Ko. Il montre que le degré moyen est $14$, avec un maximum à $80$ pour 11rrsrss.}.\\

Le calcul des résultants donne un ensemble $\mathcal{C}$ de $469808$ couples $(r,s)$ candidats en moins de $4$ minutes.
Ici encore, il y a quelques cas problématiques : $8$ erreurs d'exécution et $15$ couples de couronnes qui ne caractérisent pas un nombre fini de valeurs de $r$ et $s$.
Les $8$ erreurs sont traitées en arithmétique d'intervalles et donnent un ensemble $\mathcal{C}_{\textrm{err}}$ de $43448$ couples $(r,s)$ candidats (environ $1$ minutes).
Parmi les $15$ cas qui ne caractérisent pas $r$ et $s$ ({\em i.e.}, le système d'équations associées aux couronnes n'est pas de dimension $0$), $13$ sont formés de petite et moyenne couronnes qui n'ont pas de grand disque.
Ces seules couronnes ne permettent pas un empilement avec trois tailles de disques : il en faut donc une autre, {\em i.e.}, on retombe dans un des autres cas.
Les $2$ cas restants sont 1rr/1r1srs et 11r/111s1s.
Le premier est impossible car il correspond à une moyenne couronne 1r1r avec des s ajoutés dans les interstices, or une telle moyenne couronne est impossible.
Le second est en fait un sous cas de 1rr/1s1s1s1s, où certains trous entre deux grands disques et un moyen ne sont pas remplis par un petit disque.\\

La première passe, sur le domaine, réduit $\mathcal{C}$ à $57017$ couples et $\mathcal{C}_{\textrm{err}}$ à $5875$ couples.
La deuxième passe, sur les petites et moyennes couronnes, réduit $\mathcal{C}$ à $601$ couples et $\mathcal{C}_{\textrm{err}}$ à $6$ couples.
La troisième passe, sur les grandes couronnes, réduit  $\mathcal{C}$ à $175$ couples et $\mathcal{C}_{\textrm{err}}$ à $1$ couple.
Le tout en $4$ minutes environ.
Le dernier couple de $\mathcal{C}_{\textrm{err}}$ qui résiste (\verb+AssertionError+) correspond à la paire petite/moyenne couronnes 1rr1s/11rrs.
Résoudre les deux équations bivariées associées (via les bases de Gröbner) donne $29$ couples en 45 minutes\footnote{Ce qui montre l'intérêt d'avoir utilisé le résultant pour les $802$ autres couples !}, réduit à $1$ par les trois passes.
Il reste donc $176$ candidats, que la dernière passe valide tous (et leurs couronnes) en moins de $4$ minutes. 
Il reste à déterminer quand un empilement est possible.

\begin{lemma}
  \label{lem:1pc_impossible_1}
  Si un empilement contient une petite couronne 1rss, 11rss, 1rrss ou 1srss, alors il contient une seconde petite couronne autre que ssssss.
\end{lemma}

\begin{proof}
  La preuve ne repose pas sur les valeurs de $r$ ou $s$.
  Les quatre cas, similaires, sont illustrés Fig.~\ref{fig:1pc_impossible_1}.
  On considère un s (légèrement grisé) et sa petite couronne.
  On considère un des s de cette couronne (moyennement grisé) : il n'y a qu'une seule façon de disposer sa couronne.
  Un des s (fortement grisé) a alors dans sa couronne un facteur (1ss1 dans le premier cas, rssr dans les trois autres) qui n'apparaît ni dans une de ces quatre petites couronnes ni dans ssssss.
\end{proof}

\begin{figure}[hbtp]
  \centering
  \includegraphics[width=0.9\textwidth]{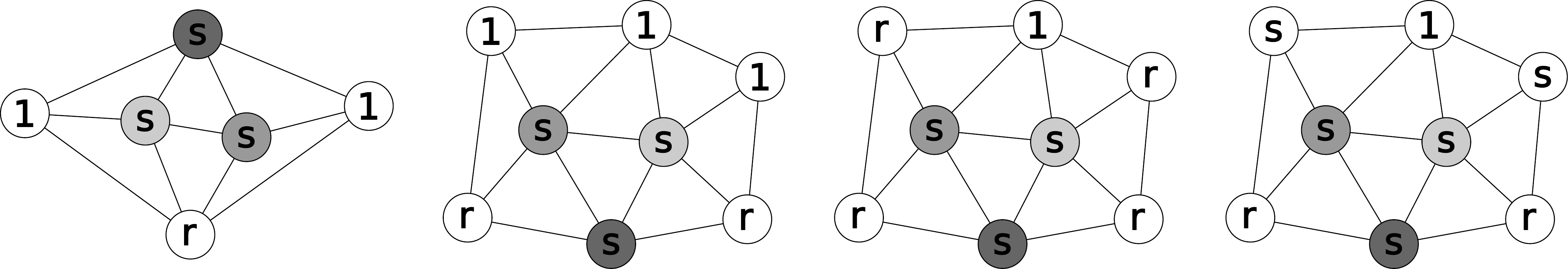}
  \caption{Quatre petites couronnes qui imposent deux petites couronnes.}
  \label{fig:1pc_impossible_1}
\end{figure}

Le lemme~\ref{lem:1pc_impossible_1} élimine $24$ couples petit/moyenne couronne qui ne permettent pas d'empilement.
Le lemme \ref{lem:1pc_impossible_2} ci-dessous en élimine encore $7$, et les $145$ restants admettent tous un empilement périodique (Annexe~\ref{sec:empilements}).

\begin{lemma}
  \label{lem:1pc_impossible_2}
  Il n'y a pas d'empilement avec des couples petite/moyenne couronnes 1rsrs/1rr1ss, 11rr/11rrs, 1rr1s/11rrs, rrrrr/1rsrsr, rrrrs/11rssr, rrrss/11rssr et rrrs/11rssr.
\end{lemma}

\begin{proof}
  On traite les cas un par un.\\
  {\bf 1rsrs/1rr1ss.}
  Les valeurs de $r$ et $s$ ne sont compatibles avec aucune autre moyenne couronne.
  Le r entre deux s de la petite couronne, qui touche trois s (ses voisins et le centre) n'a pas de couronne.\\
  {\bf 11rr/11rrs.}
  Les valeurs de $r$ et $s$ ne sont compatibles avec aucune autre moyenne couronne.
  Dans la moyenne couronne, la couronne du r avec un voisin s impose 1rr1 dans la couronne de son voisin r : c'est impossible.\\
  {\bf 1rr1s/11rrs.}
  Les valeurs de $r$ et $s$ ne sont compatibles avec aucune autre moyenne couronne.
  Dans la moyenne couronne, la couronne du r voisin de 1 impose le facteur srrs dans la couronne du r voisin de s.
  C'est impossible car la moyenne couronne n'a pas de facteur srrs.\\
  {\bf rrrrr/1rsrsr.}
  Les valeurs de $r$ et $s$ ne sont compatibles avec aucune autre moyenne couronne.
  Il y a deux façons symétriques de placer la couronne d'un r dans une petite couronne.
  Ce choix impose alors, de proche en proche, un placement unique des couronnes des autres r de la petite couronne.
  Il n'y a pas de couronne possible pour le dernier r.\\
  {\bf rrrrs/11rssr.}
  Les valeurs de $r$ et $s$ sont compatibles avec deux autres moyennes couronnes, rsrsrss et 1111r, ainsi qu'avec la grande couronne, 1r1r1rr.
  Appelons disque $r_1$ un moyen disque qui touche au moins un grand disque.
  Si un empilement ne contient pas la moyenne couronne 11rssr, alors en partant d'un disque $1$ les autres couronnes ne permettent d'atteindre que des disques $1$ ou $r_1$, donc il n'y a pas de petit disque dans l'empilement : contradiction.
  Considérons donc une moyenne couronne 11rssr.
  Soit deux disques r et s voisins dans cette couronne.
  Il y a une seule façon de placer la couronne du r et elle impose au s un facteur sss dans sa propre couronne, ce que ne permet pas la petite couronne.\\
  {\bf rrrss/11rssr.}
  Les valeurs de $r$ et $s$ sont compatibles avec les moyennes couronnes rrsrrss, rrrsrss et 111rr, ainsi qu'avec les grandes couronnes 11r11rr et 111r1rr.
  Le même argument que pour rrrrs/rsrsrss impose la présence d'une moyenne couronne 11rssr.
  Soit deux disques r et s voisins dans cette couronne.
  Il y a une seule façon de placer la couronne du r et elle impose au s un facteur srrs dans sa propre couronne, ce que ne permet pas la petite couronne.\\
  {\bf rrrs/11rssr.}
  Les valeurs de $r$ et $s$ sont compatibles avec les moyennes couronnes rsrsrsrsss, rsrsrssrss, rsrssrsrss, rrrrsrss, rrrsrrss, rrrsrss, 1rsrsssr, 1rssrssr et 111r.
  La petite couronne interdit trois s consécutifs dans une couronne, ce qui élimine 1rsrsssr (et rsrsrsrsss).
  Le même argument que pour rrrrs/rsrsrss impose la présence d'une moyenne couronne 1rssrssr ou 11rssr.
  Considérons une de ces couronnes (l'autre cas est similaire).
  La petite couronne force les paires de s voisins dans cette couronne a être encerclés par quatre r.
  Considérons un des deux r qui, dans cette couronne, est entre 1 et s : il a alors un facteur 1rsr dans sa couronne, ce qui est impossible.
\end{proof}

\section{Conclusion}

Il y a donc $18$ couples $(r,s)$ qui permettent un empilement à deux phases, un seul qui permet deux petites couronnes et admet un empilement (qui ne contient qu'une de ces petites couronnes) et $145$ qui ne permettent qu'une seule petite couronne et admettent un empilement avec une seule phase.
Au total, il y a donc $164$ couples $(r,s)$ qui permettent un empilement : le théorème~\ref{th:164} est démontré.
La figure~\ref{fig:repartition_r_sr} illustre la répartition des rayons (les numéros correspondent à ceux de l'annexe~\ref{sec:empilements}) et la table~\ref{tab:degre_rayons} celle de leur degré algébrique.\\

\begin{figure}[hbtp]
  \centering
  \includegraphics[width=0.85\textwidth]{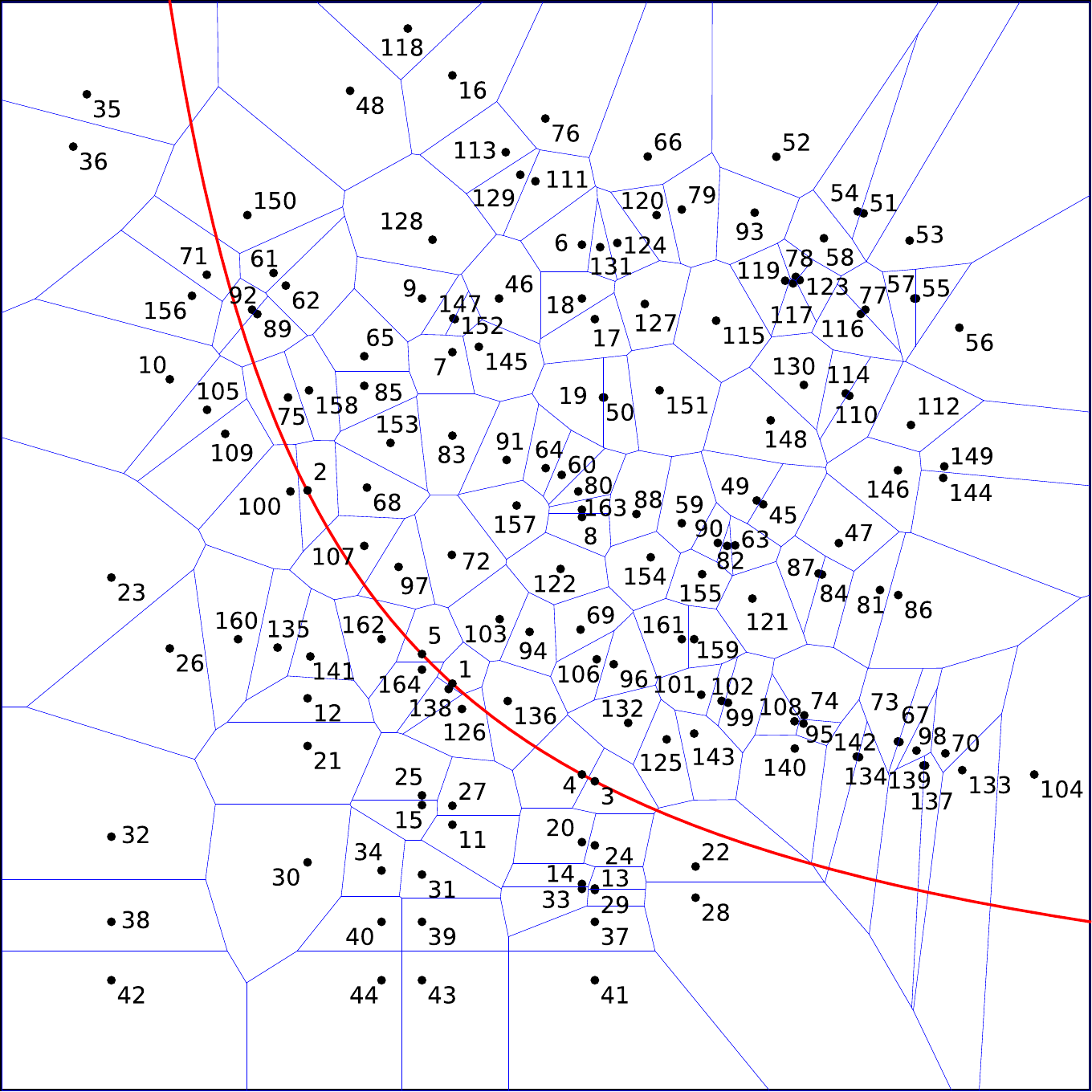}
  \label{fig:repartition_r_sr}
  \caption{
    Répartition des $164$ couples $(r,s)$ qui permettent un empilement, avec $r$ en abscisse et $\tfrac{s}{r}$ en ordonnée.
    Les couples sous la courbe rouge sont {\em interstitiels} : le petit disque tient dans le trou entre trois grands disques.
    Les cellules de Voronoï donnent une idée de la proximité entre couples.
  }
\end{figure}

\begin{table}[hbtp]
  \centering
  \resizebox{\textwidth}{!}{  
  \begin{tabular}{c|c|c|c|c|c|c|c|c|c|c|c|c|c|c|c|c}
    $1$ & $2$ & $3$ & $4$ & $5$ & $6$ & $7$ & $8$ & $9$ & $10$ & $11$ & $12$ & $14$ & $16$ & $20$ & $22$ & $24$\\
    $10$ & $52$ & $29$ & $79$ & $20$ & $30$ & $4$ & $47$ & $2$ & $10$ & $6$ & $18$ & $2$ & $9$ & $4$ & $2$ & $4$
  \end{tabular}}
  \caption{Nombre de rayons $r$ et $s$ en fonction de leur degré algébrique.}
  \label{tab:degre_rayons}
\end{table}


\newpage
\appendix

\section{Empilements}
\label{sec:empilements}

On présente ici un exemple d'empilement pour chacun des $164$ couple $(r,s)$ possibles.
Ils sont tous périodiques et un cadre indique sur chaque dessin un domaine fondamental.
On donne pour chacun les codages d'une petite et d'une moyenne couronnes (à partir desquelles le couple $(r,s)$ se déduit).
Les $18$ premiers sont ceux à deux phases, le $19^\textrm{ème}$ correspond au couple $(r,s)$ qui permet deux petites couronnes (dont une seule utilisable dans un empilement), les $145$ suivant n'admettent qu'une petite couronne et n'ont pas deux phases.\\

\noindent
\begin{tabular}{lll}
  1\hfill 111 / 1111 & 2\hfill 111 / 111r & 3\hfill 111 / 111rr\\
  \includegraphics[width=0.3\textwidth]{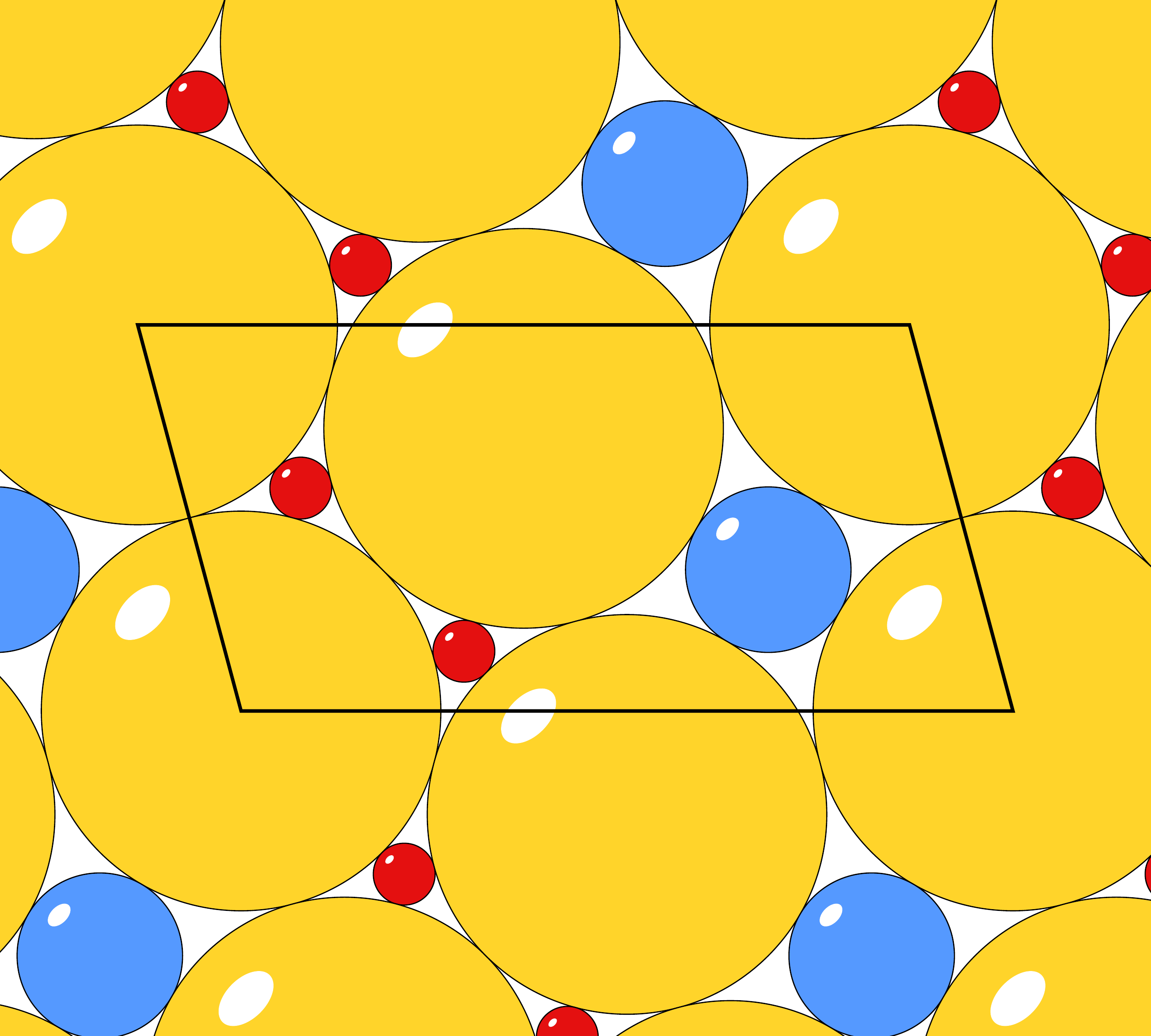} &
  \includegraphics[width=0.3\textwidth]{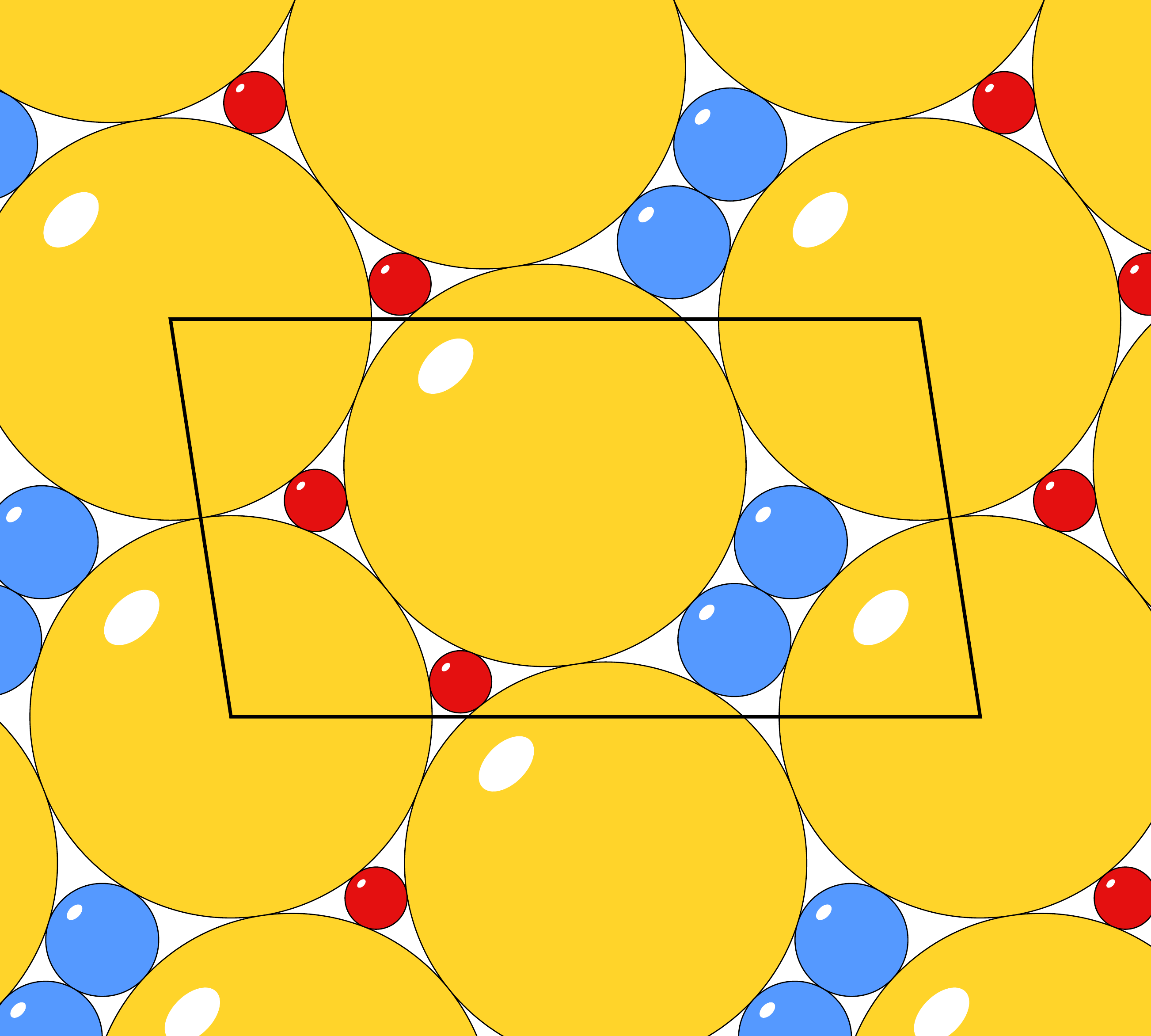} &
  \includegraphics[width=0.3\textwidth]{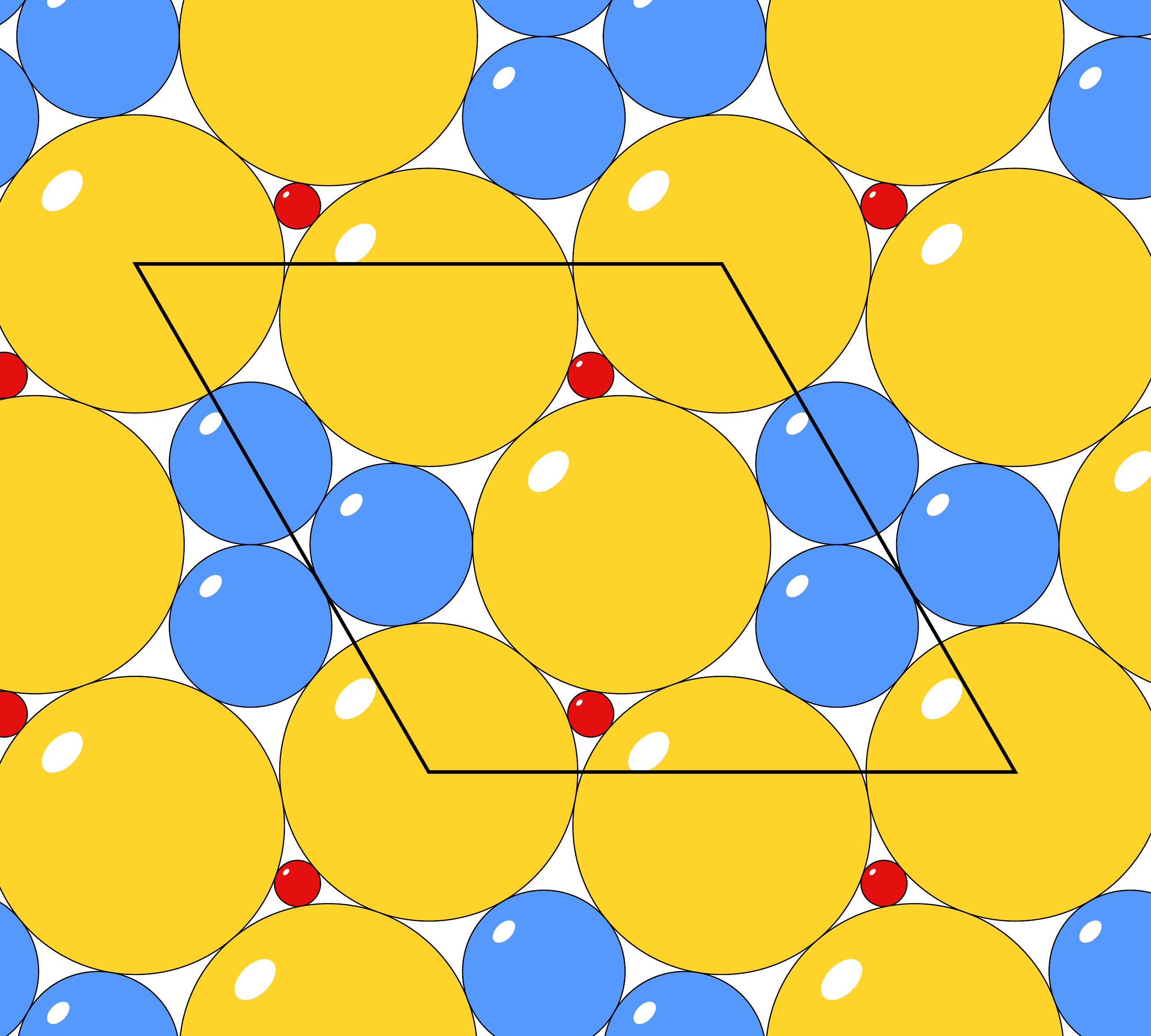}
\end{tabular}
\noindent
\begin{tabular}{lll}
  4\hfill 111 / 11r1r & 5\hfill 111 / 11rrr & 6\hfill 1111 / 11r1r\\
  \includegraphics[width=0.3\textwidth]{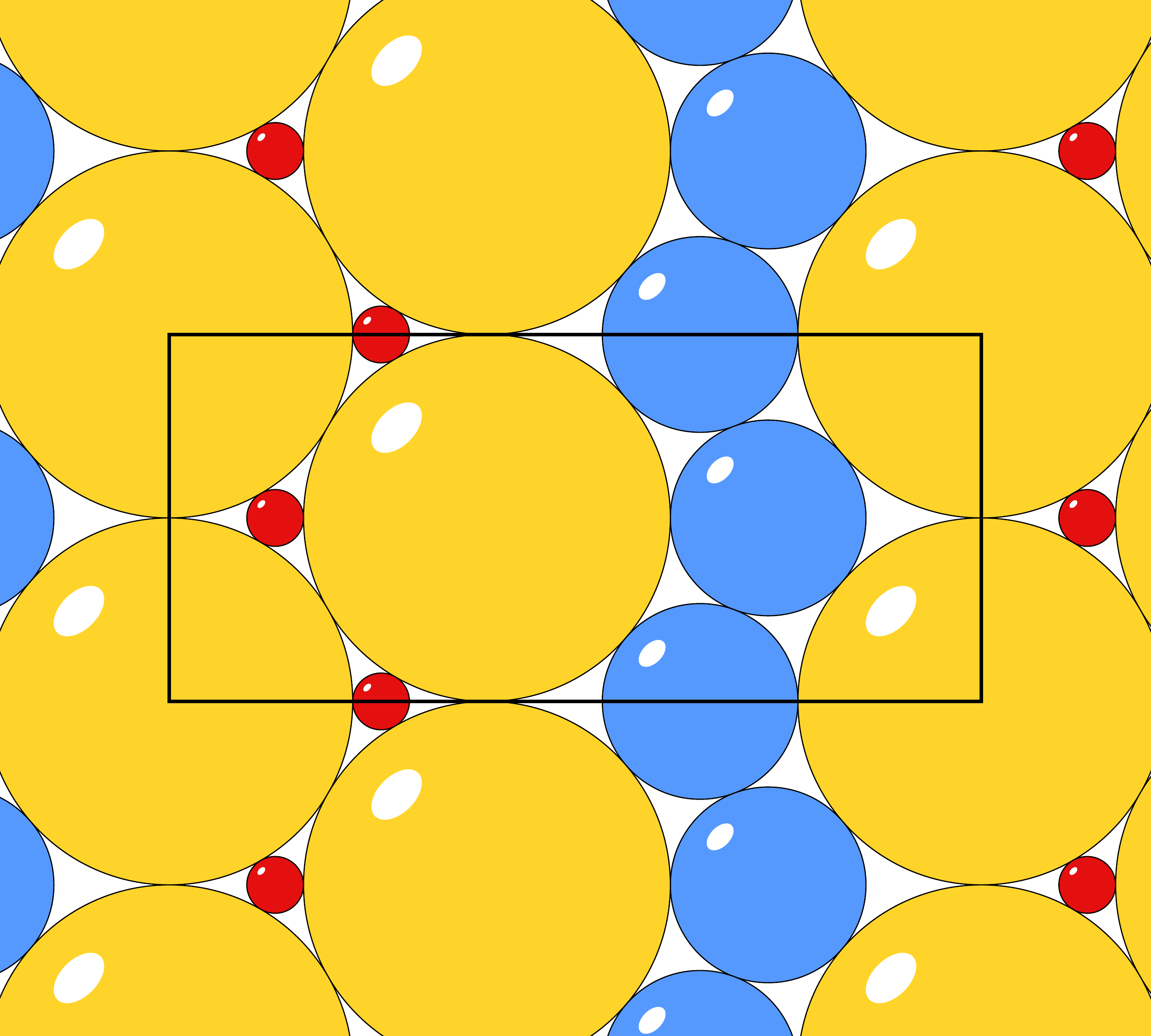} &
  \includegraphics[width=0.3\textwidth]{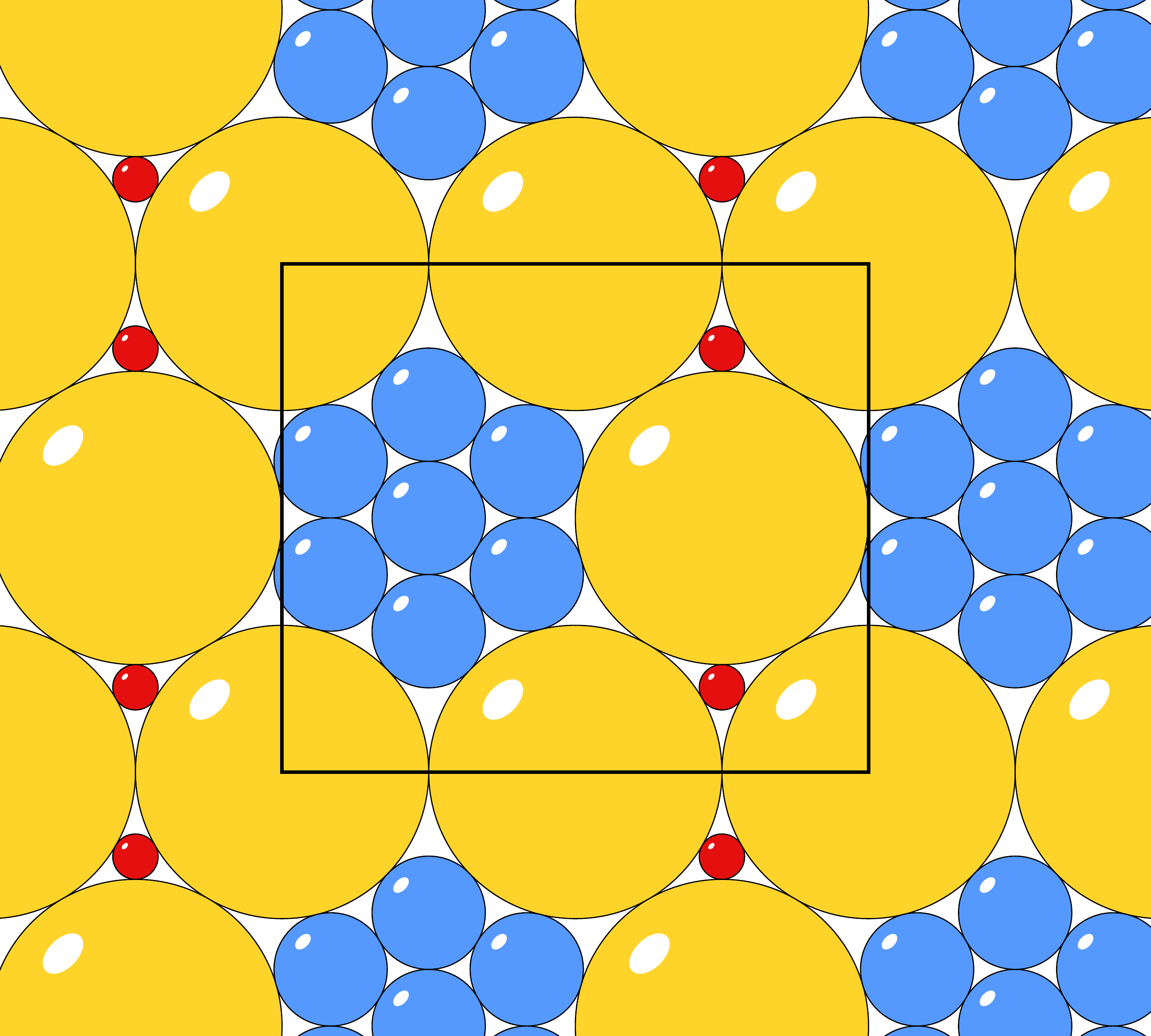} &
  \includegraphics[width=0.3\textwidth]{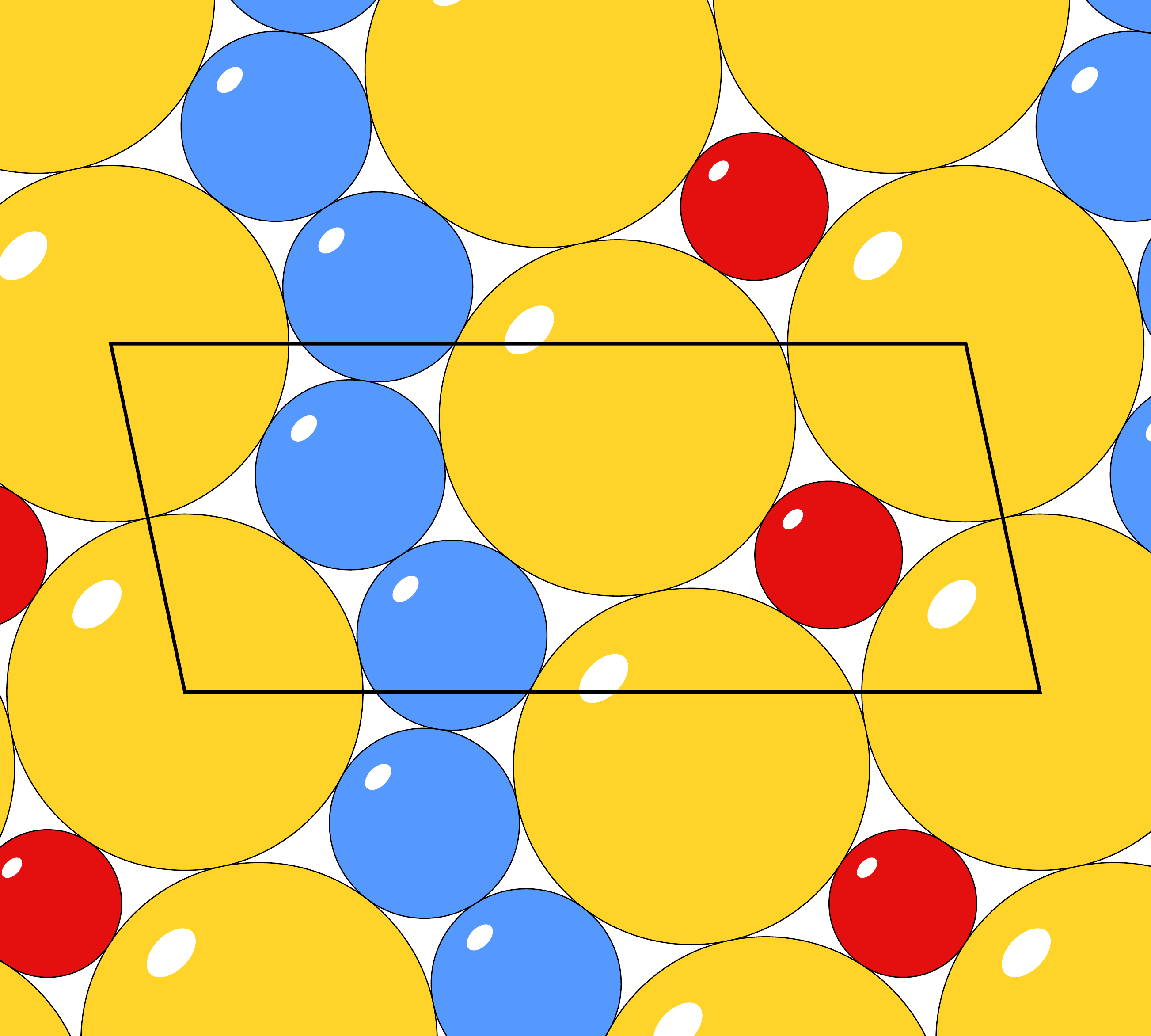}
\end{tabular}
\noindent
\begin{tabular}{lll}
  7\hfill 111s / 1111 & 8\hfill 111s / 11r1r & 9\hfill 111s / 11rrr\\
  \includegraphics[width=0.3\textwidth]{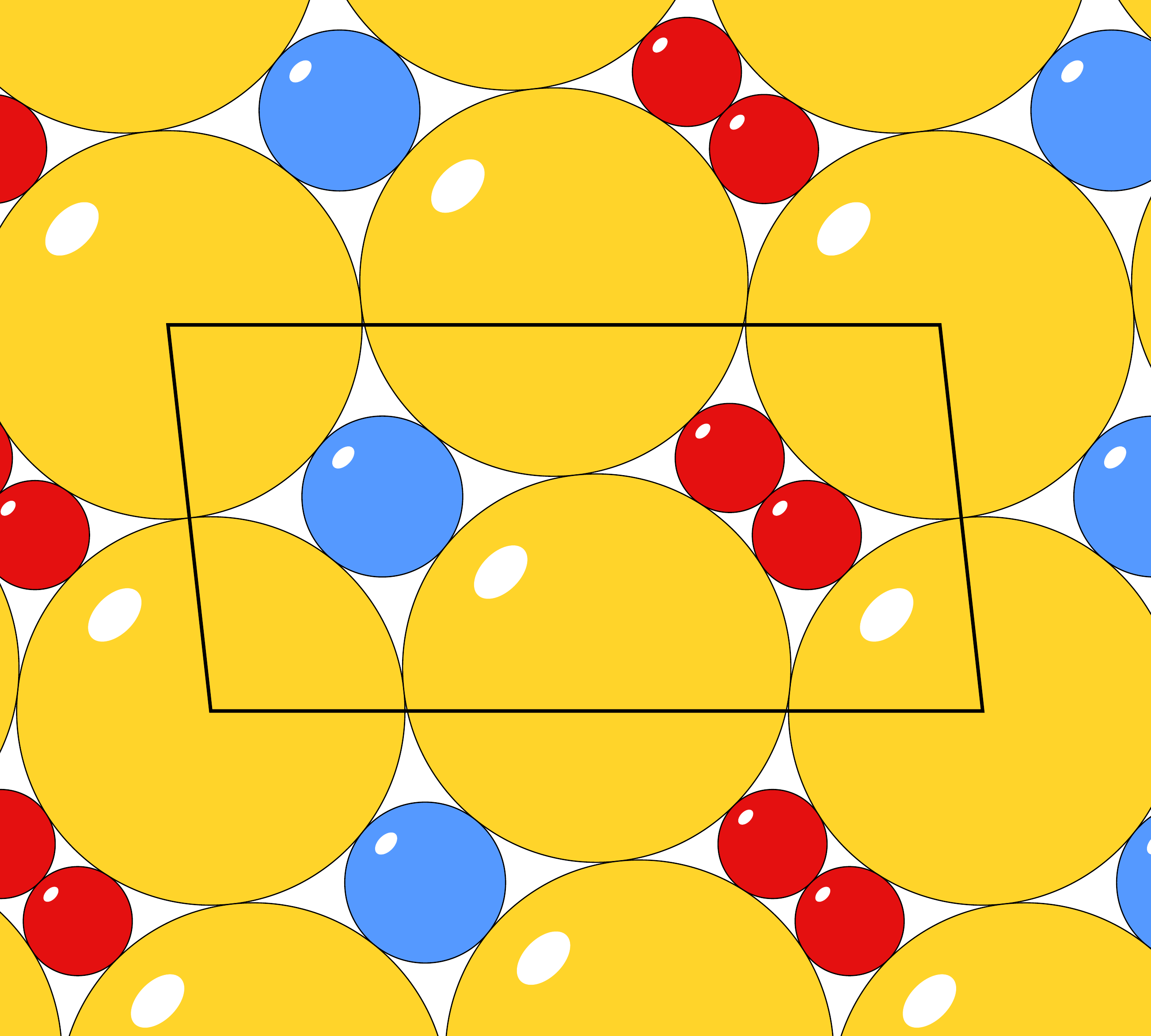} &
  \includegraphics[width=0.3\textwidth]{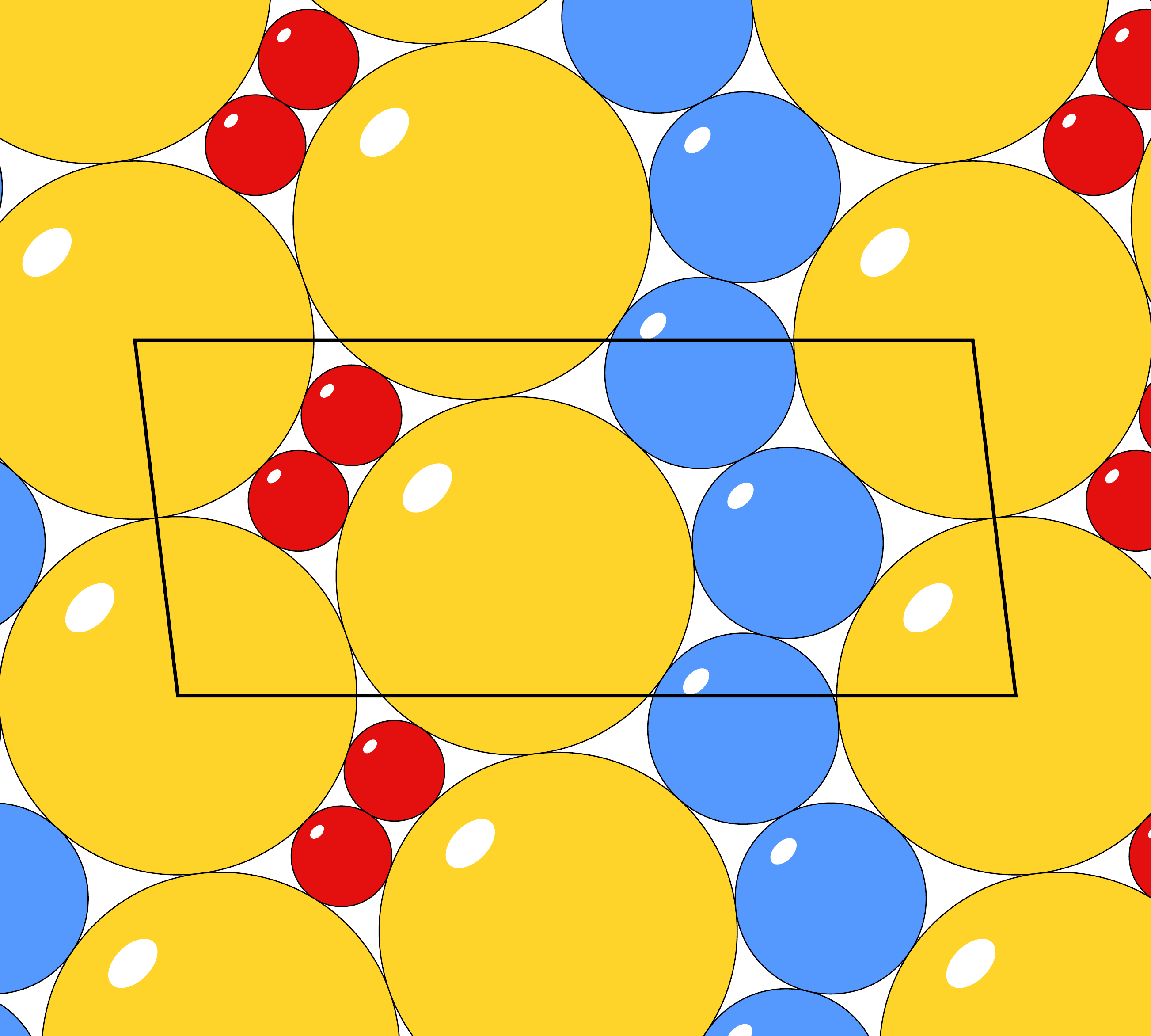} &
  \includegraphics[width=0.3\textwidth]{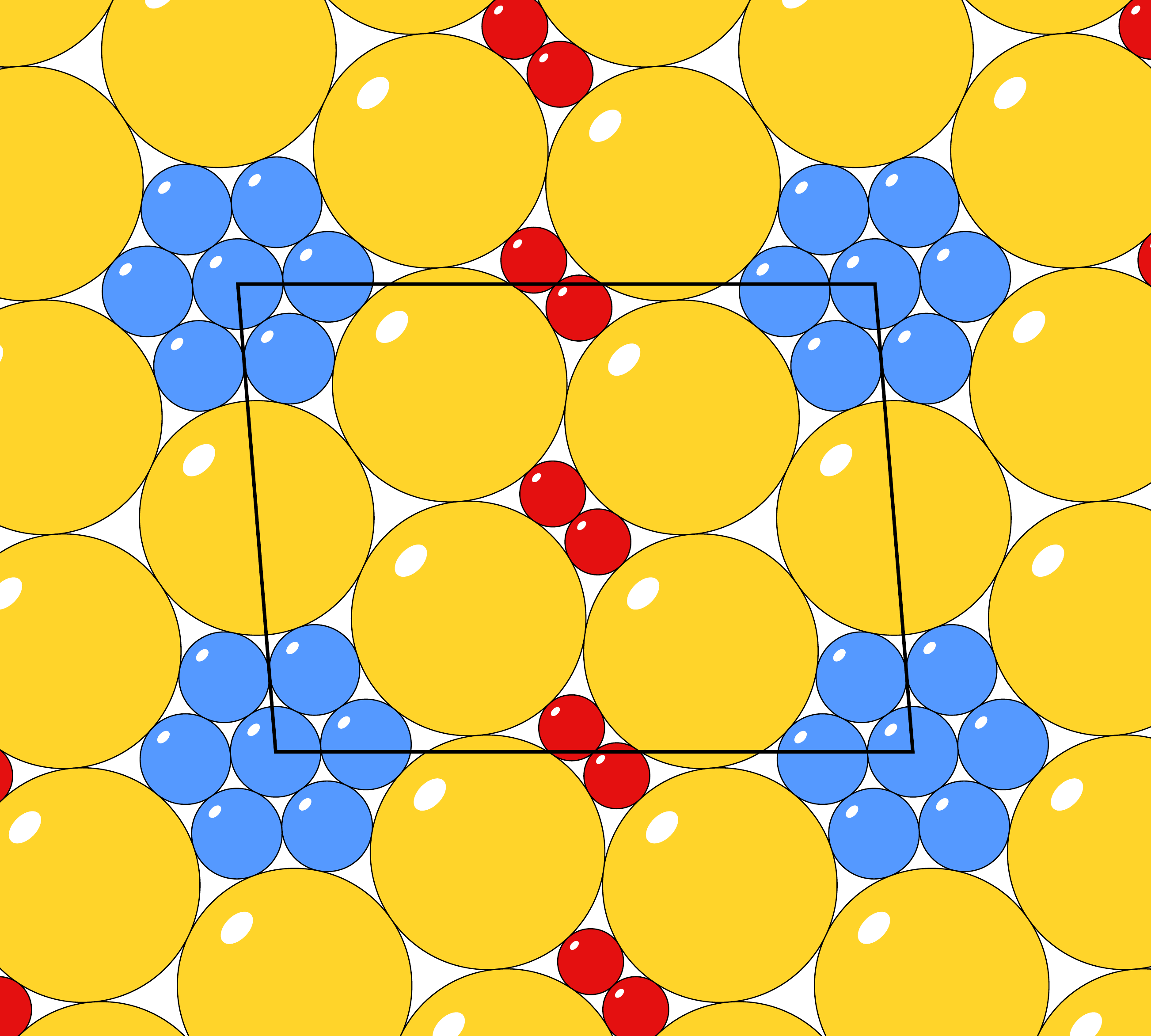}
\end{tabular}
\noindent
\begin{tabular}{lll}
  10\hfill 11ss / 111 & 11\hfill 11ss / 1111 & 12\hfill 11ss / 111r\\
  \includegraphics[width=0.3\textwidth]{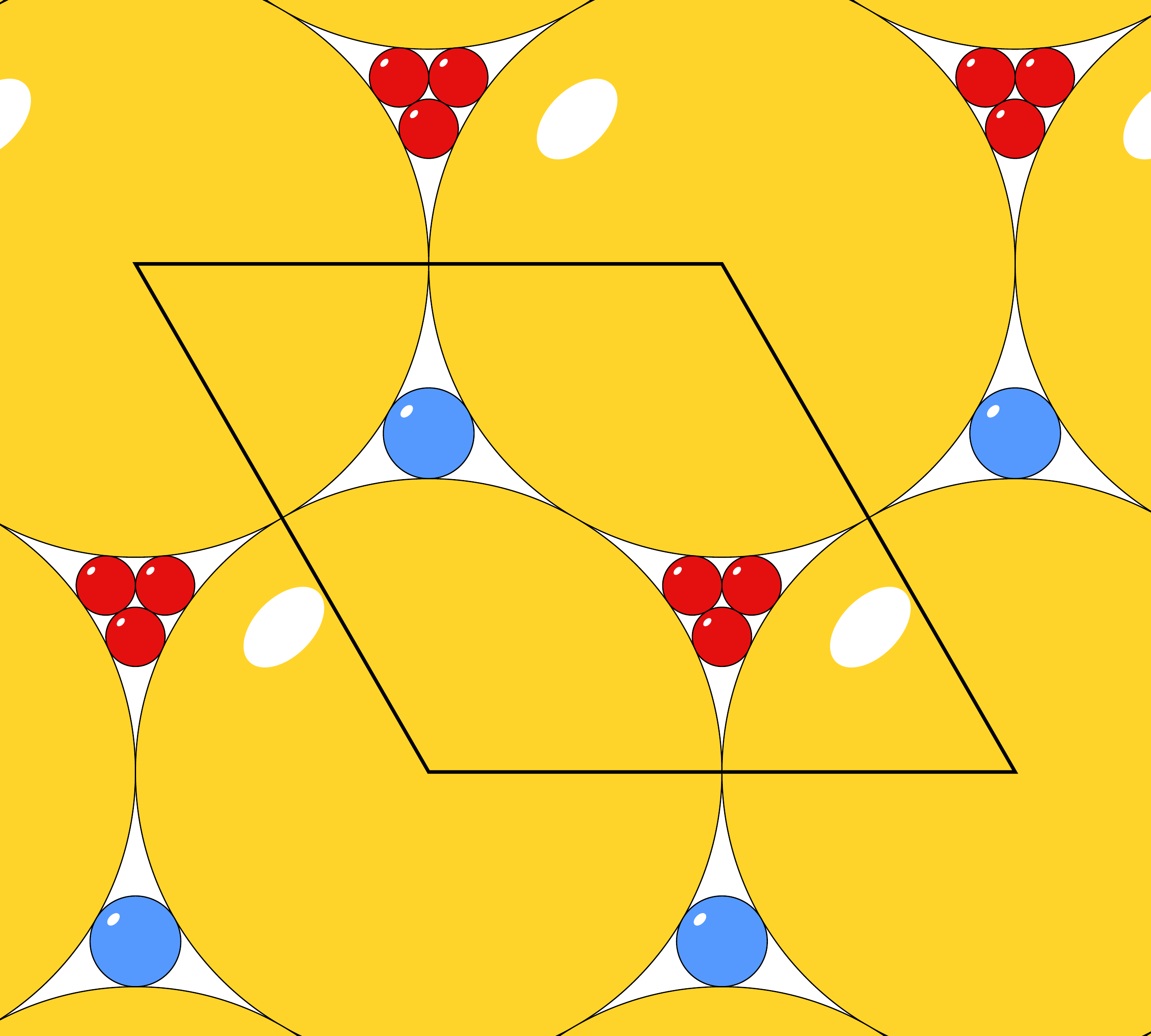} &
  \includegraphics[width=0.3\textwidth]{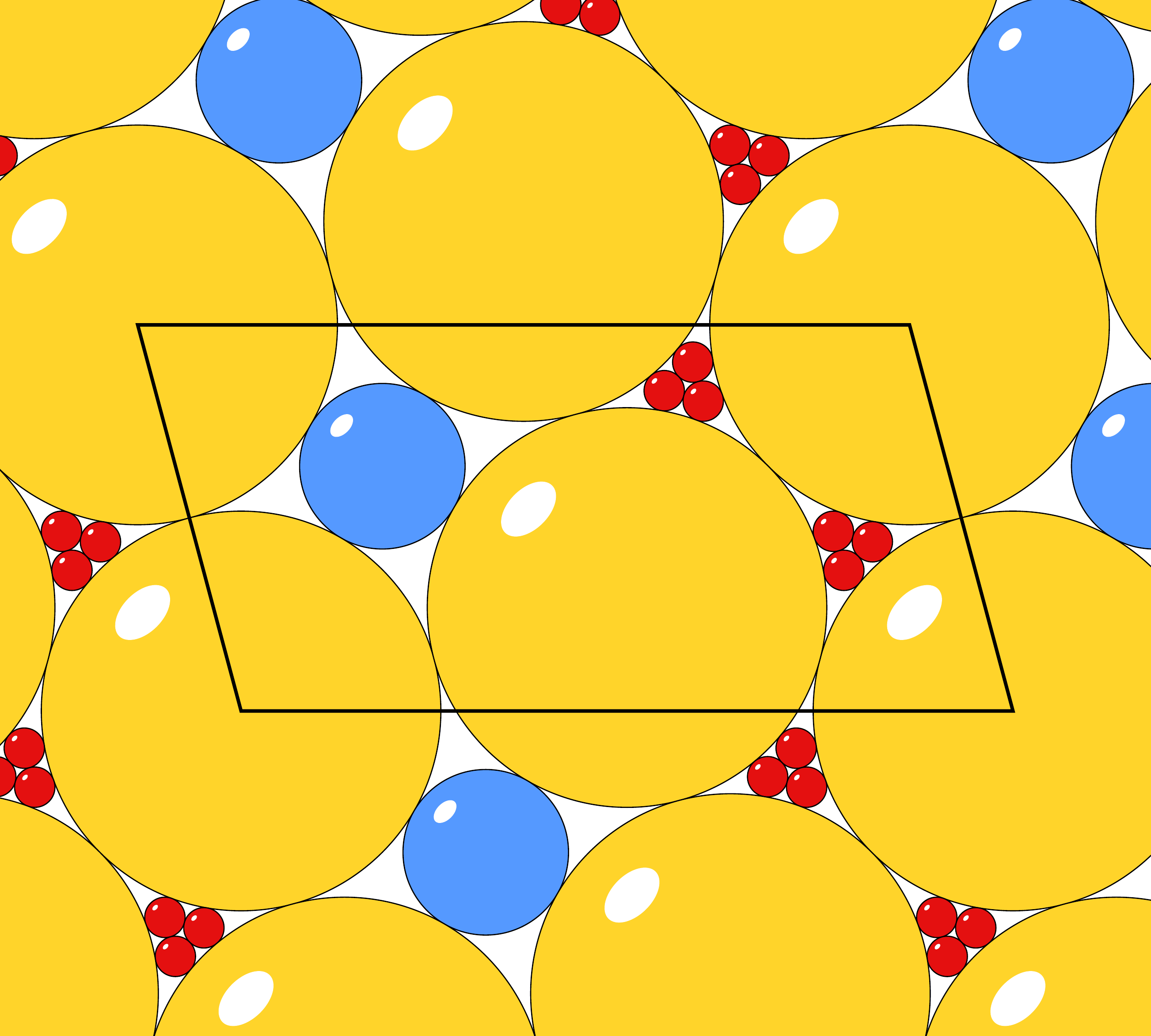} &
  \includegraphics[width=0.3\textwidth]{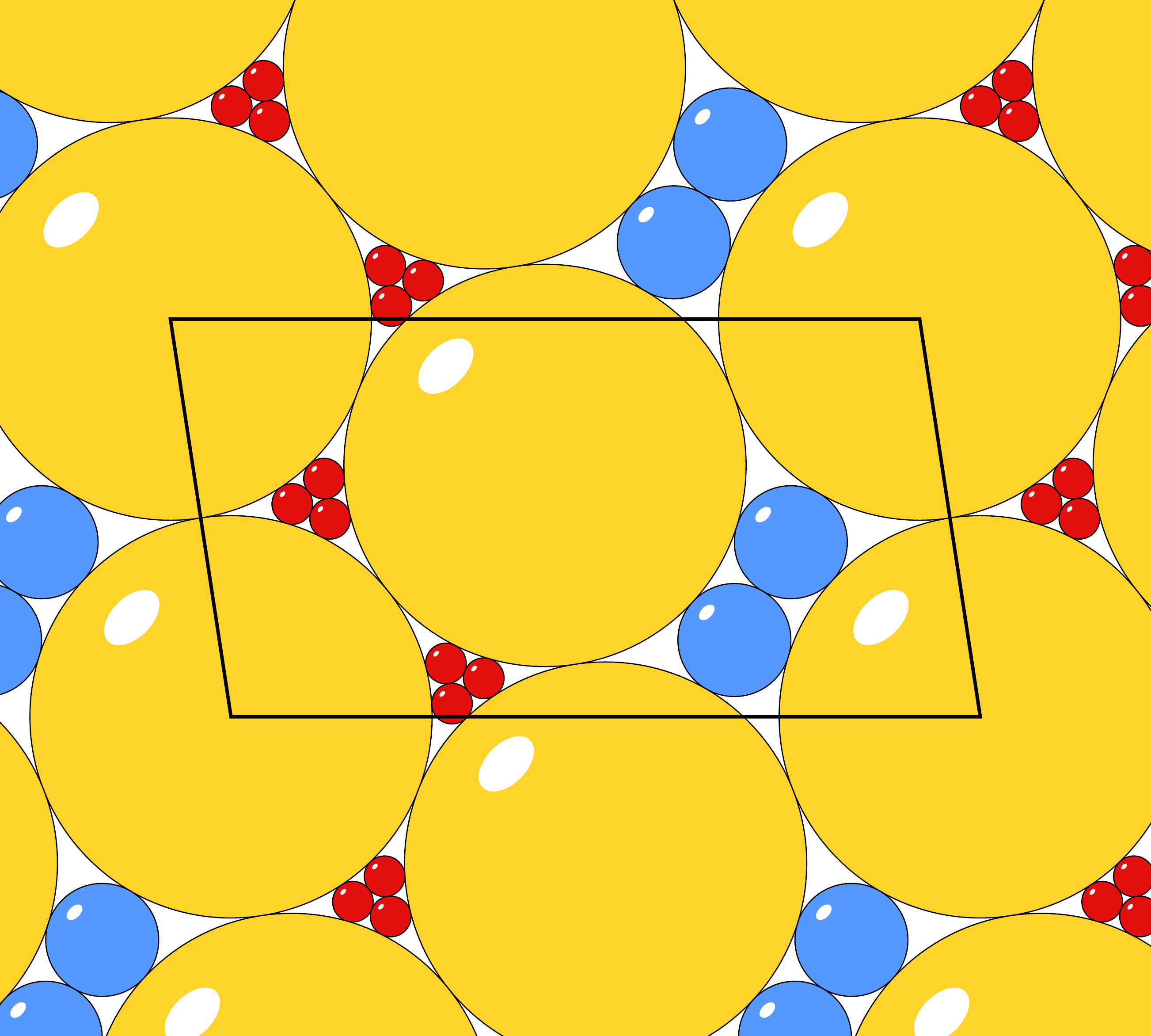}
\end{tabular}
\noindent
\begin{tabular}{lll}
  13\hfill 11ss / 111rr & 14\hfill 11ss / 11r1r & 15\hfill 11ss / 11rrr\\
  \includegraphics[width=0.3\textwidth]{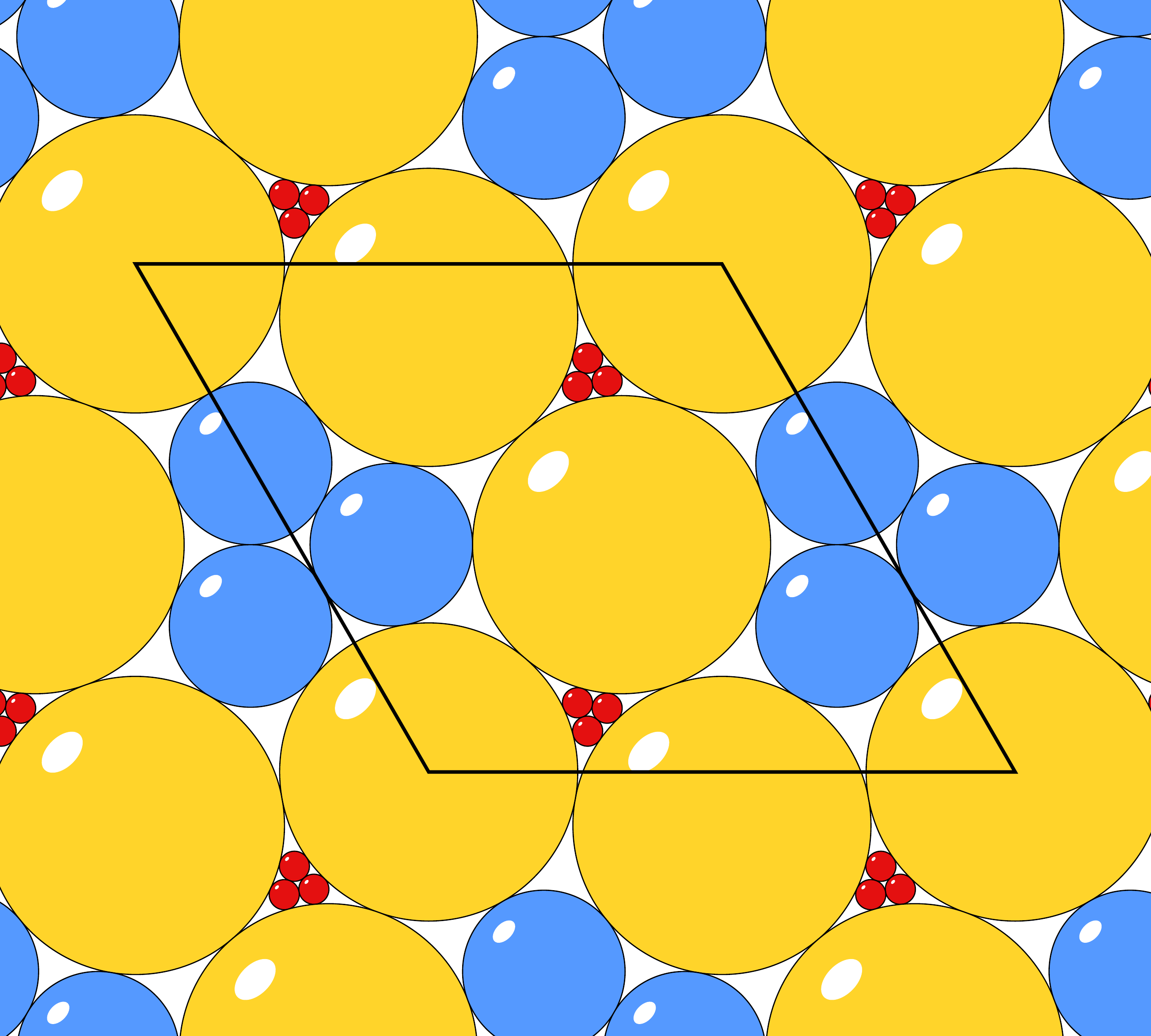} &
  \includegraphics[width=0.3\textwidth]{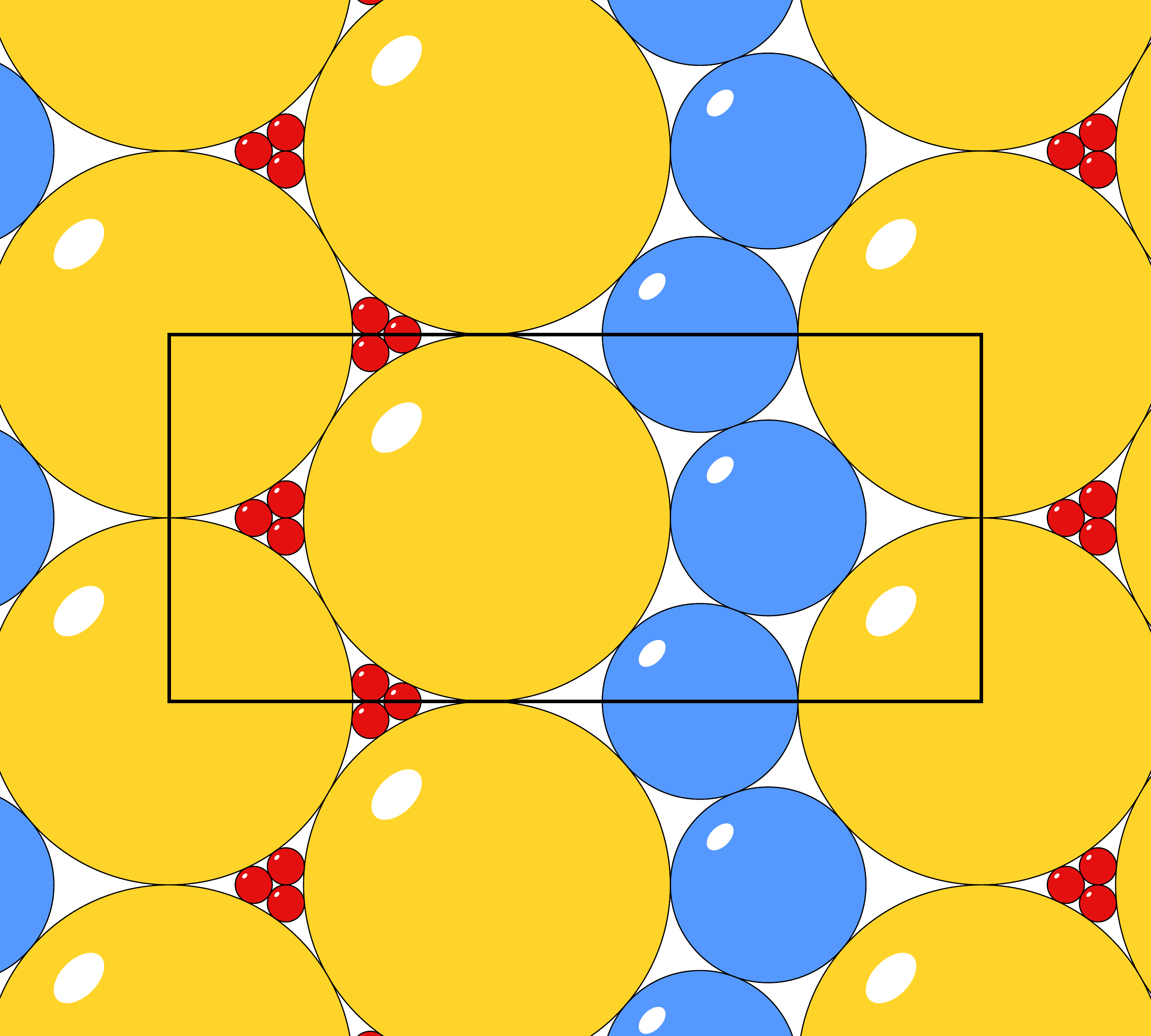} &
  \includegraphics[width=0.3\textwidth]{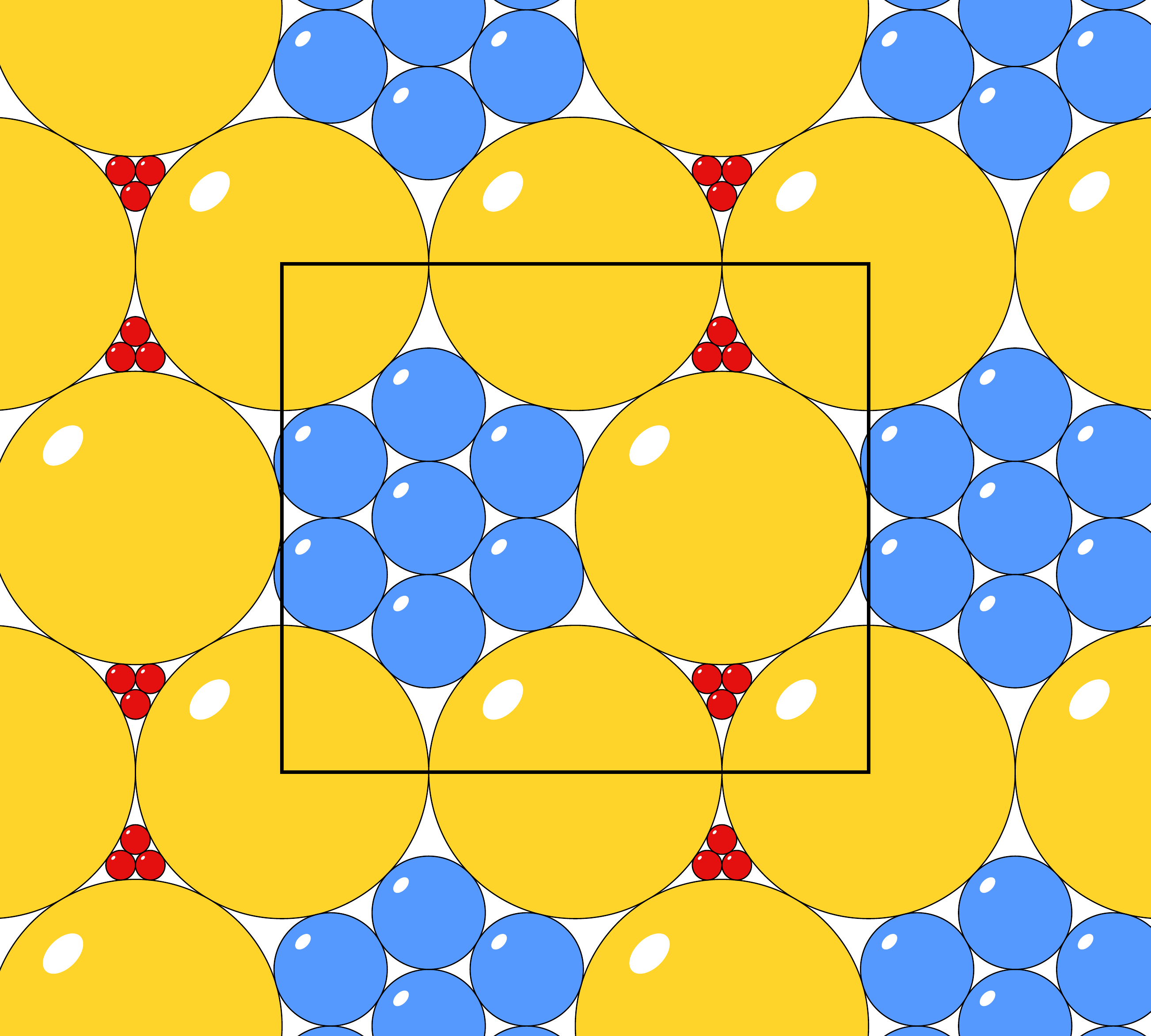}
\end{tabular}
\noindent
\begin{tabular}{lll}
  16\hfill 11sss / 1111 & 17\hfill 11sss / 111rr & 18\hfill 11sss / 11r1r\\
  \includegraphics[width=0.3\textwidth]{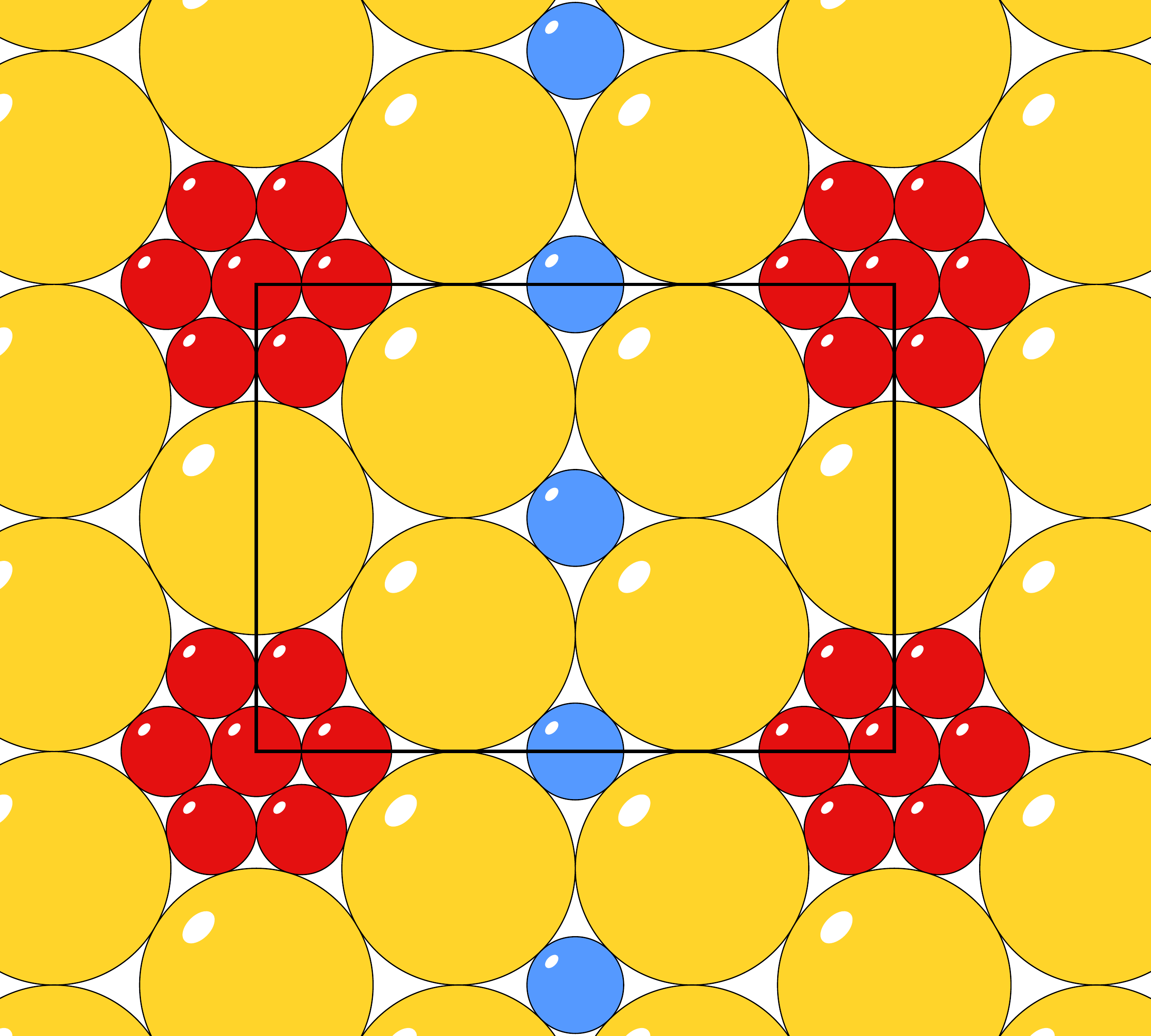} &
  \includegraphics[width=0.3\textwidth]{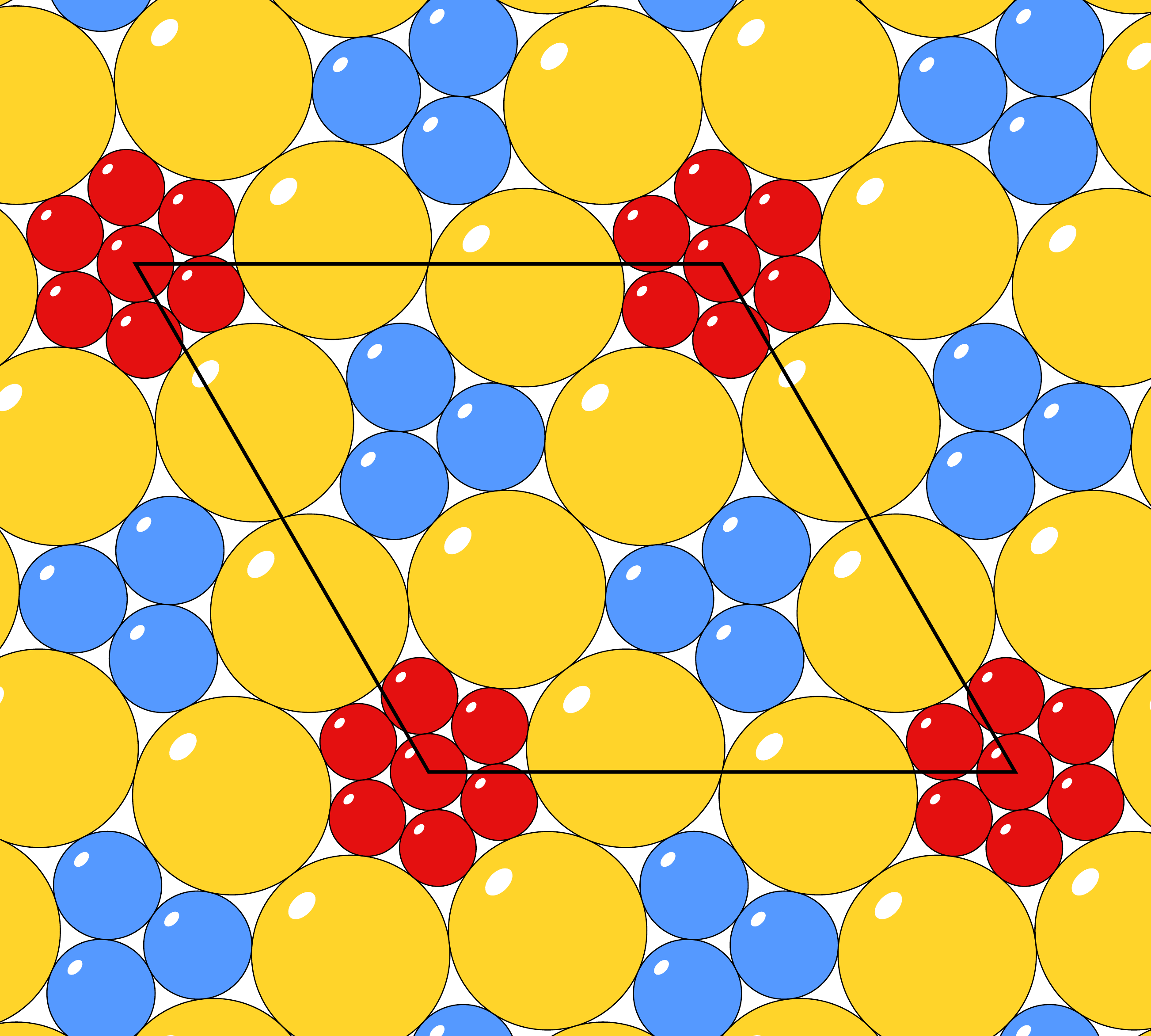} &
  \includegraphics[width=0.3\textwidth]{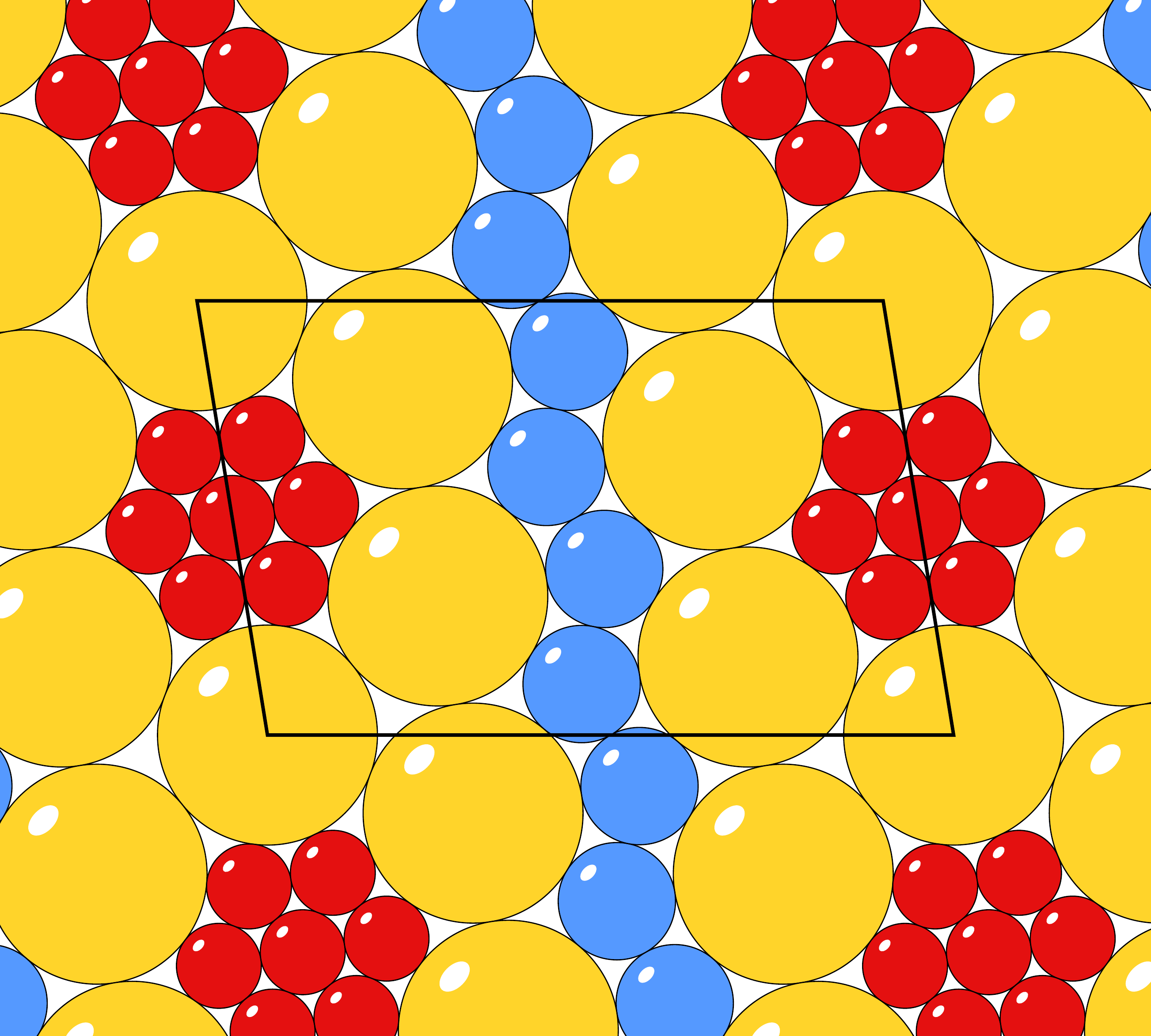}
\end{tabular}
\noindent
\begin{tabular}{lll}
  19\hfill 1srrs / rrsrsss & 20\hfill 11r / 1r1s1r & 21\hfill 11r / 1r1s1s\\
  \includegraphics[width=0.3\textwidth]{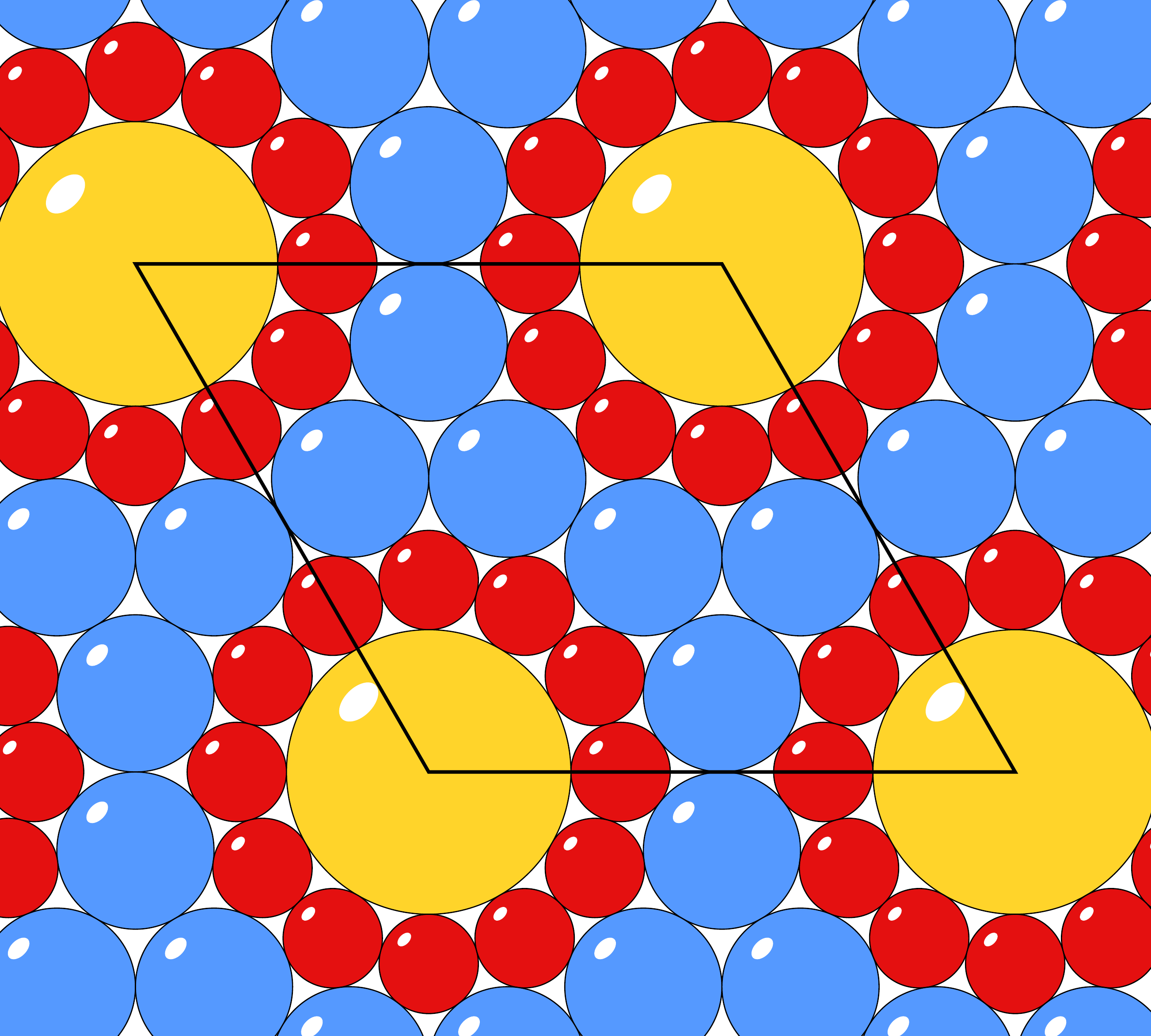} &
  \includegraphics[width=0.3\textwidth]{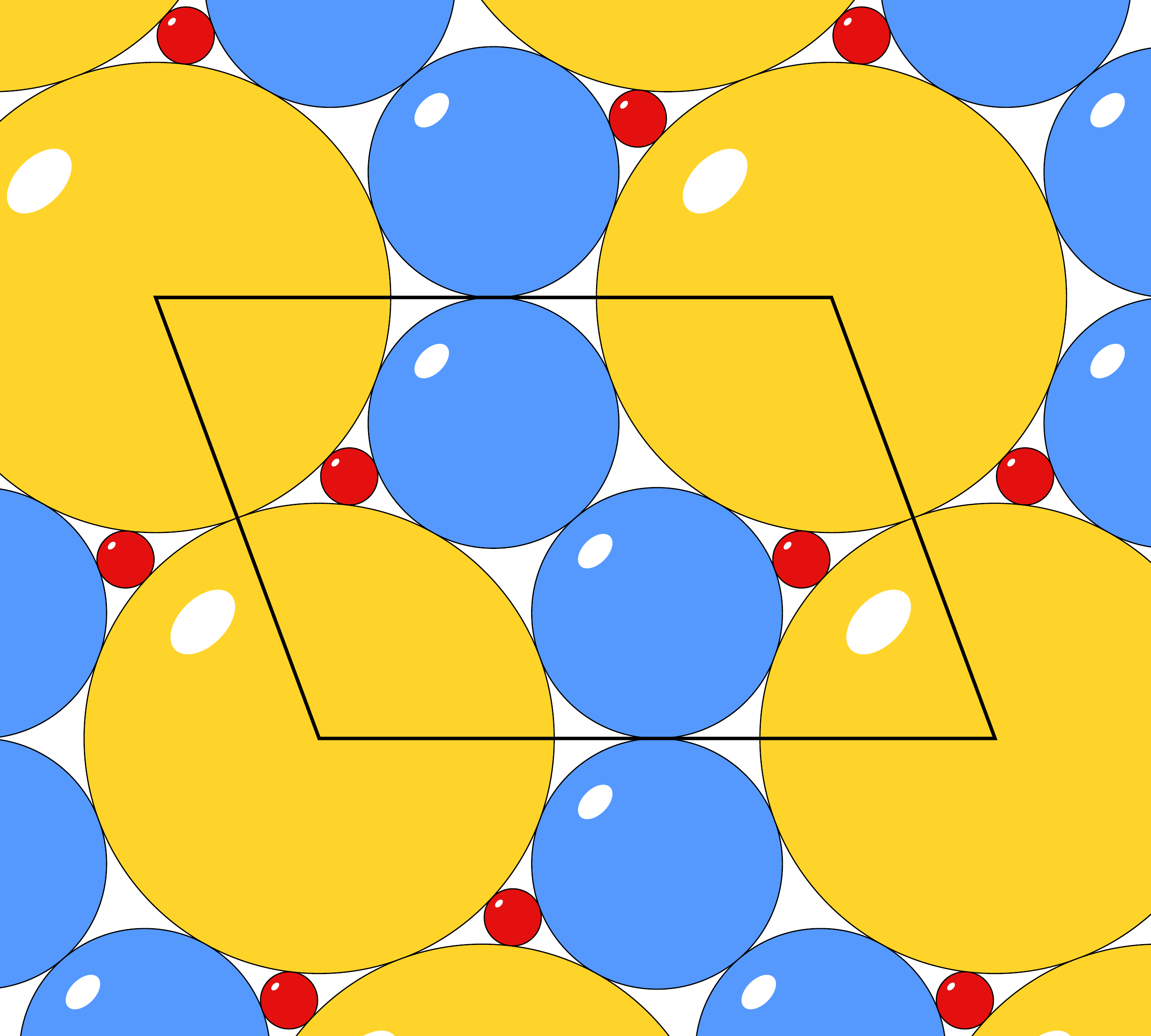} &
  \includegraphics[width=0.3\textwidth]{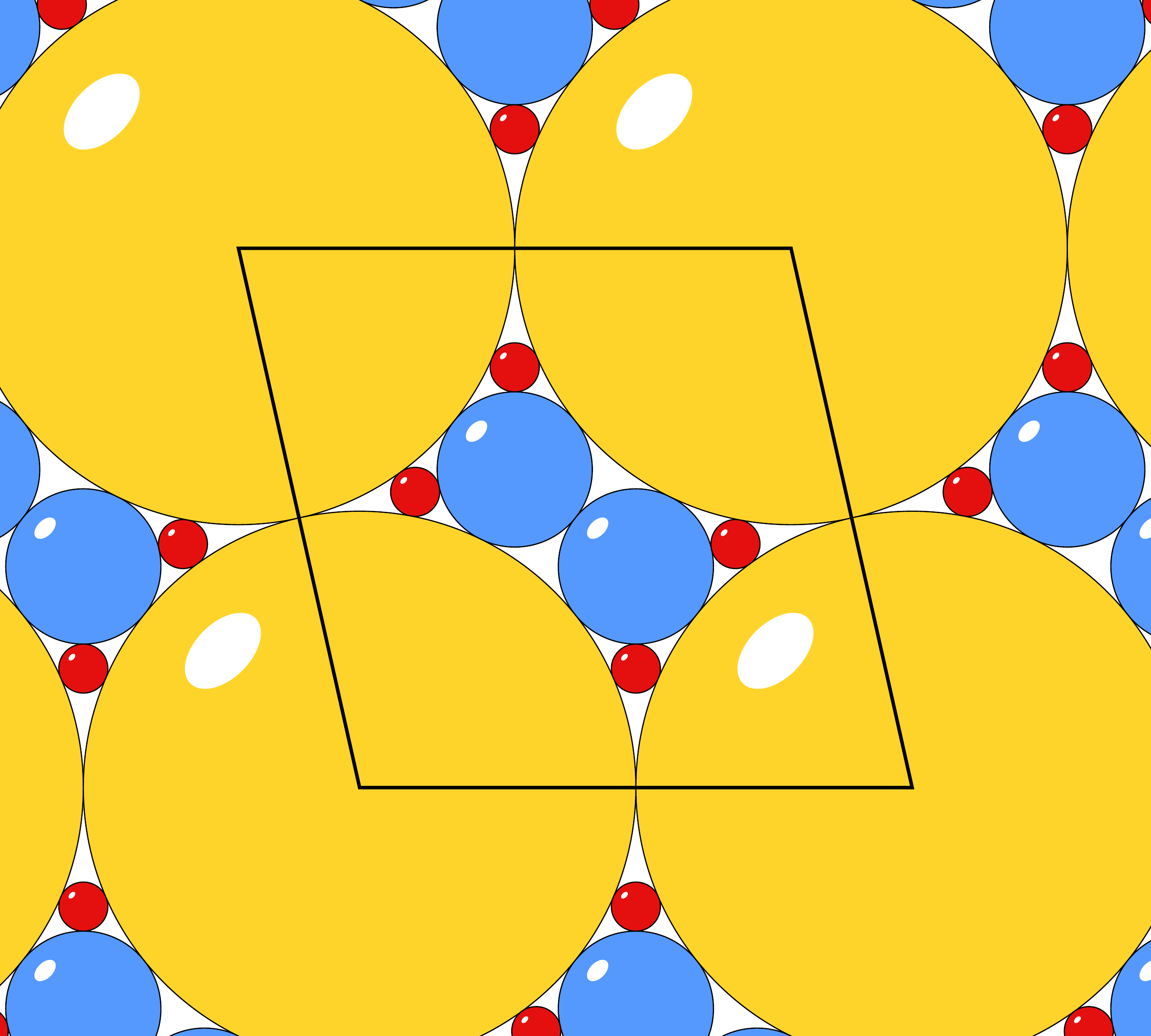}
\end{tabular}
\noindent
\begin{tabular}{lll}
  22\hfill 11r / 1r1s1s1s & 23\hfill 11r / 1rr1s & 24\hfill 11r / 1rr1s1s\\
  \includegraphics[width=0.3\textwidth]{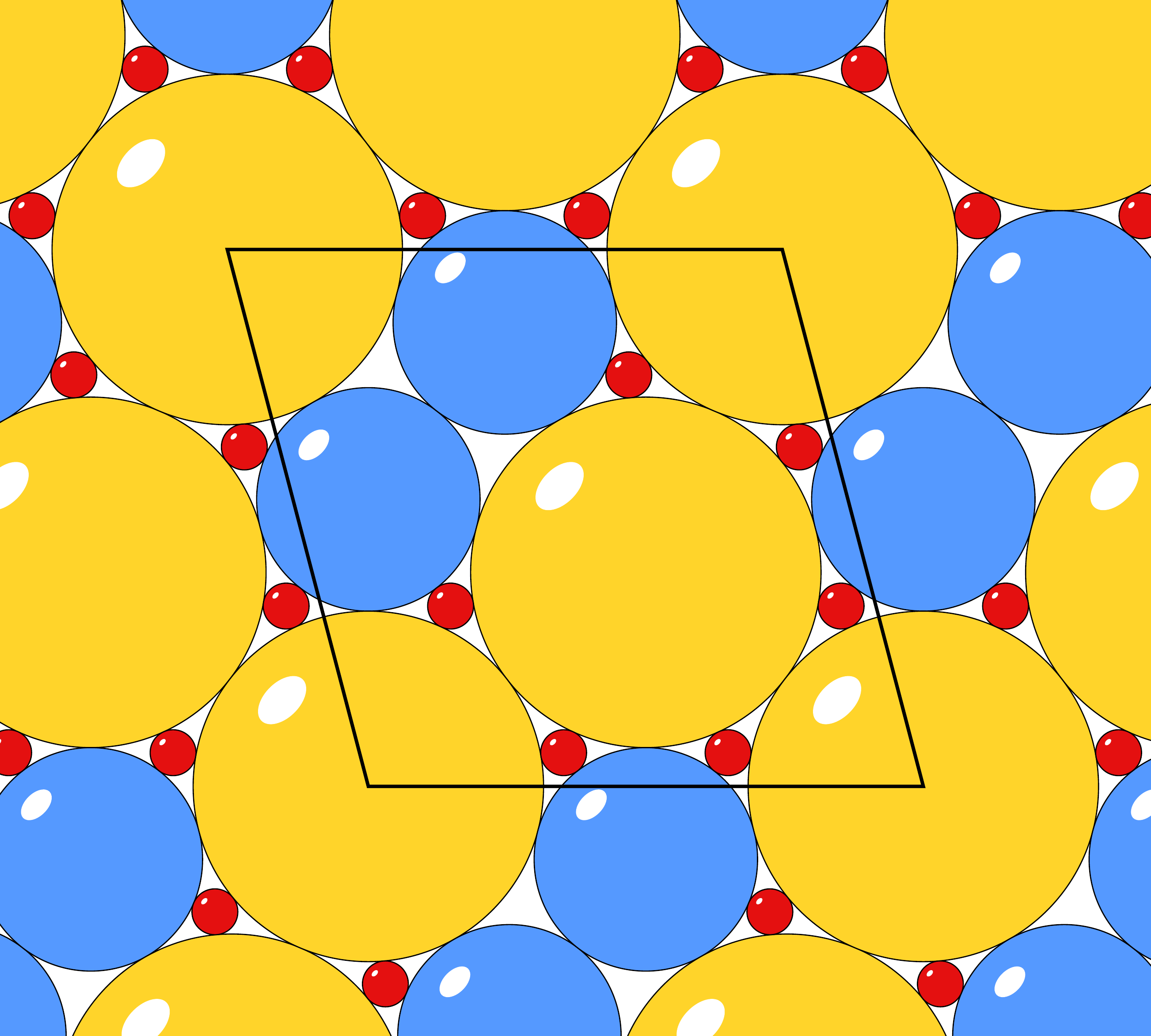} &
  \includegraphics[width=0.3\textwidth]{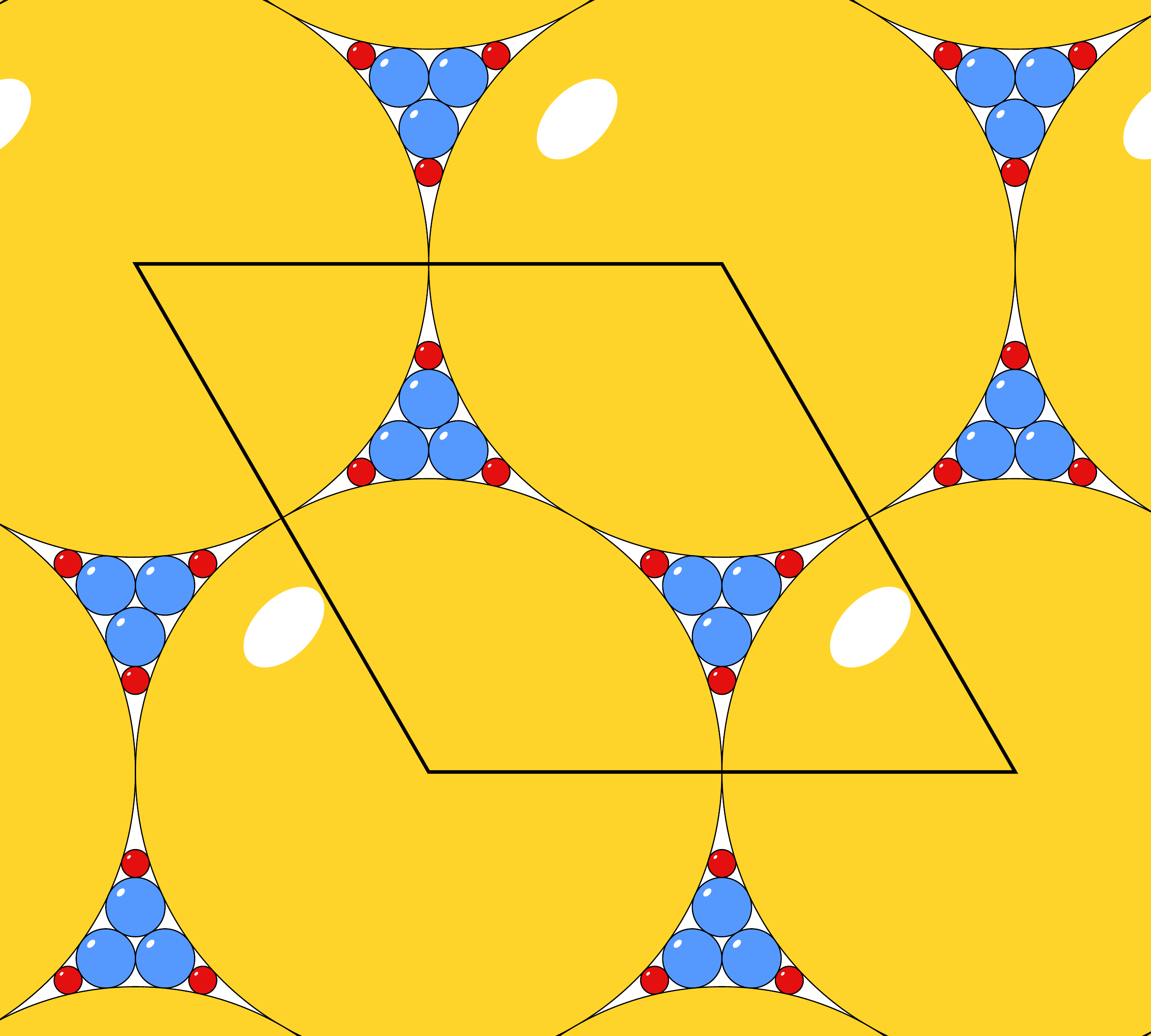} &
  \includegraphics[width=0.3\textwidth]{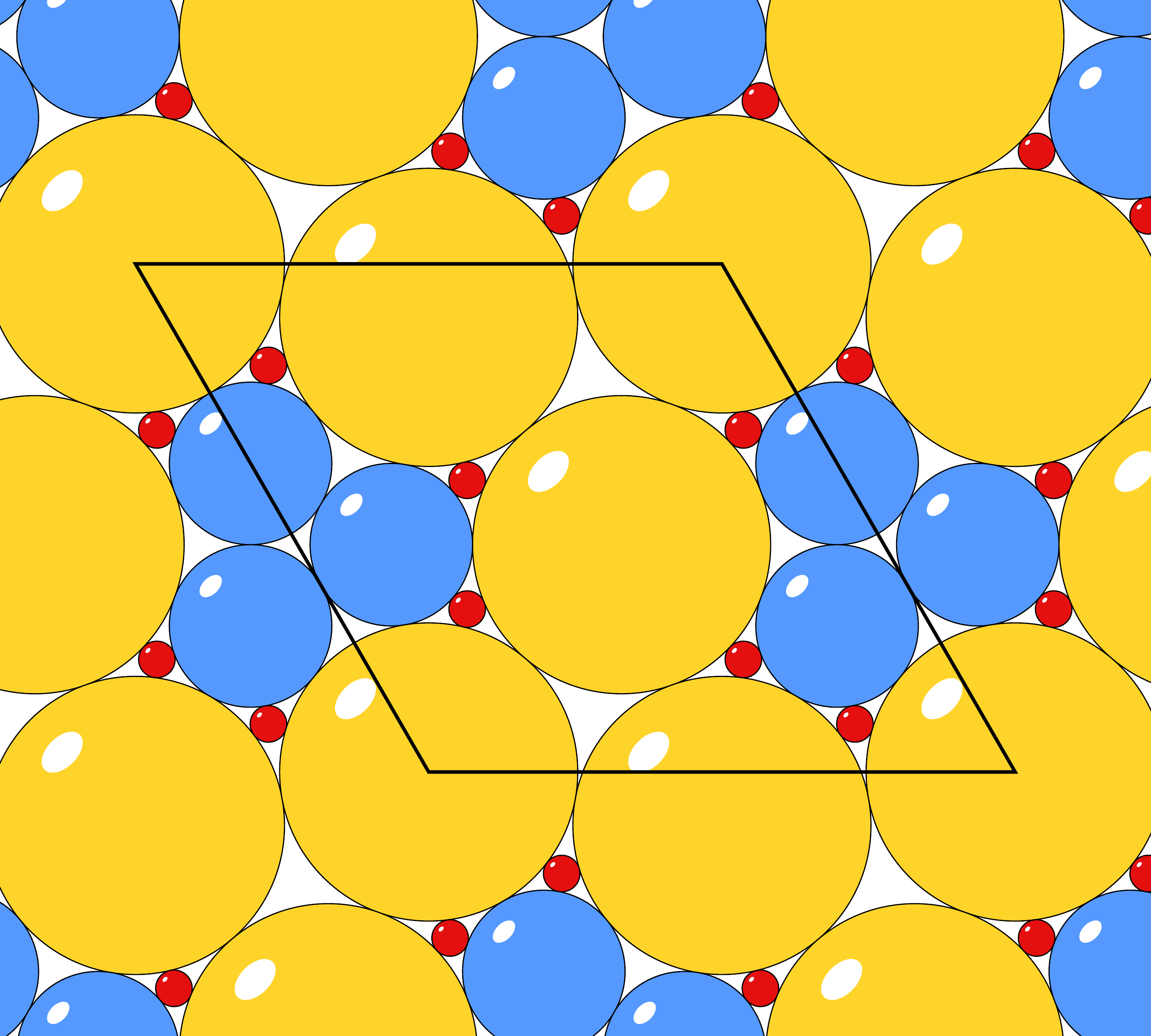}
\end{tabular}
\noindent
\begin{tabular}{lll}
  25\hfill 11r / 1rrr1s & 26\hfill 11r / 1s1s1s & 27\hfill 11r / 1s1s1s1s\\
  \includegraphics[width=0.3\textwidth]{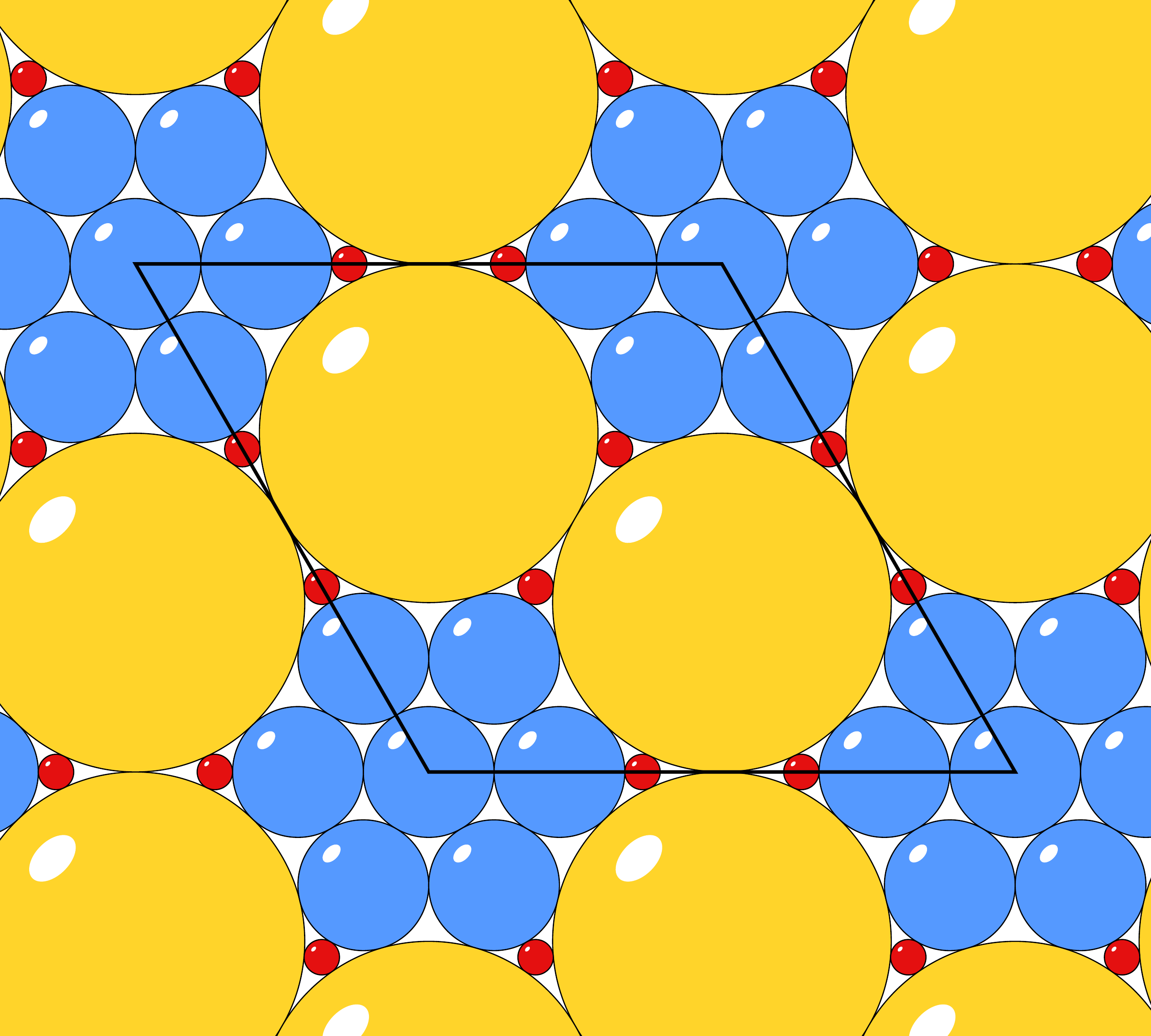} &
  \includegraphics[width=0.3\textwidth]{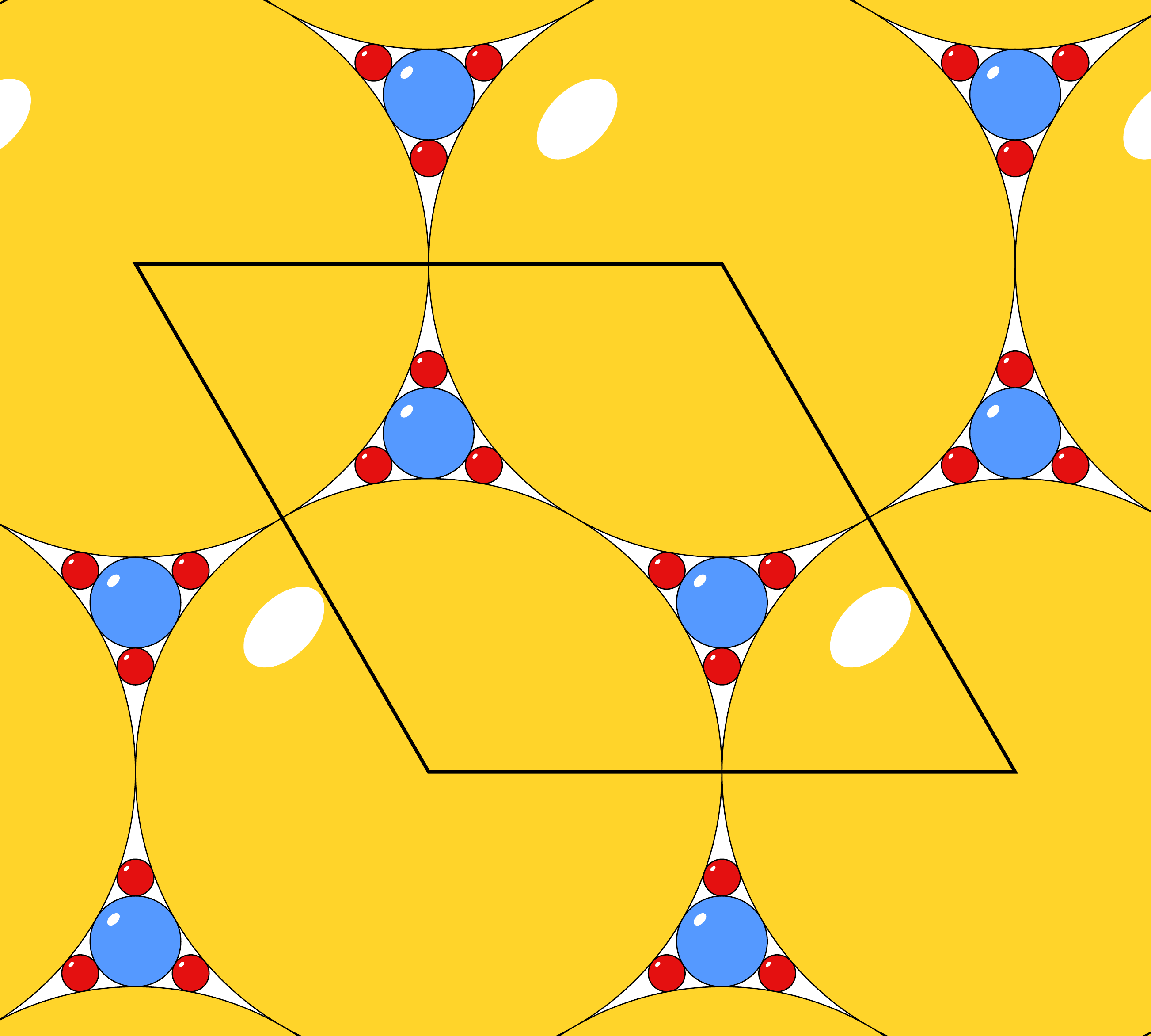} &
  \includegraphics[width=0.3\textwidth]{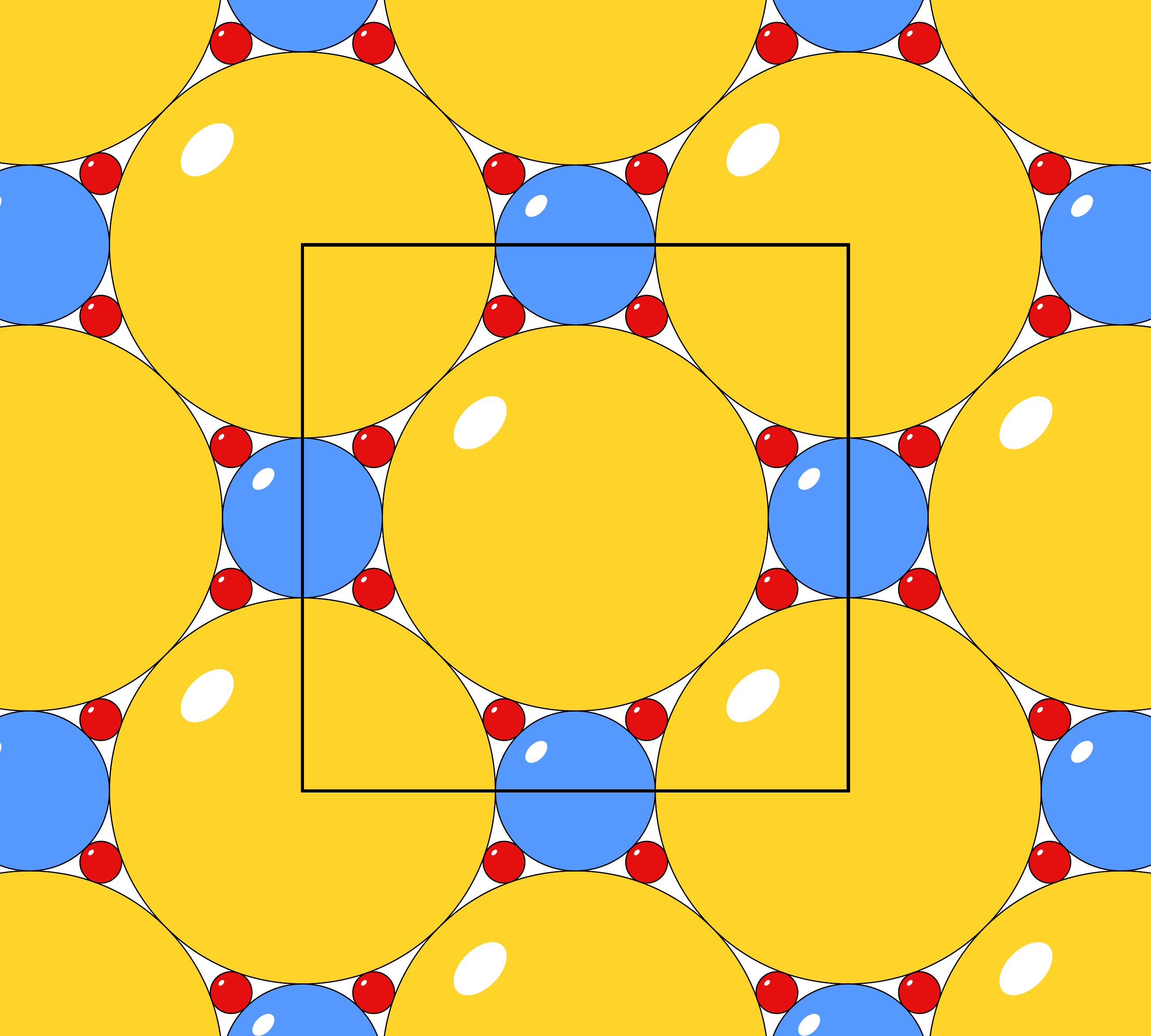}
\end{tabular}
\noindent
\begin{tabular}{lll}
  28\hfill 1rr / 1111srs & 29\hfill 1rr / 111srrs & 30\hfill 1rr / 111srs\\
  \includegraphics[width=0.3\textwidth]{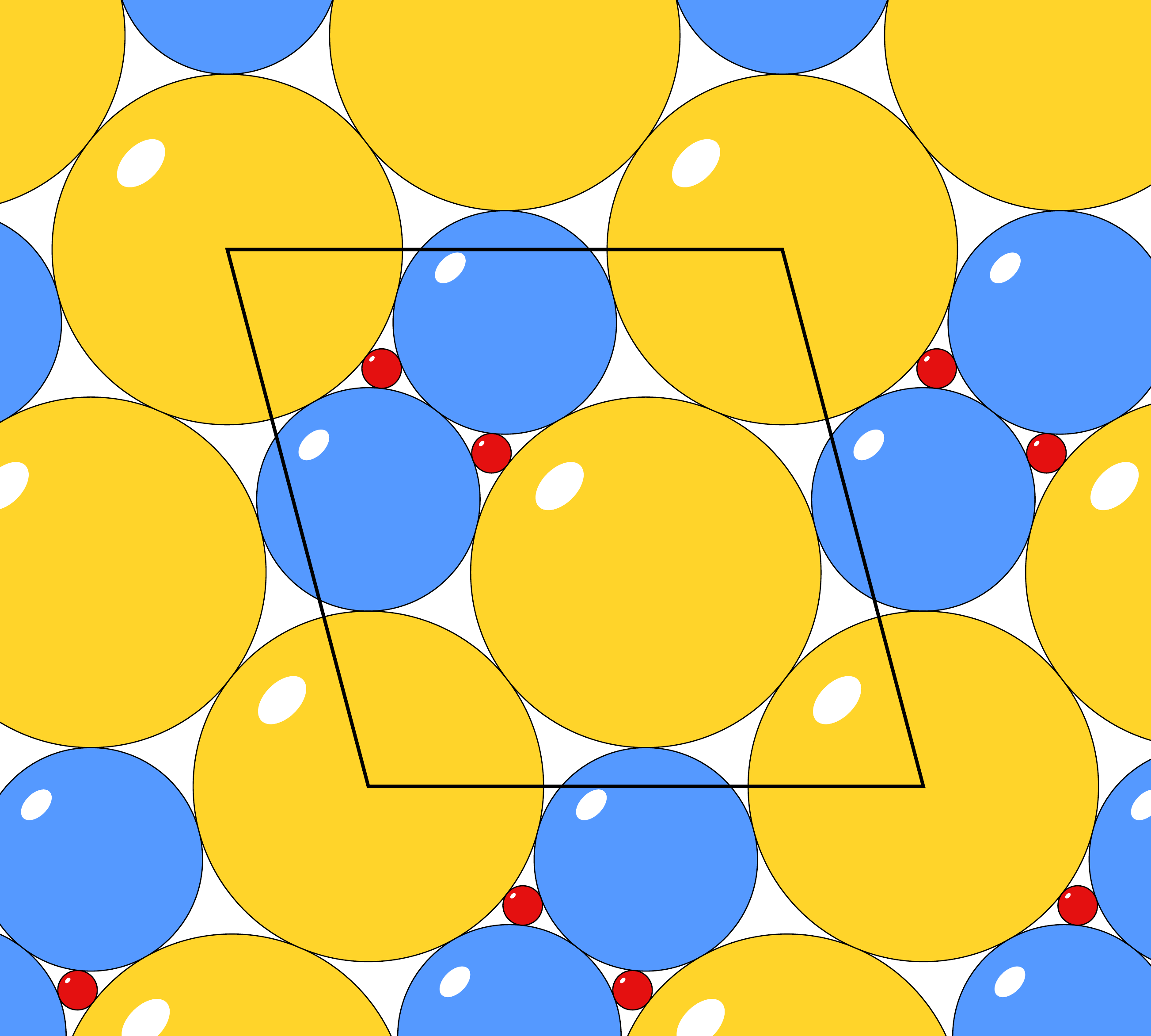} &
  \includegraphics[width=0.3\textwidth]{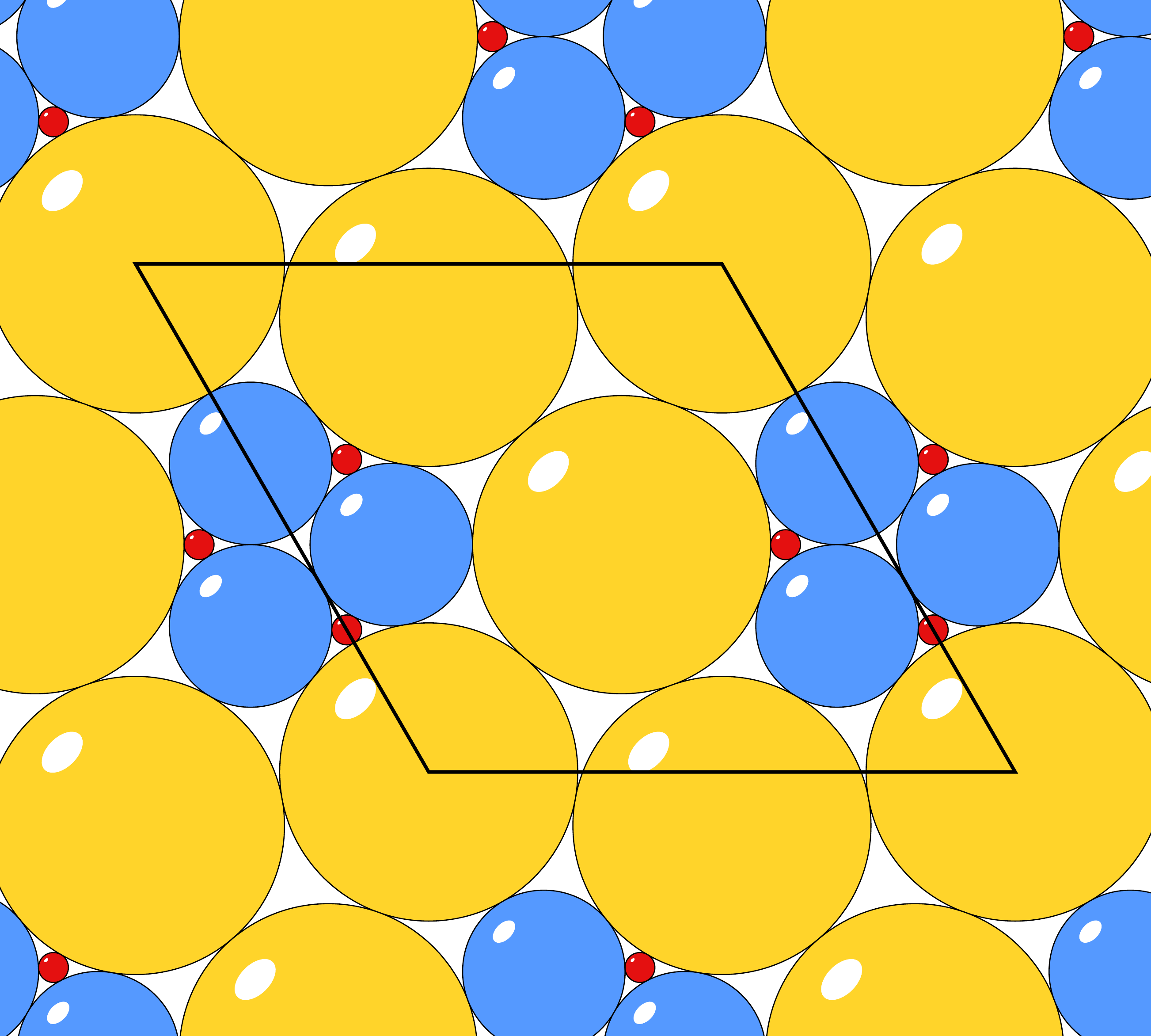} &
  \includegraphics[width=0.3\textwidth]{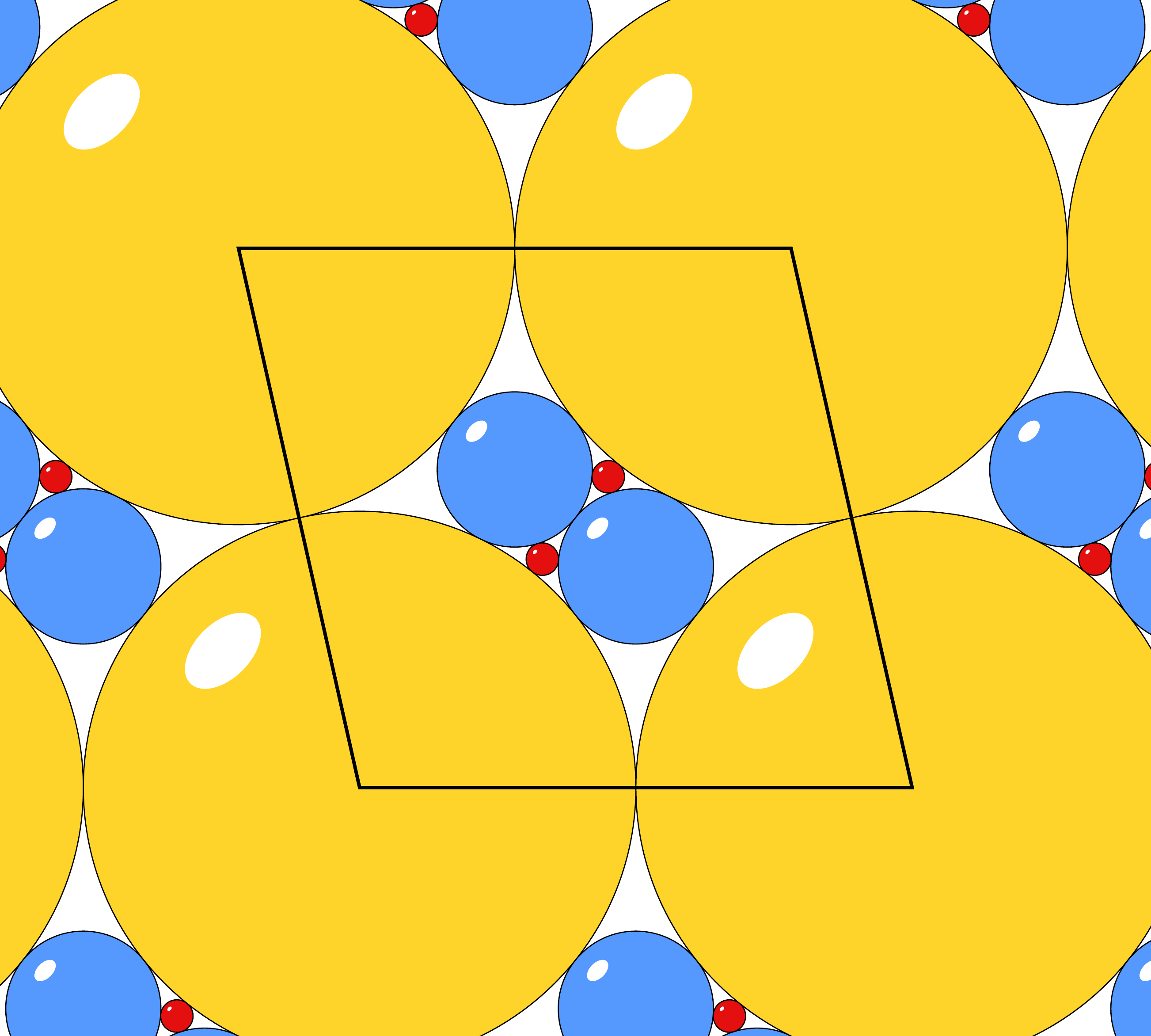}
\end{tabular}
\noindent
\begin{tabular}{lll}
  31\hfill 1rr / 11srrrs & 32\hfill 1rr / 11srrs & 33\hfill 1rr / 11srs1srs\\
  \includegraphics[width=0.3\textwidth]{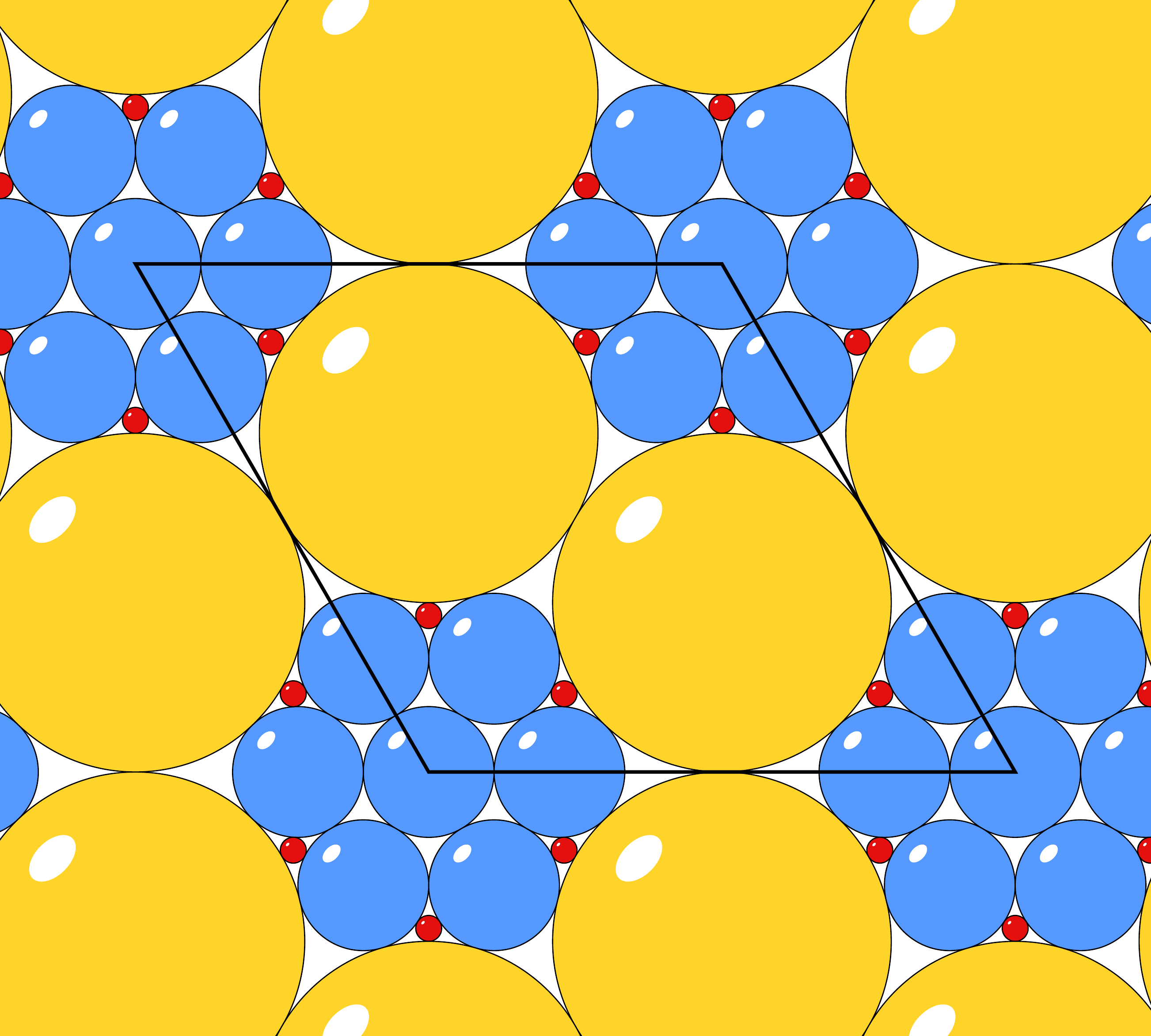} &
  \includegraphics[width=0.3\textwidth]{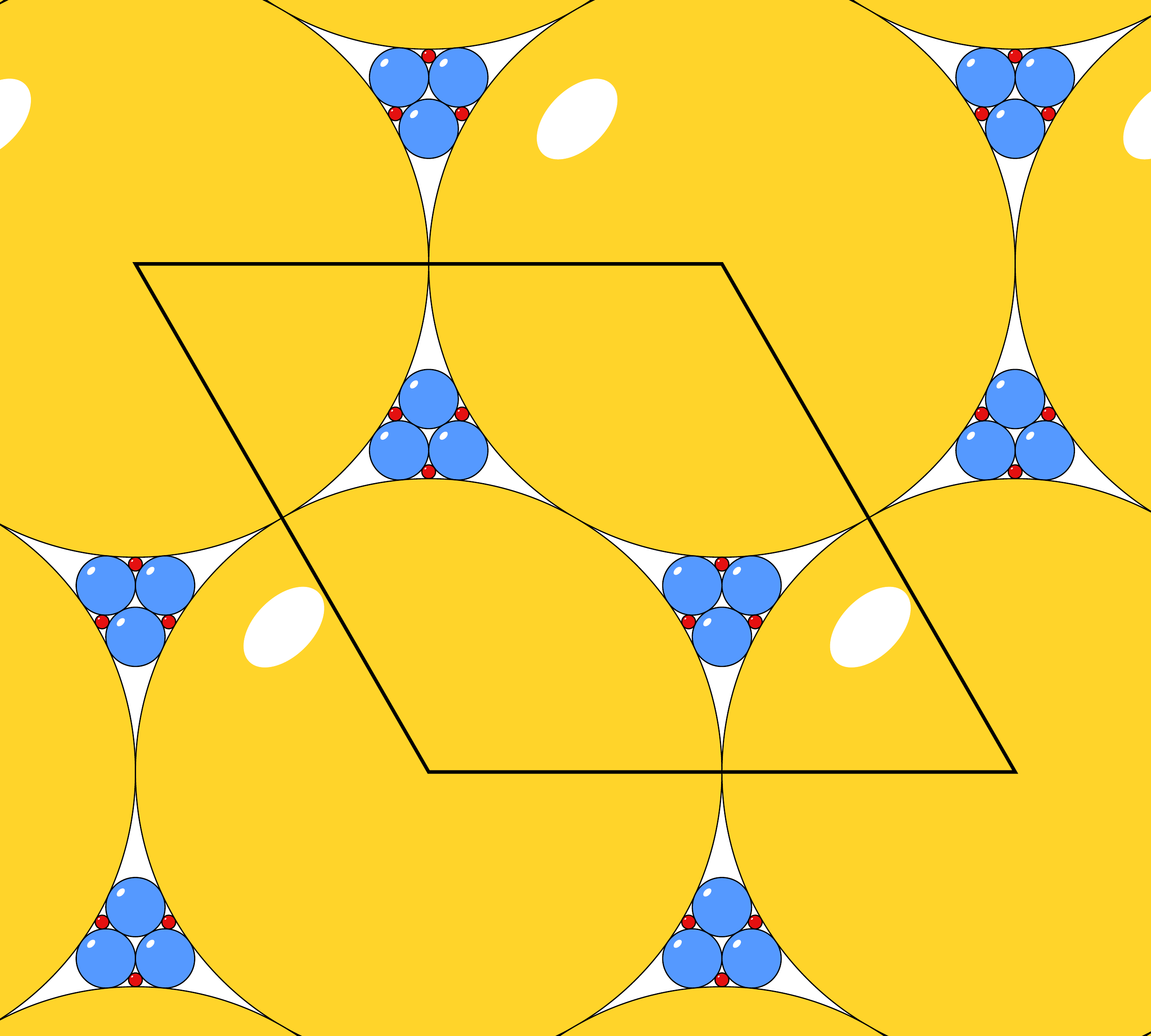} &
  \includegraphics[width=0.3\textwidth]{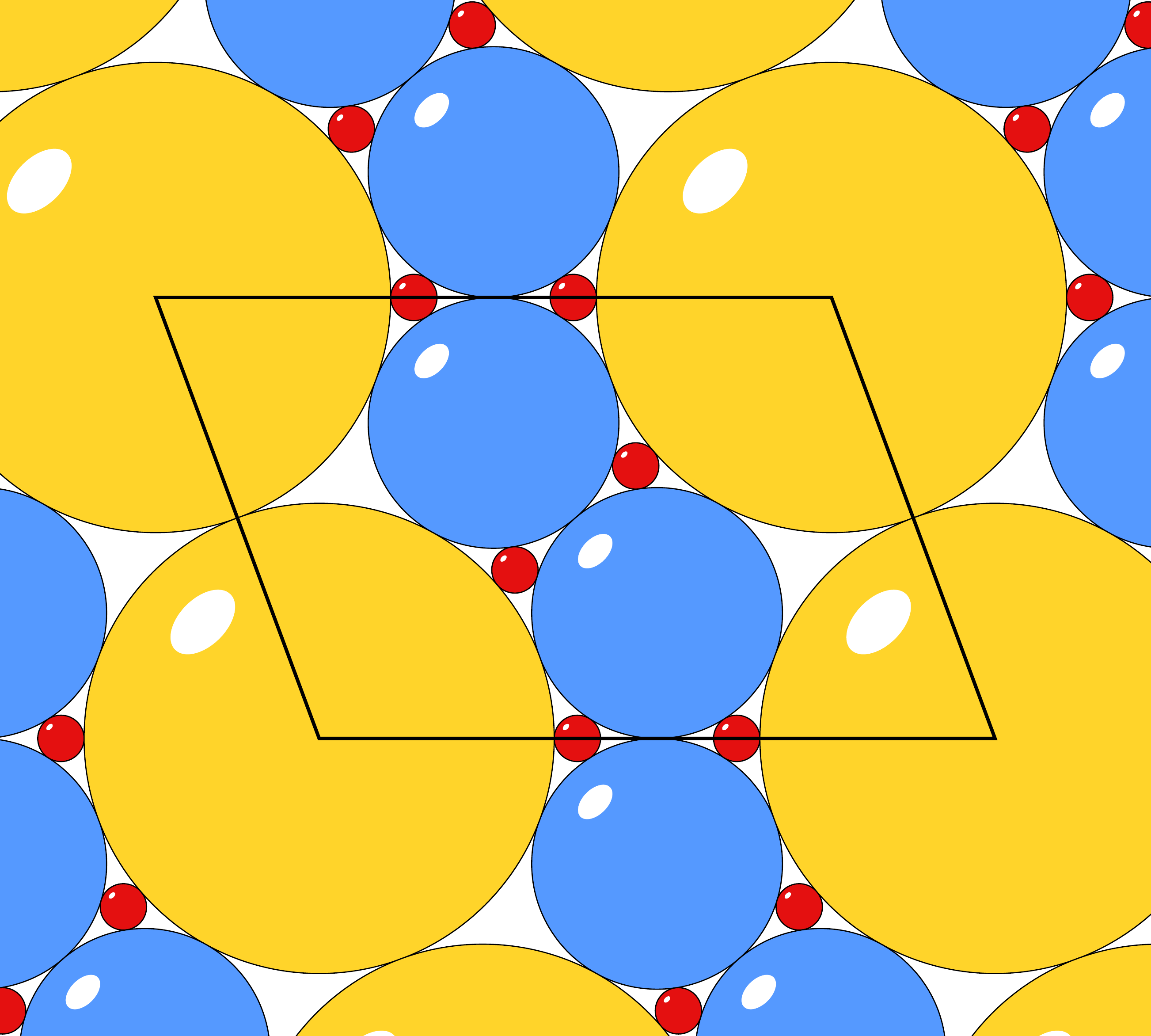}
\end{tabular}
\noindent
\begin{tabular}{lll}
  34\hfill 1rr / 1srrs1srs & 35\hfill 1rssr / 11ss & 36\hfill 1rsss / 11ss\\
  \includegraphics[width=0.3\textwidth]{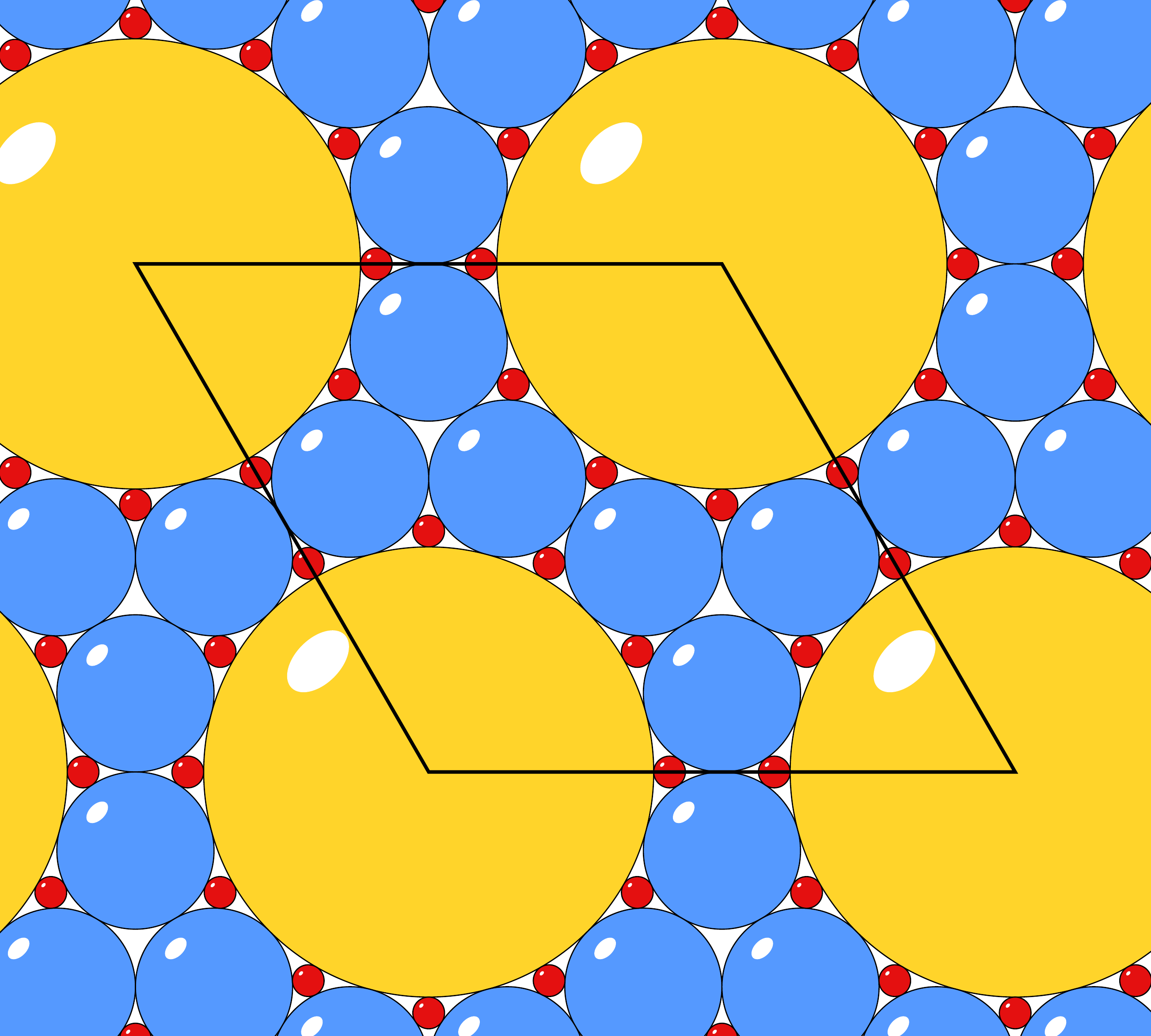} &
  \includegraphics[width=0.3\textwidth]{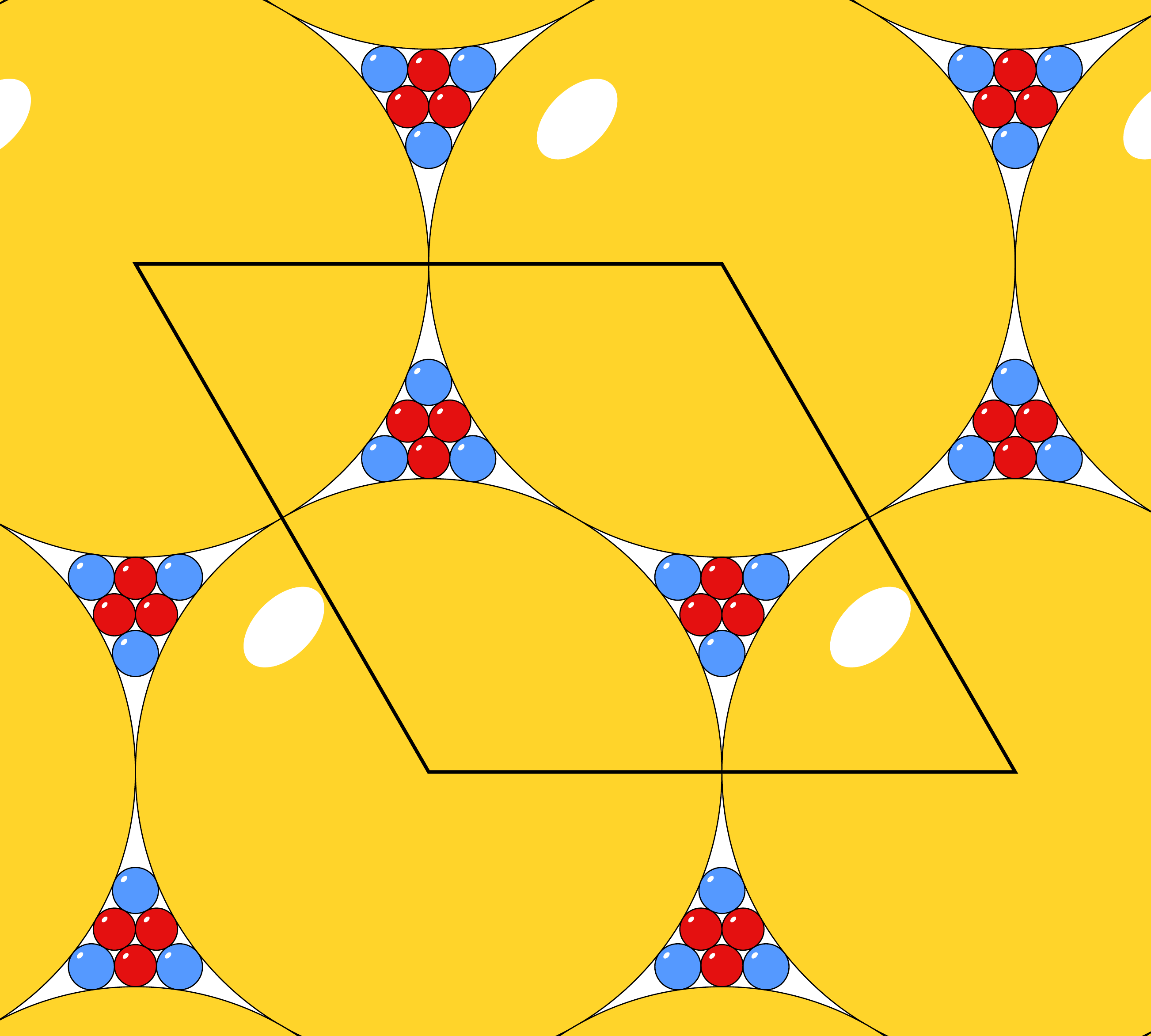} &
  \includegraphics[width=0.3\textwidth]{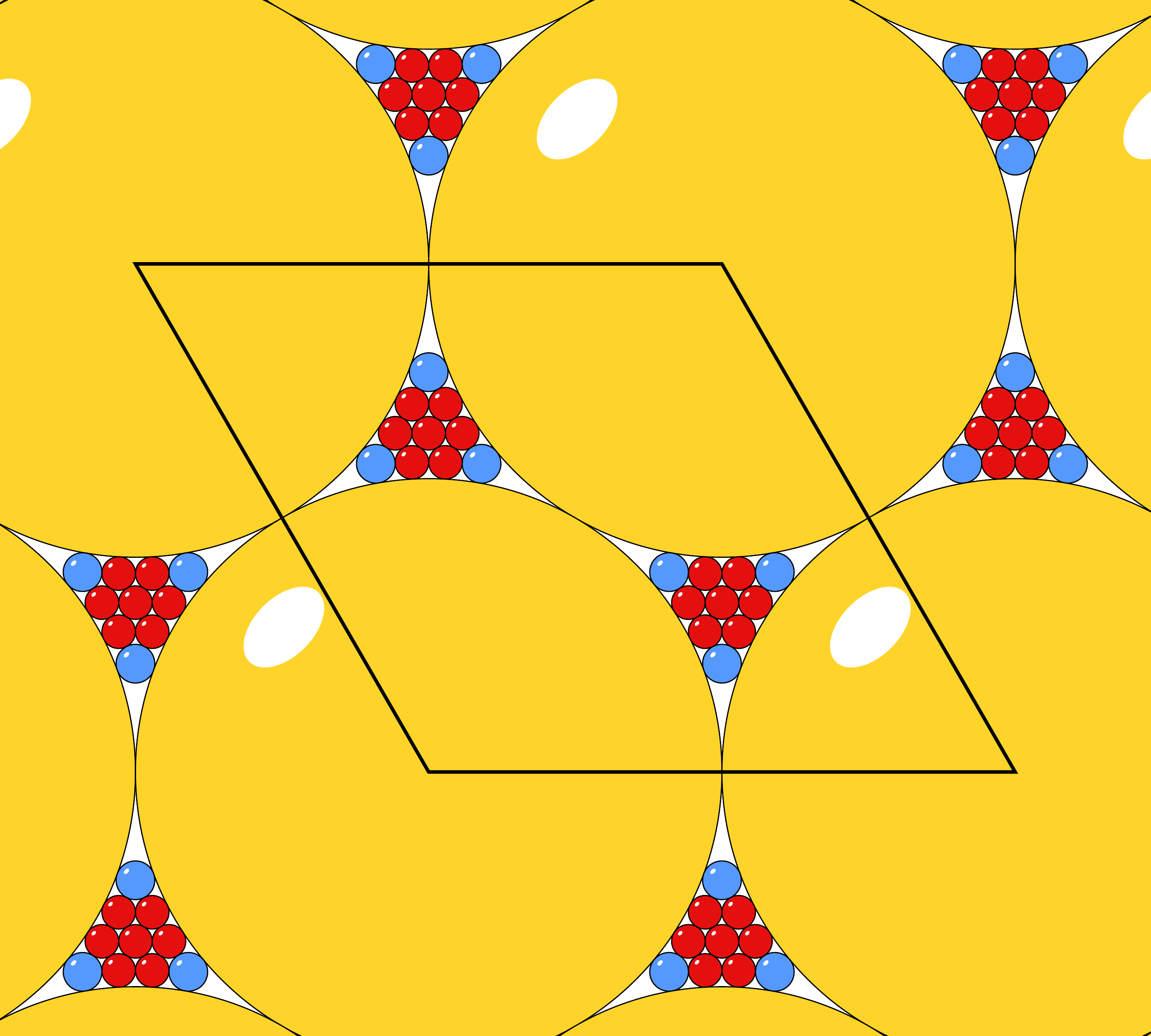}
\end{tabular}
\noindent
\begin{tabular}{lll}
  37\hfill rrr / 111rsr & 38\hfill rrr / 11rsr & 39\hfill rrr / 11rsrsr\\
  \includegraphics[width=0.3\textwidth]{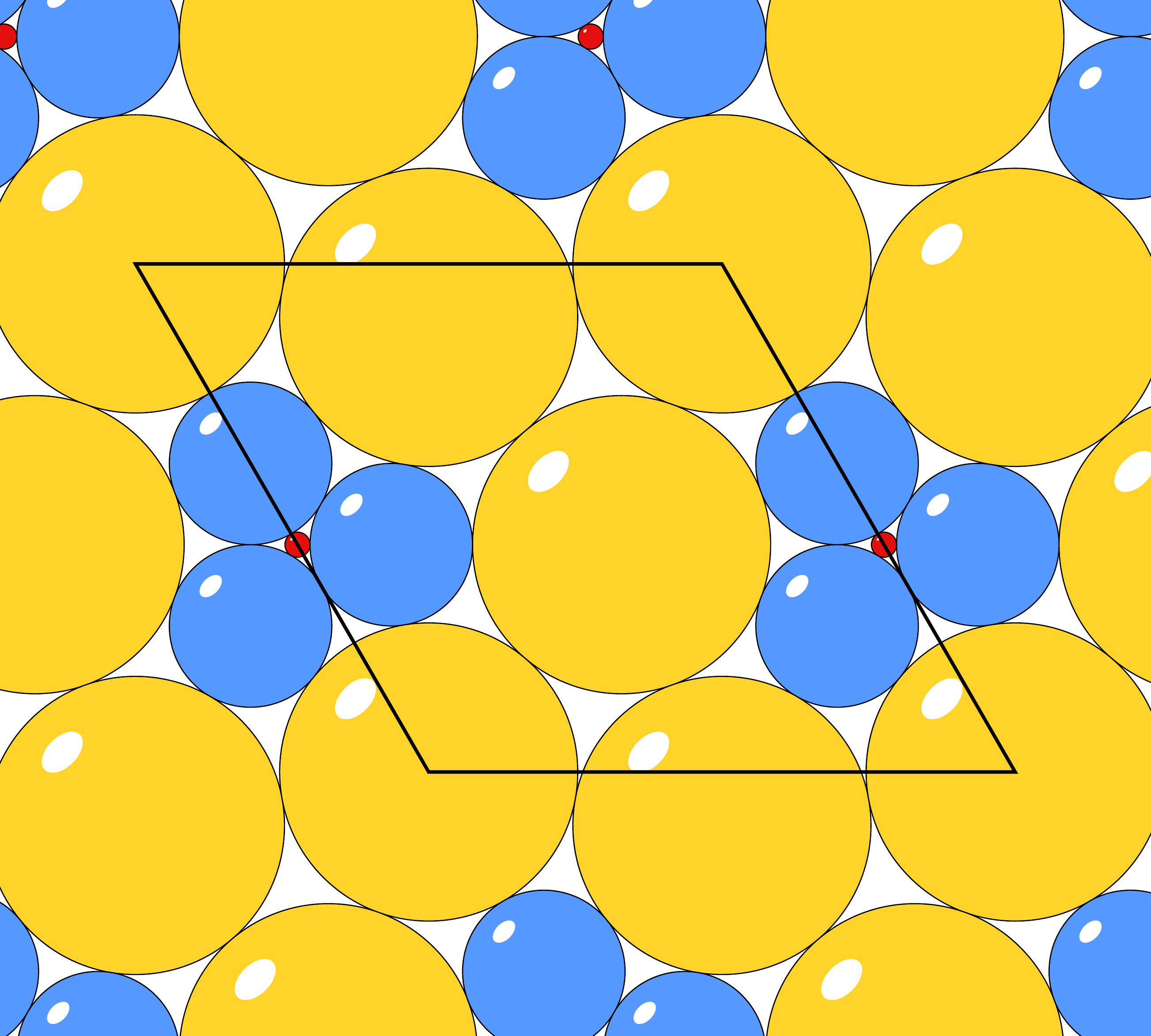} &
  \includegraphics[width=0.3\textwidth]{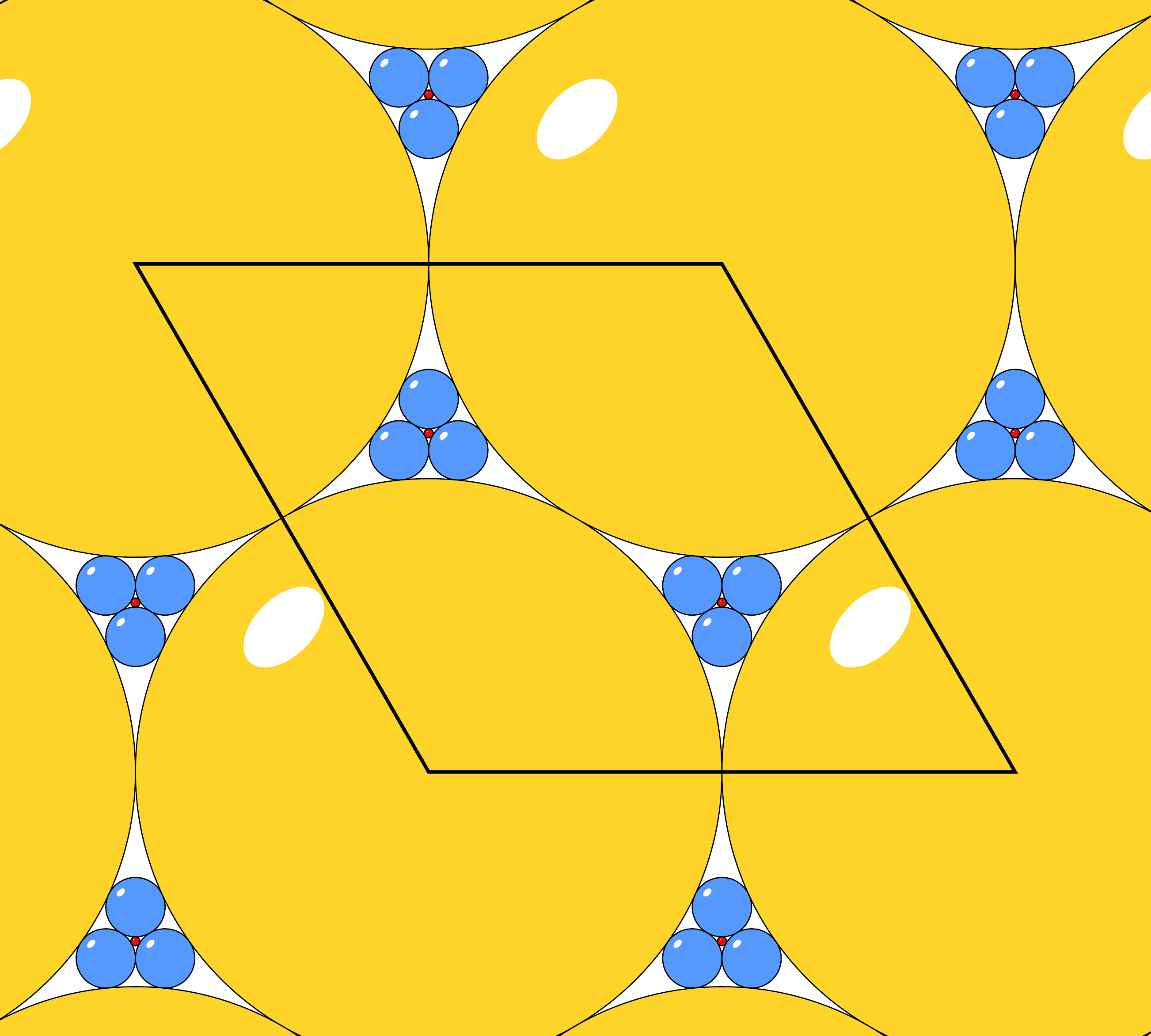} &
  \includegraphics[width=0.3\textwidth]{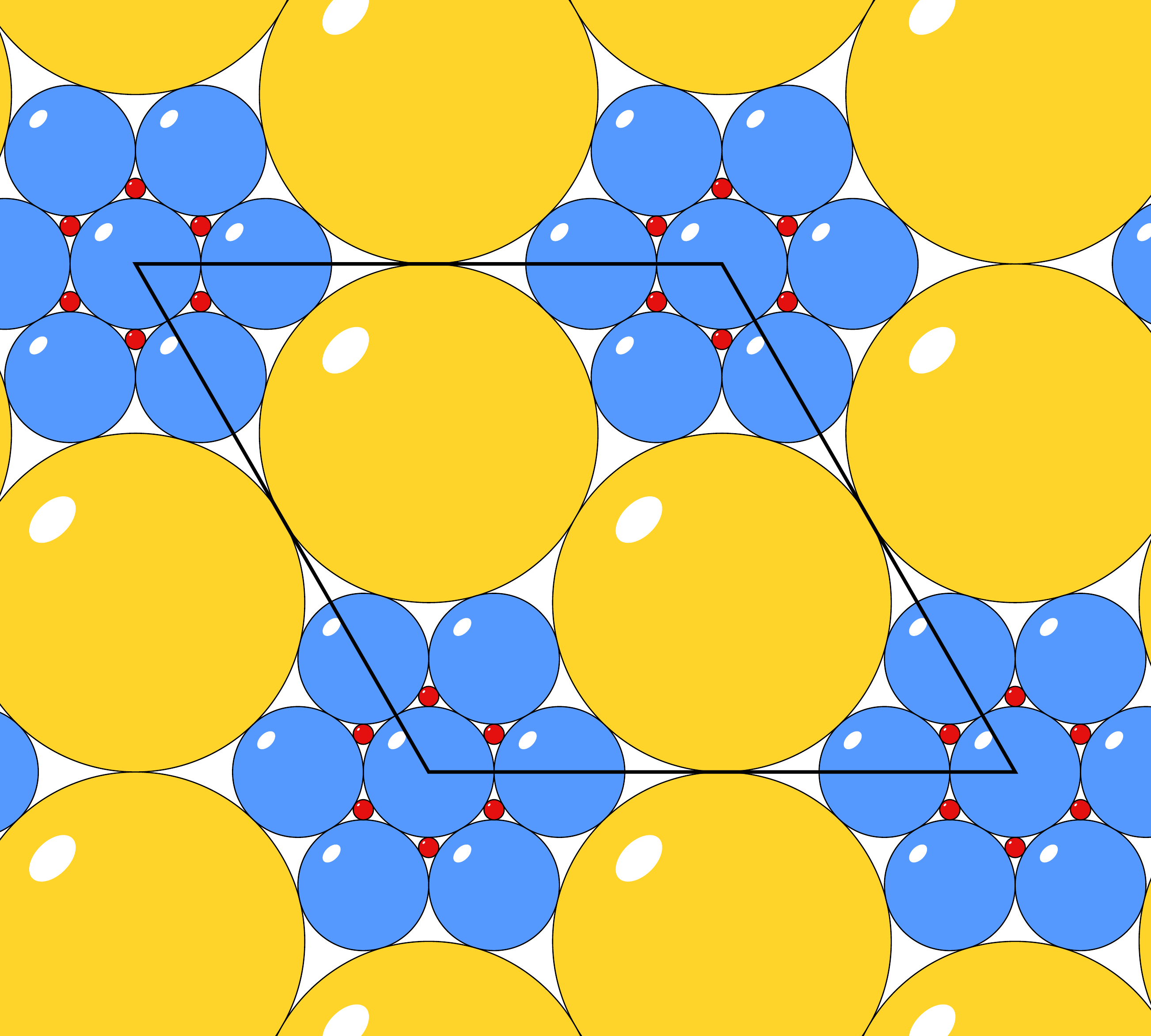}
\end{tabular}
\noindent
\begin{tabular}{lll}
  40\hfill rrr / 1r1rsr & 41\hfill rrss / 111rssr & 42\hfill rrss / 11rssr\\
  \includegraphics[width=0.3\textwidth]{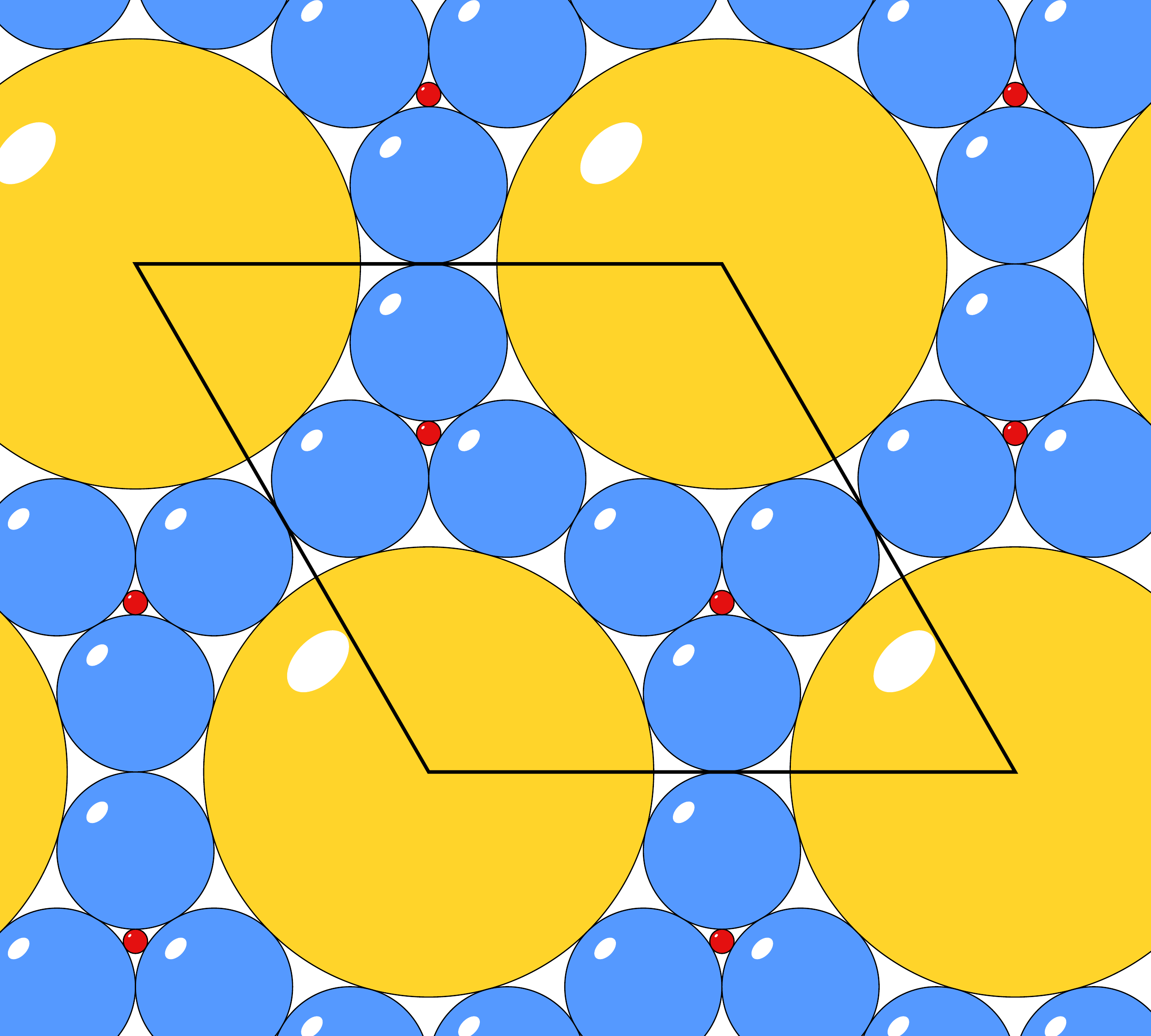} &
  \includegraphics[width=0.3\textwidth]{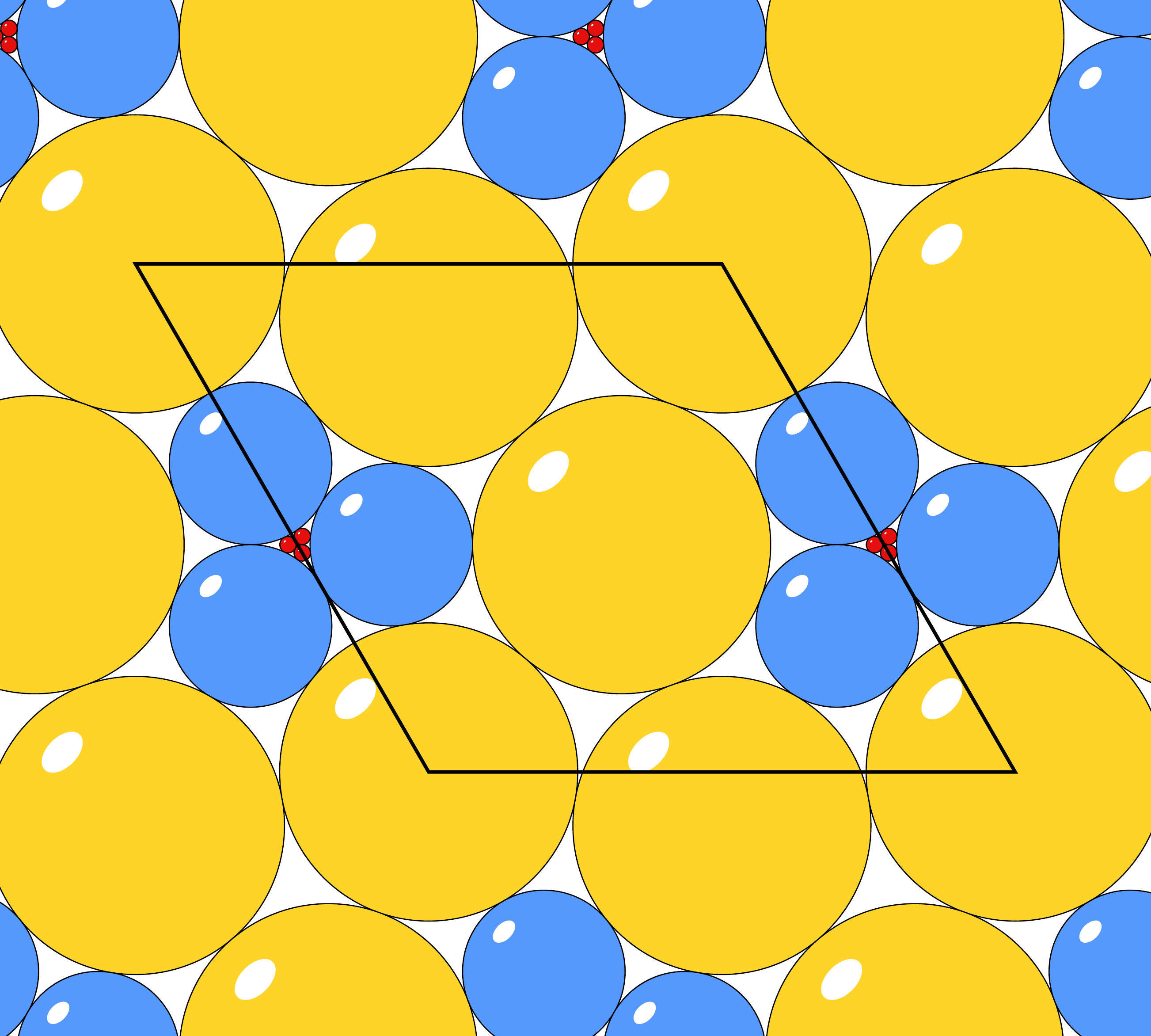} &
  \includegraphics[width=0.3\textwidth]{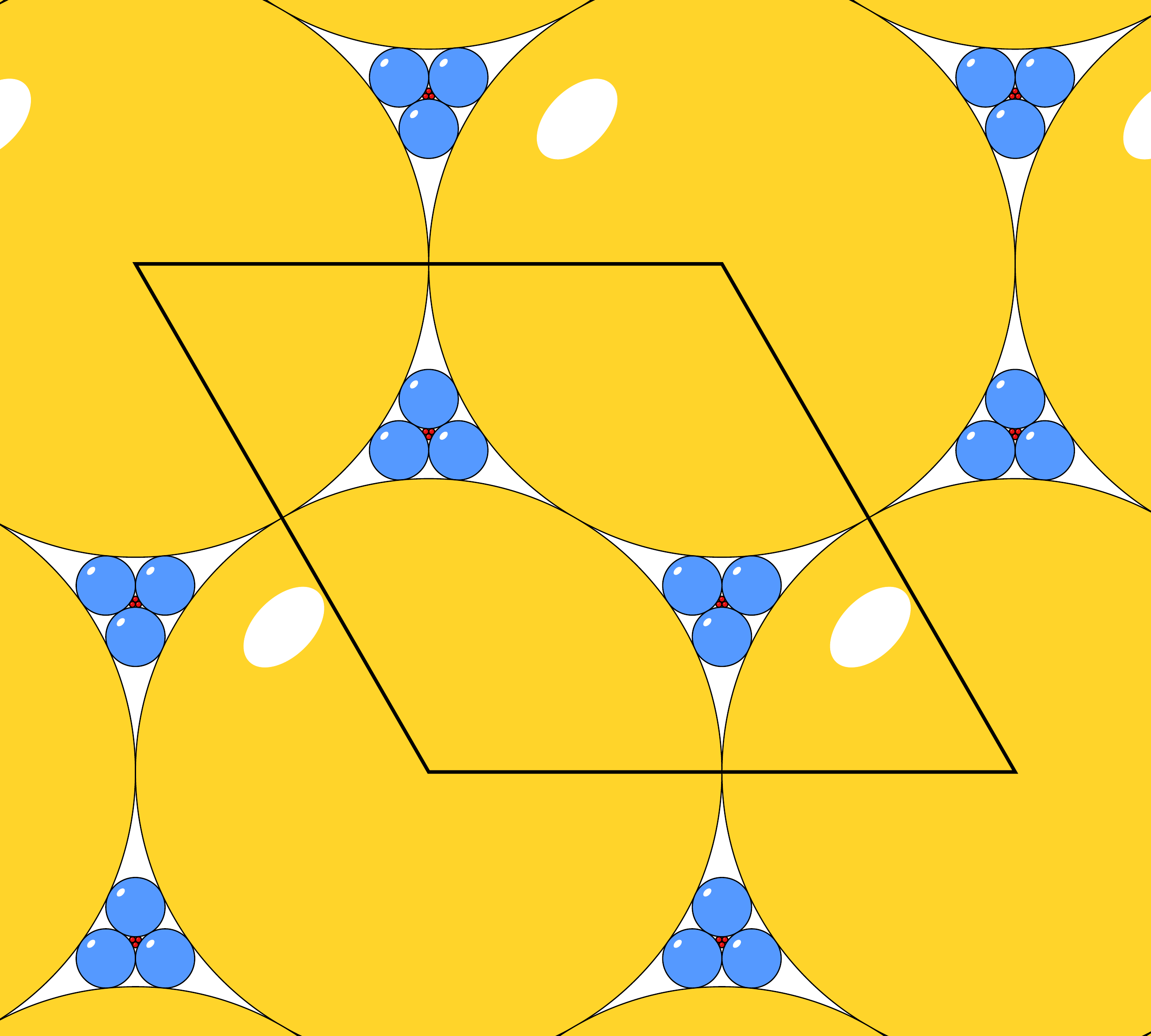}
\end{tabular}
\noindent
\begin{tabular}{lll}
  43\hfill rrss / 11rssrssr & 44\hfill rrss / 1r1rssr & 45\hfill 111r / 111s1s\\
  \includegraphics[width=0.3\textwidth]{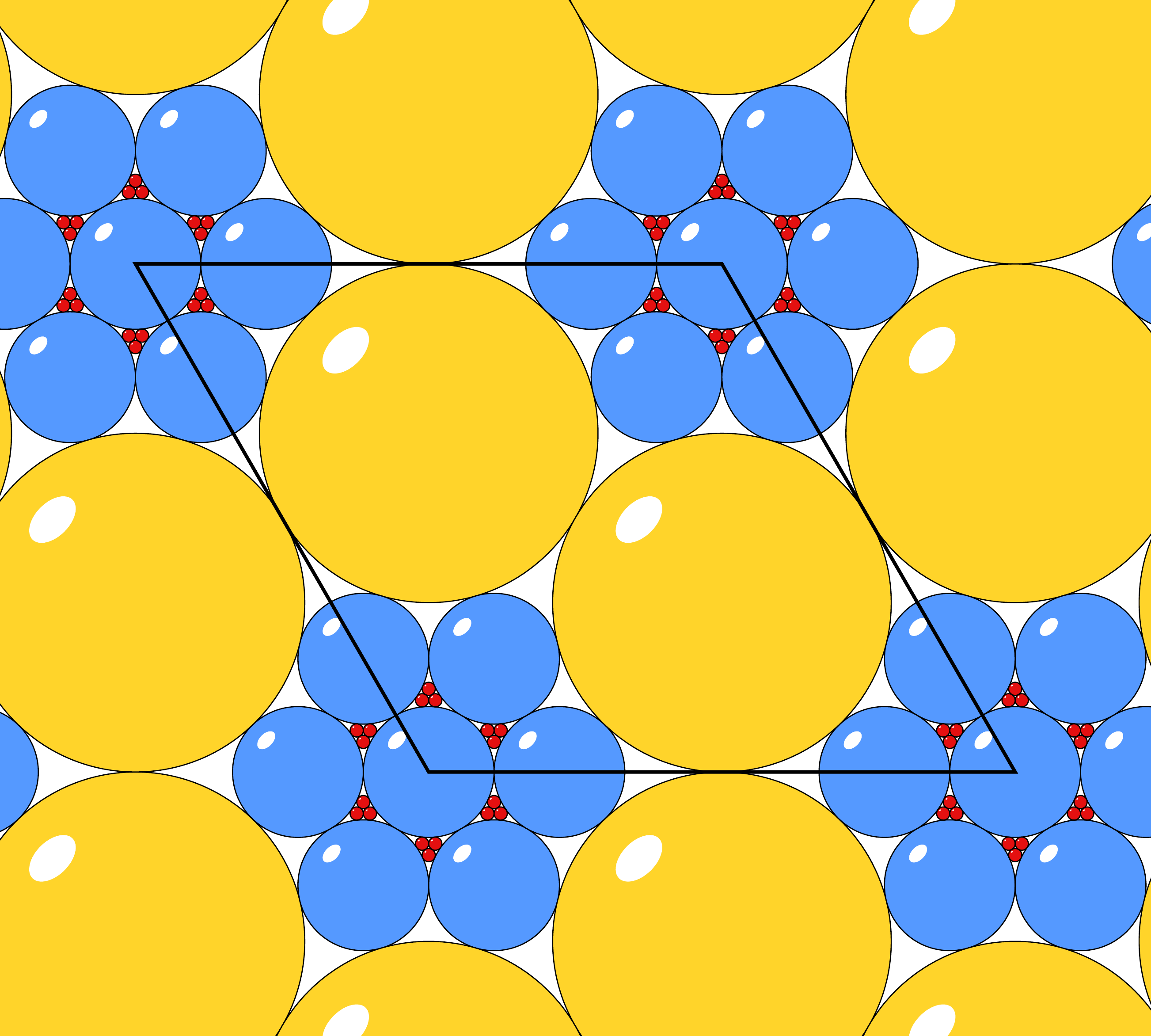} &
  \includegraphics[width=0.3\textwidth]{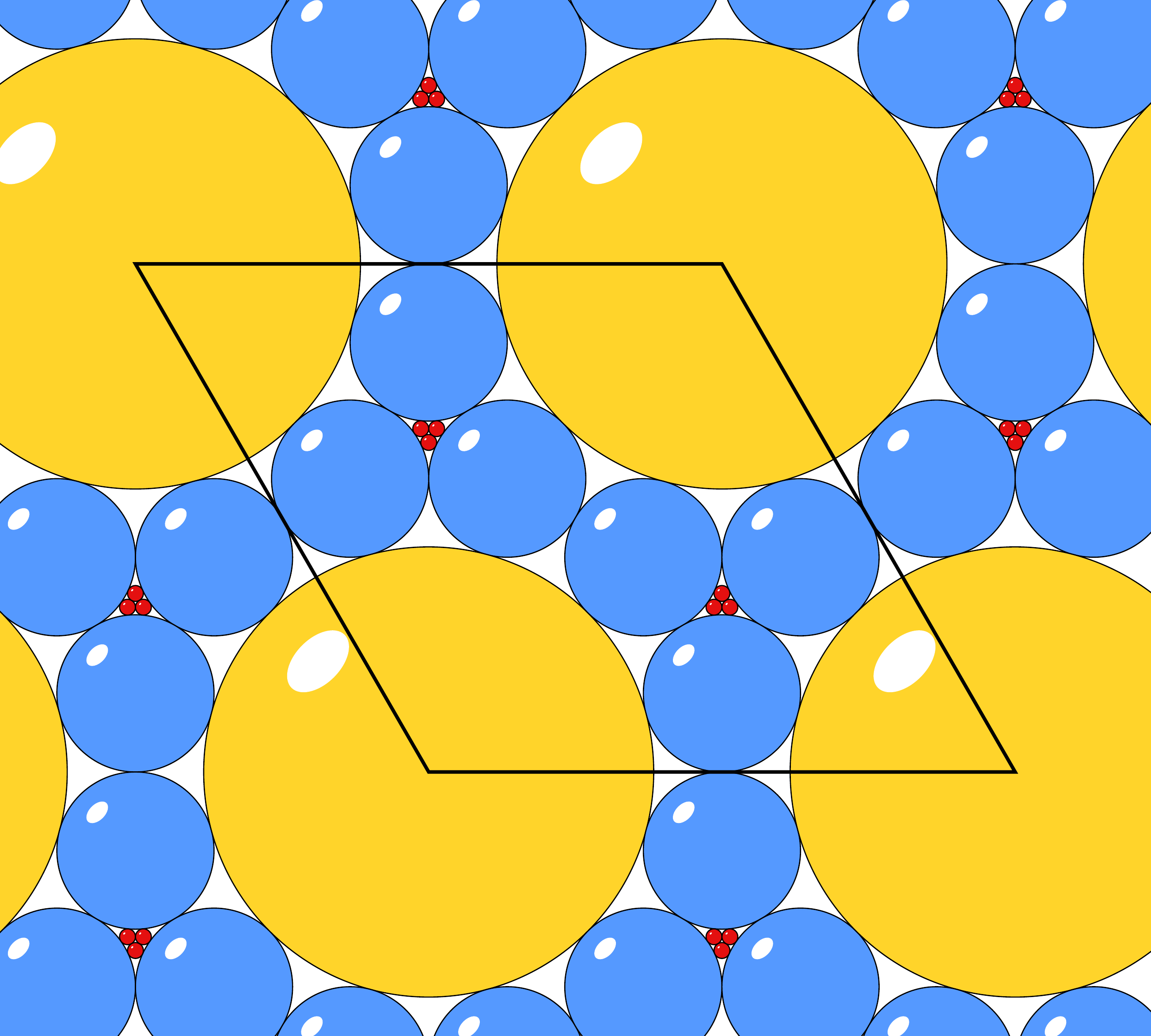} &
  \includegraphics[width=0.3\textwidth]{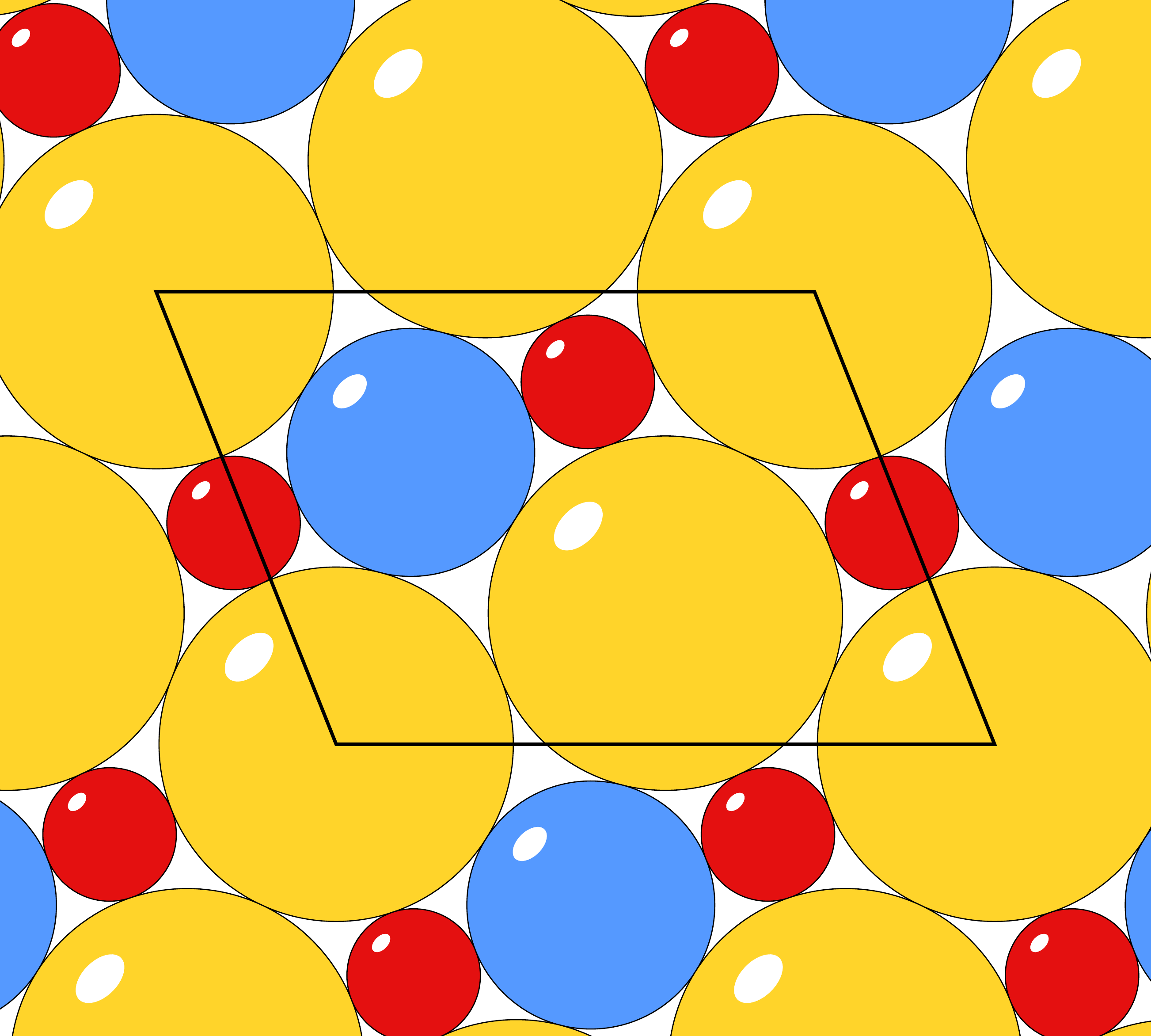}
\end{tabular}
\noindent
\begin{tabular}{lll}
  46\hfill 111r / 11r1s & 47\hfill 111r / 1r1r1s & 48\hfill 111r / 1rr1s\\
  \includegraphics[width=0.3\textwidth]{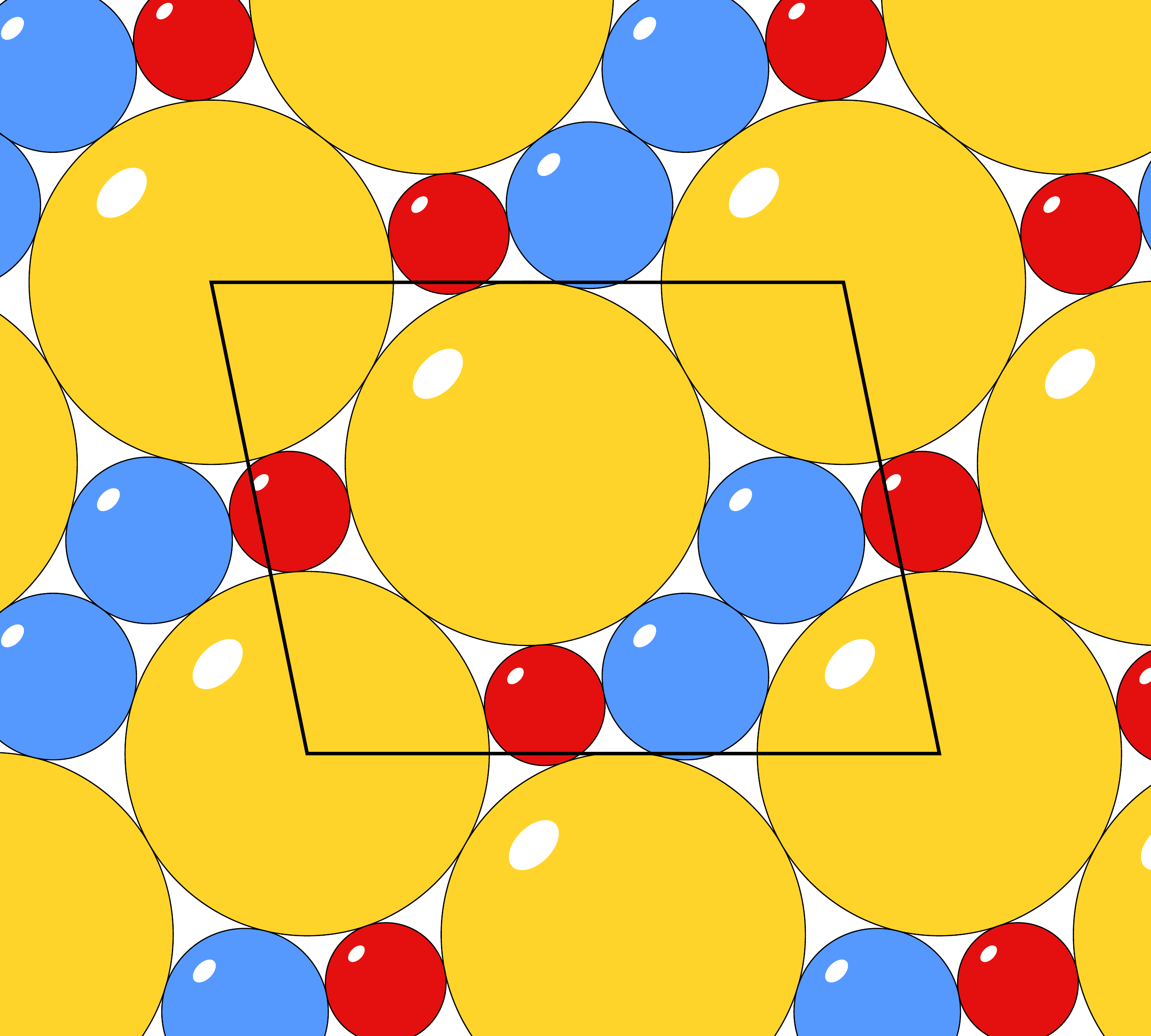} &
  \includegraphics[width=0.3\textwidth]{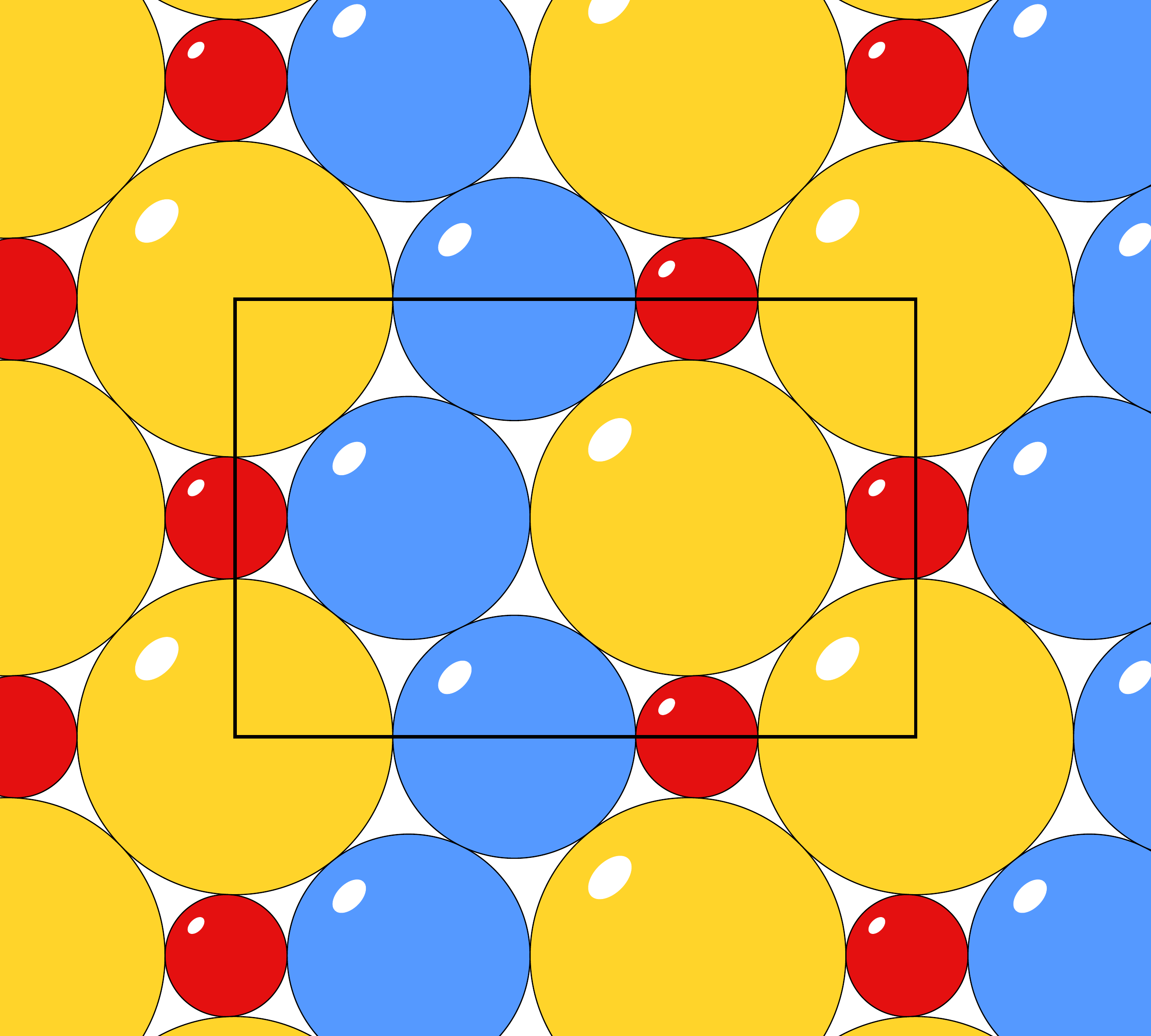} &
  \includegraphics[width=0.3\textwidth]{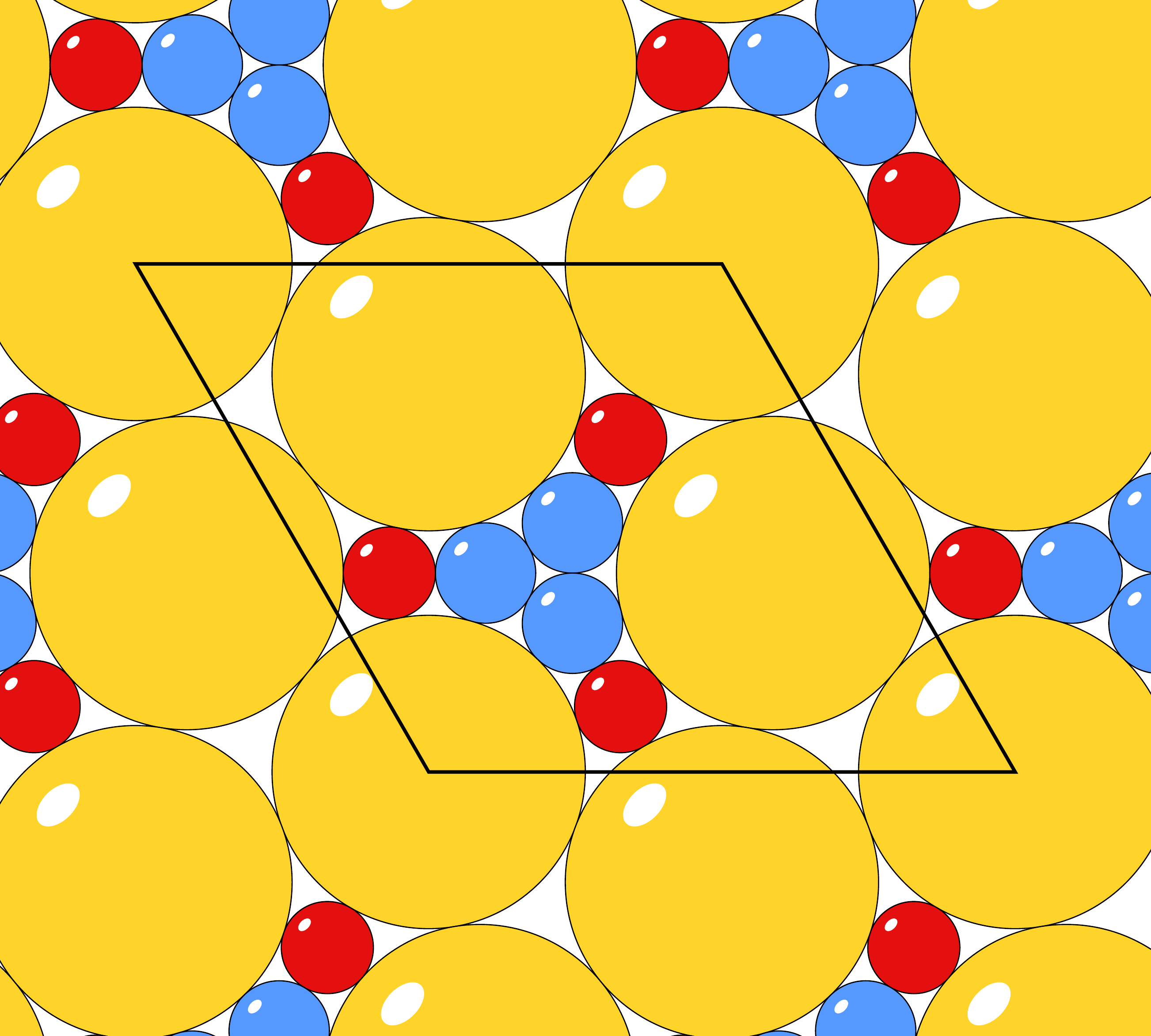}
\end{tabular}
\noindent
\begin{tabular}{lll}
  49\hfill 111r / 1rrr1s & 50\hfill 111r / 1s1s1s & 51\hfill 111rr / 1rrrrs\\
  \includegraphics[width=0.3\textwidth]{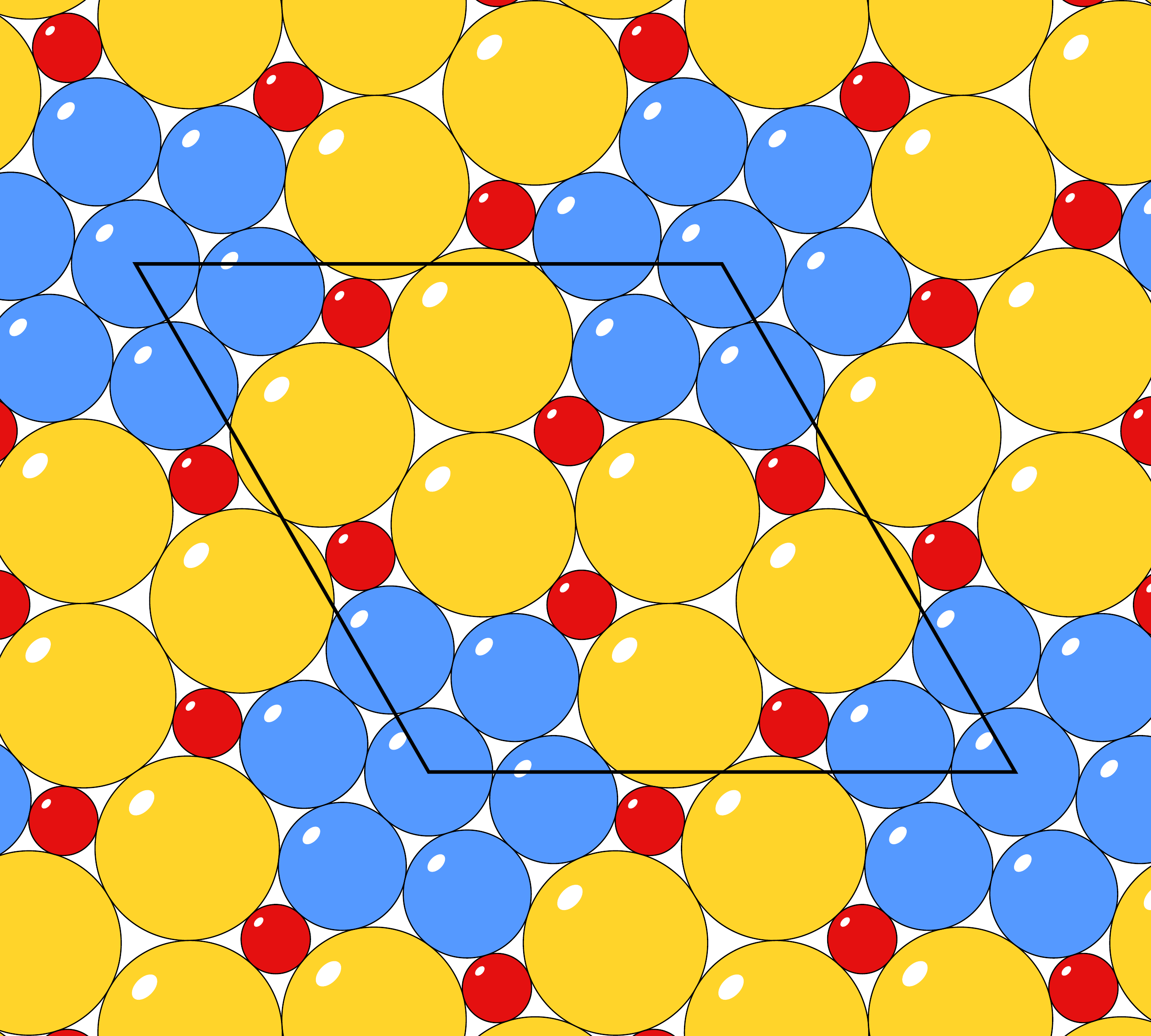} &
  \includegraphics[width=0.3\textwidth]{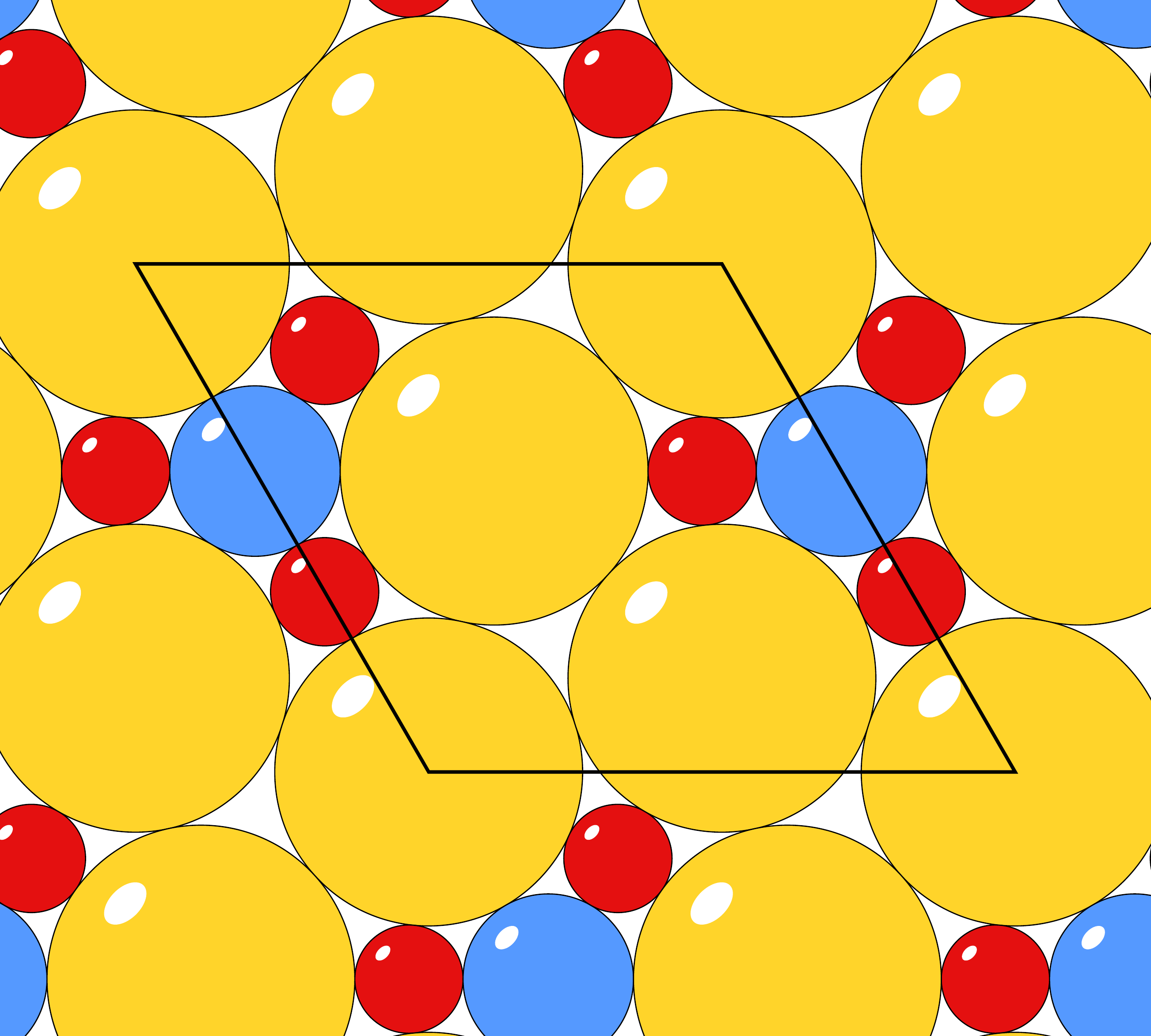} &
  \includegraphics[width=0.3\textwidth]{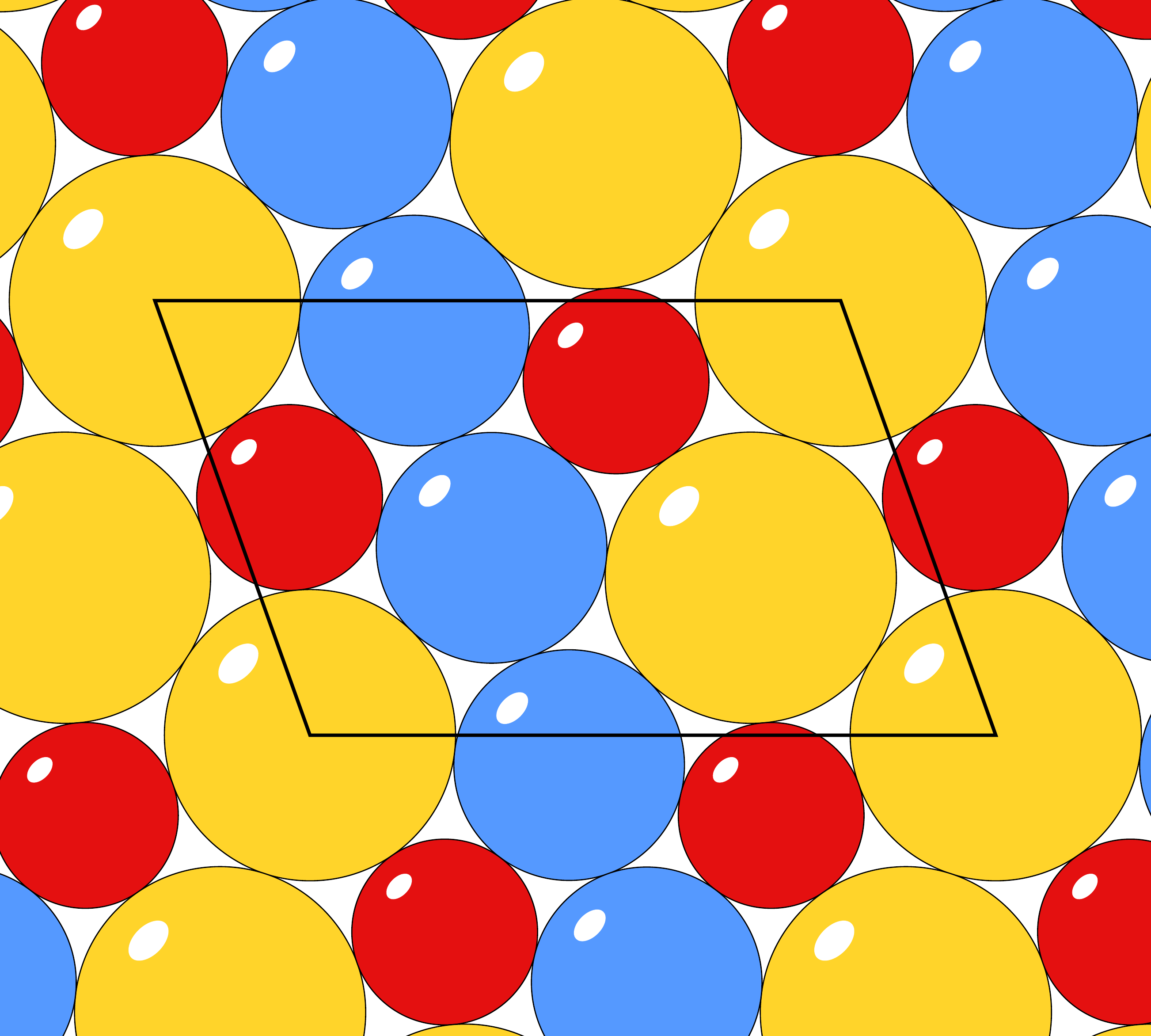}
\end{tabular}
\noindent
\begin{tabular}{lll}
  52\hfill 111rr / 1srrrs & 53\hfill 11r1r / 1r1s1s & 54\hfill 11r1r / 1s1s1s\\
  \includegraphics[width=0.3\textwidth]{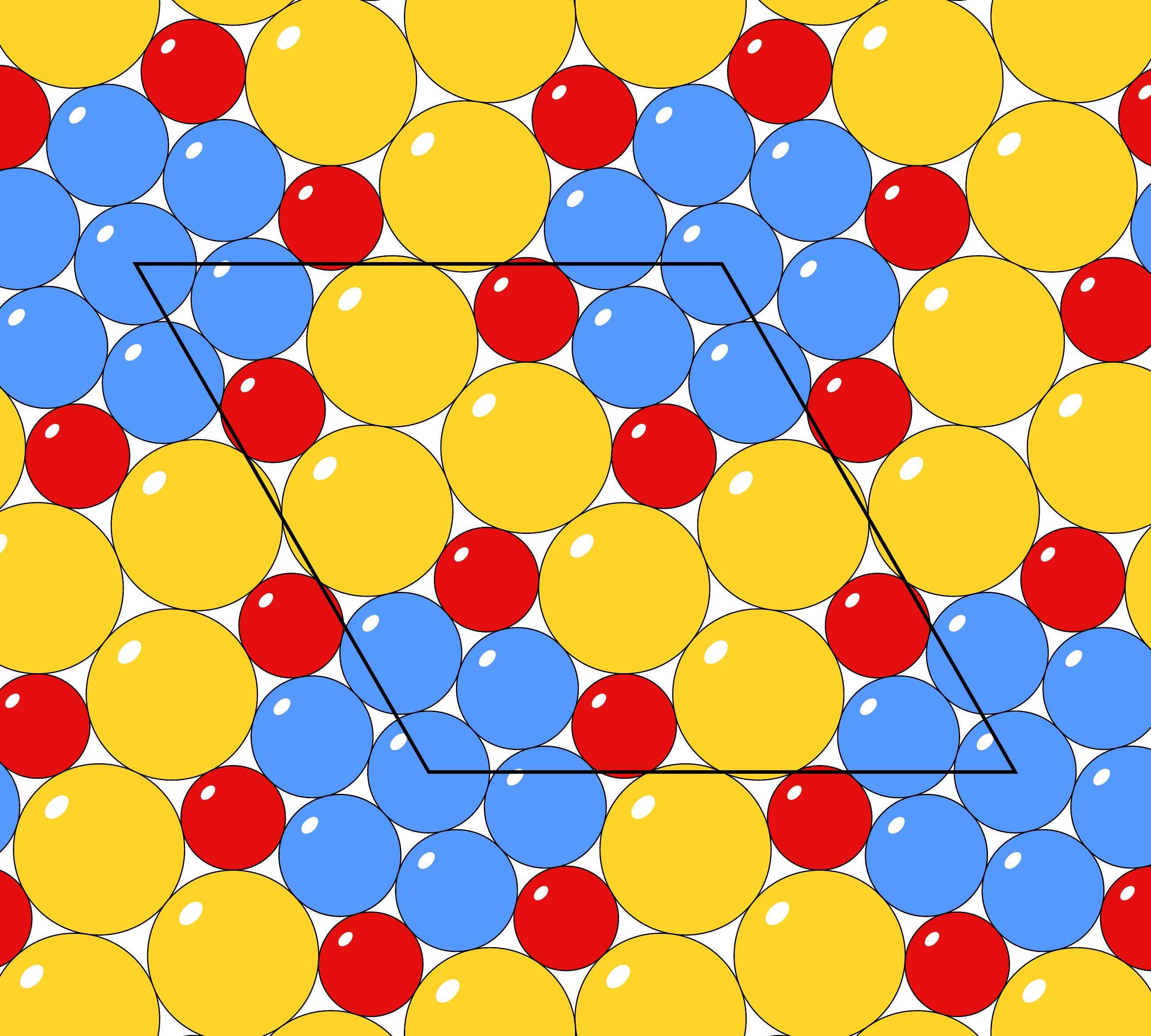} &
  \includegraphics[width=0.3\textwidth]{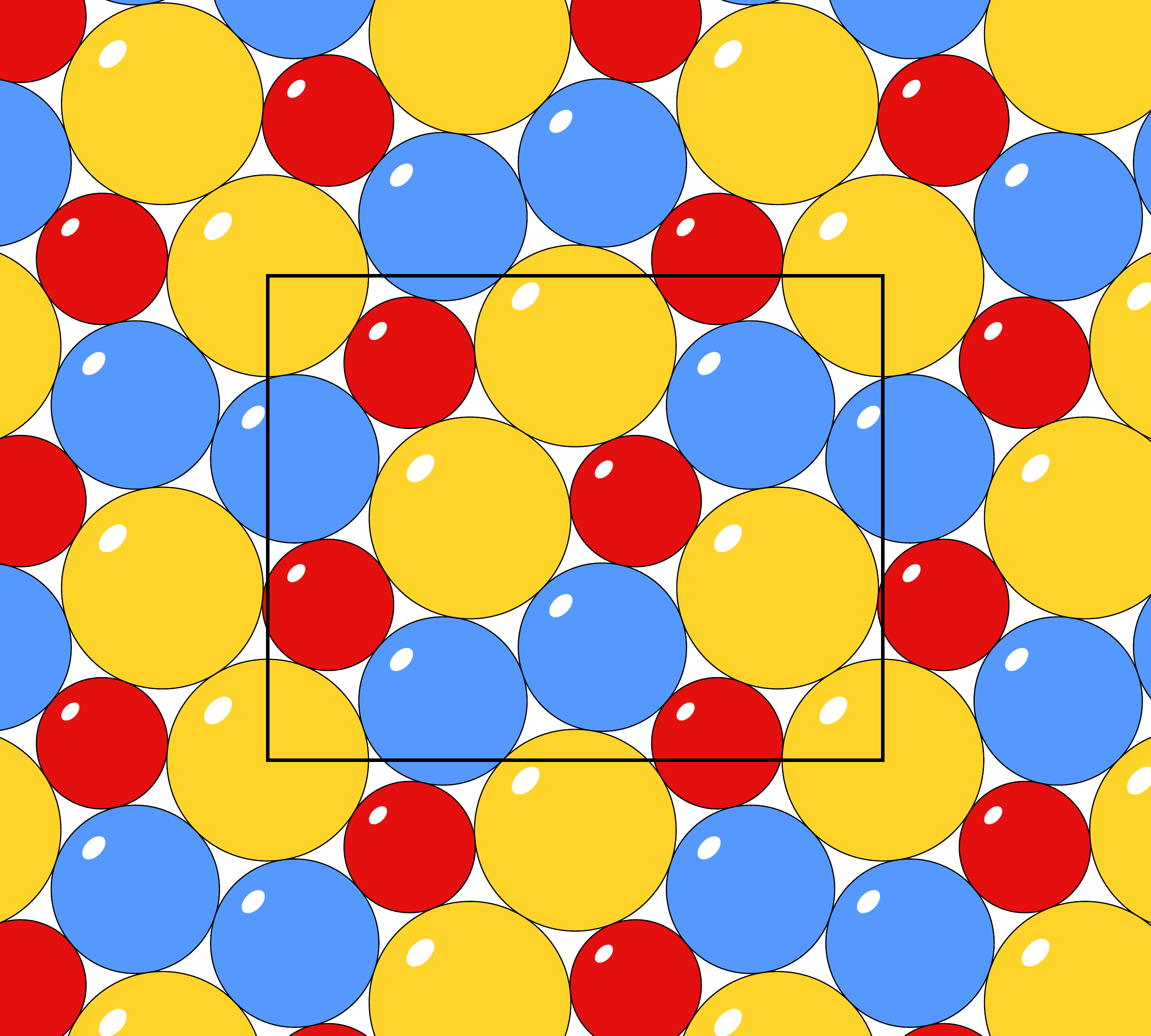} &
  \includegraphics[width=0.3\textwidth]{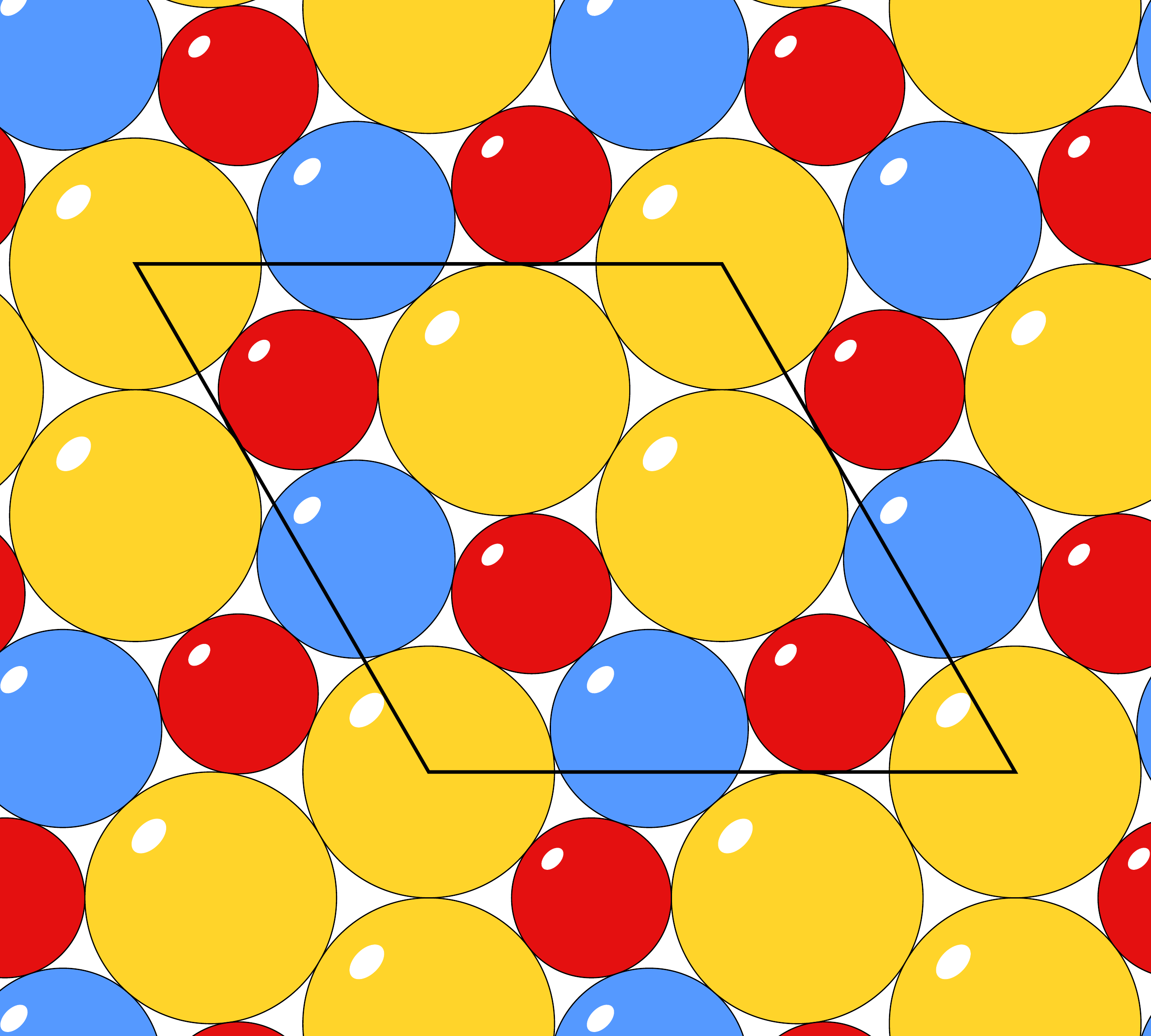}
\end{tabular}
\noindent
\begin{tabular}{lll}
  55\hfill 11r1s / 111s1s & 56\hfill 11r1s / 1r1r1s & 57\hfill 11r1s / 1rrr1s\\
  \includegraphics[width=0.3\textwidth]{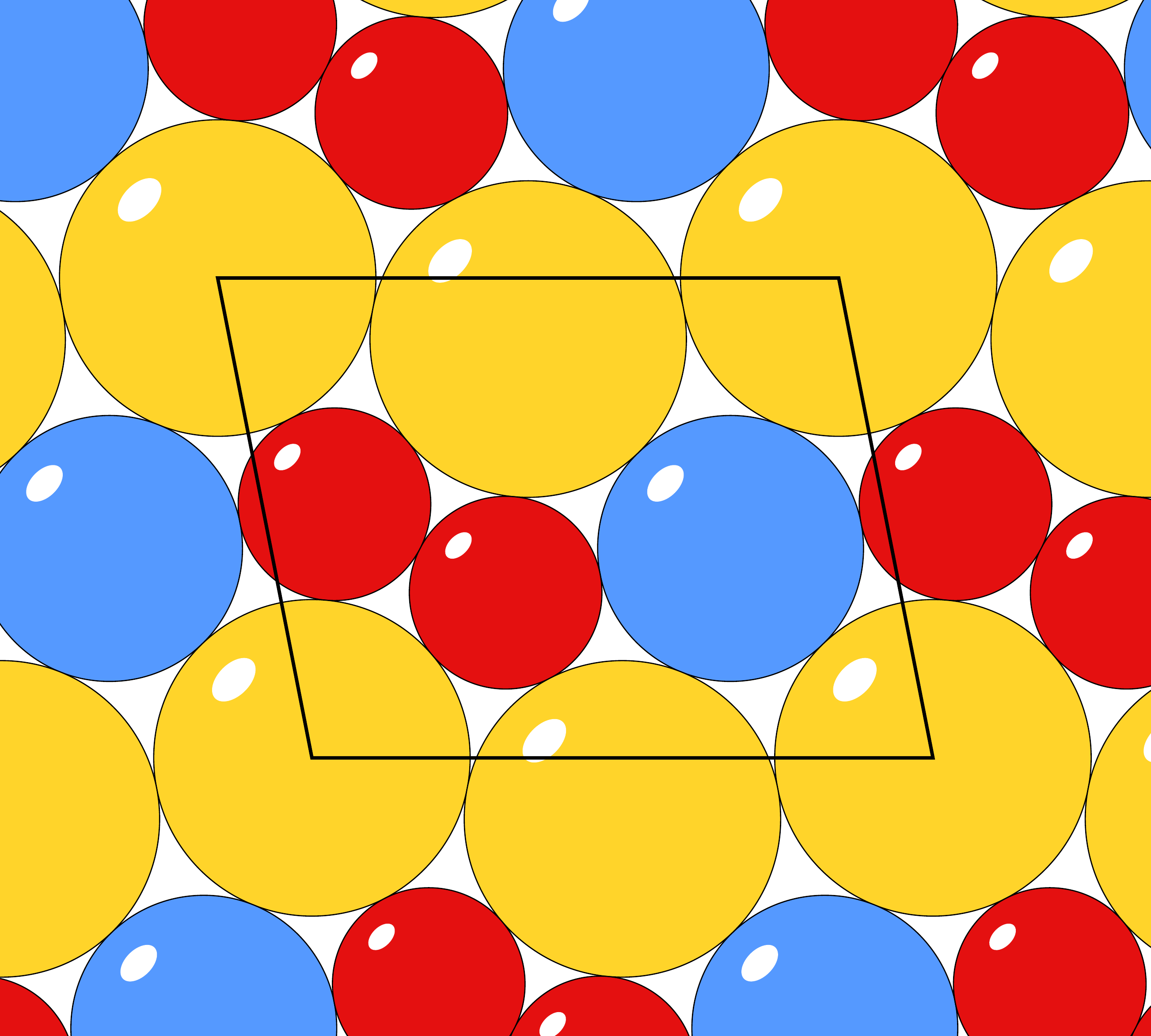} &
  \includegraphics[width=0.3\textwidth]{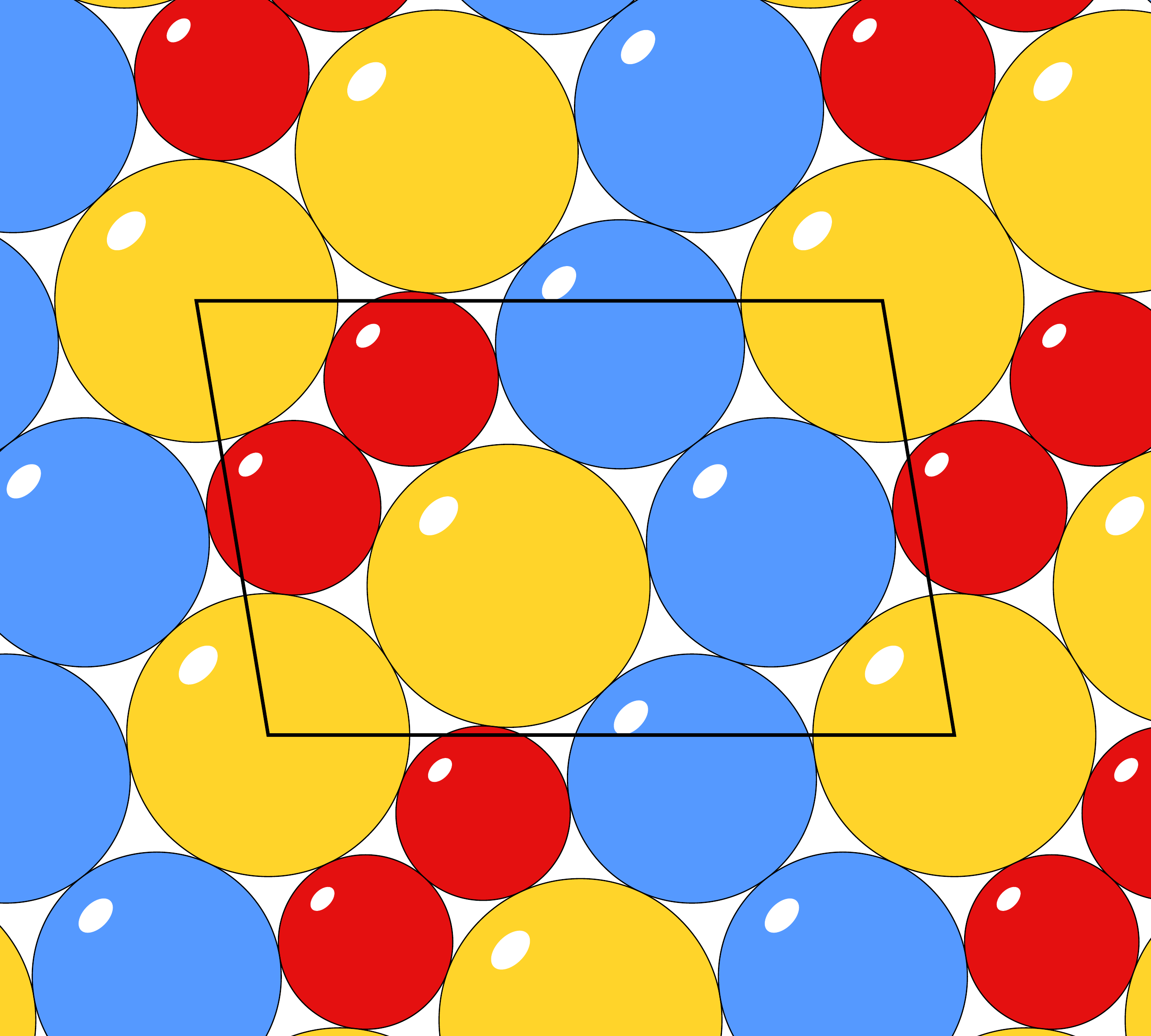} &
  \includegraphics[width=0.3\textwidth]{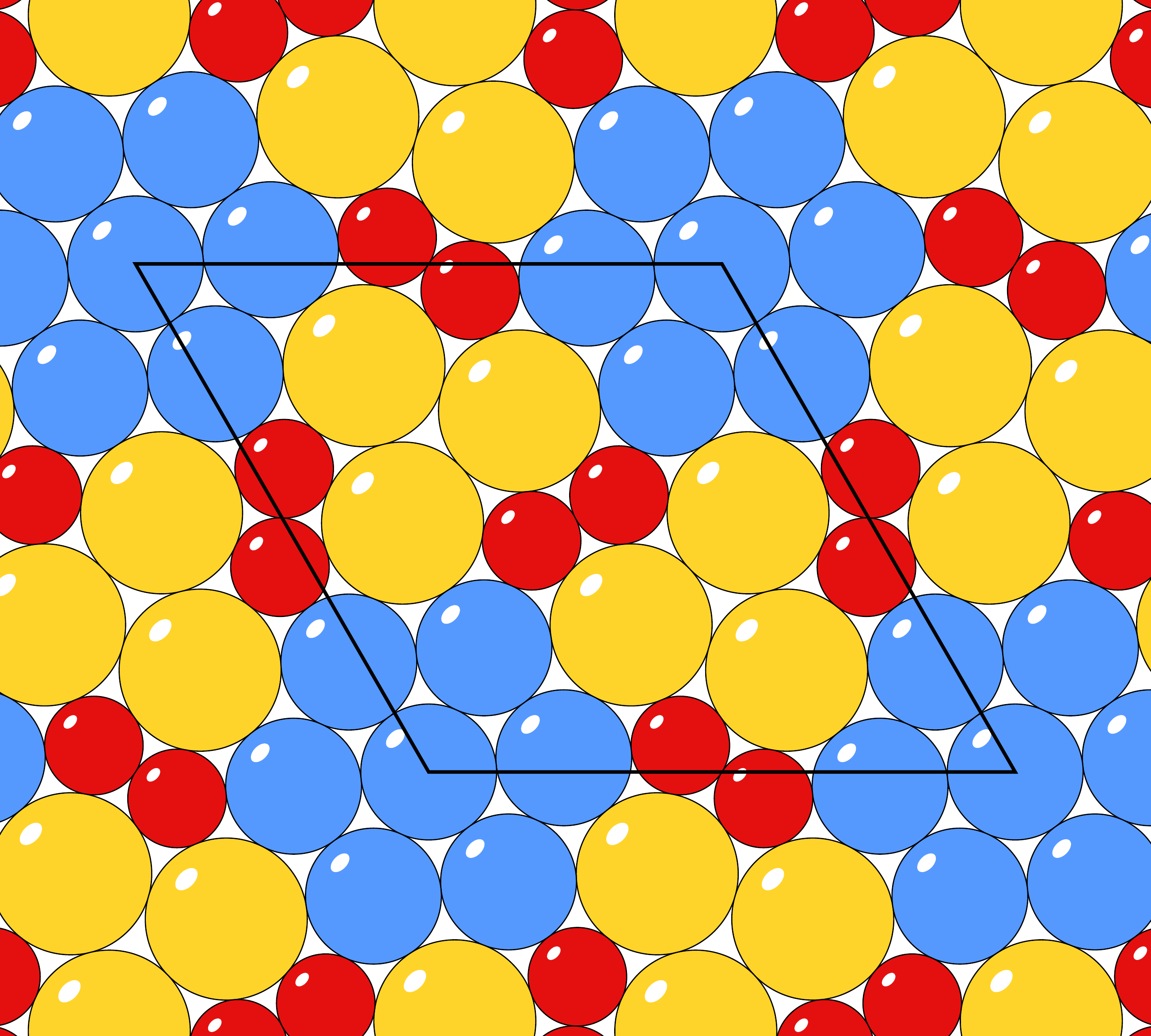}
\end{tabular}
\noindent
\begin{tabular}{lll}
  58\hfill 11r1s / 1s1s1s & 59\hfill 11rr / 111srs & 60\hfill 11rr / 11srrs\\
  \includegraphics[width=0.3\textwidth]{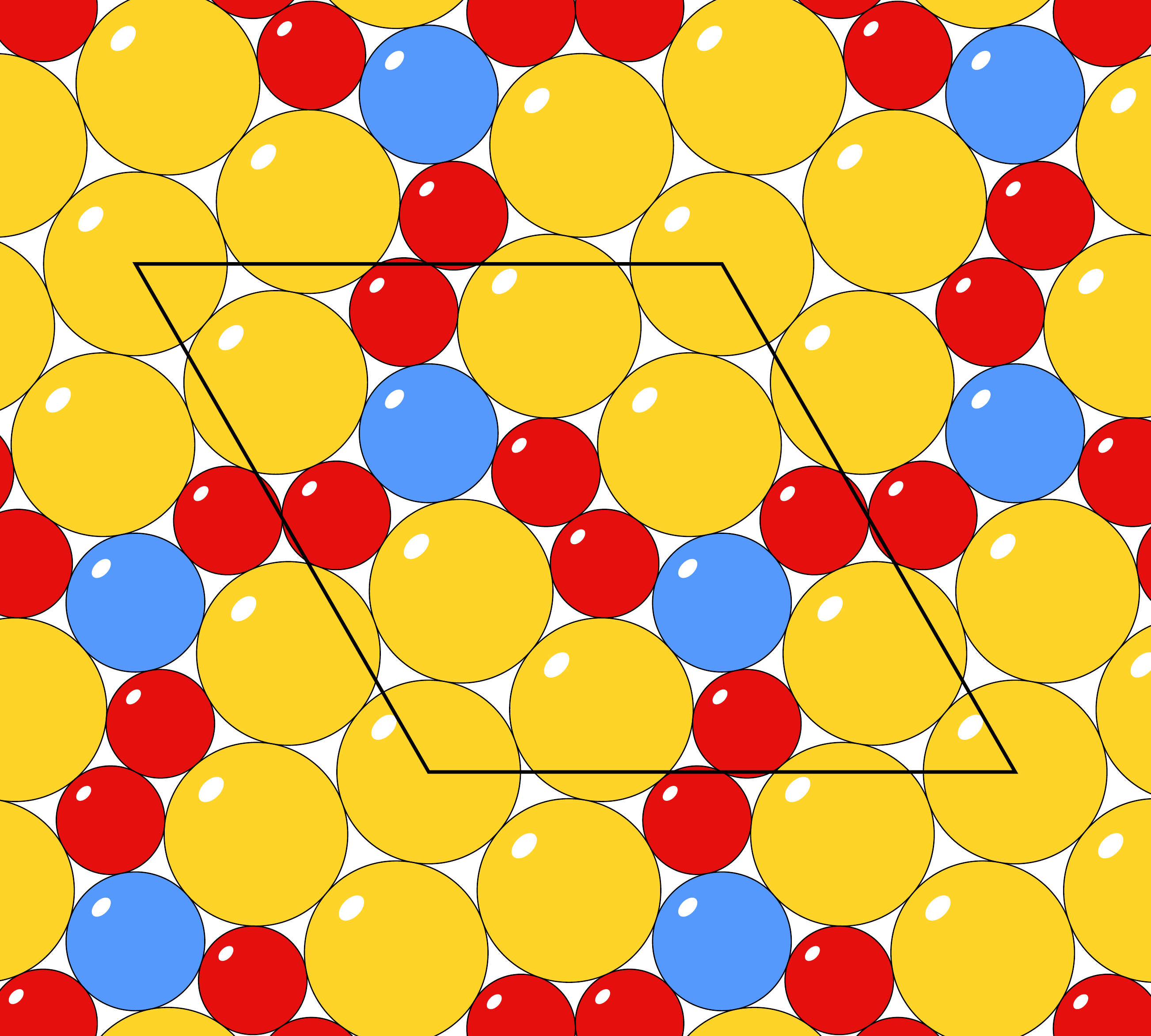} &
  \includegraphics[width=0.3\textwidth]{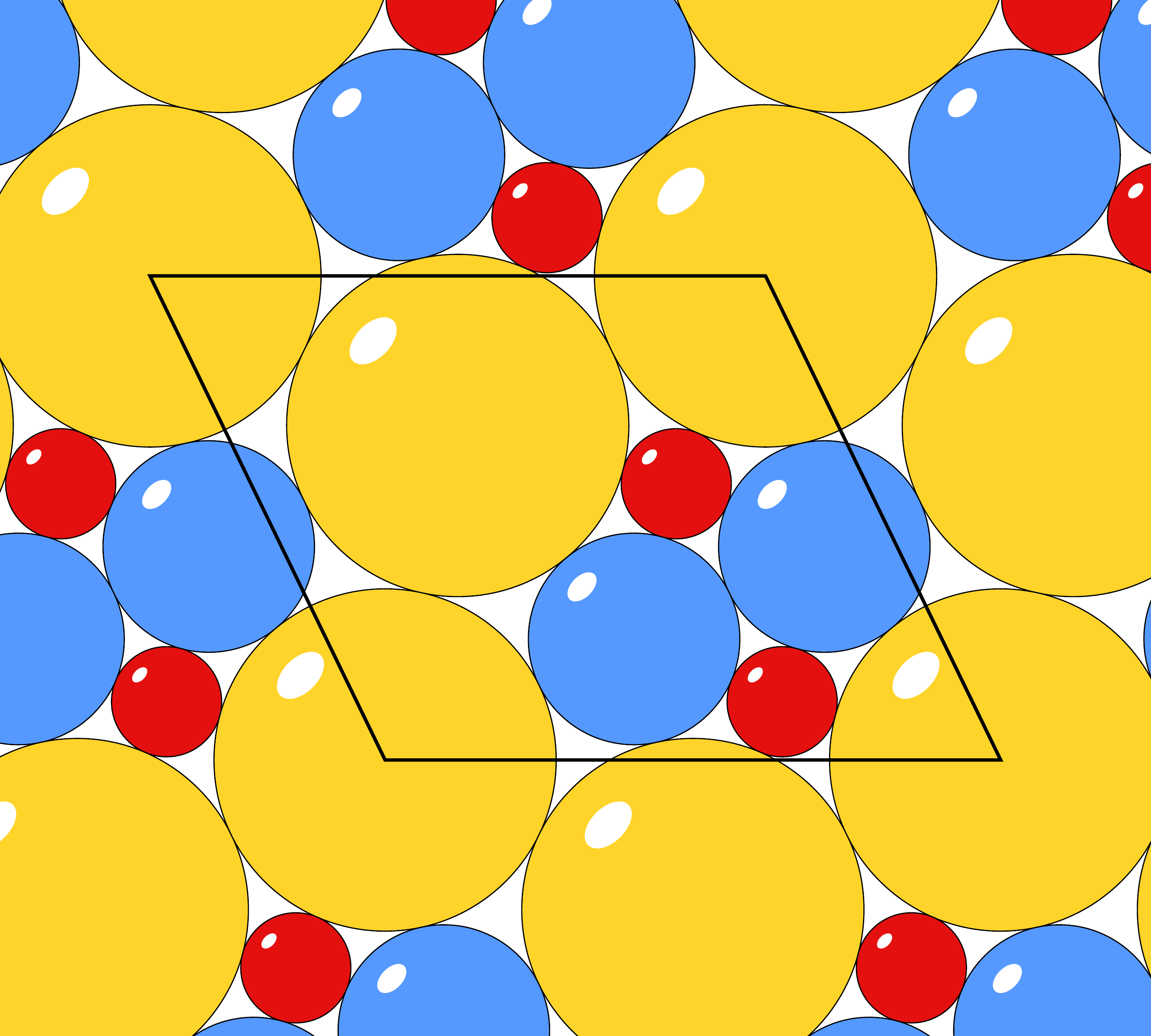} &
  \includegraphics[width=0.3\textwidth]{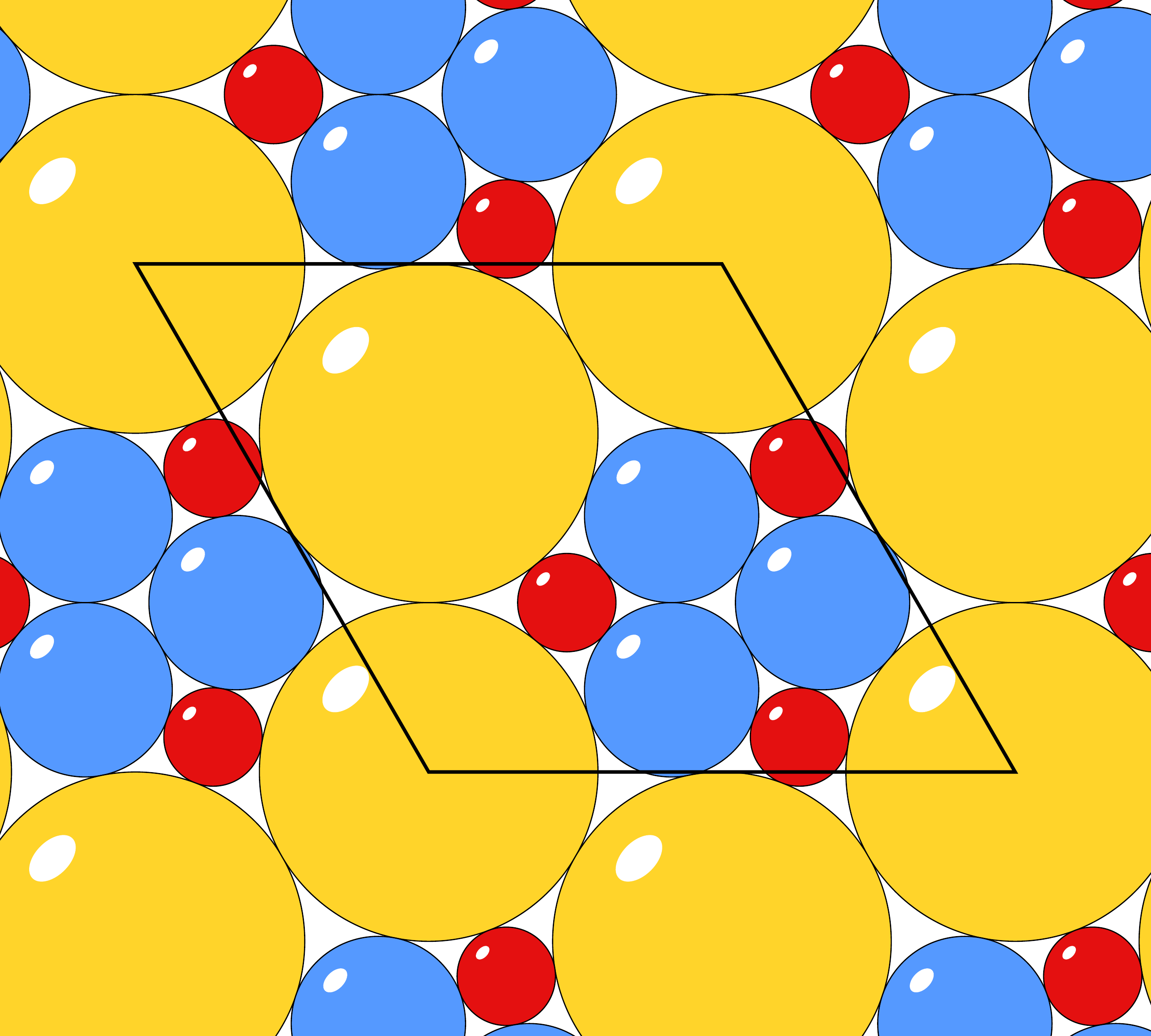}
\end{tabular}
\noindent
\begin{tabular}{lll}
  61\hfill 11rr / 11srs & 62\hfill 11rr / 1r1rs & 63\hfill 11rr / 1rr1rs\\
  \includegraphics[width=0.3\textwidth]{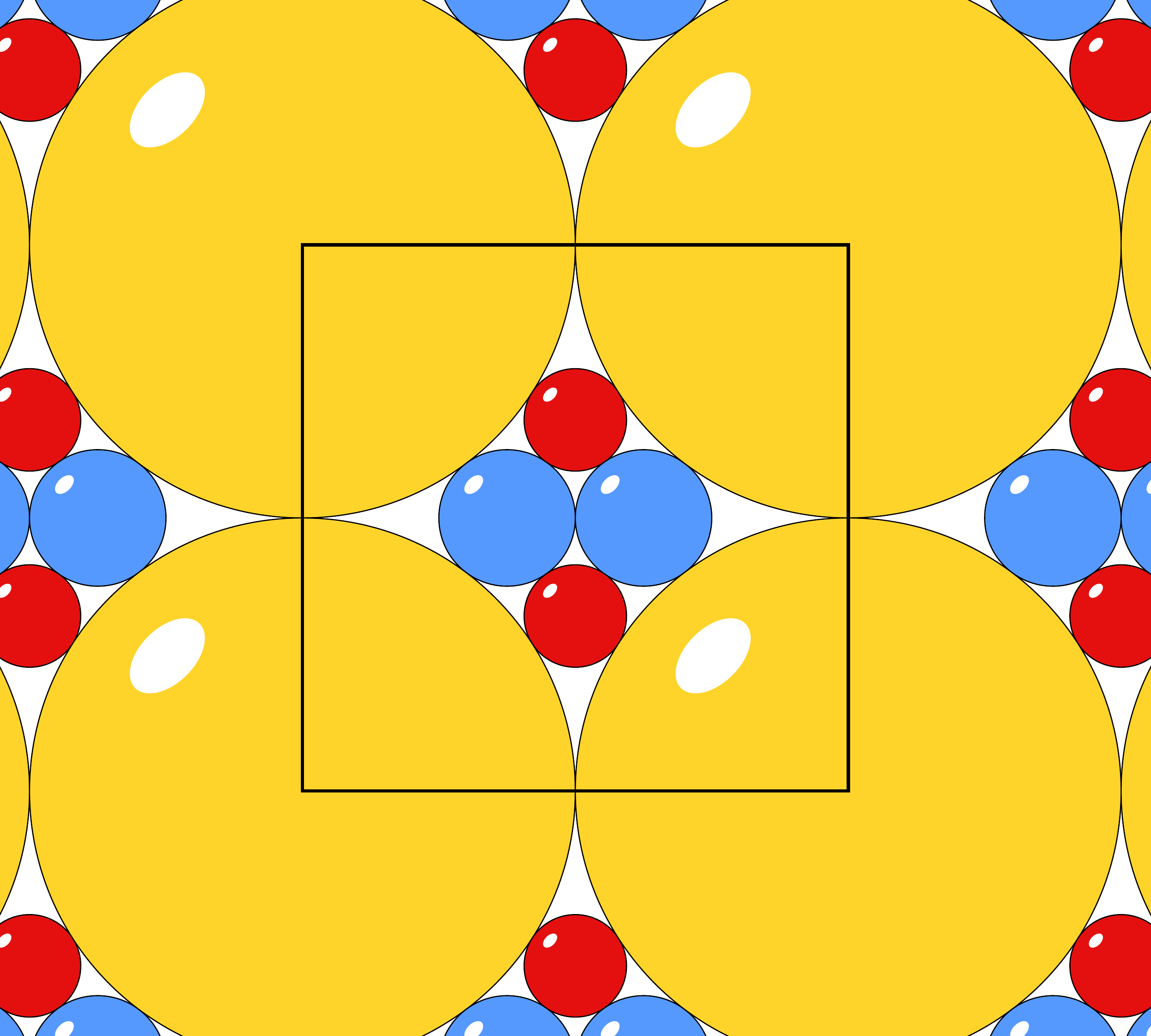} &
  \includegraphics[width=0.3\textwidth]{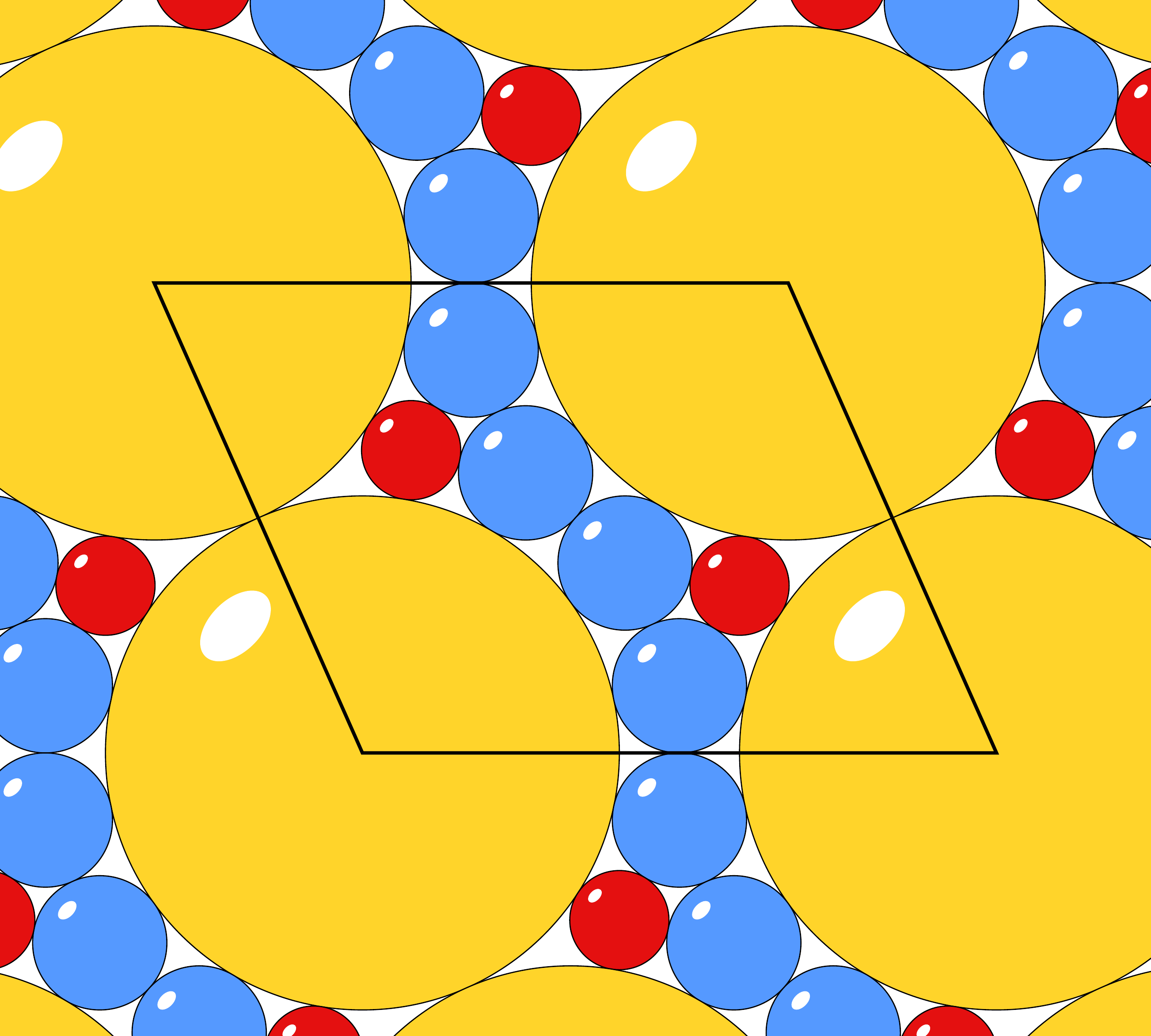} &
  \includegraphics[width=0.3\textwidth]{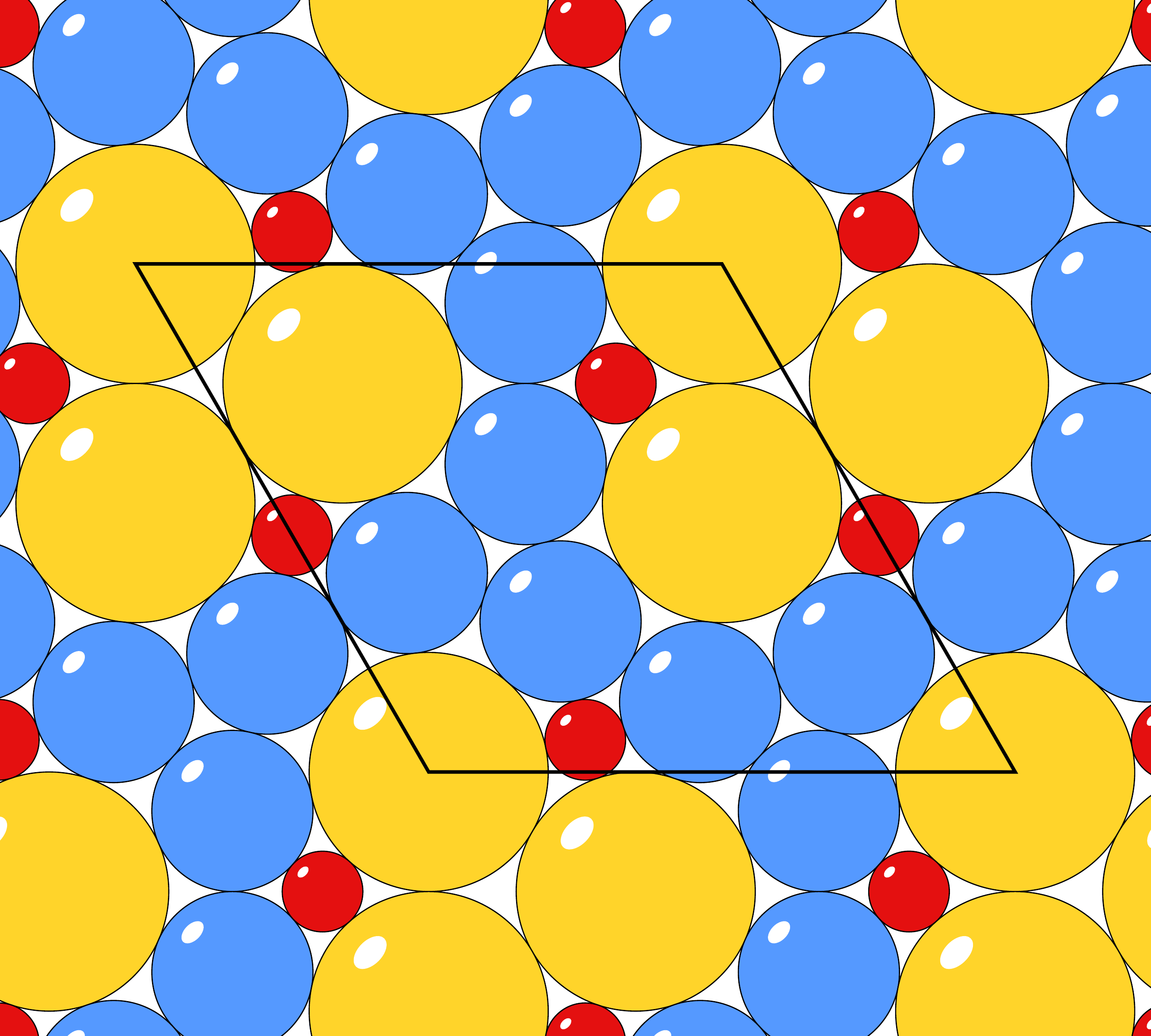}
\end{tabular}
\noindent
\begin{tabular}{lll}
  64\hfill 11rr / 1rrrrs & 65\hfill 11rr / 1srrrs & 66\hfill 11rrr / 1srsrs\\
  \includegraphics[width=0.3\textwidth]{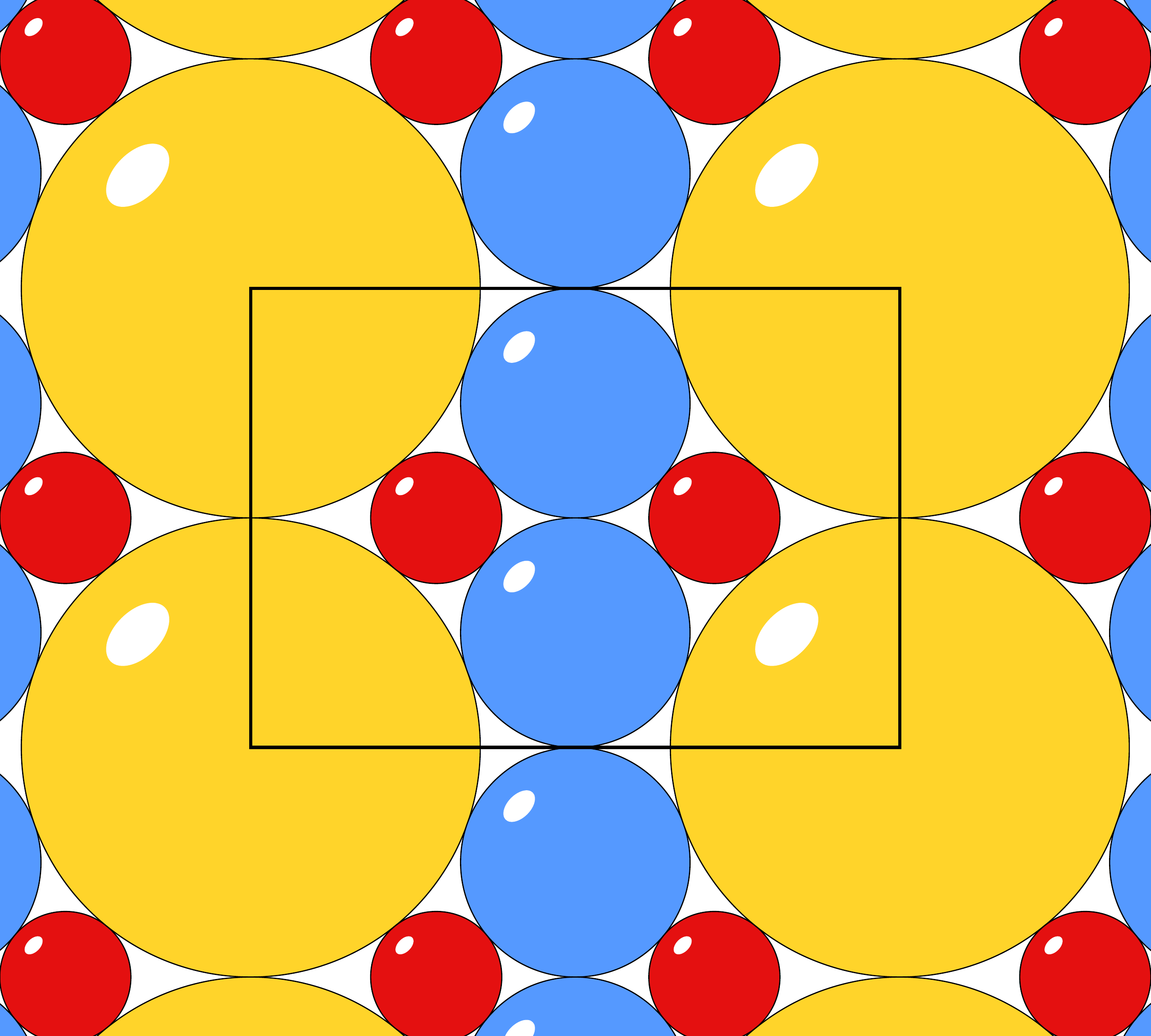} &
  \includegraphics[width=0.3\textwidth]{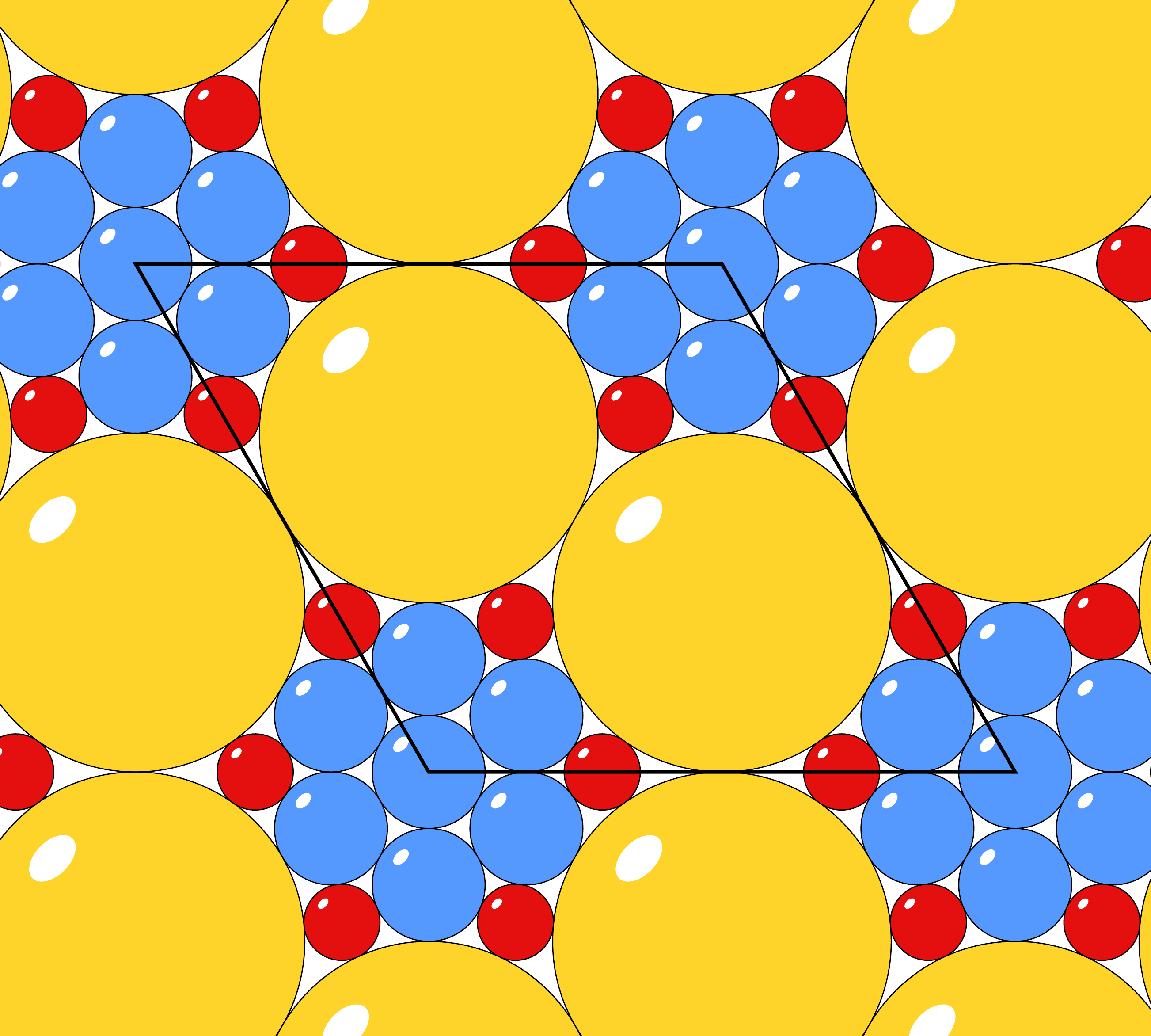} &
  \includegraphics[width=0.3\textwidth]{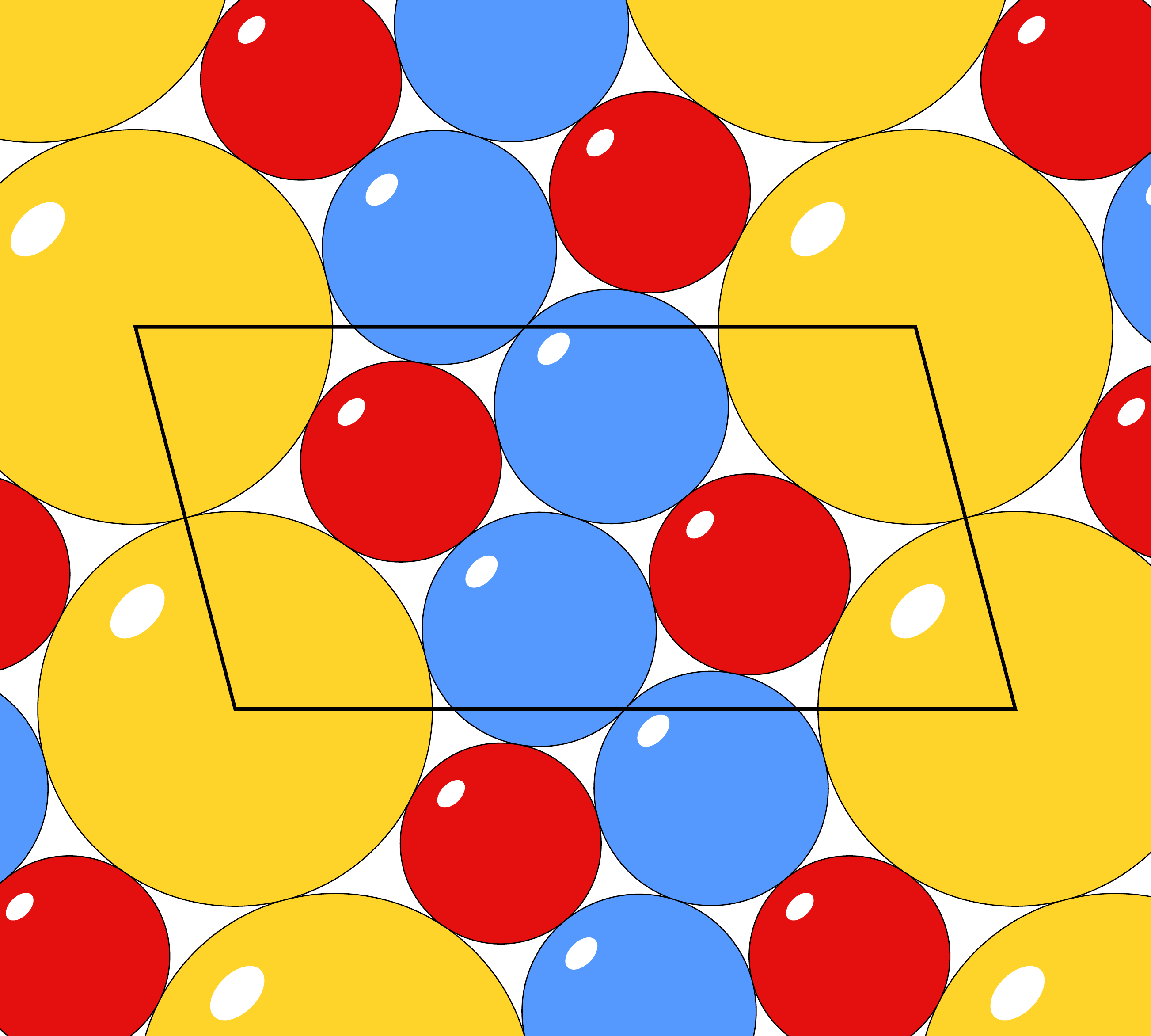}
\end{tabular}
\noindent
\begin{tabular}{lll}
  67\hfill 11rs / 111s1sss & 68\hfill 11rs / 111ss & 69\hfill 11rs / 11r1ss\\
  \includegraphics[width=0.3\textwidth]{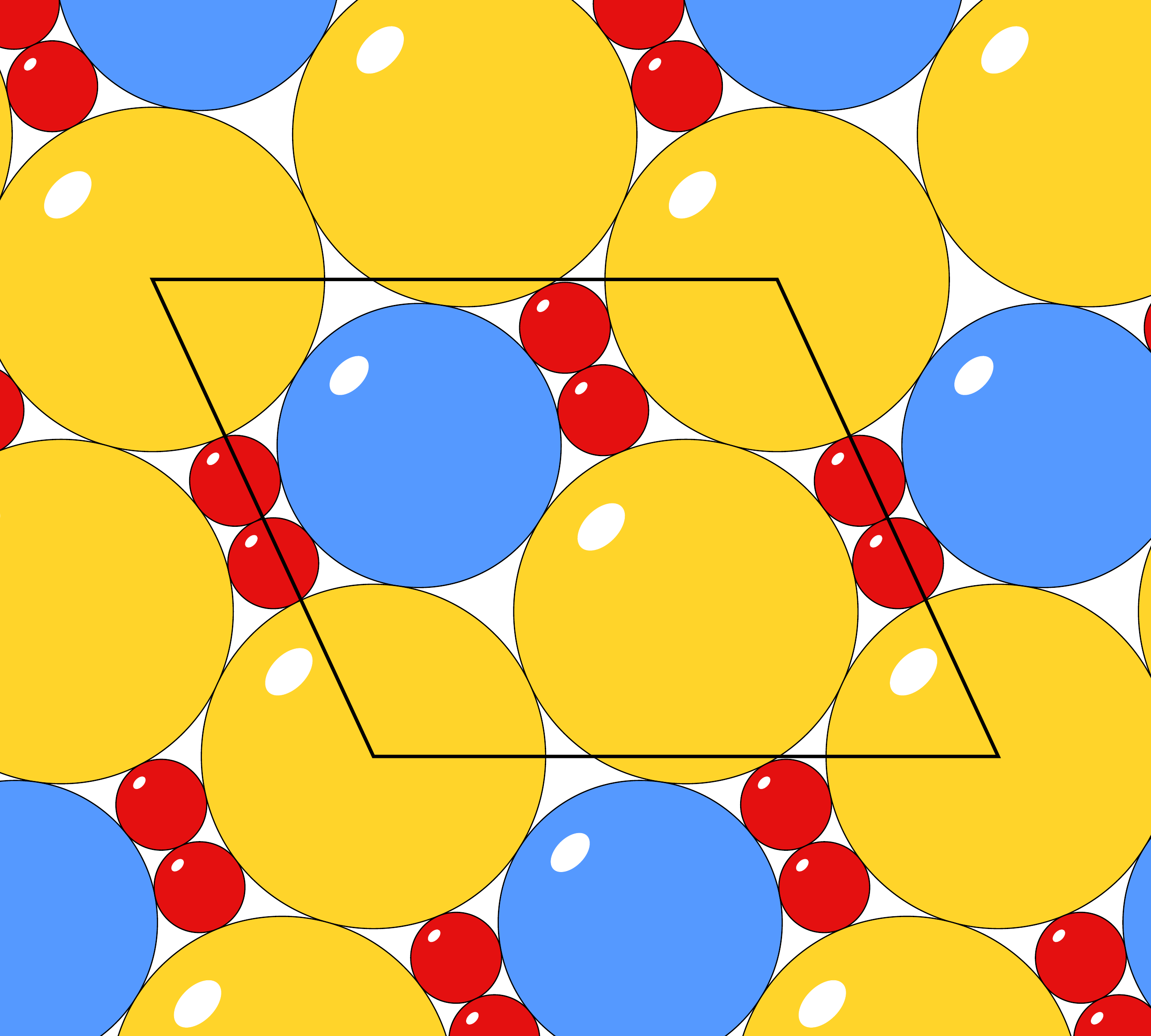} &
  \includegraphics[width=0.3\textwidth]{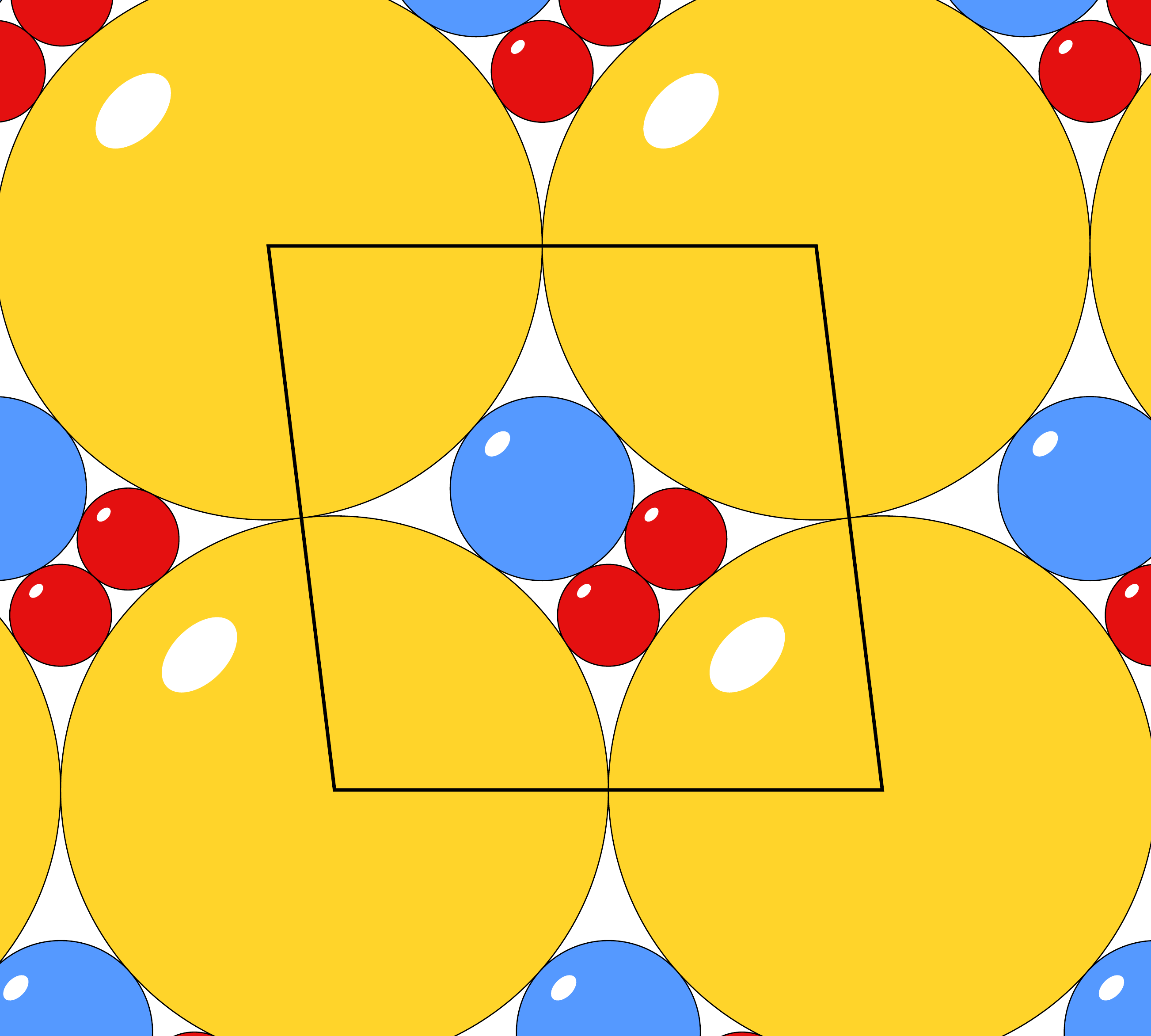} &
  \includegraphics[width=0.3\textwidth]{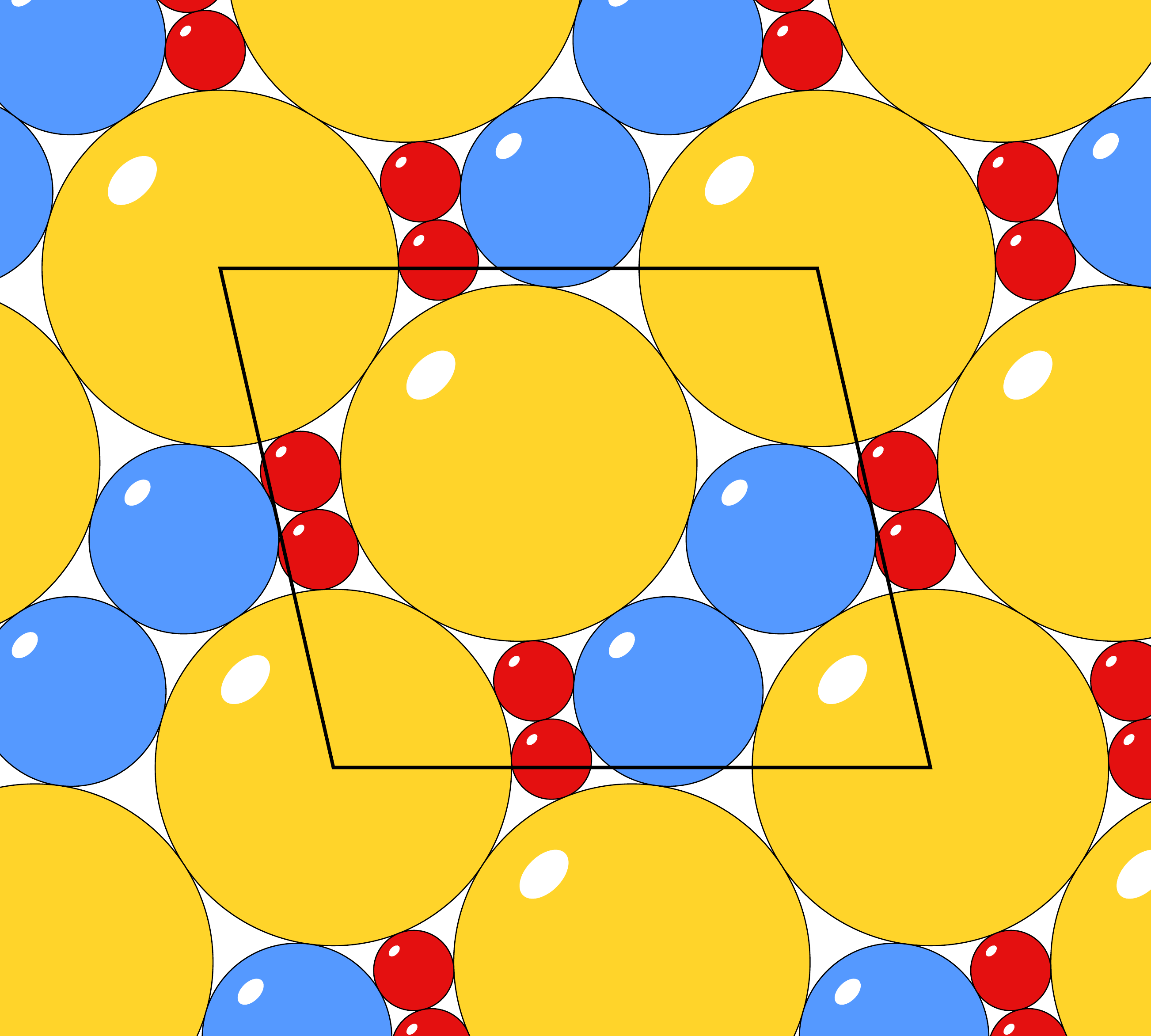}
\end{tabular}
\noindent
\begin{tabular}{lll}
  70\hfill 11rs / 1r1r1ss & 71\hfill 11rs / 1r1ss & 72\hfill 11rs / 1rr1ss\\
  \includegraphics[width=0.3\textwidth]{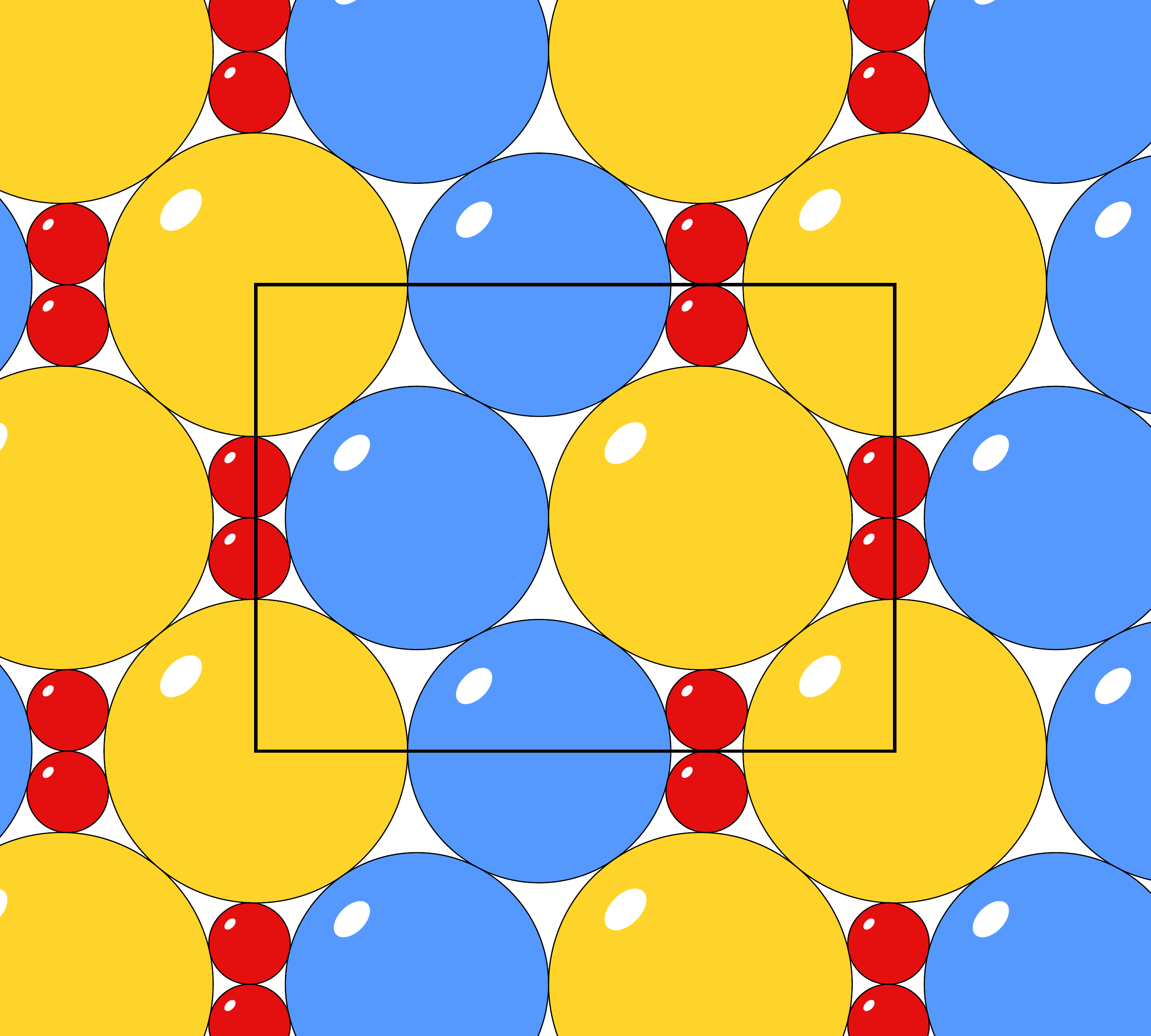} &
  \includegraphics[width=0.3\textwidth]{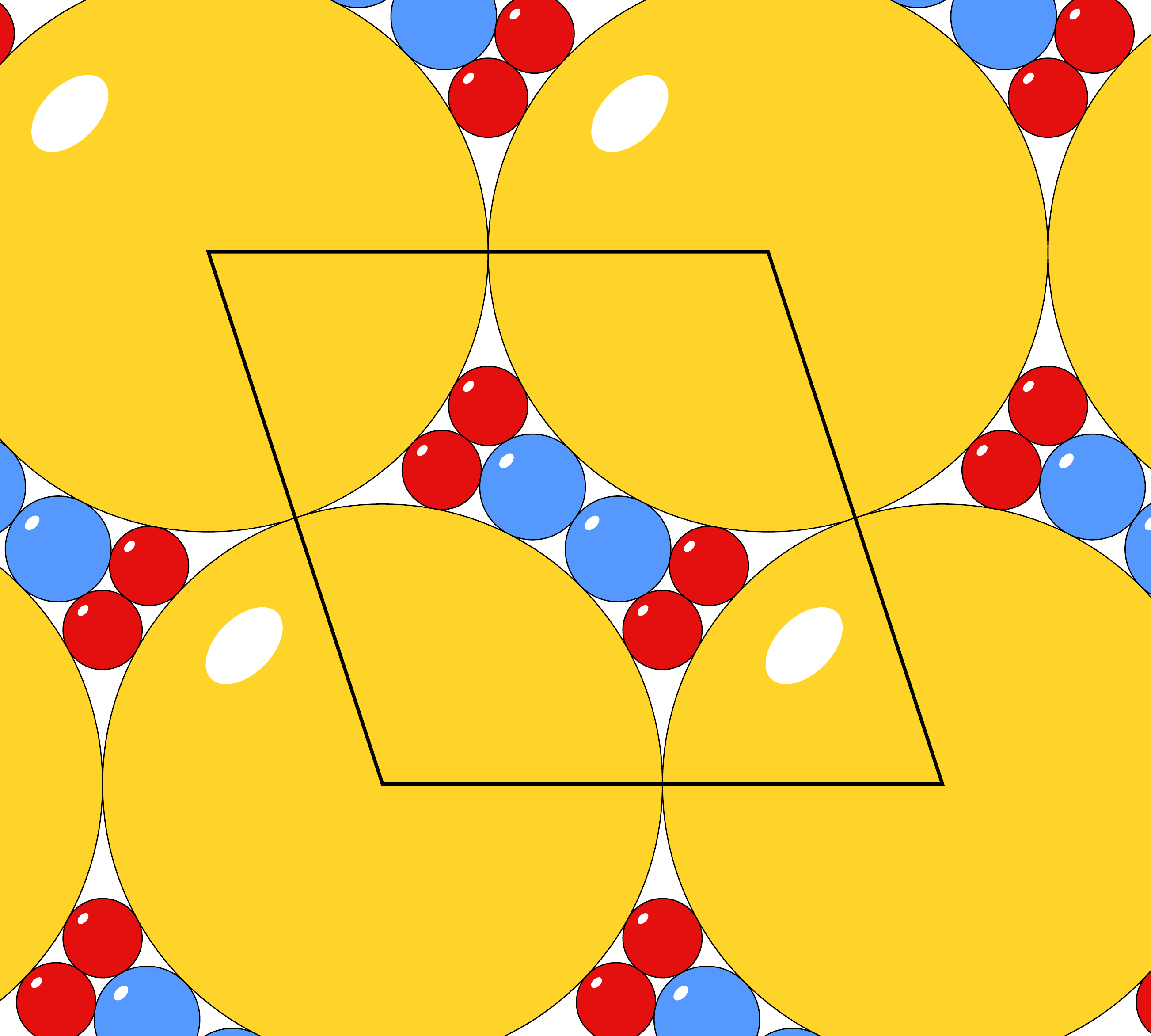} &
  \includegraphics[width=0.3\textwidth]{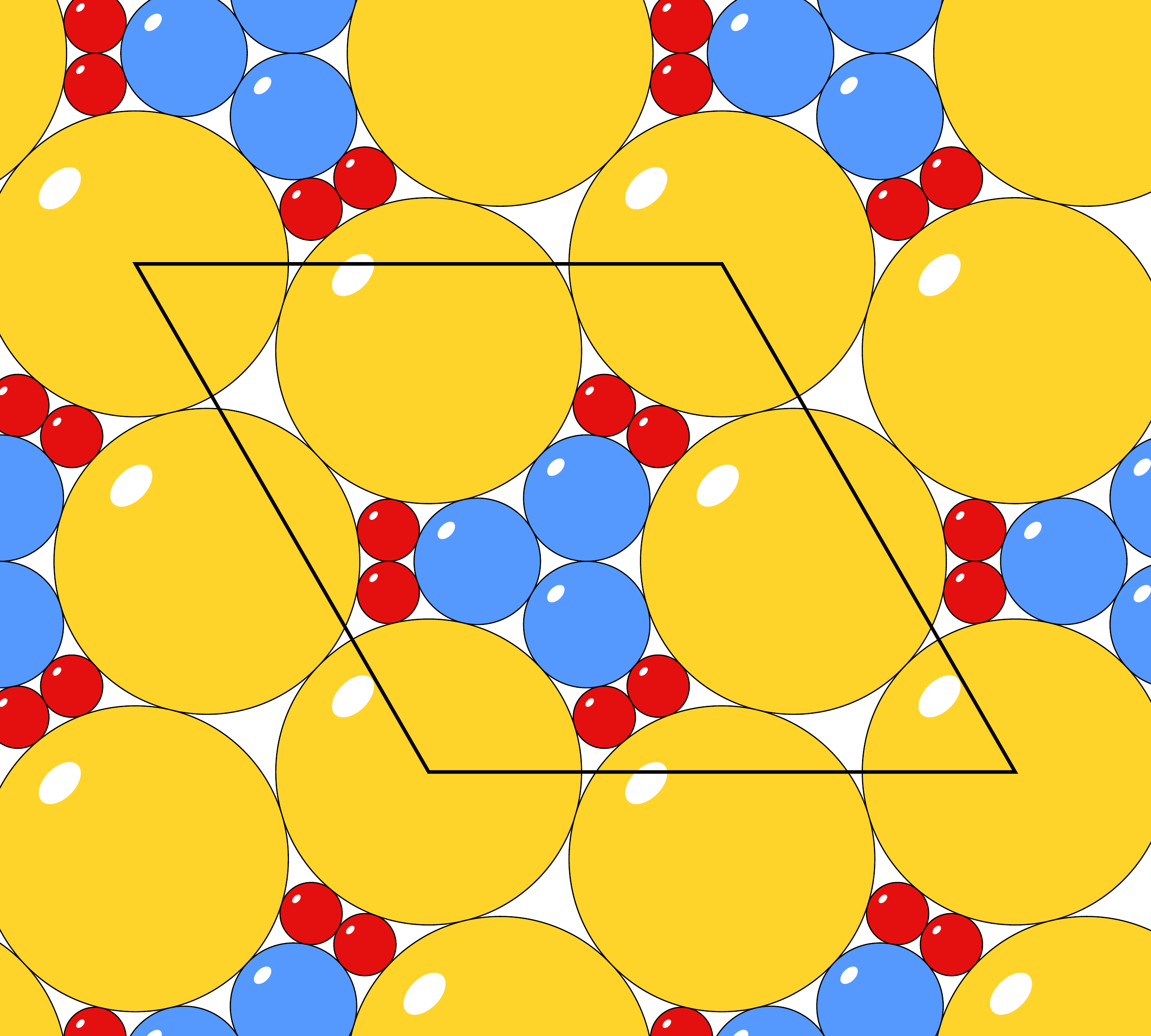}
\end{tabular}
\noindent
\begin{tabular}{lll}
  73\hfill 11rs / 1rrr1ss & 74\hfill 11rs / 1s1s1ssss & 75\hfill 11rs / 1s1sss\\
  \includegraphics[width=0.3\textwidth]{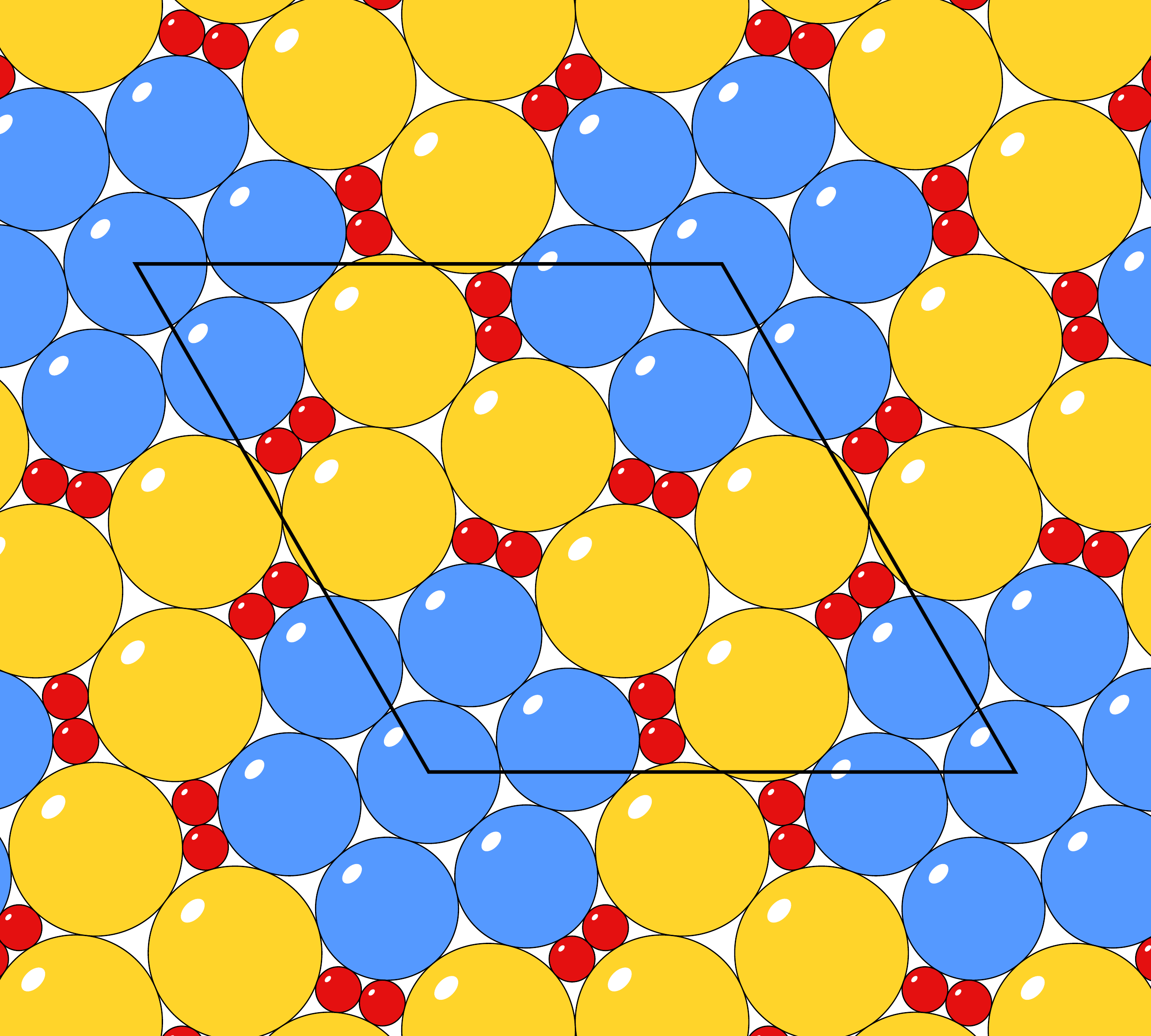} &
  \includegraphics[width=0.3\textwidth]{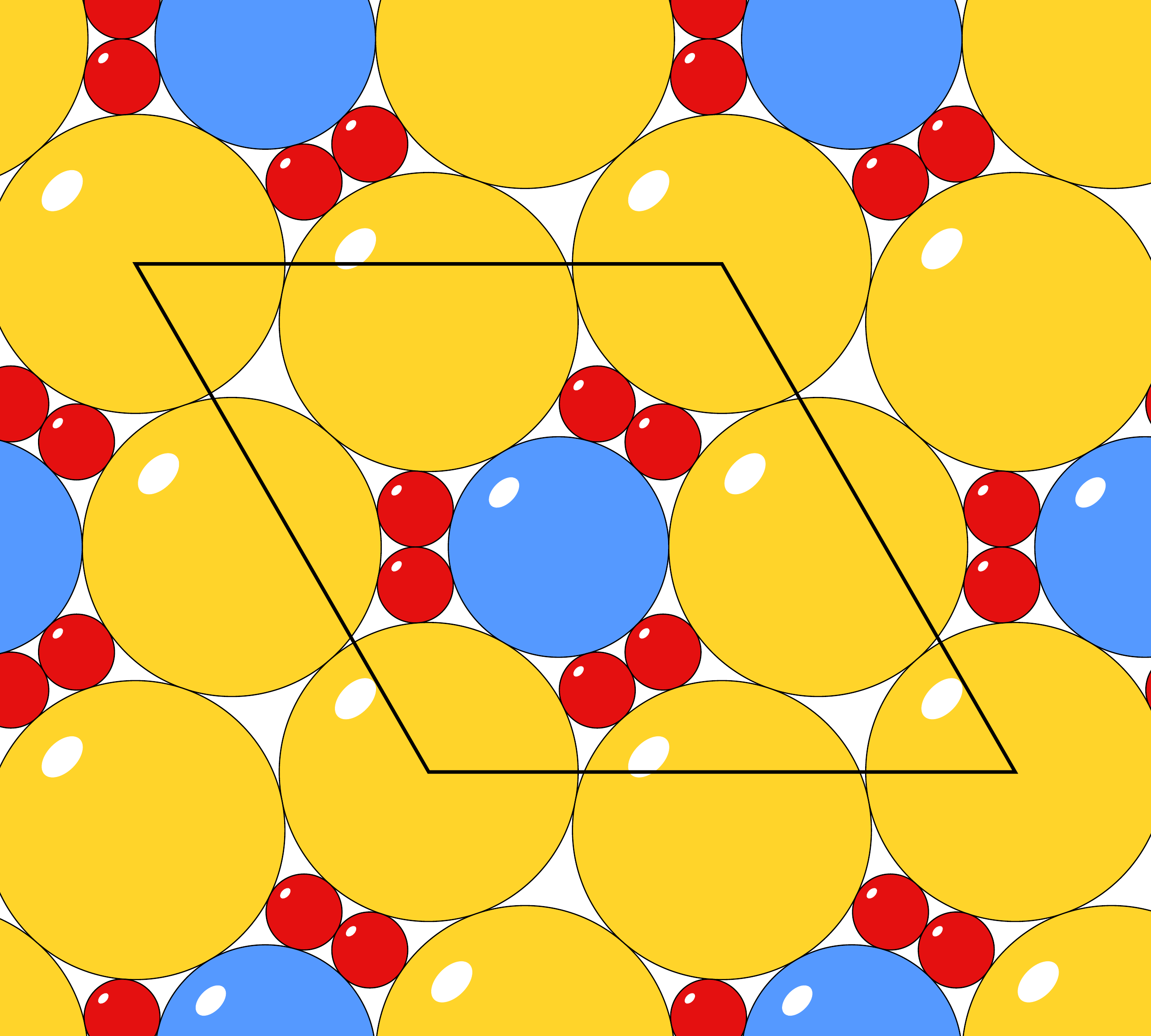} &
  \includegraphics[width=0.3\textwidth]{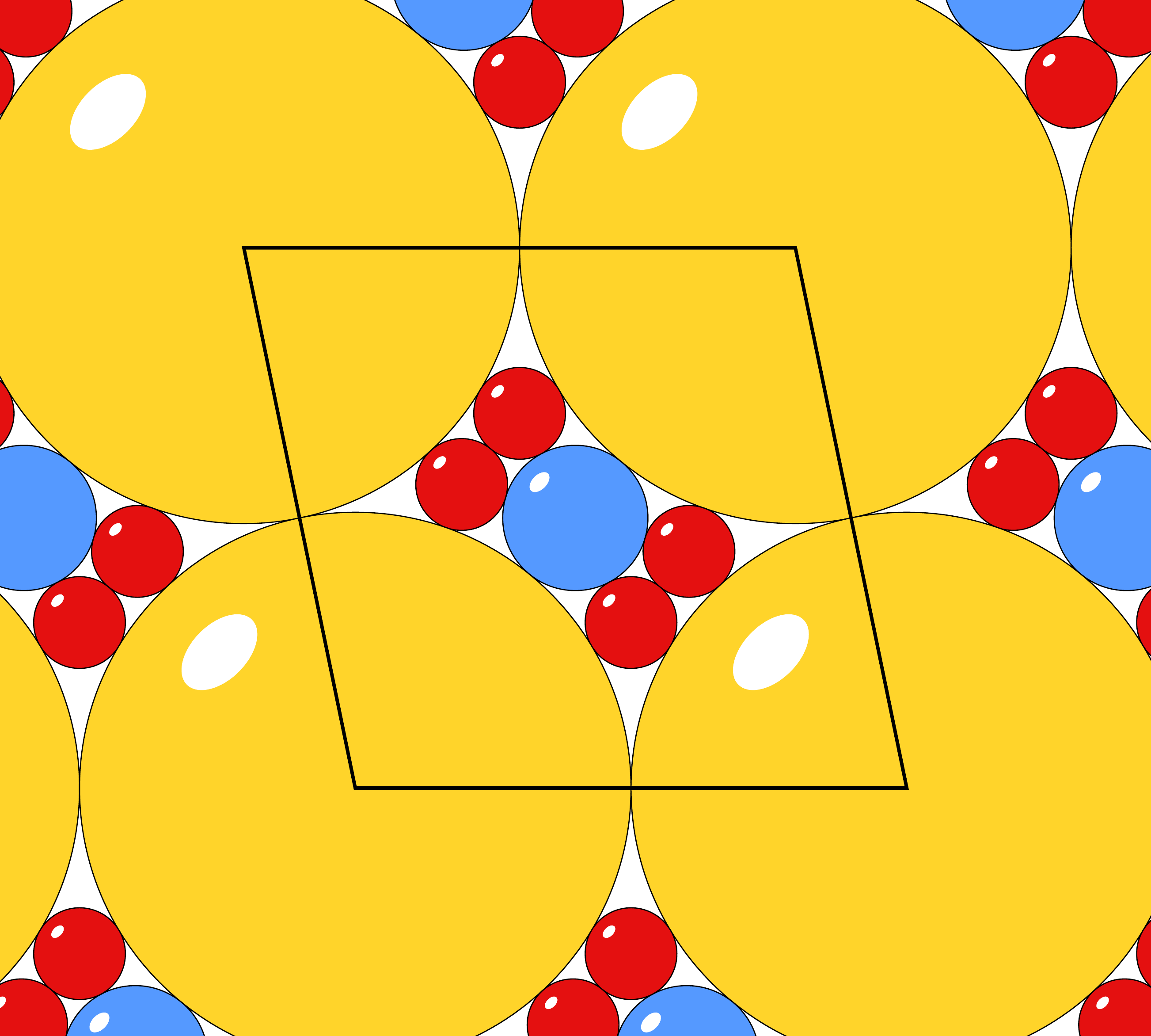}
\end{tabular}
\noindent
\begin{tabular}{lll}
  76\hfill 11rsr / 111ss & 77\hfill 11rsr / 11r1ss & 78\hfill 11rsr / 1rr1ss\\
  \includegraphics[width=0.3\textwidth]{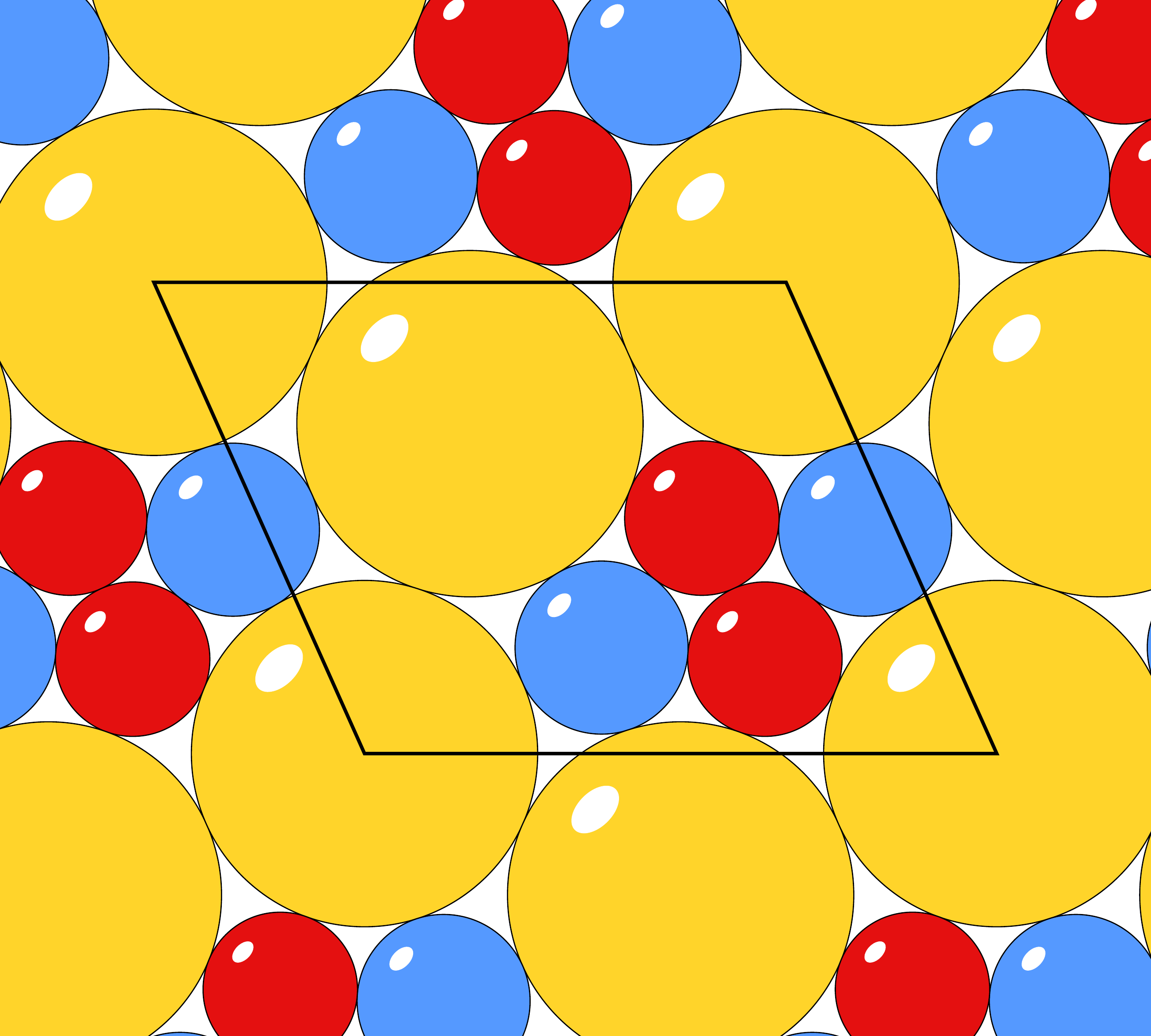} &
  \includegraphics[width=0.3\textwidth]{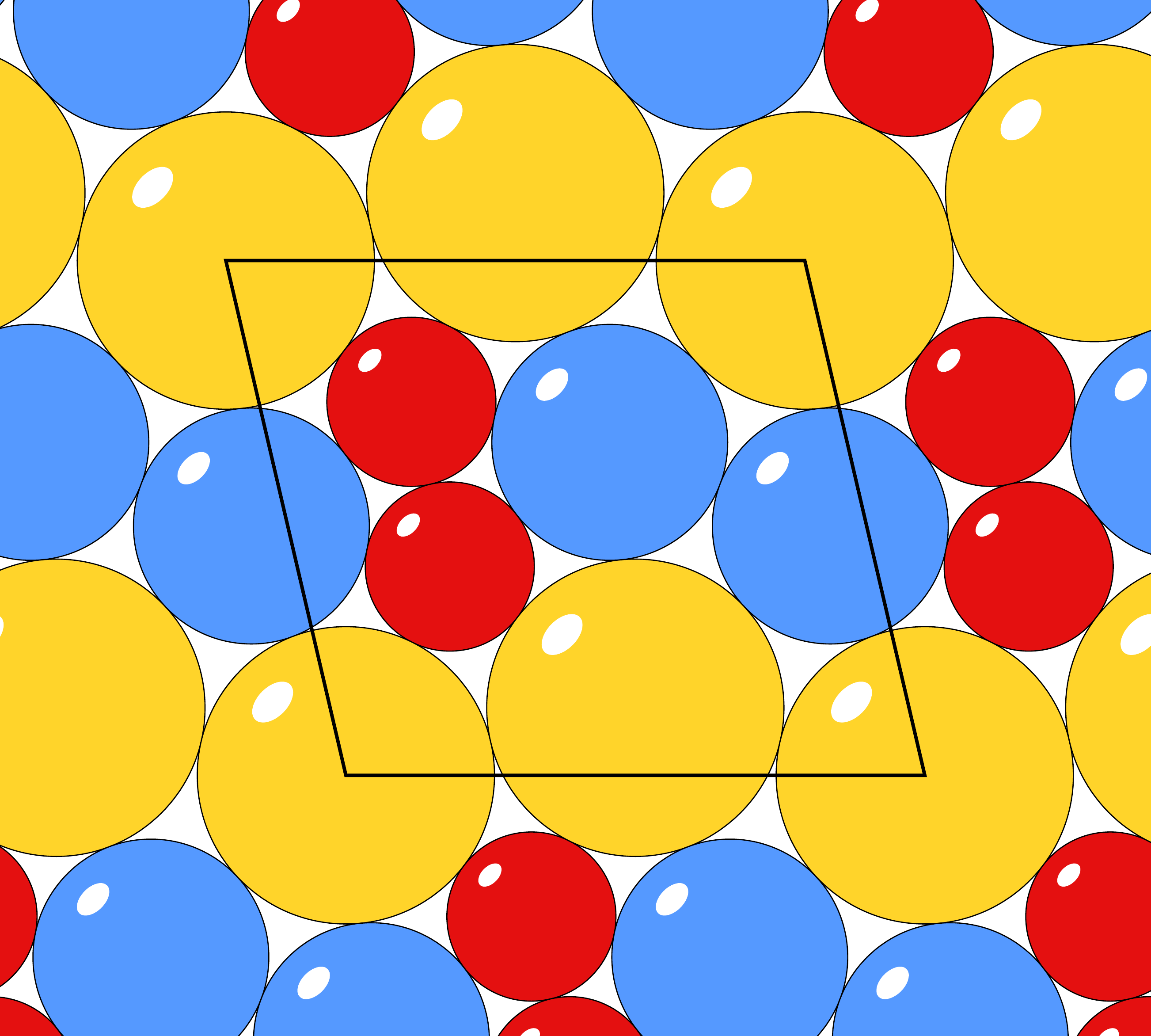} &
  \includegraphics[width=0.3\textwidth]{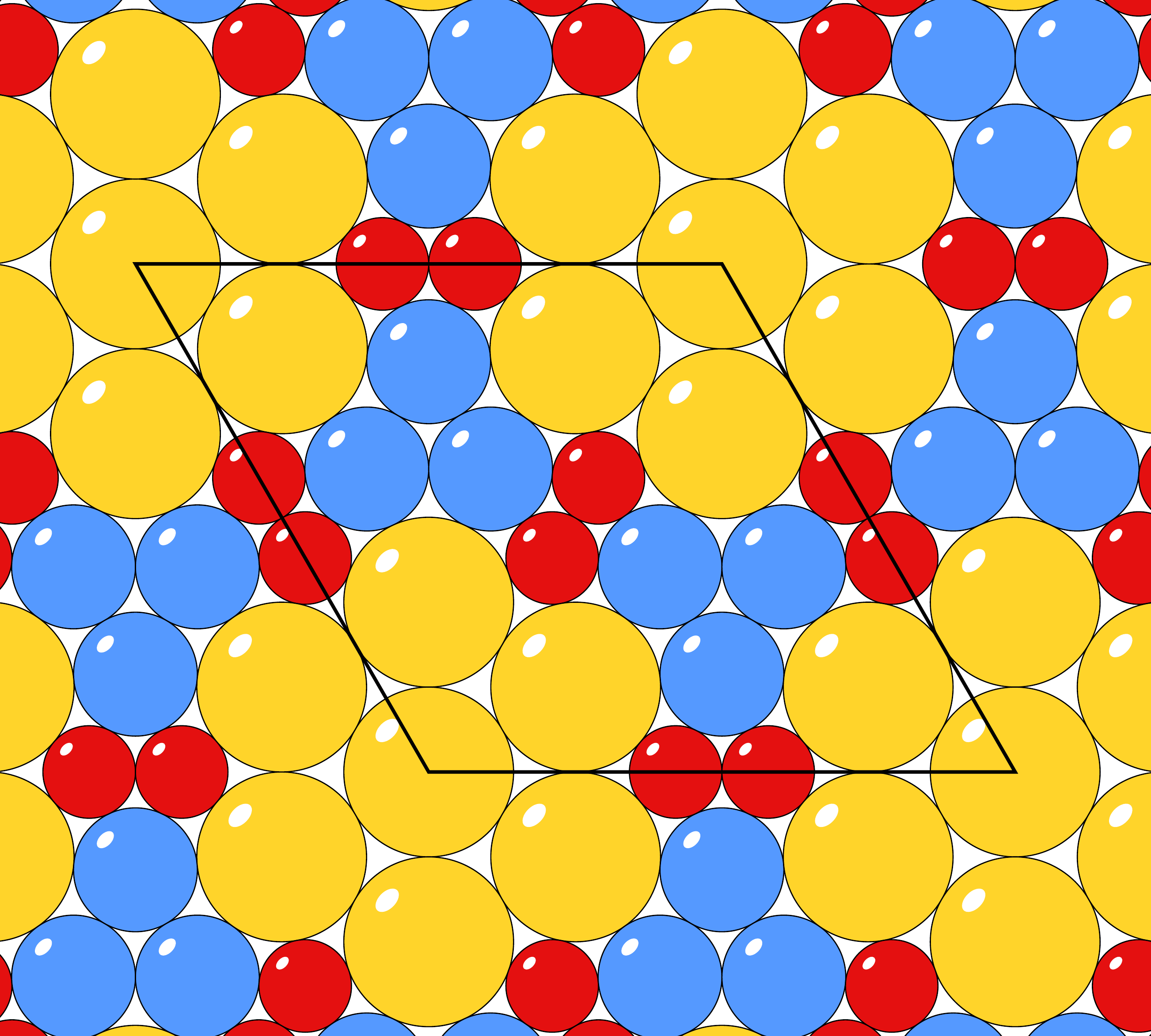}
\end{tabular}
\noindent
\begin{tabular}{lll}
  79\hfill 11rsr / 1s1sss & 80\hfill 1r1r / 1111s & 81\hfill 1r1r / 111r1s\\
  \includegraphics[width=0.3\textwidth]{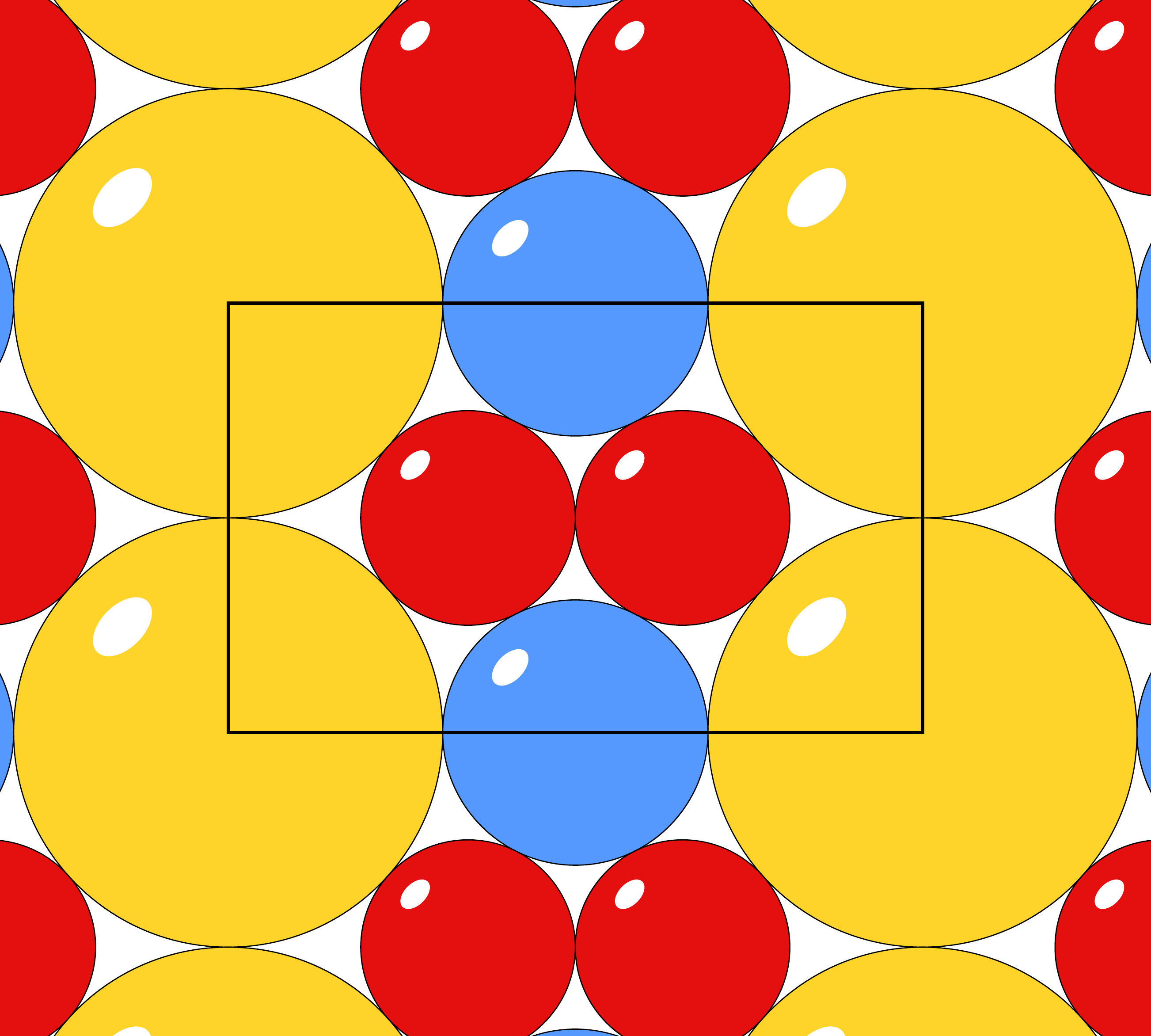} &
  \includegraphics[width=0.3\textwidth]{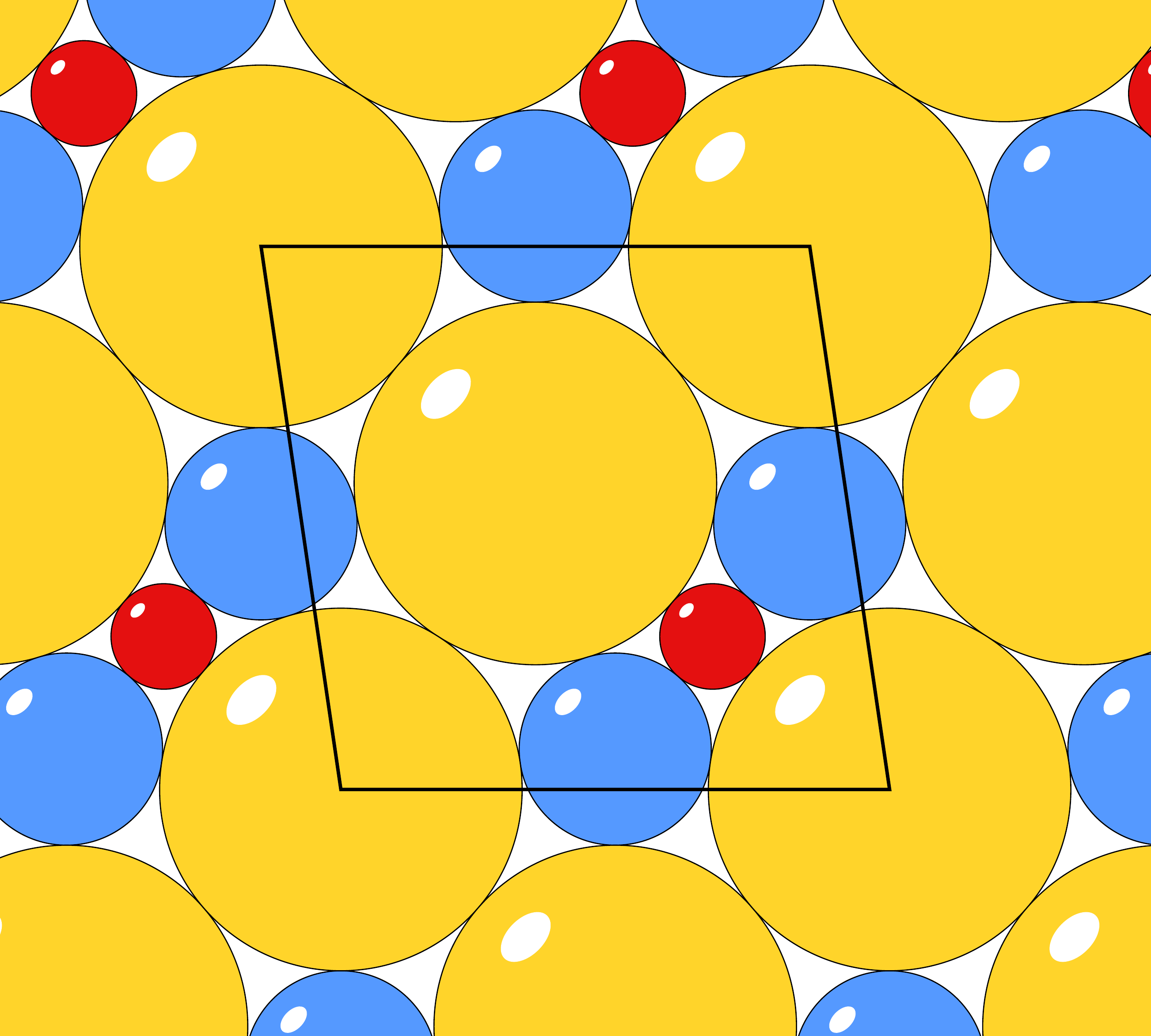} &
  \includegraphics[width=0.3\textwidth]{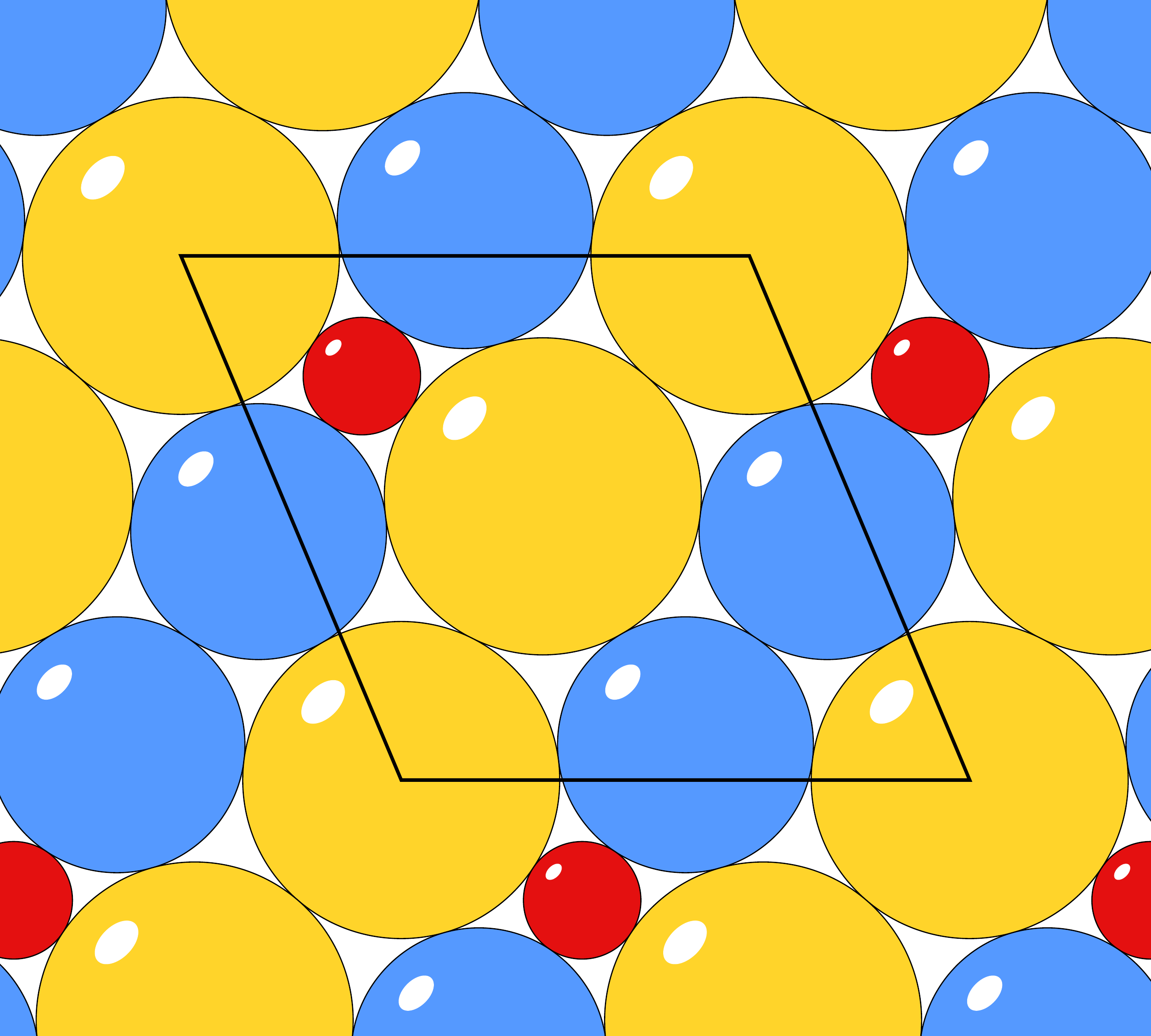}
\end{tabular}
\noindent
\begin{tabular}{lll}
  82\hfill 1r1r / 111s1s & 83\hfill 1r1r / 11r1s & 84\hfill 1r1r / 11rr1s\\
  \includegraphics[width=0.3\textwidth]{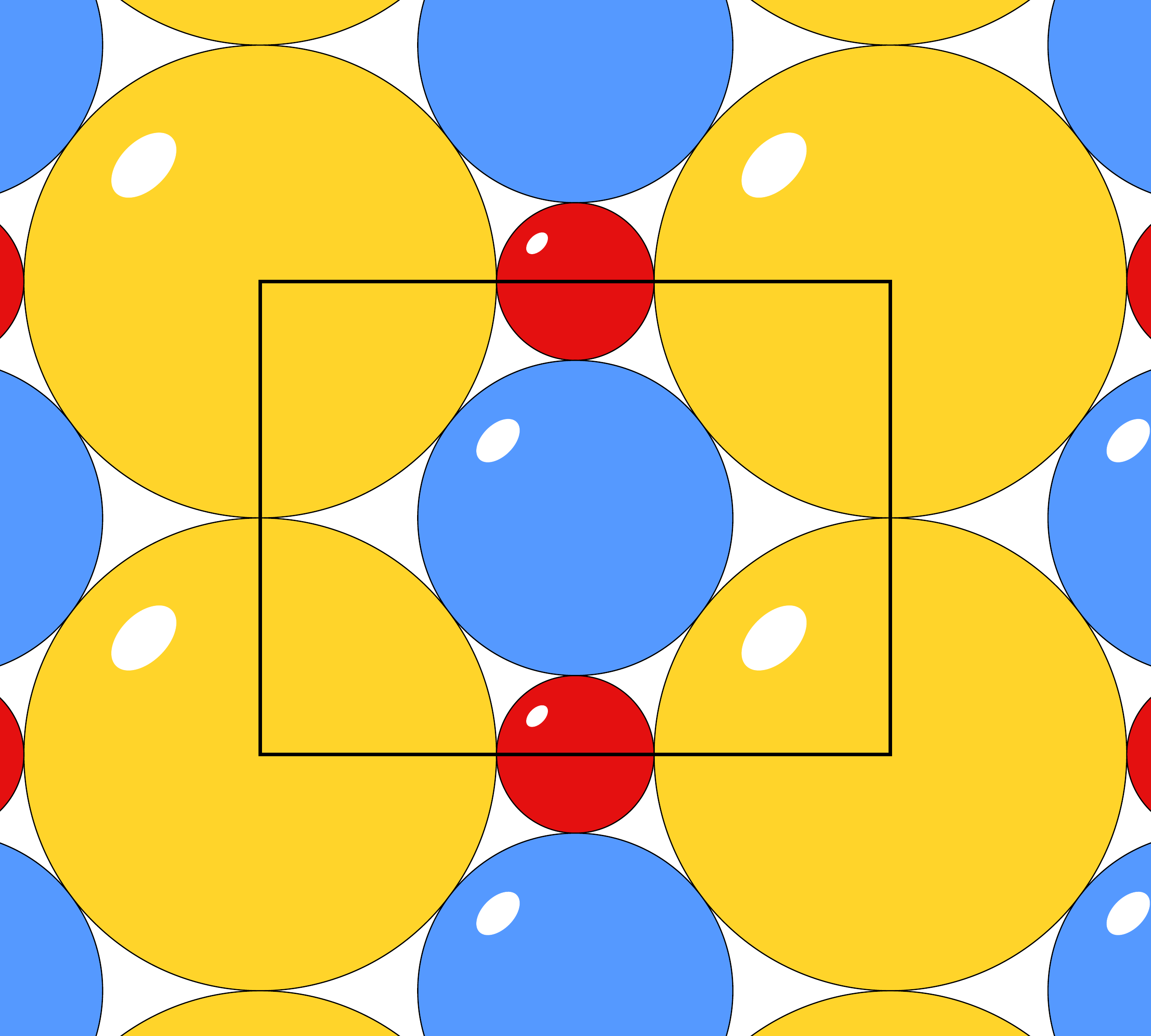} &
  \includegraphics[width=0.3\textwidth]{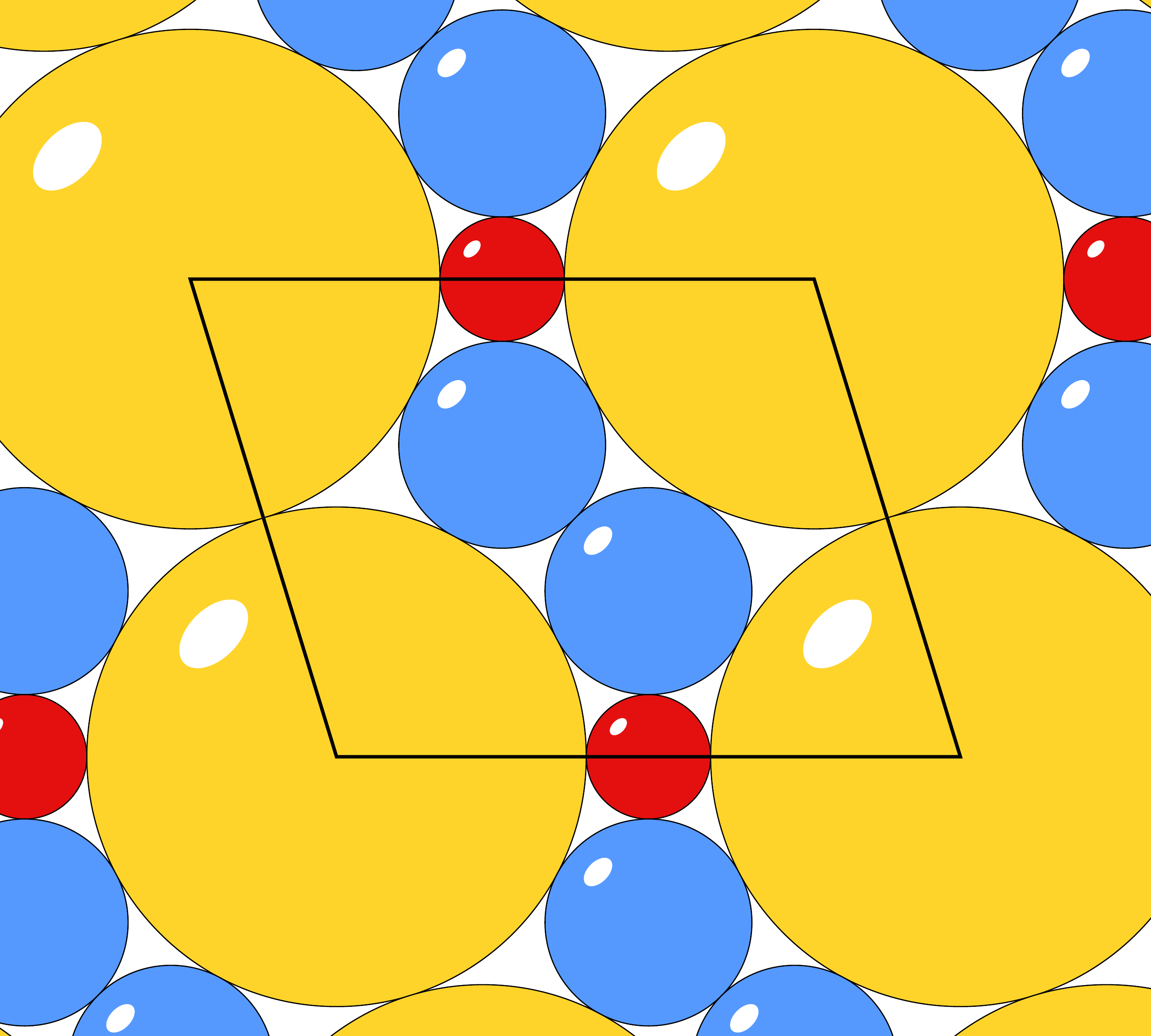} &
  \includegraphics[width=0.3\textwidth]{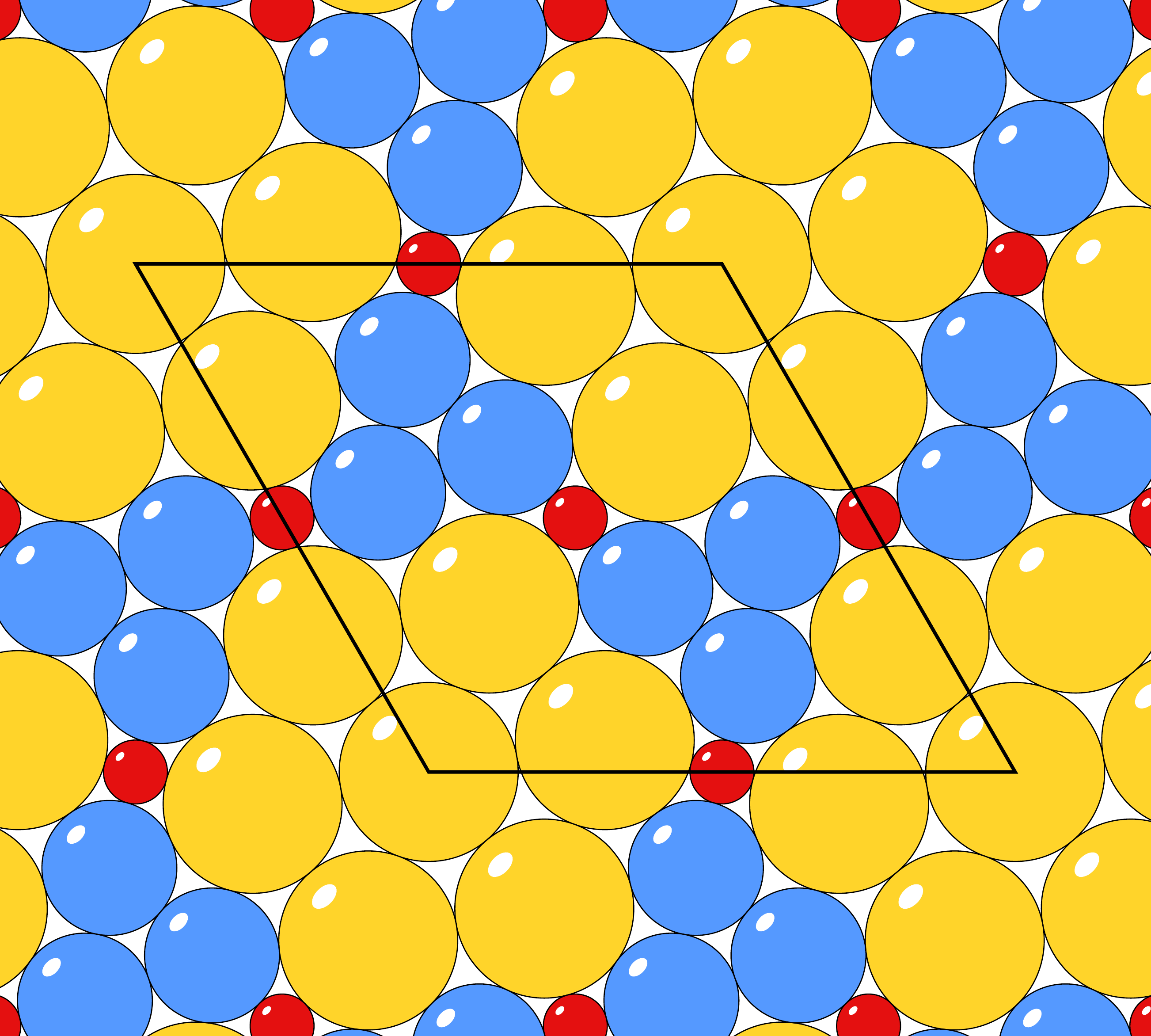}
\end{tabular}
\noindent
\begin{tabular}{lll}
  85\hfill 1r1r / 11s1s & 86\hfill 1r1r / 11s1s1s & 87\hfill 1r1r / 1r1r1s\\
  \includegraphics[width=0.3\textwidth]{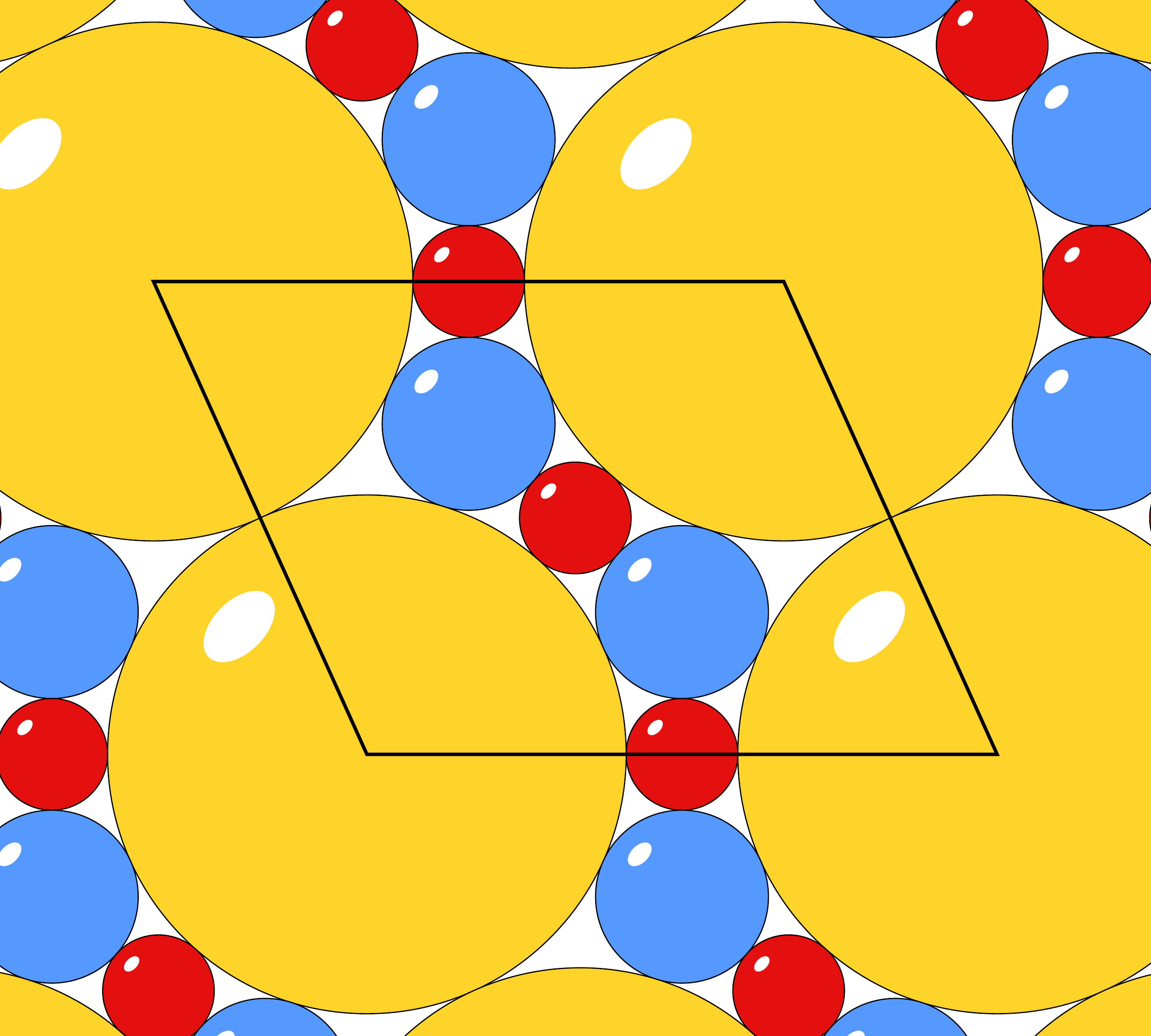} &
  \includegraphics[width=0.3\textwidth]{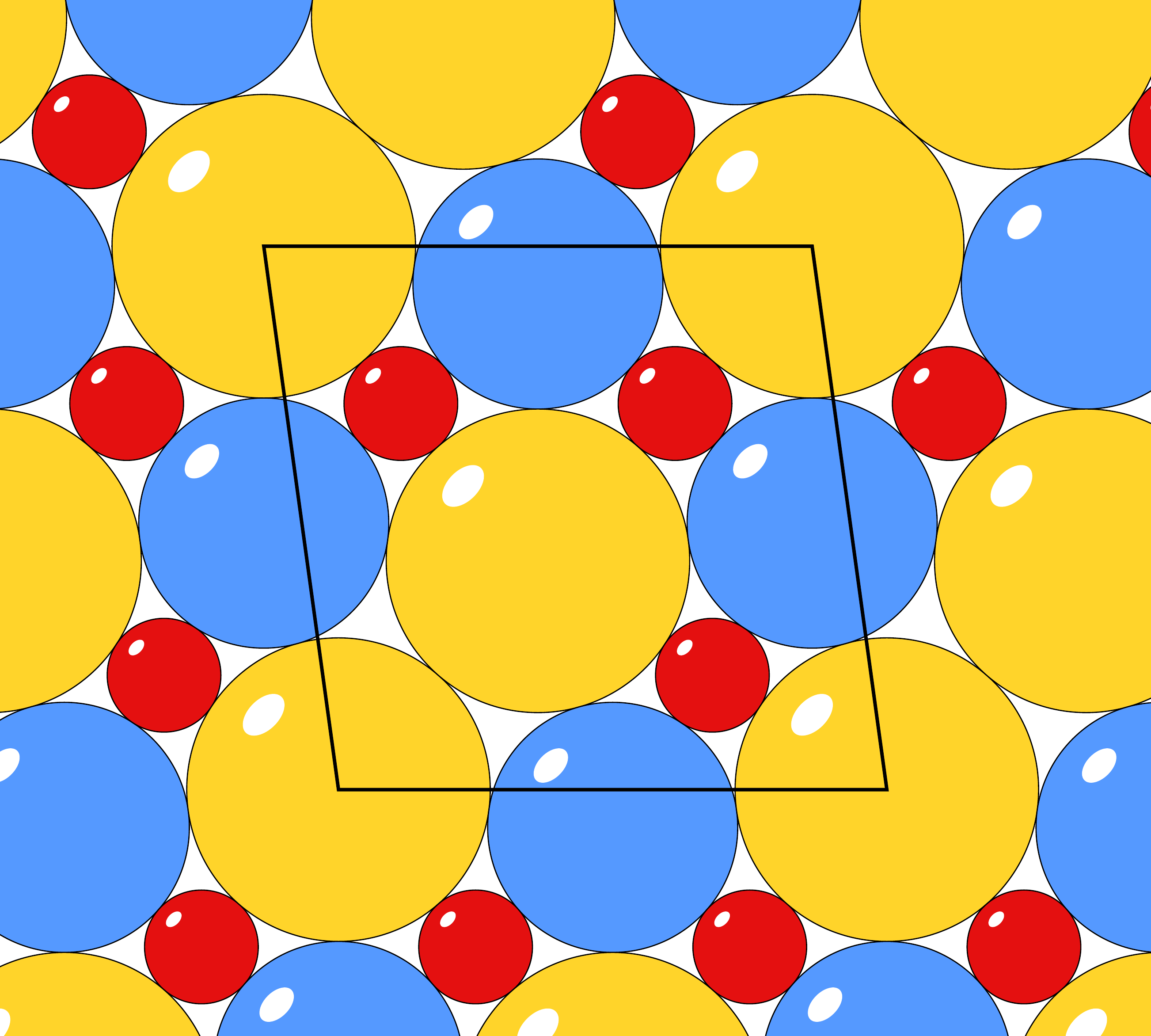} &
  \includegraphics[width=0.3\textwidth]{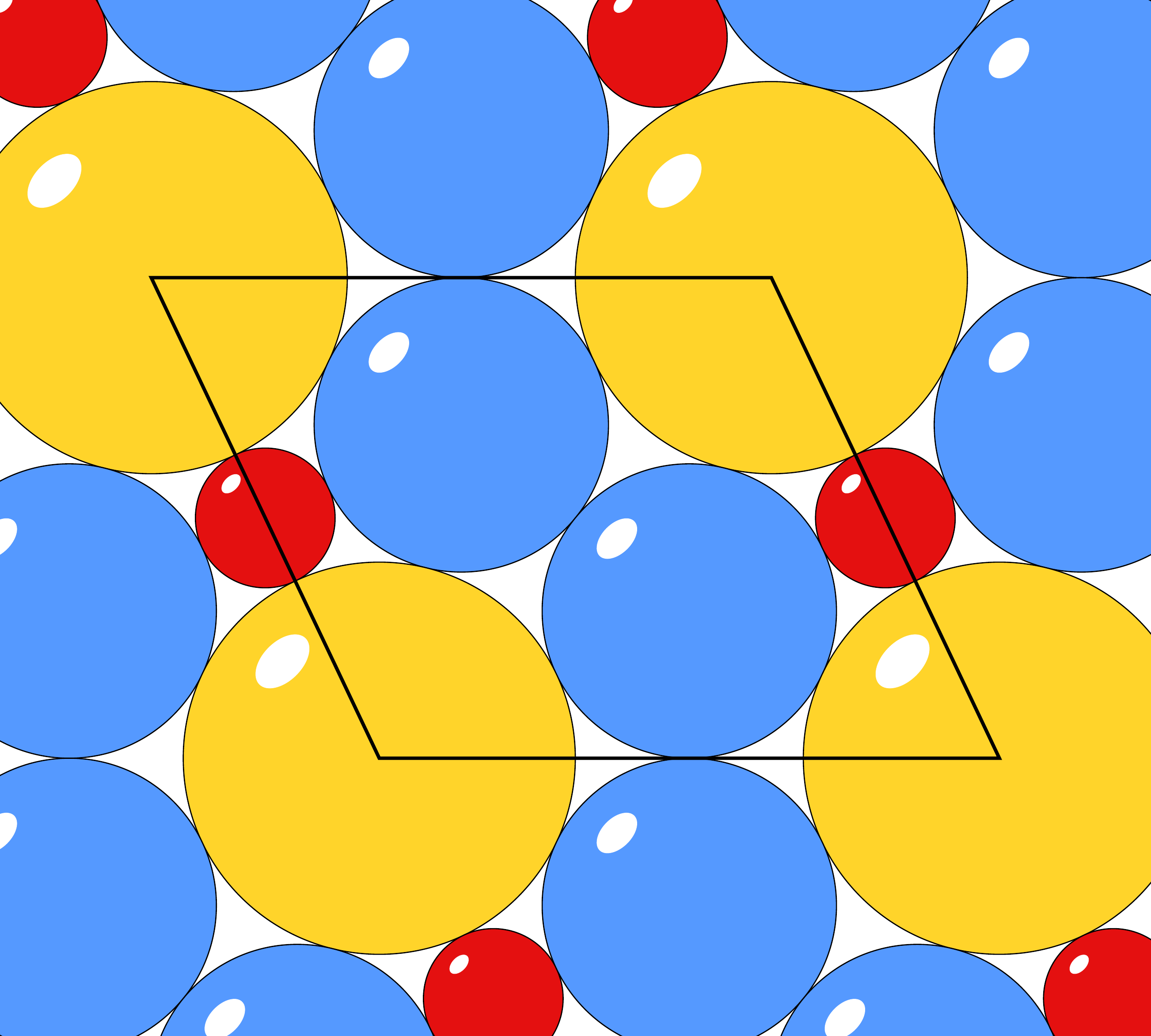}
\end{tabular}
\noindent
\begin{tabular}{lll}
  88\hfill 1r1r / 1r1s1s & 89\hfill 1r1r / 1rr1s & 90\hfill 1r1r / 1rrr1s\\
  \includegraphics[width=0.3\textwidth]{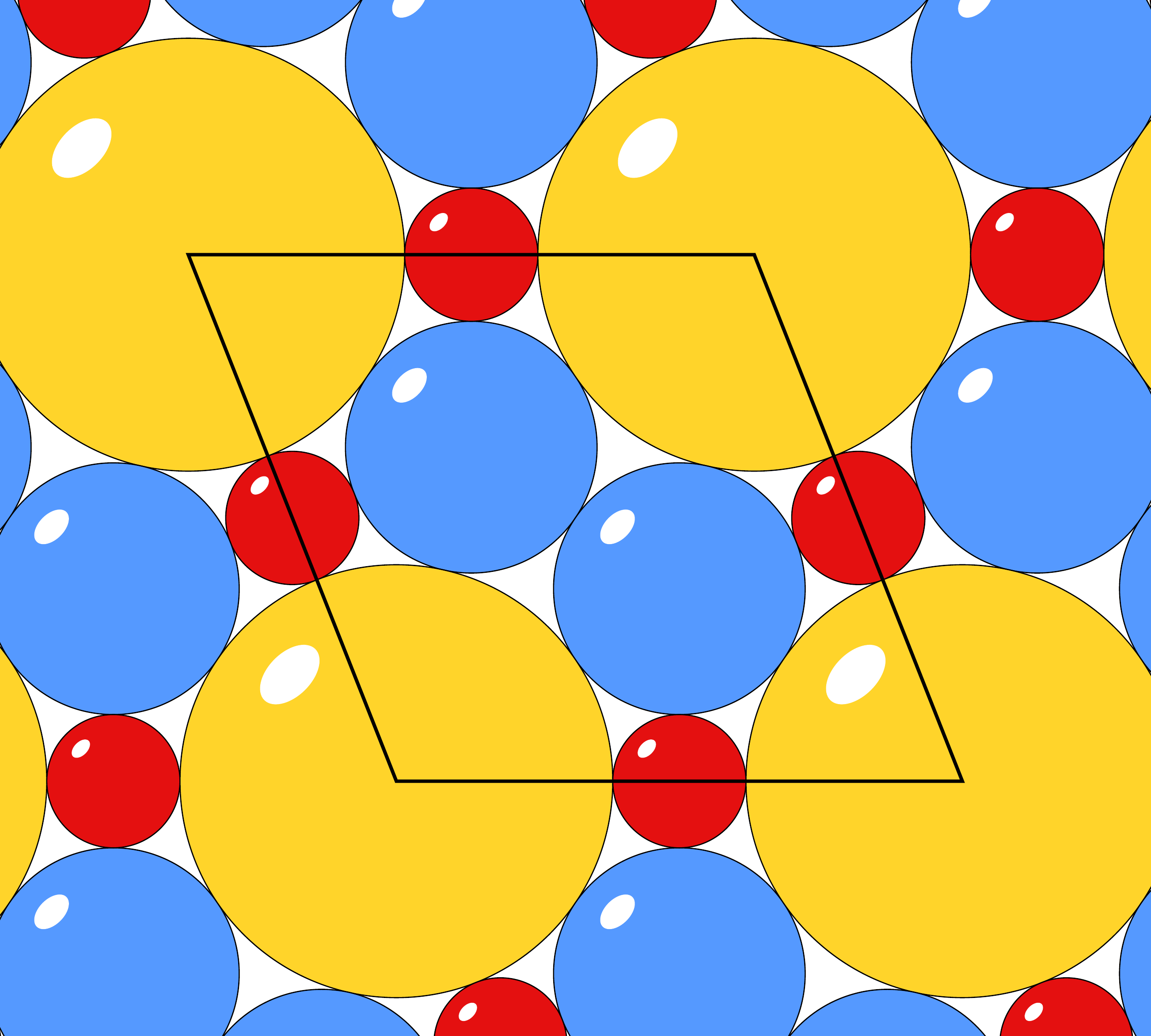} &
  \includegraphics[width=0.3\textwidth]{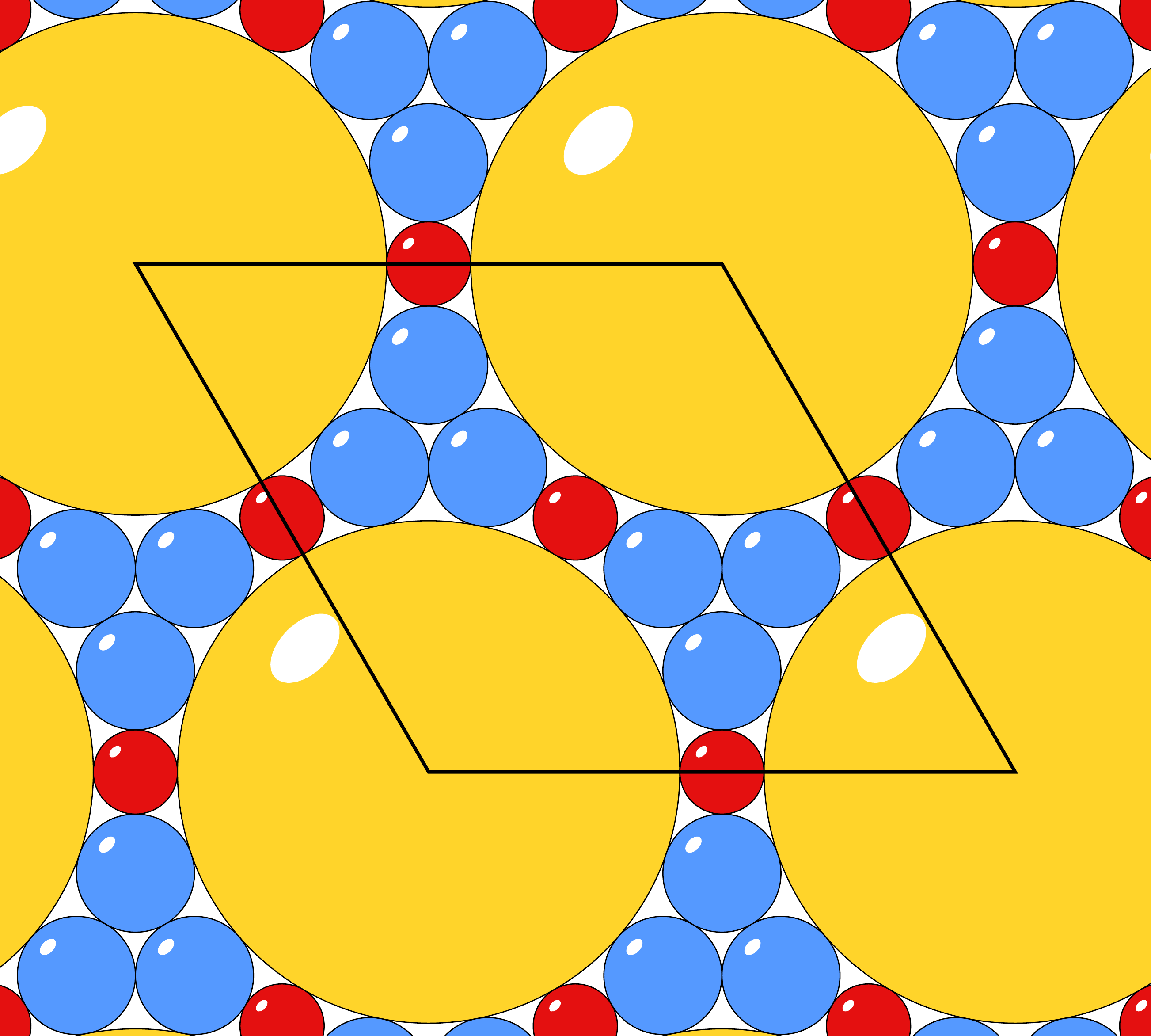} &
  \includegraphics[width=0.3\textwidth]{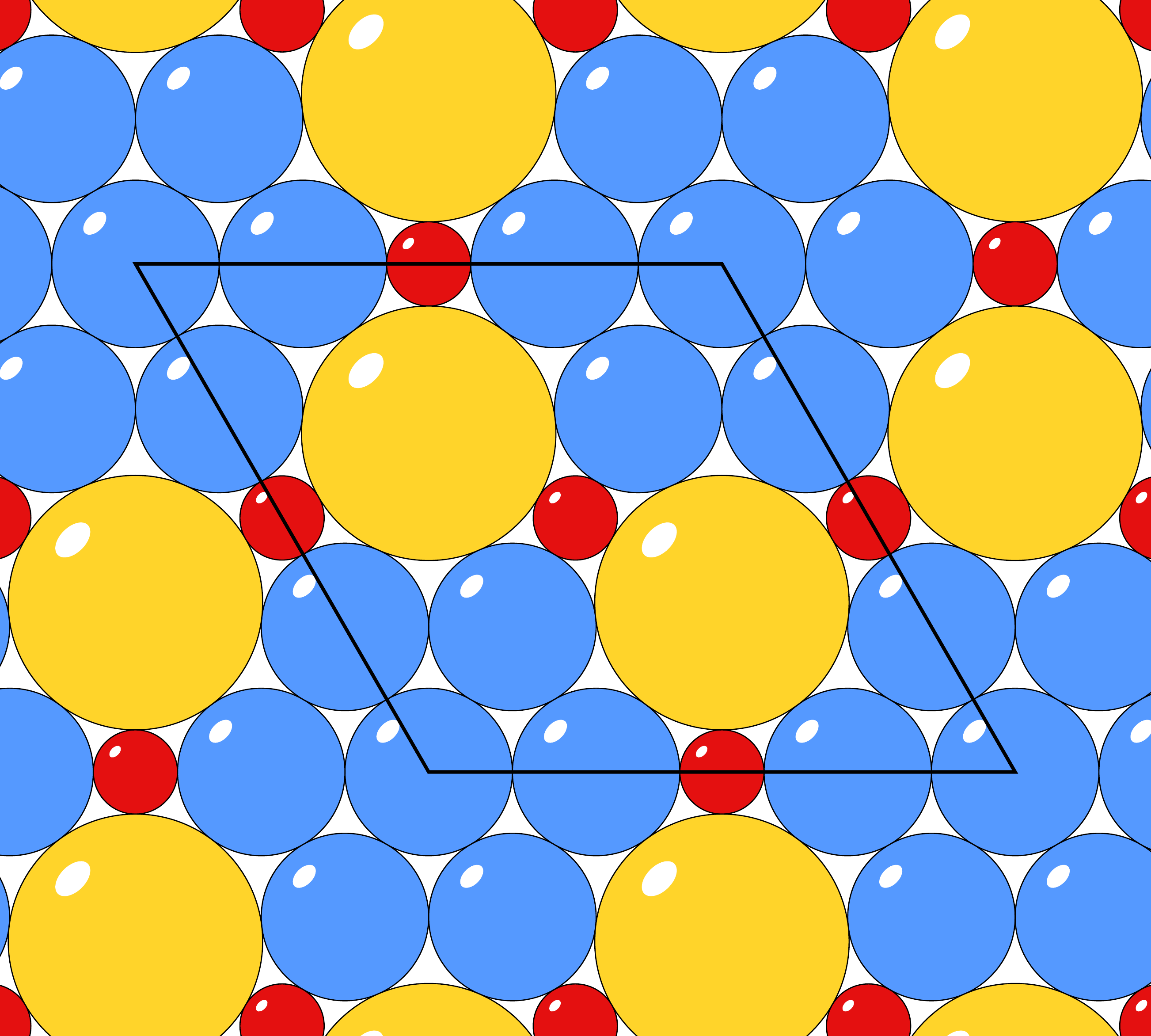}
\end{tabular}
\noindent
\begin{tabular}{lll}
  91\hfill 1r1r / 1s1s1s & 92\hfill 1r1r / 1srsrs & 93\hfill 1r1rr / 1s1srs\\
  \includegraphics[width=0.3\textwidth]{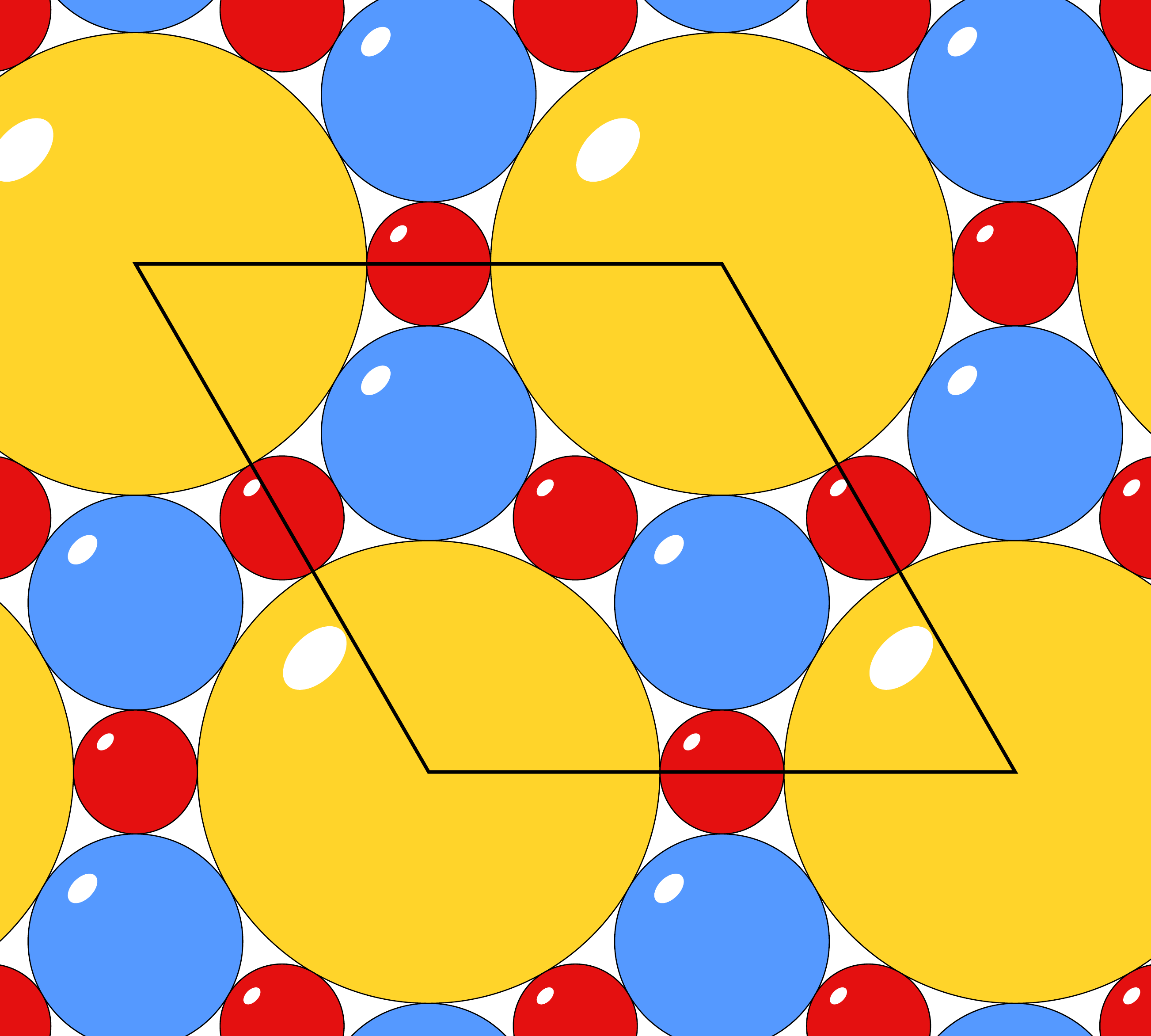} &
  \includegraphics[width=0.3\textwidth]{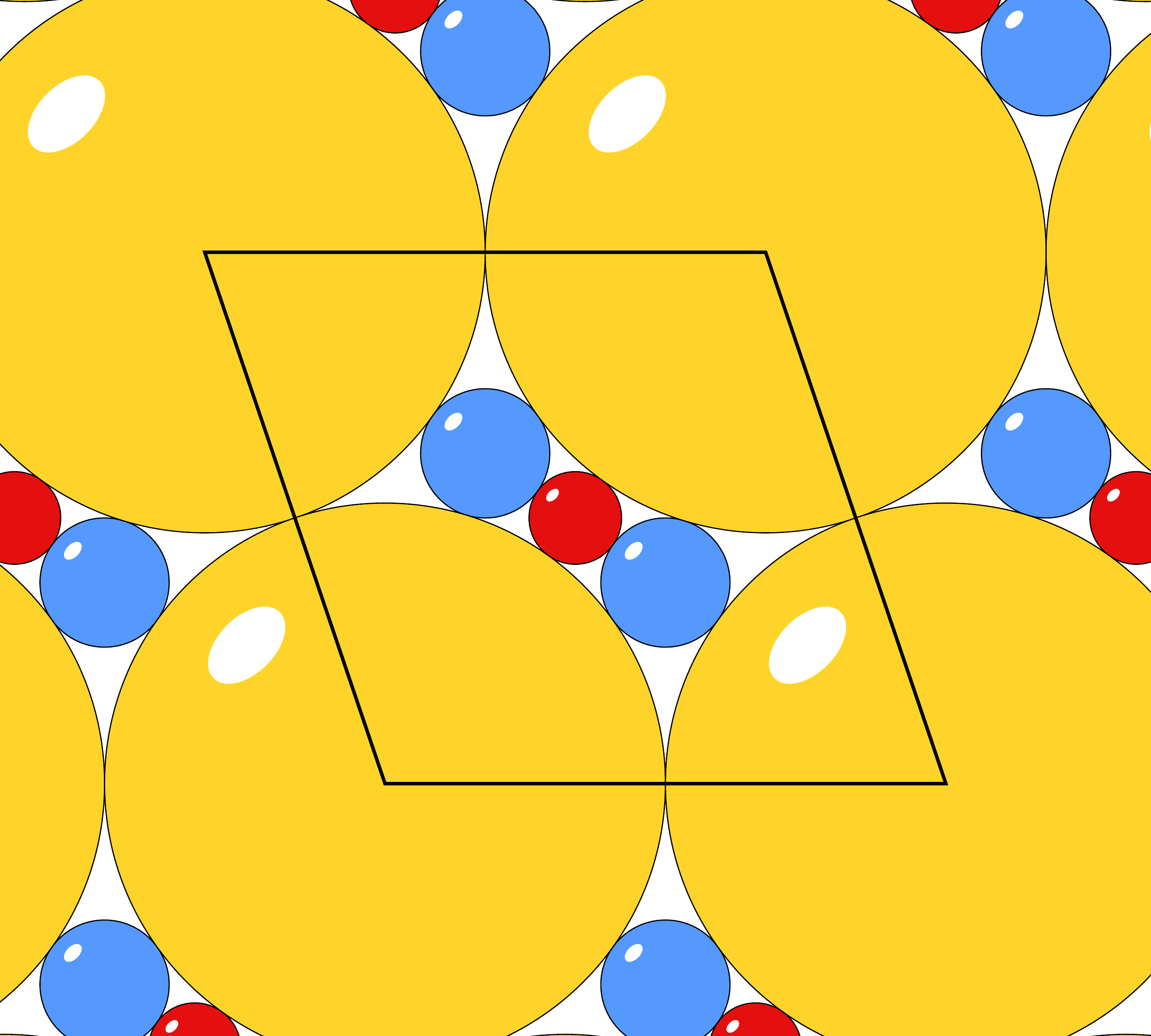} &
  \includegraphics[width=0.3\textwidth]{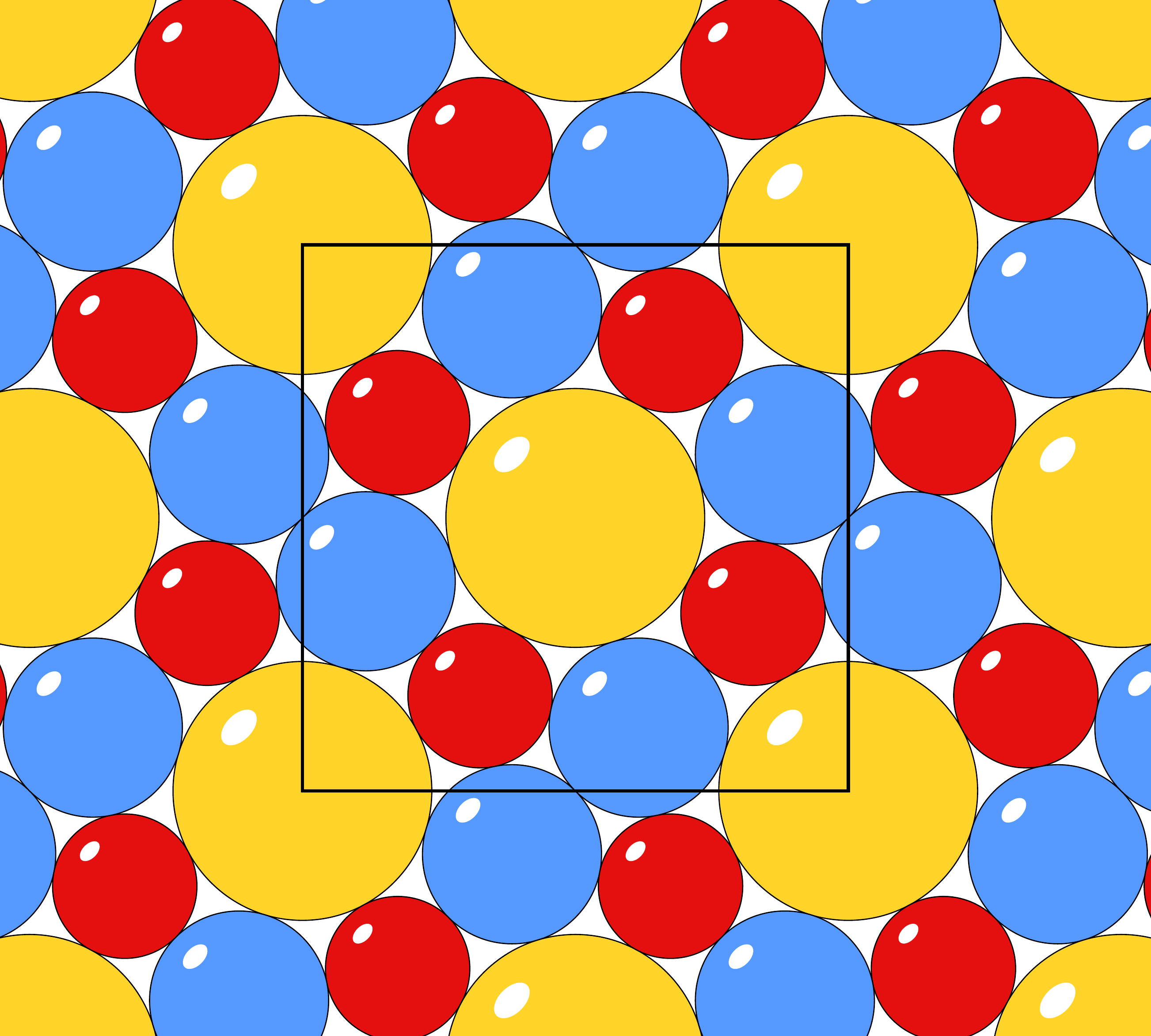}
\end{tabular}
\noindent
\begin{tabular}{lll}
  94\hfill 1r1s / 1111s & 95\hfill 1r1s / 111r1s & 96\hfill 1r1s / 111s1s\\
  \includegraphics[width=0.3\textwidth]{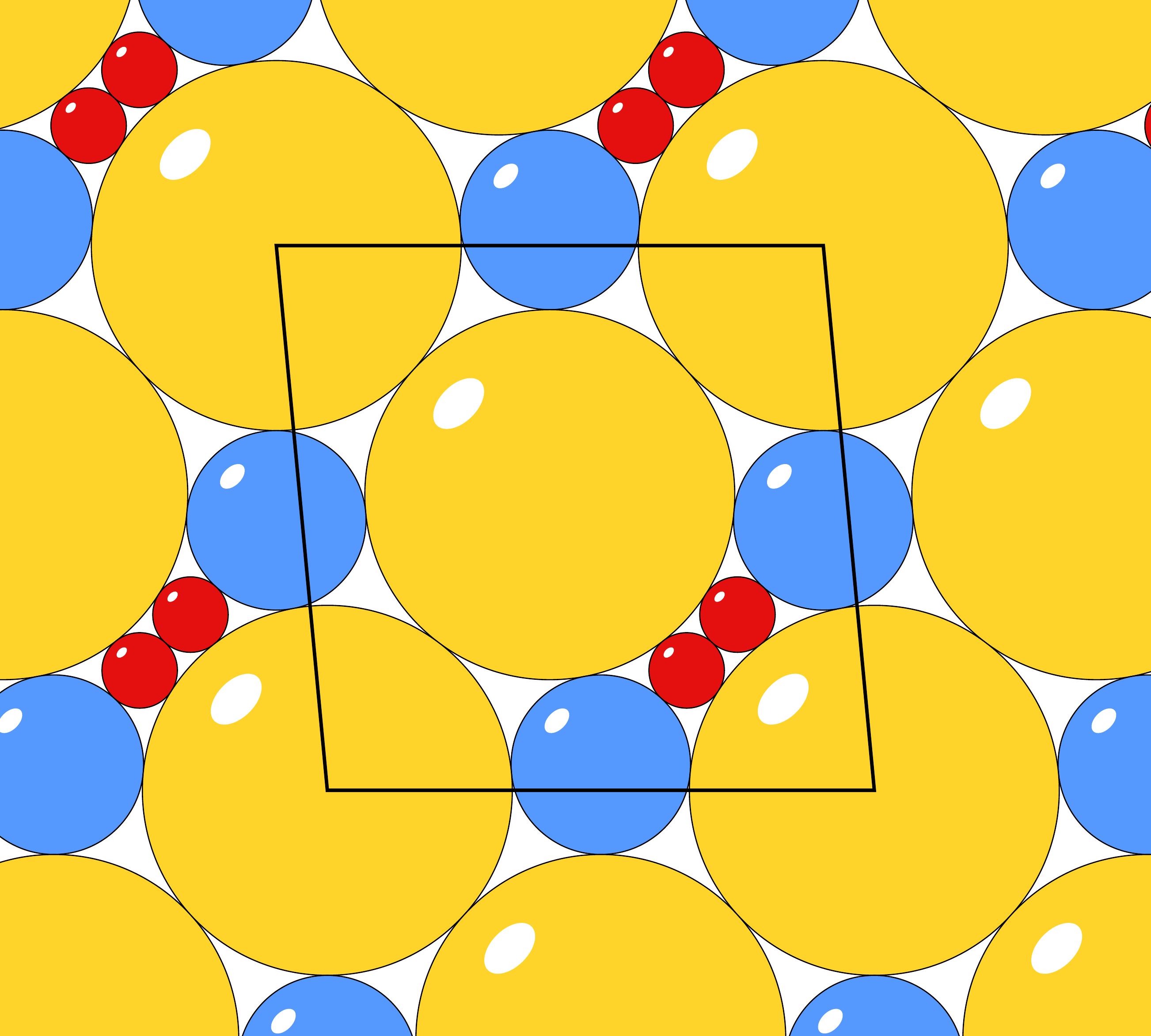} &
  \includegraphics[width=0.3\textwidth]{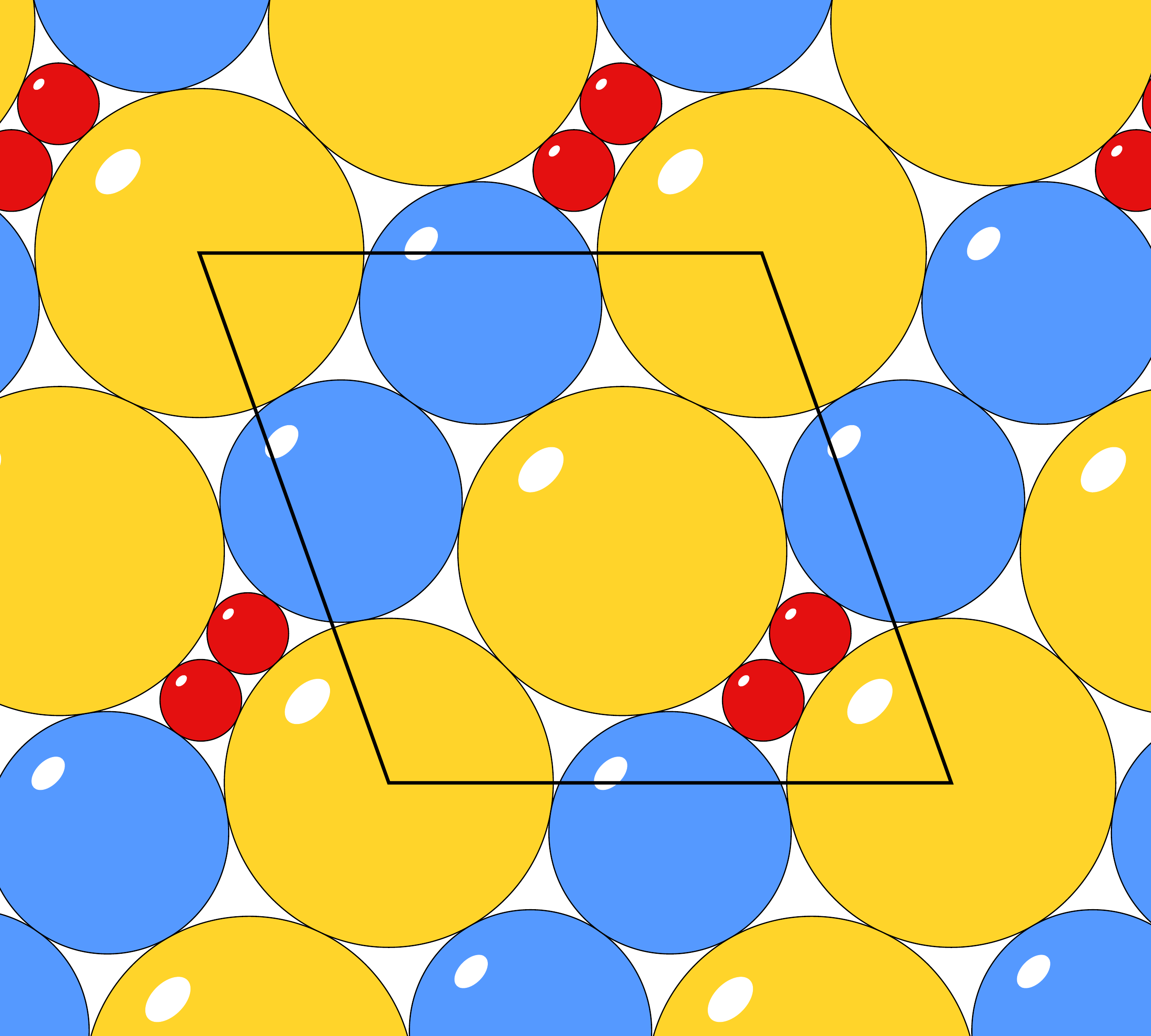} &
  \includegraphics[width=0.3\textwidth]{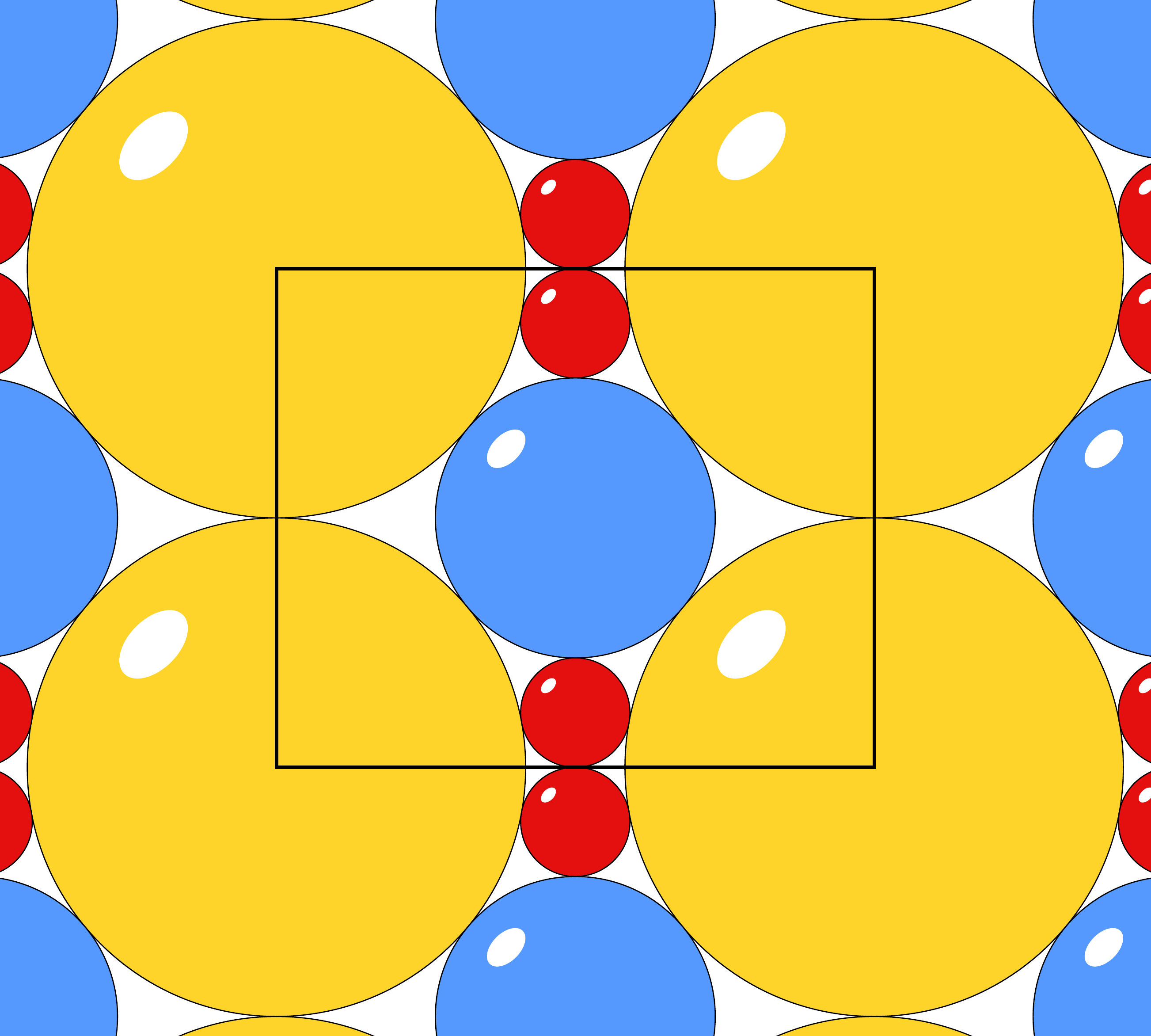}
\end{tabular}
\noindent
\begin{tabular}{lll}
  97\hfill 1r1s / 11r1s & 98\hfill 1r1s / 11r1s1s & 99\hfill 1r1s / 11rr1s\\
  \includegraphics[width=0.3\textwidth]{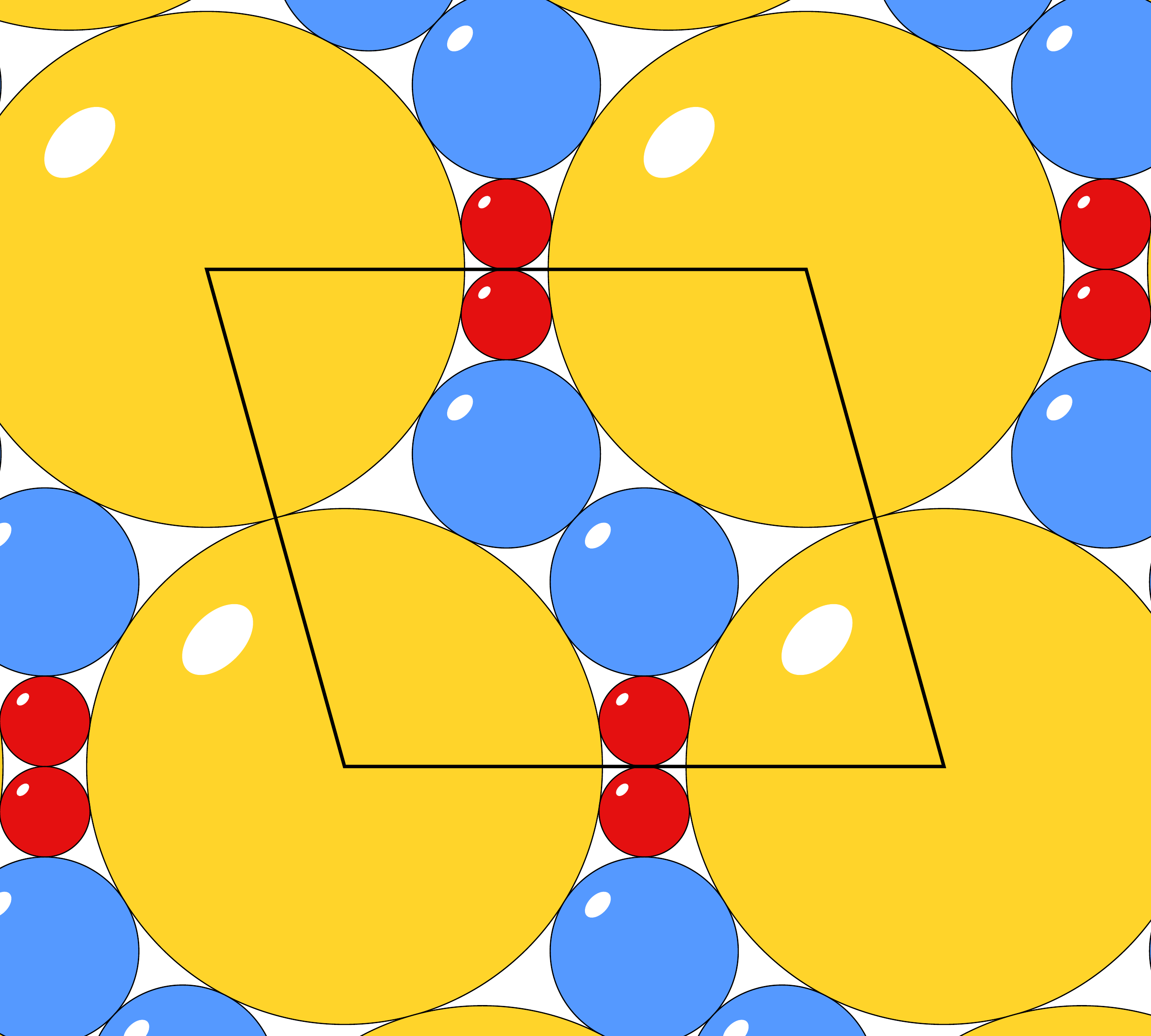} &
  \includegraphics[width=0.3\textwidth]{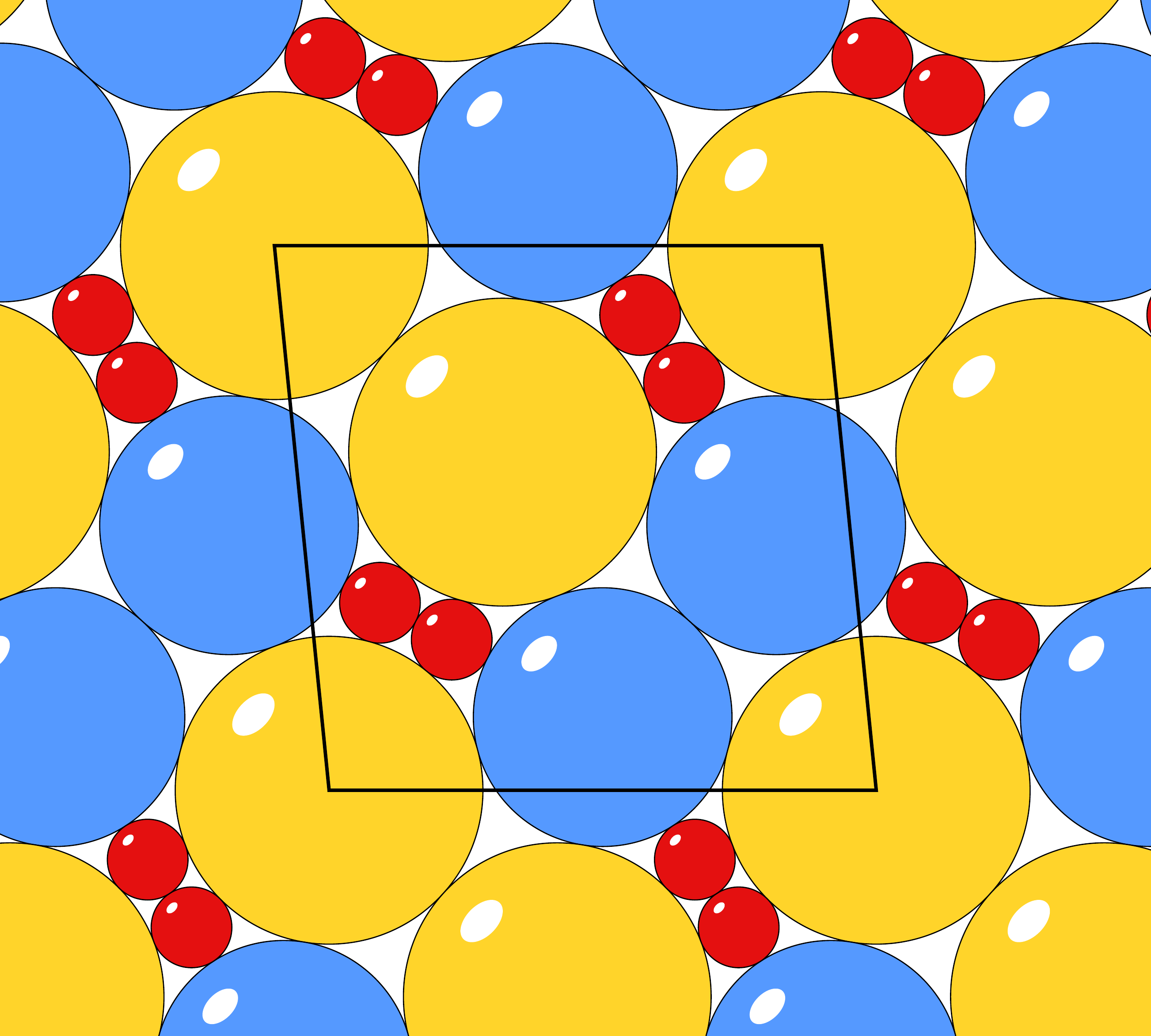} &
  \includegraphics[width=0.3\textwidth]{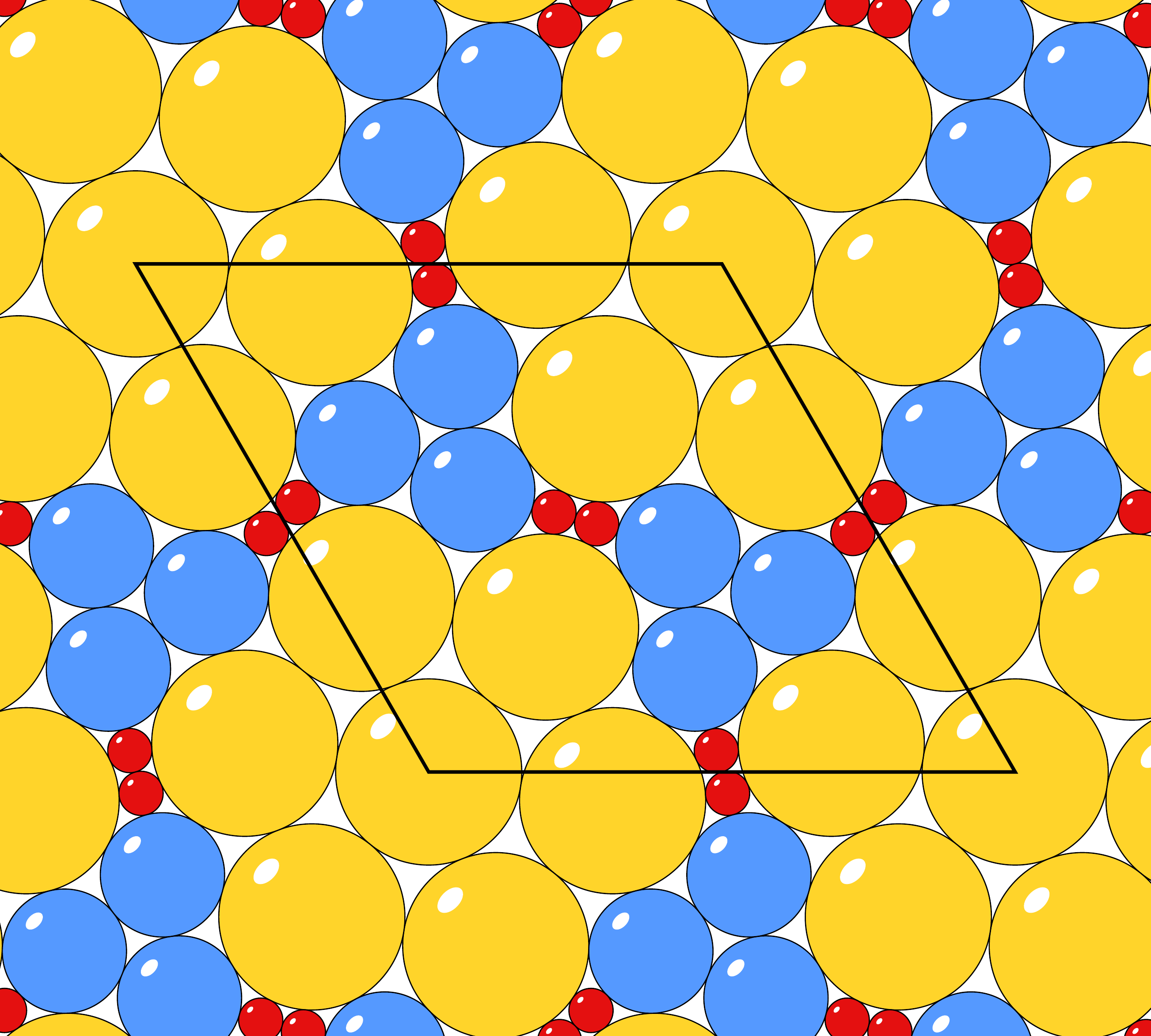}
\end{tabular}
\noindent
\begin{tabular}{lll}
  100\hfill 1r1s / 11s1s & 101\hfill 1r1s / 11s1s1s & 102\hfill 1r1s / 1r1r1s\\
  \includegraphics[width=0.3\textwidth]{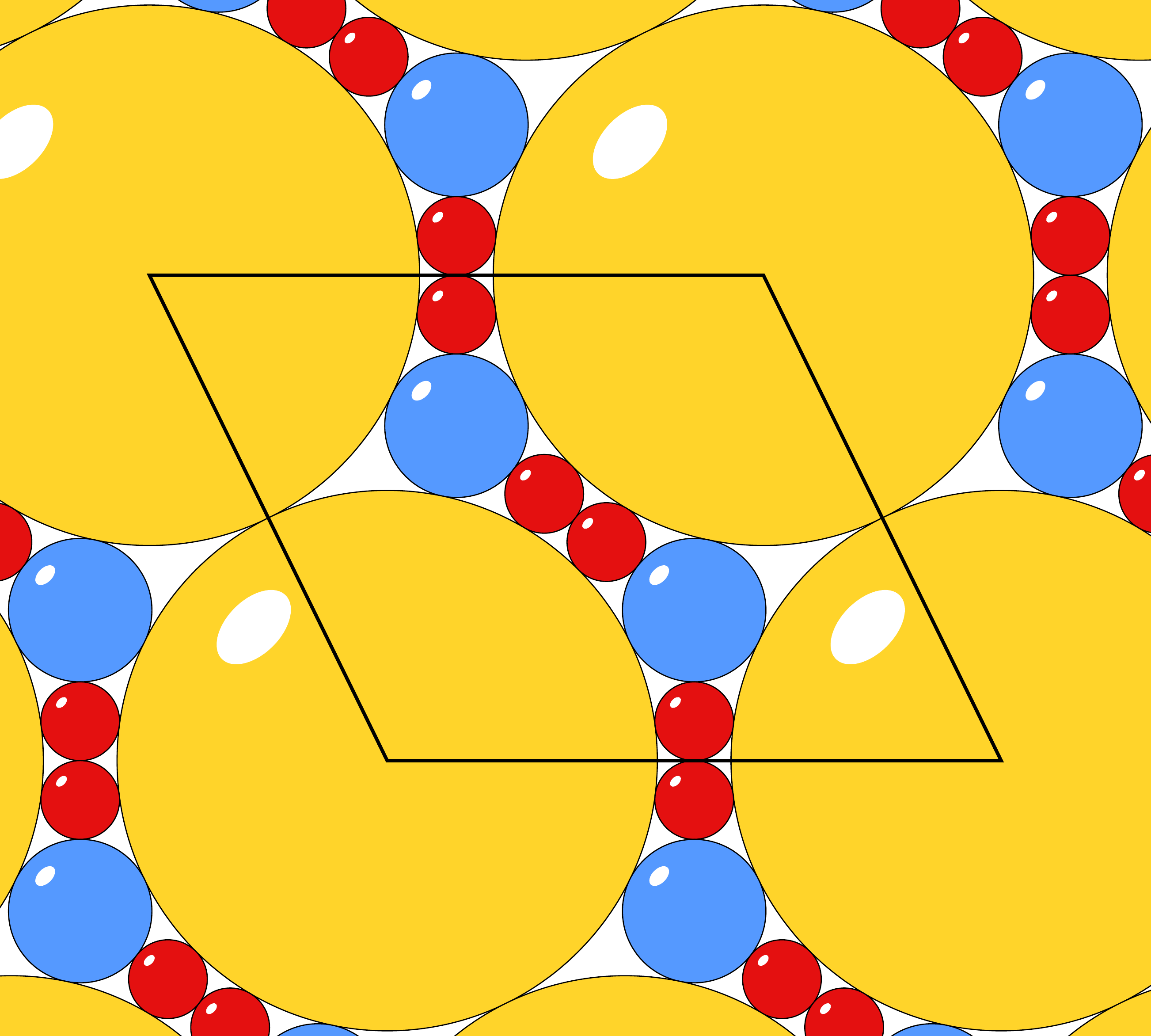} &
  \includegraphics[width=0.3\textwidth]{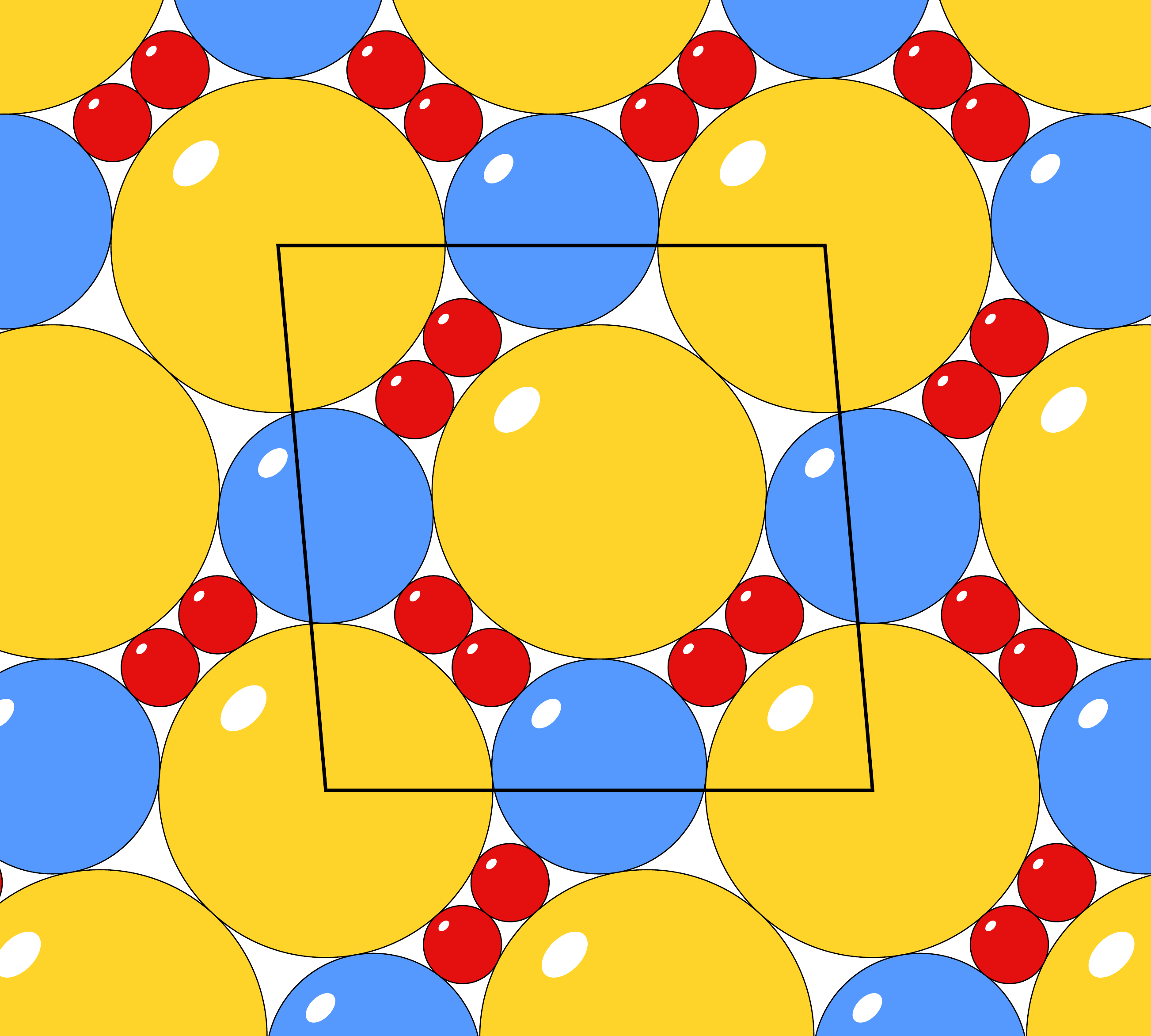} &
  \includegraphics[width=0.3\textwidth]{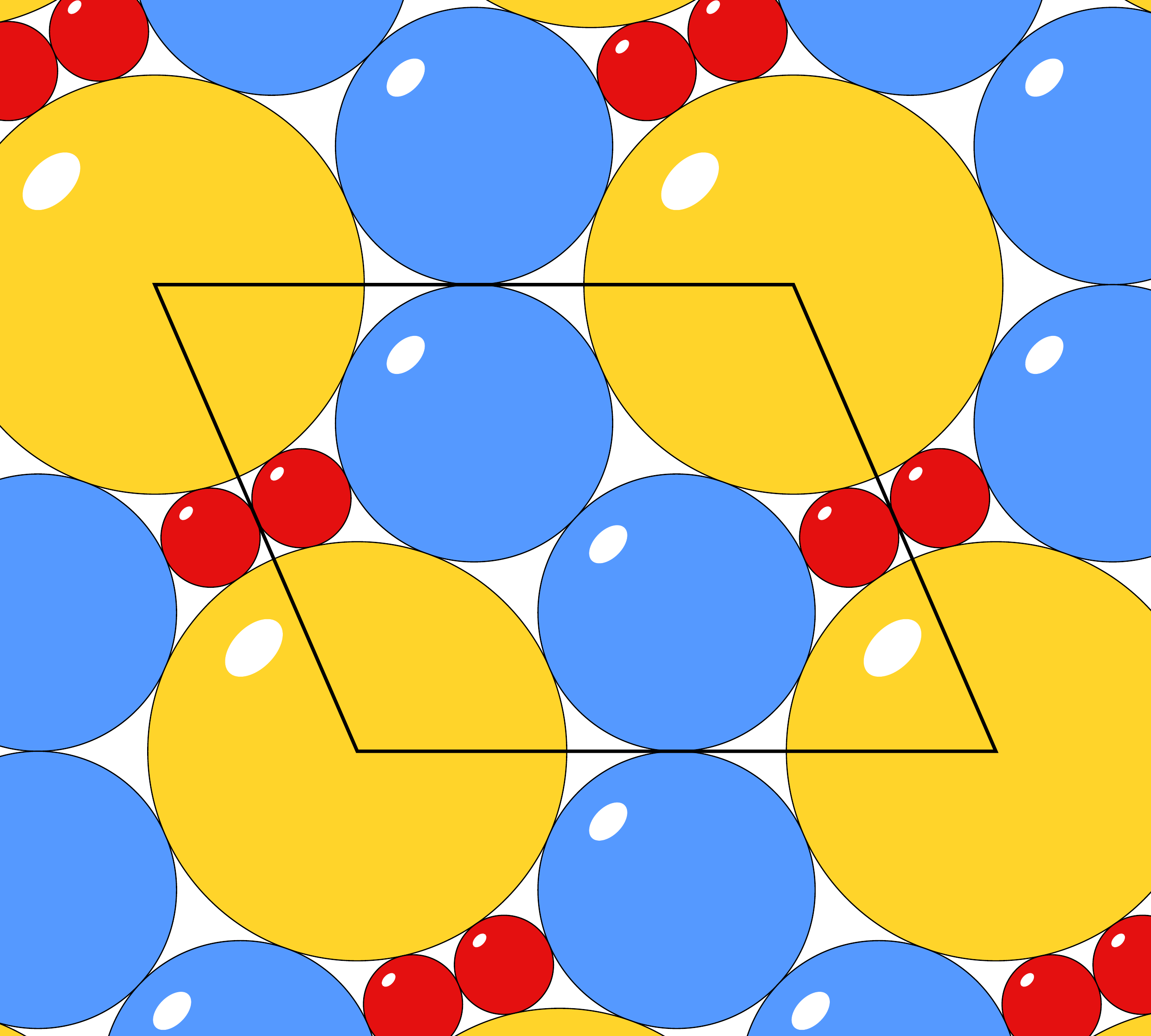}
\end{tabular}
\noindent
\begin{tabular}{lll}
  103\hfill 1r1s / 1r1s1s & 104\hfill 1r1s / 1r1s1s1s & 105\hfill 1r1s / 1rr1s\\
  \includegraphics[width=0.3\textwidth]{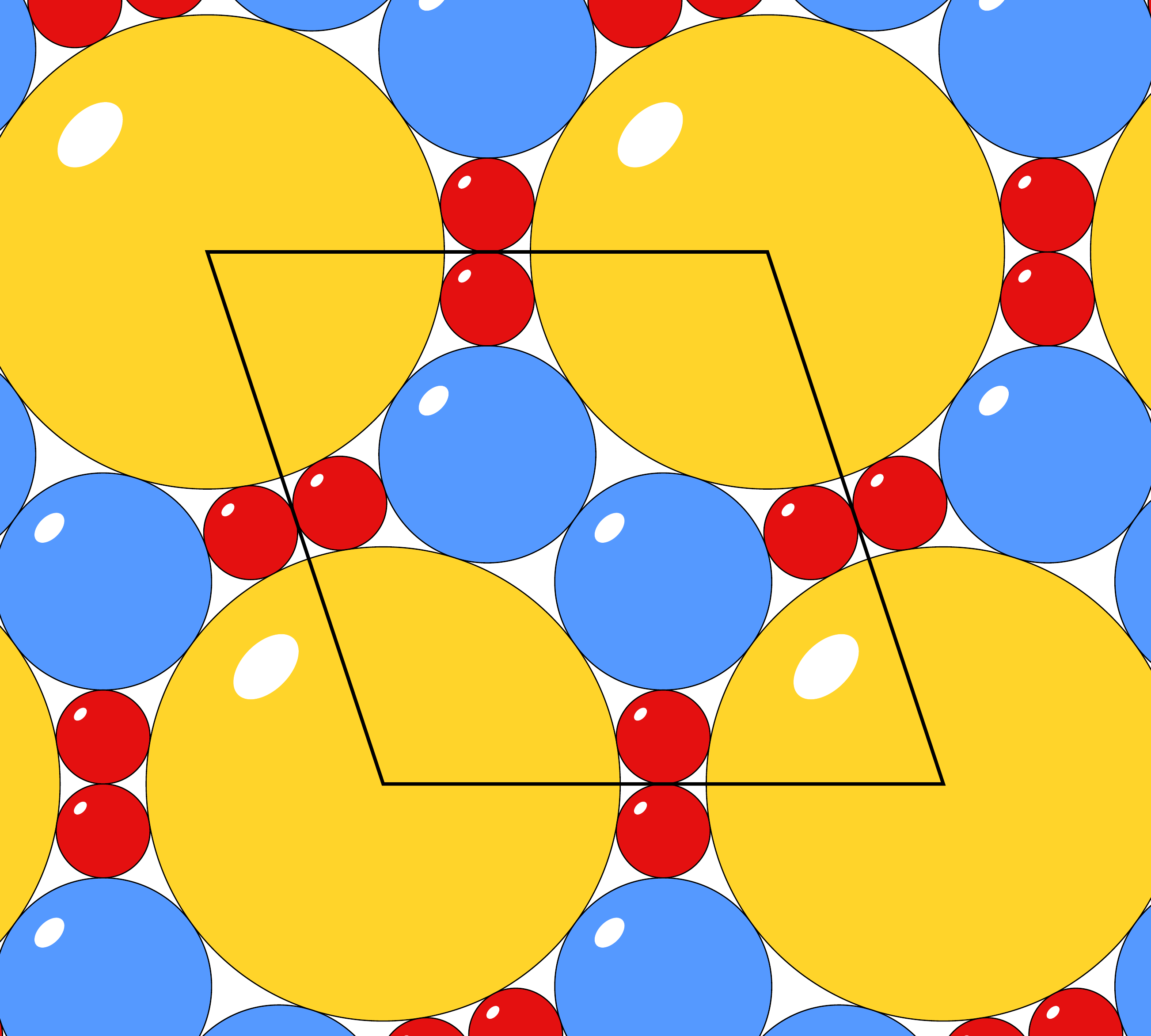} &
  \includegraphics[width=0.3\textwidth]{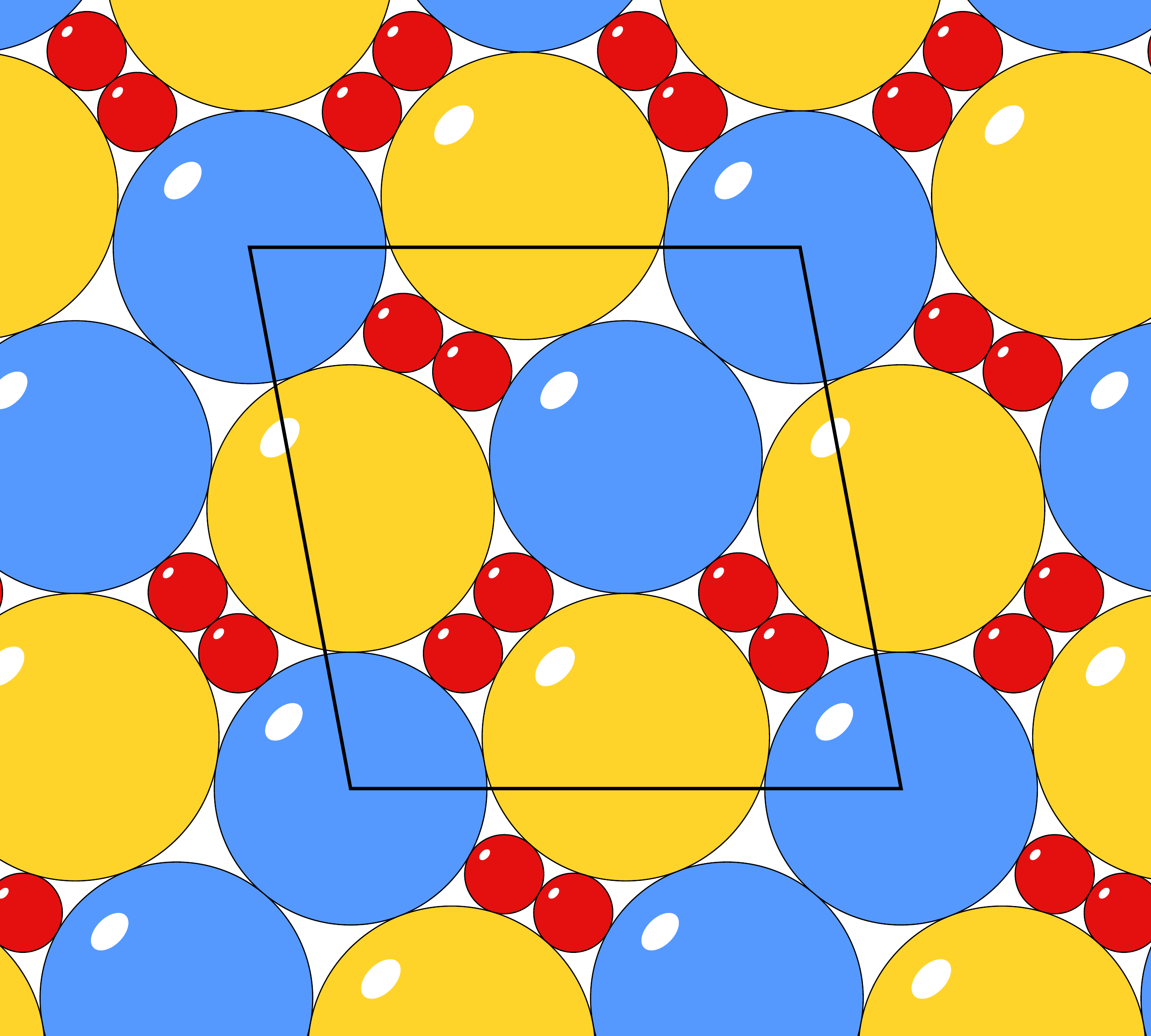} &
  \includegraphics[width=0.3\textwidth]{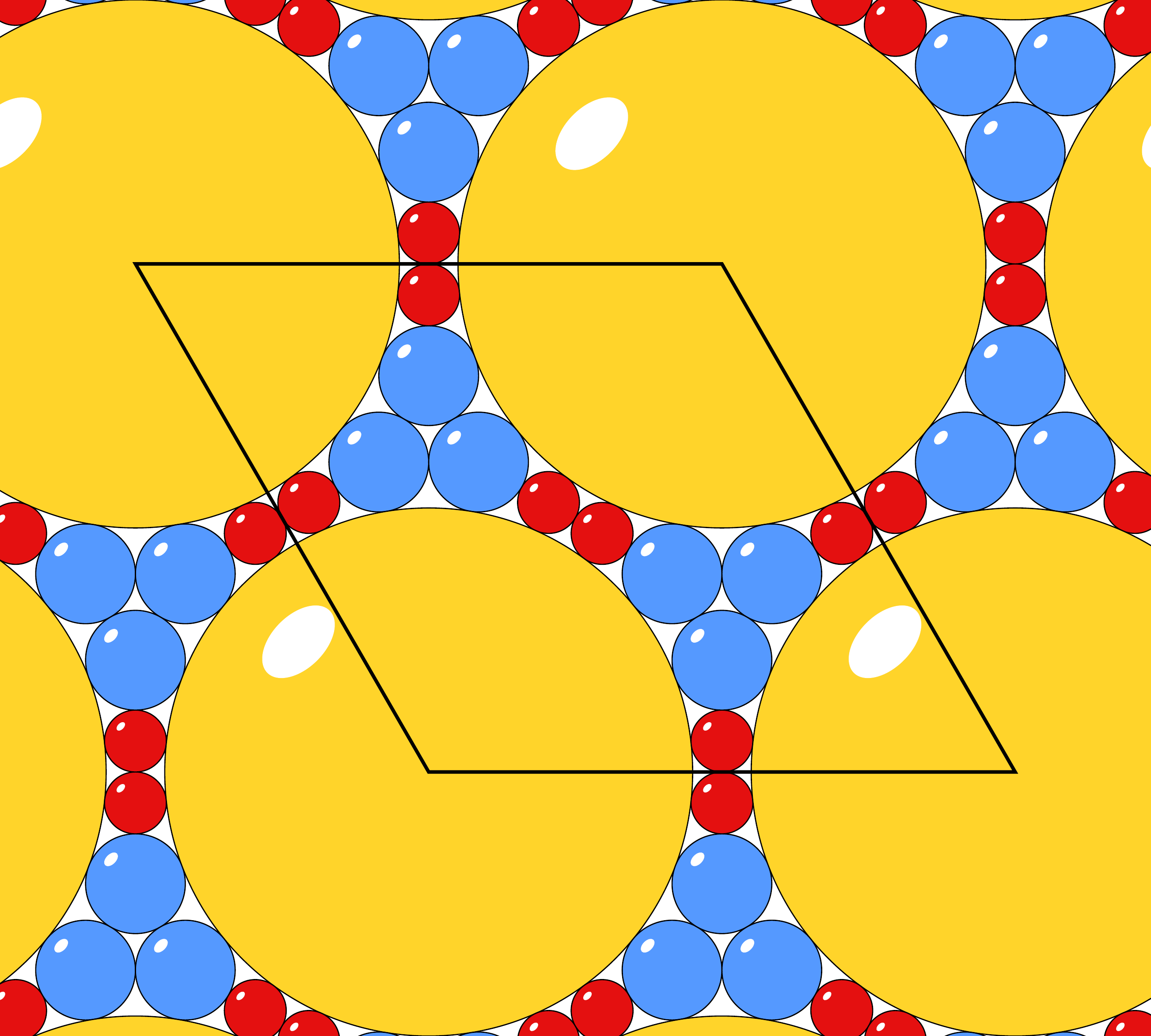}
\end{tabular}
\noindent
\begin{tabular}{lll}
  106\hfill 1r1s / 1rrr1s & 107\hfill 1r1s / 1s1s1s & 108\hfill 1r1s / 1s1s1s1s\\
  \includegraphics[width=0.3\textwidth]{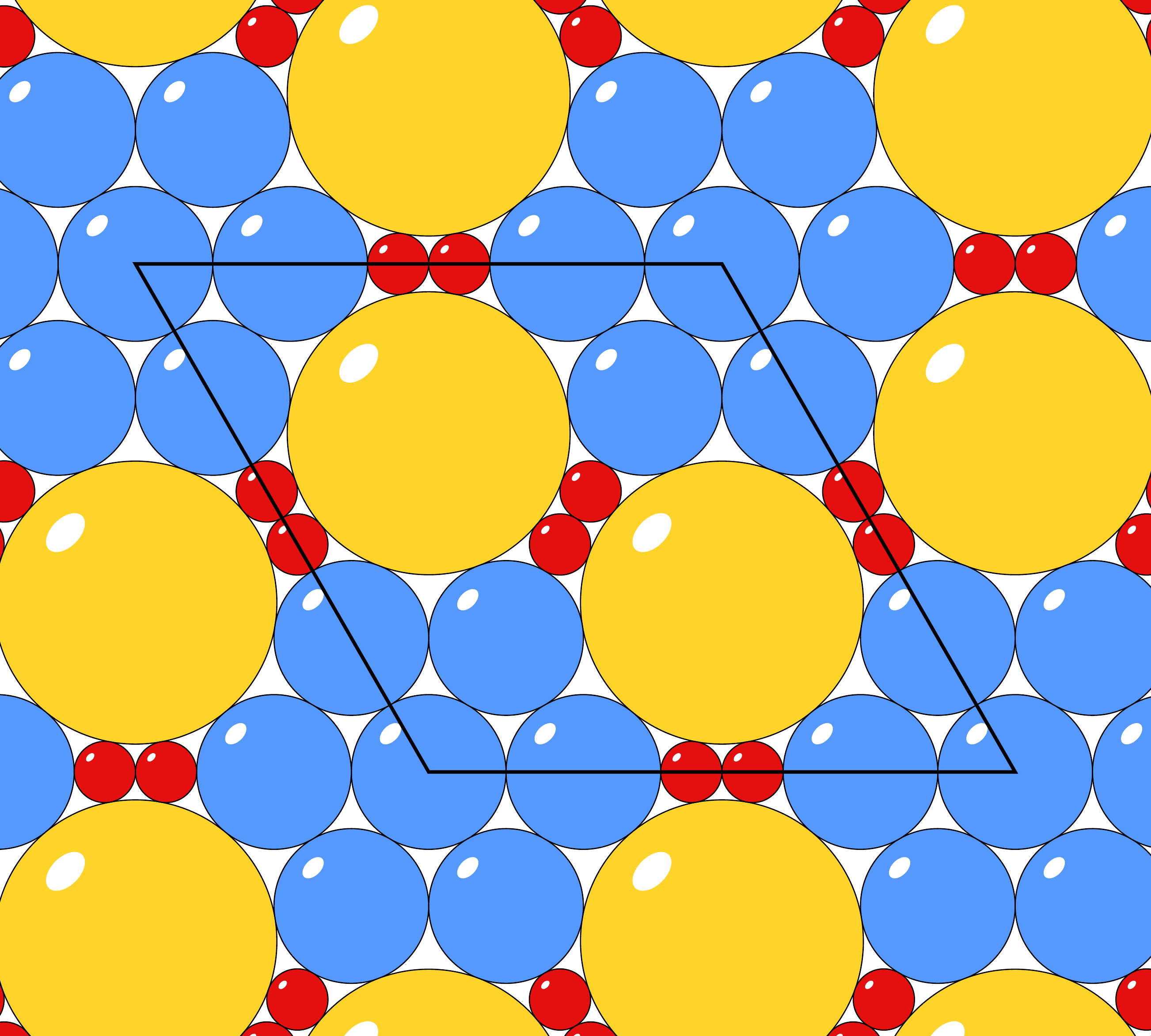} &
  \includegraphics[width=0.3\textwidth]{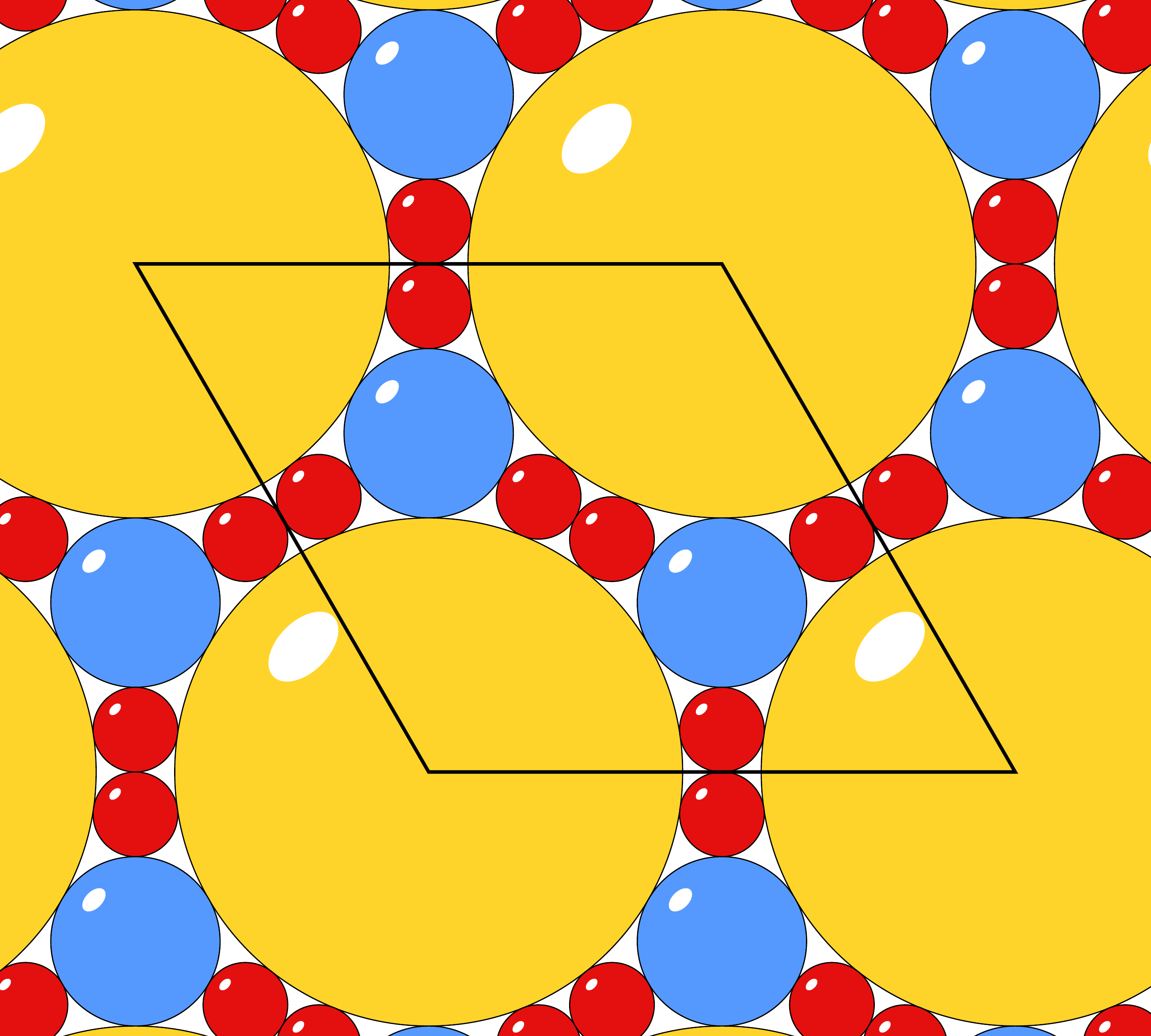} &
  \includegraphics[width=0.3\textwidth]{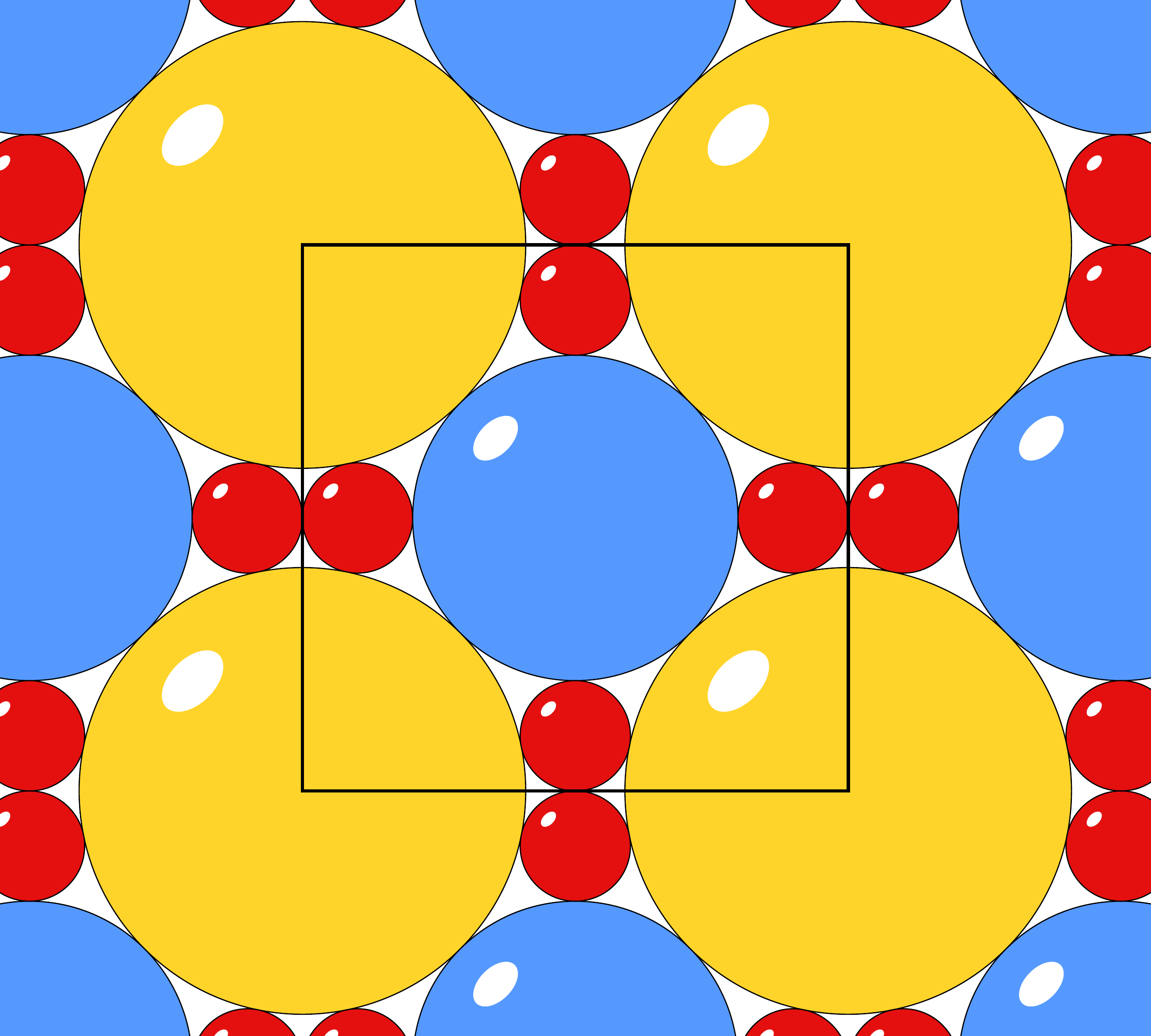}
\end{tabular}
\noindent
\begin{tabular}{lll}
  109\hfill 1r1s / 1s1sss & 110\hfill 1r1ss / 111s1s & 111\hfill 1r1ss / 11r1s\\
  \includegraphics[width=0.3\textwidth]{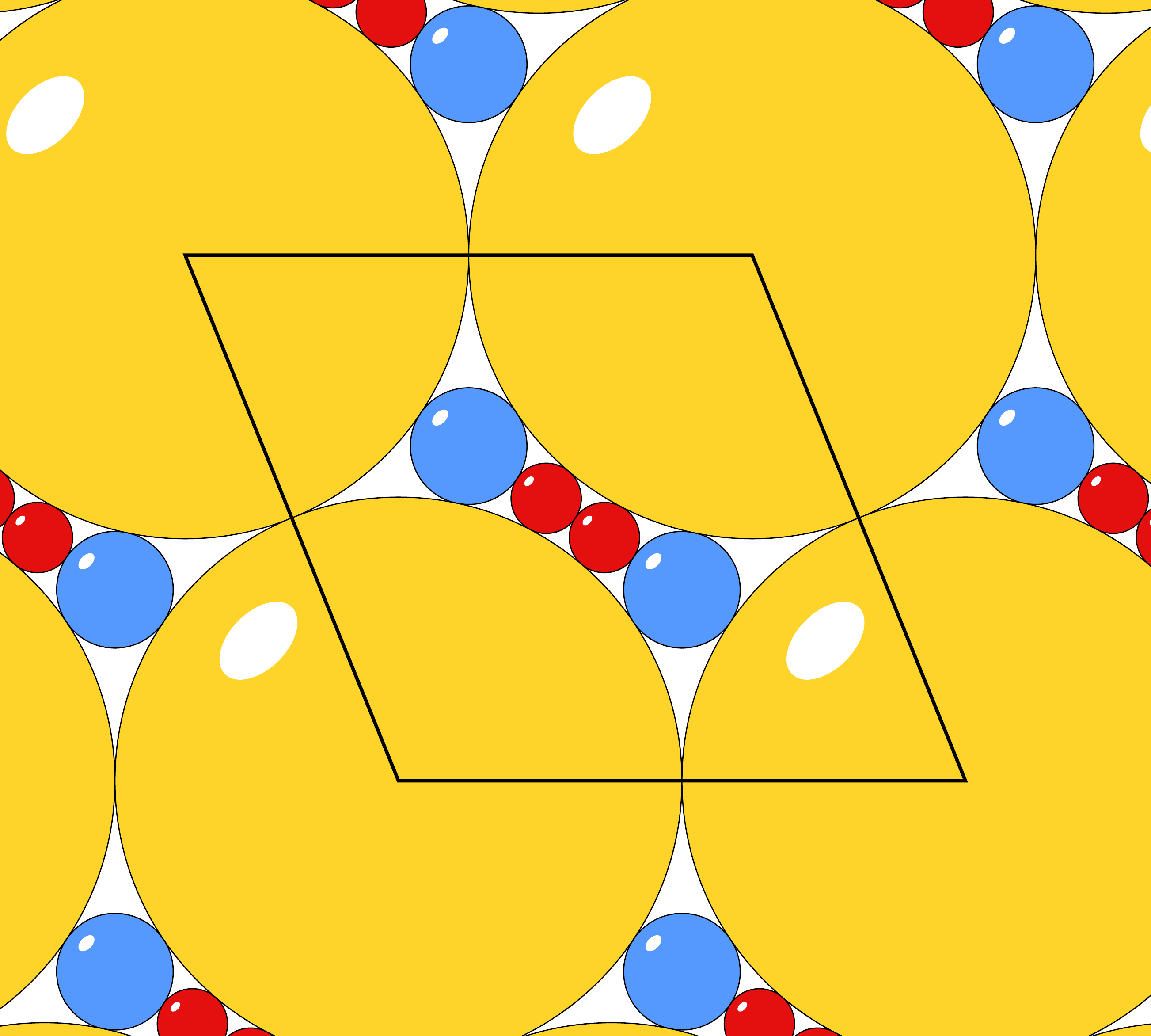} &
  \includegraphics[width=0.3\textwidth]{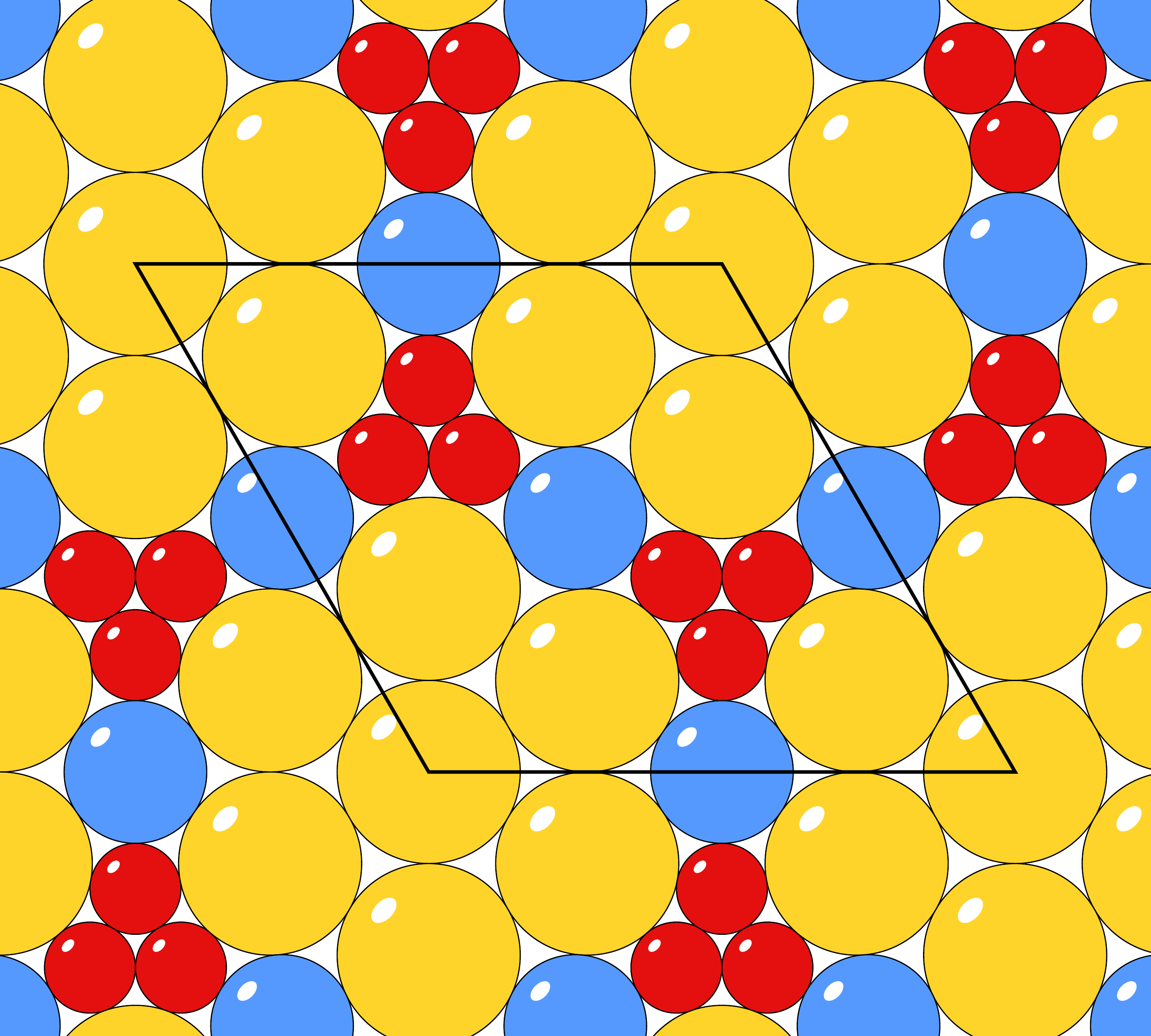} &
  \includegraphics[width=0.3\textwidth]{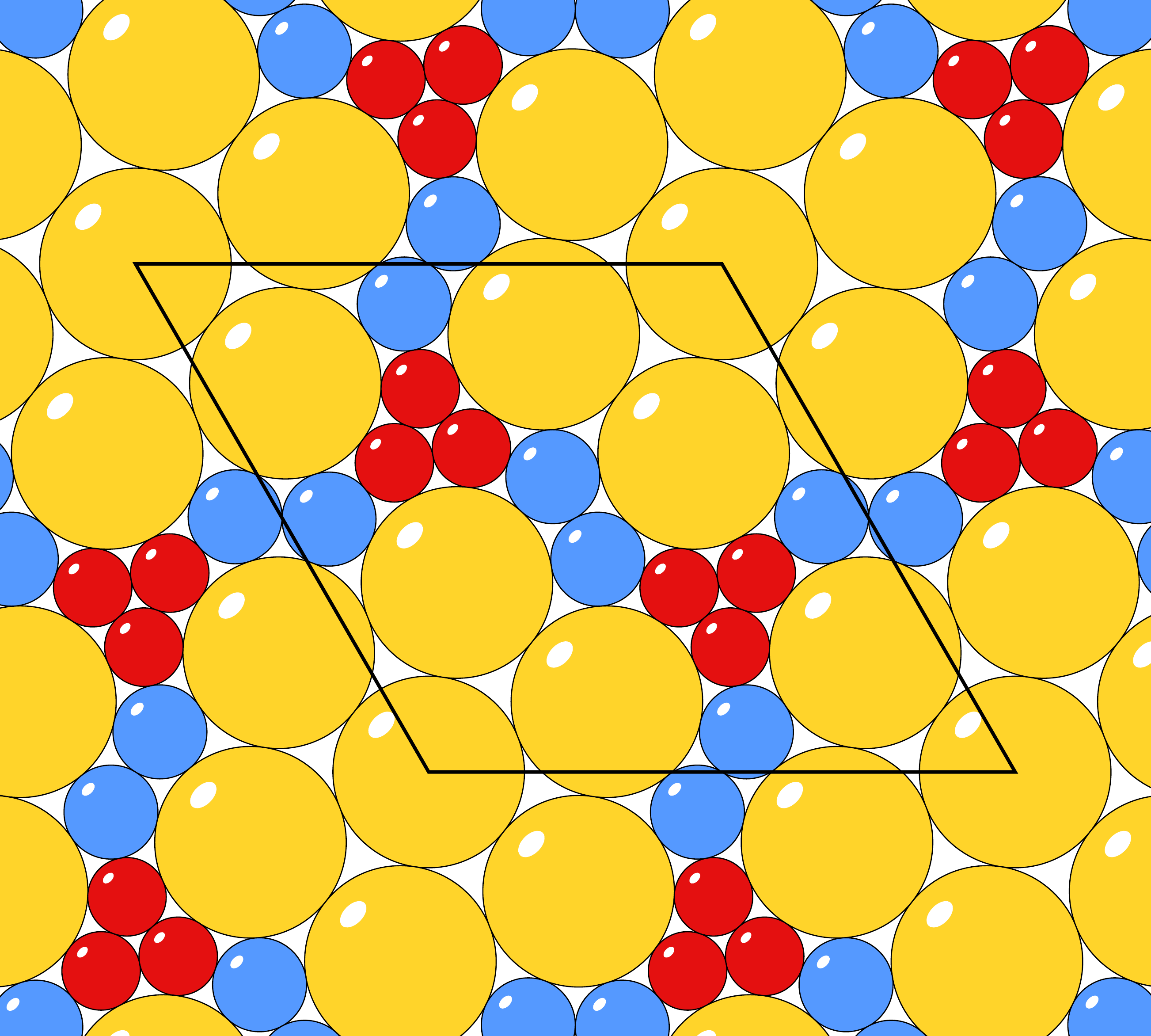}
\end{tabular}
\noindent
\begin{tabular}{lll}
  112\hfill 1r1ss / 11rr1s & 113\hfill 1r1ss / 11s1s & 114\hfill 1r1ss / 1rrr1s\\
  \includegraphics[width=0.3\textwidth]{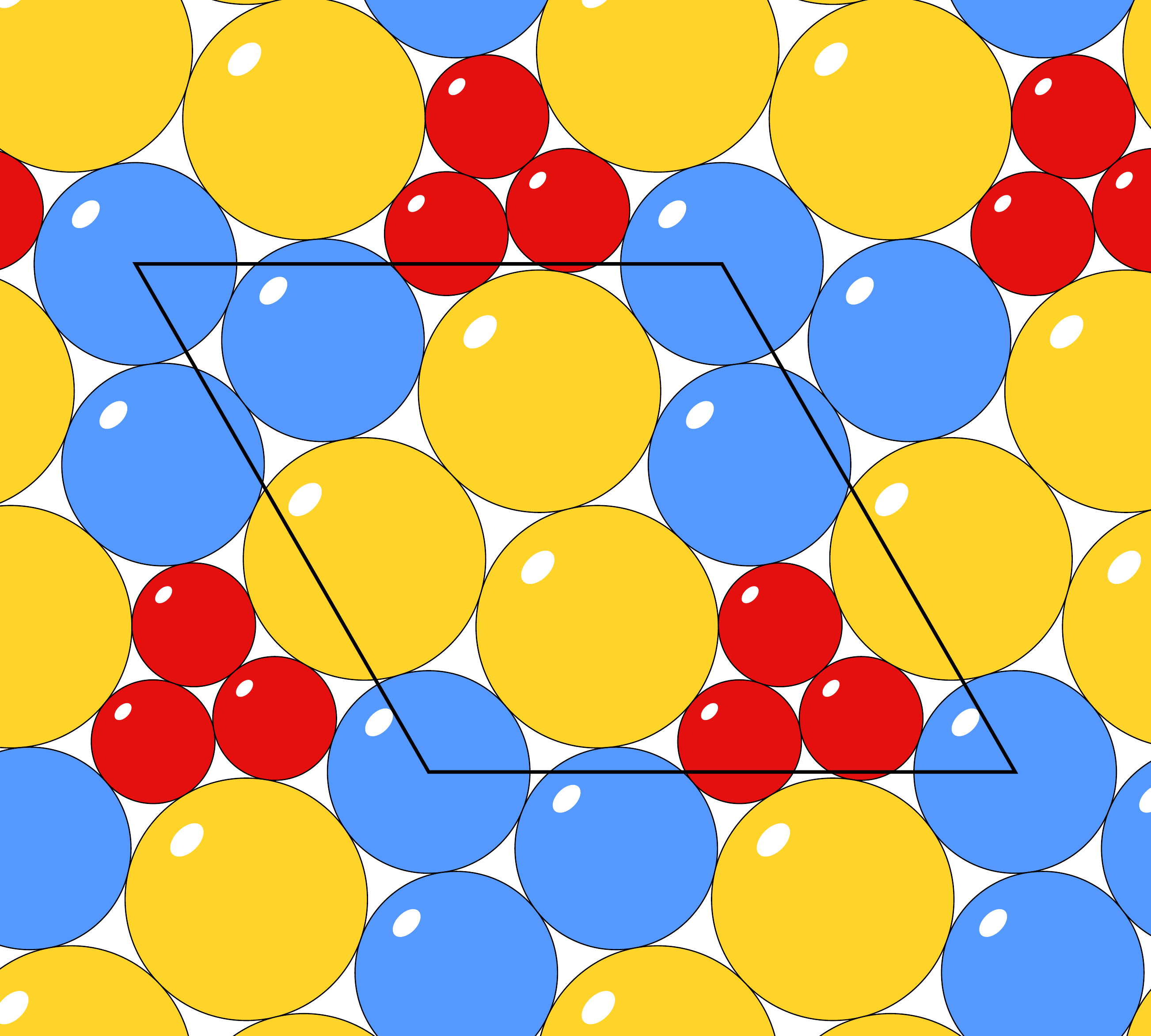} &
  \includegraphics[width=0.3\textwidth]{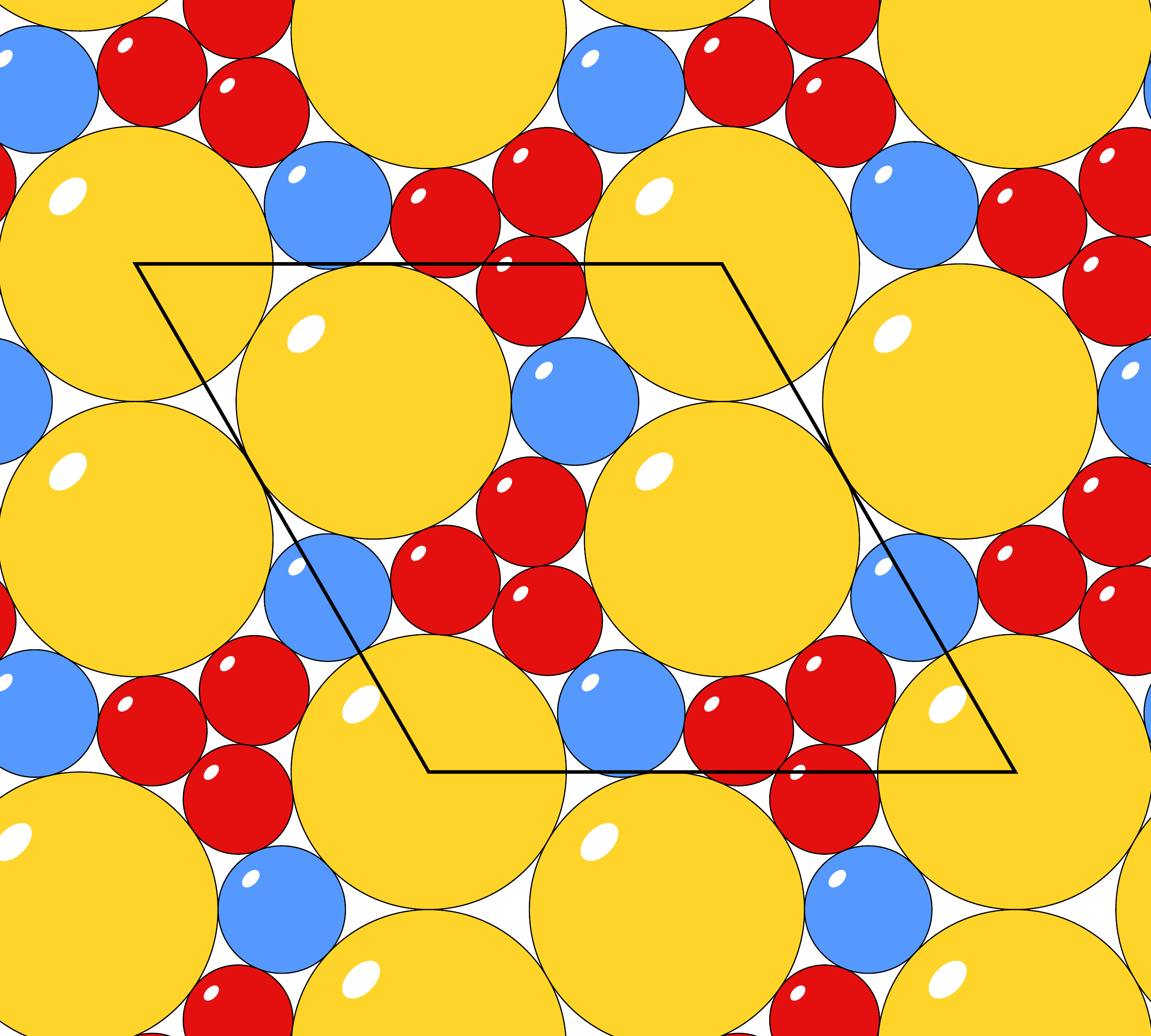} &
  \includegraphics[width=0.3\textwidth]{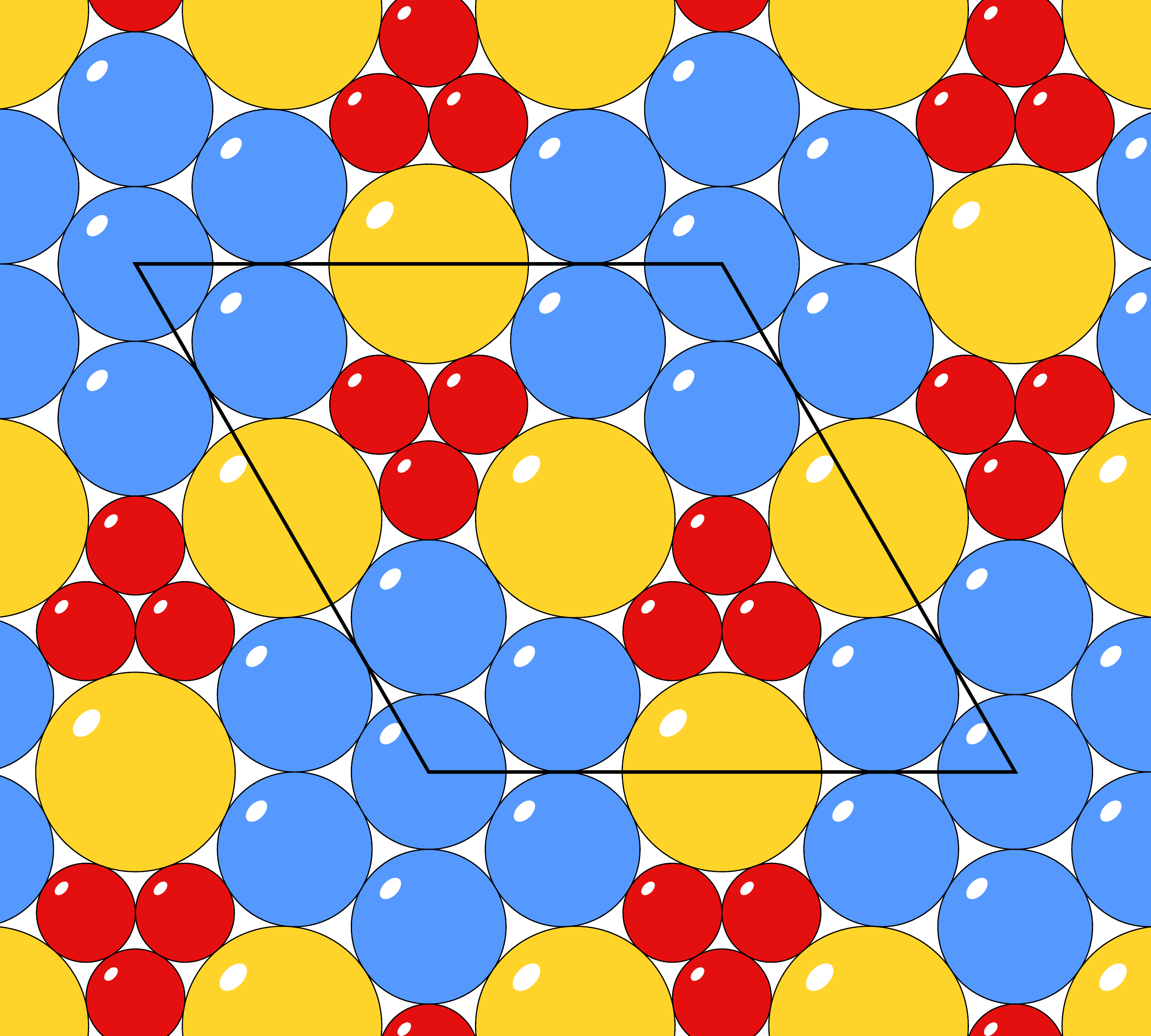}
\end{tabular}
\noindent
\begin{tabular}{lll}
  115\hfill 1r1ss / 1s1s1s & 116\hfill 1rr1s / 111srs & 117\hfill 1rr1s / 11srrs\\
  \includegraphics[width=0.3\textwidth]{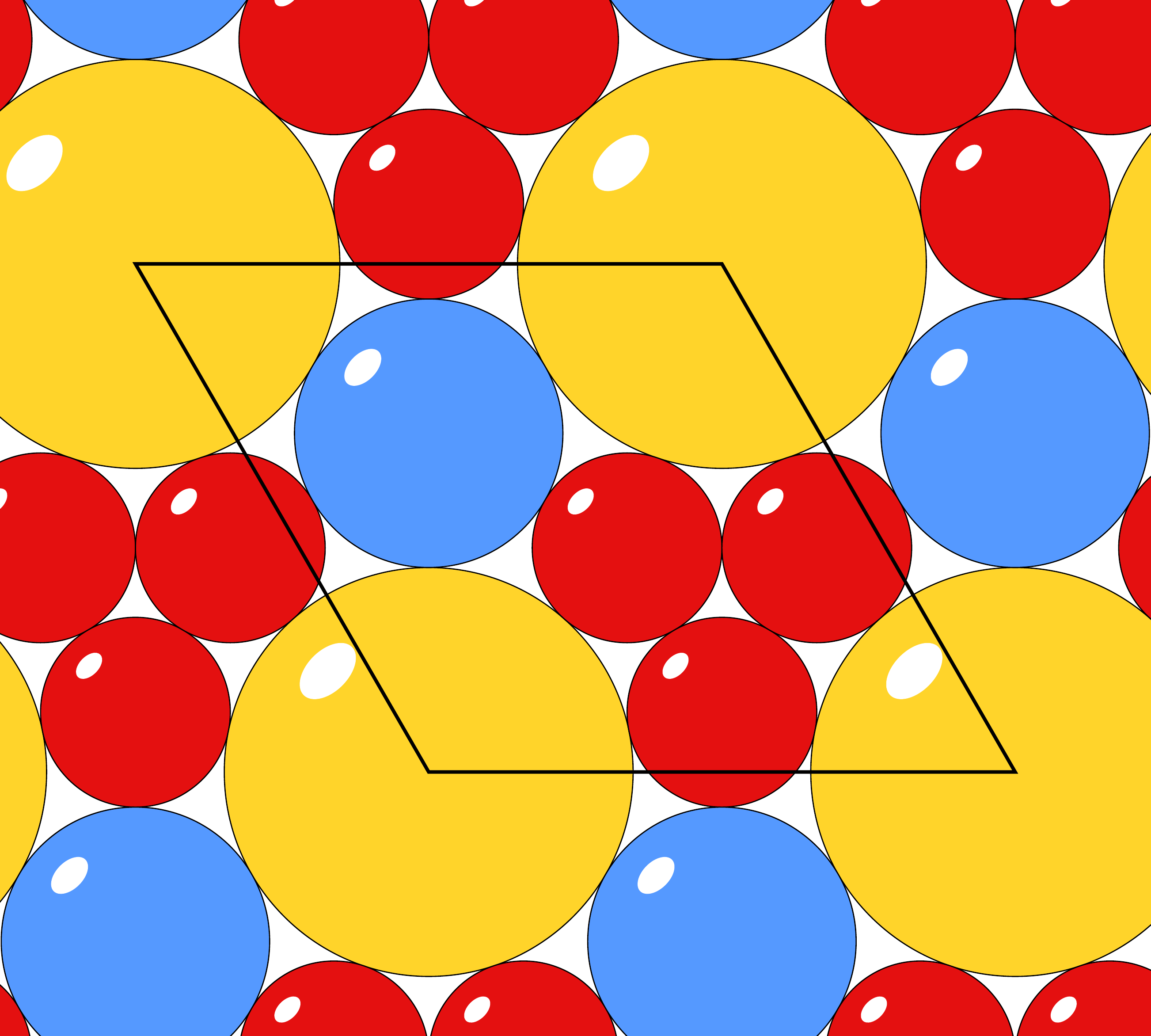} &
  \includegraphics[width=0.3\textwidth]{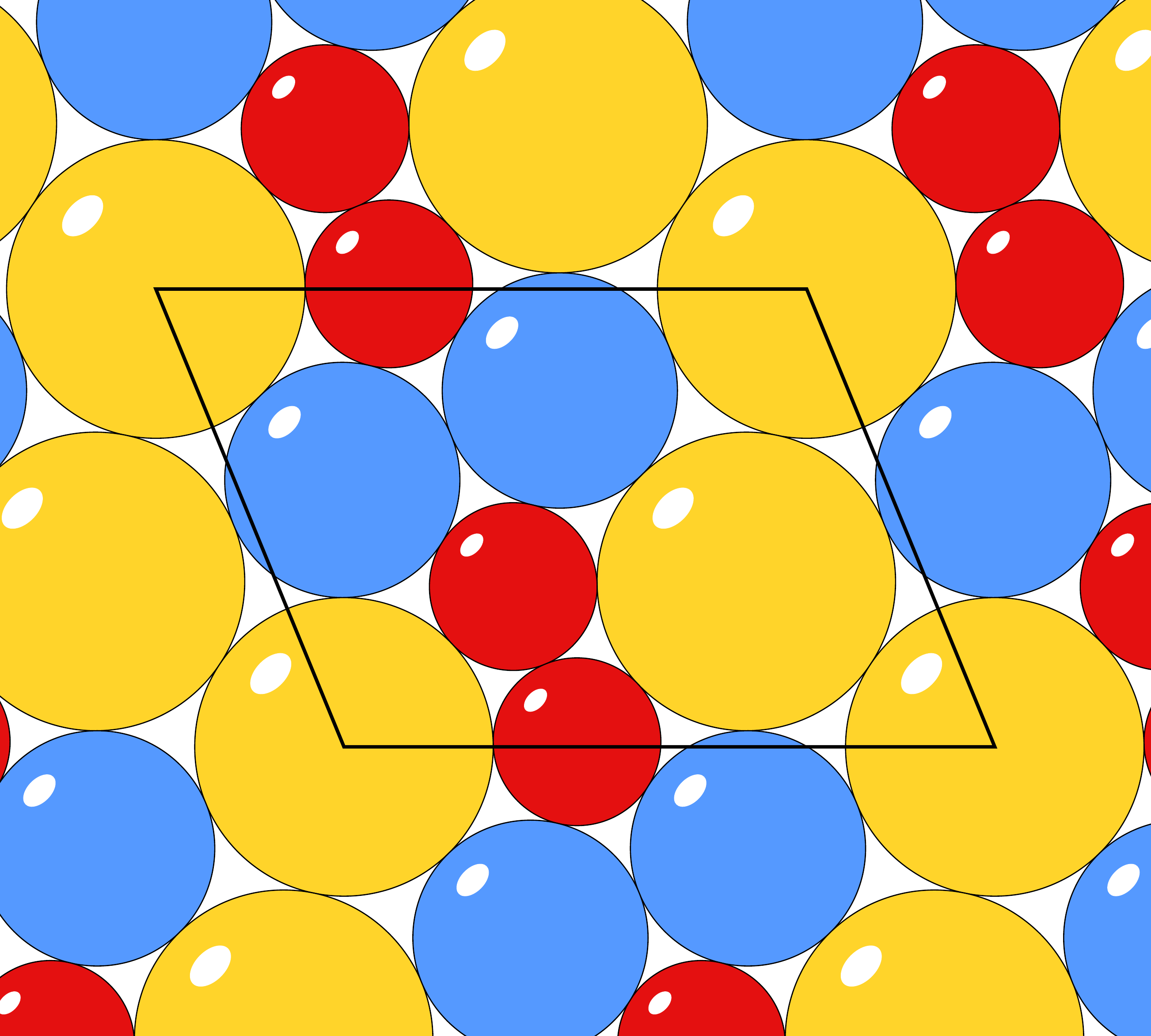} &
  \includegraphics[width=0.3\textwidth]{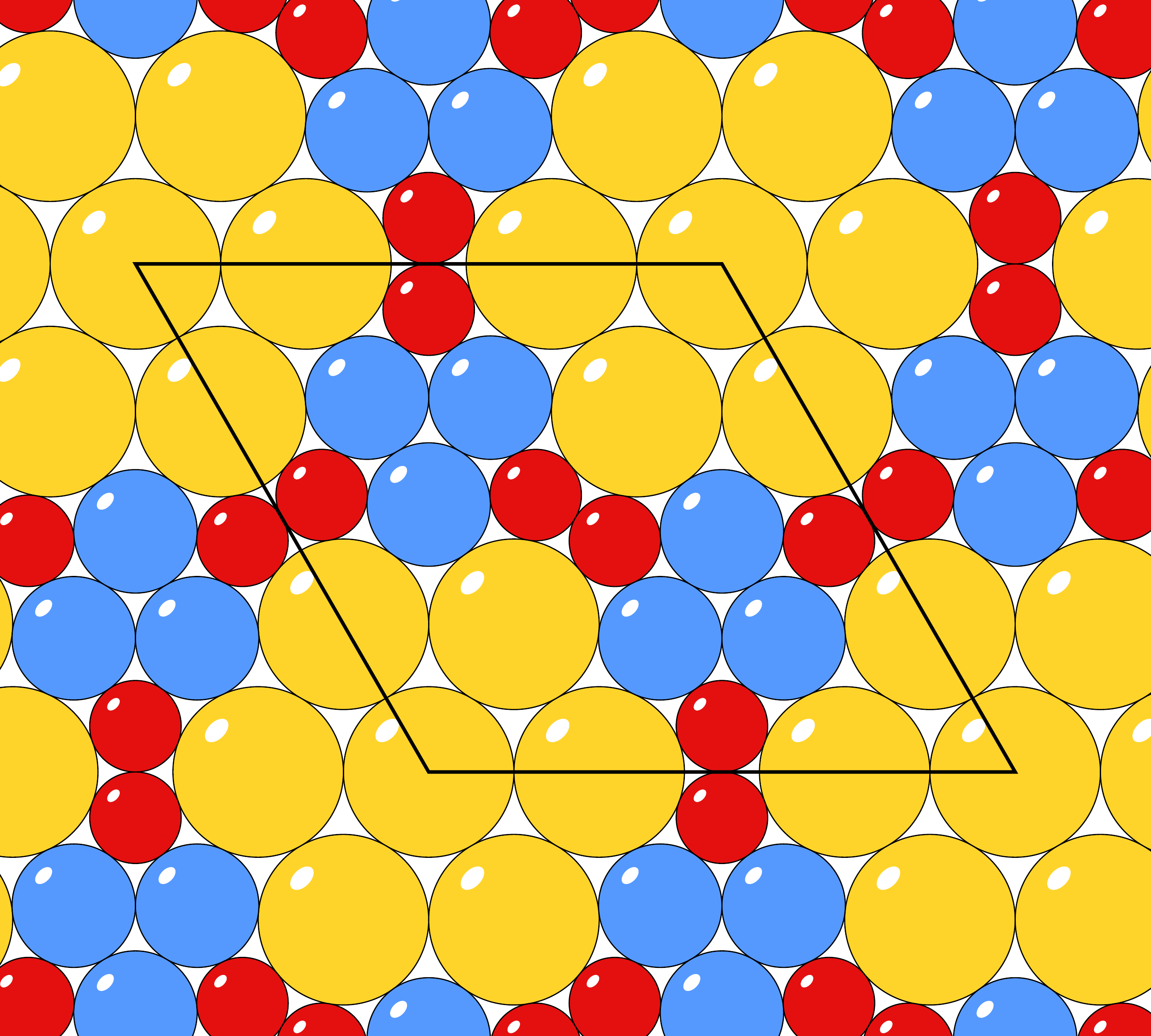}
\end{tabular}
\noindent
\begin{tabular}{lll}
  118\hfill 1rr1s / 11srs & 119\hfill 1rr1s / 1rrrrs & 120\hfill 1rr1s / 1srrrs\\
  \includegraphics[width=0.3\textwidth]{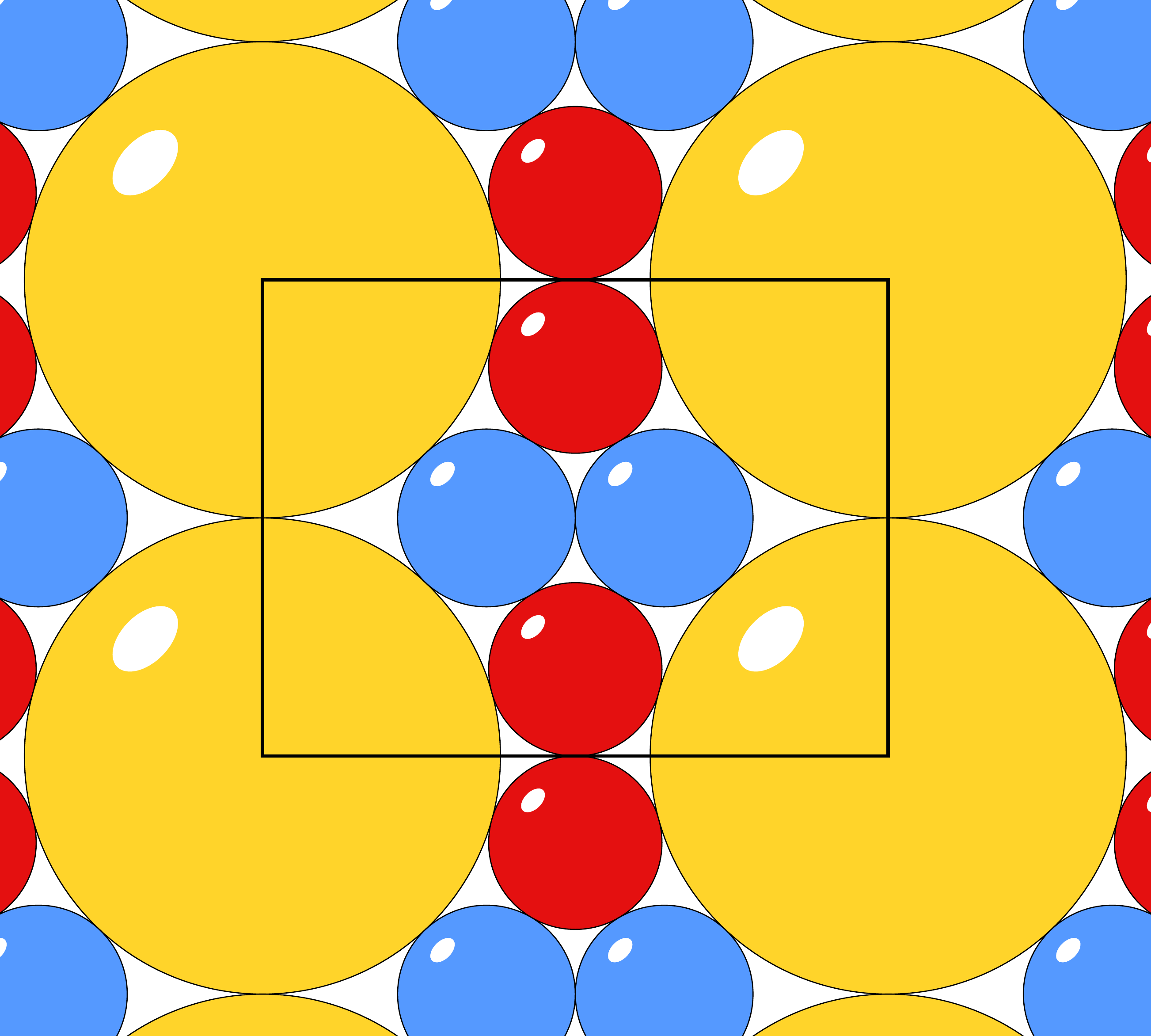} &
  \includegraphics[width=0.3\textwidth]{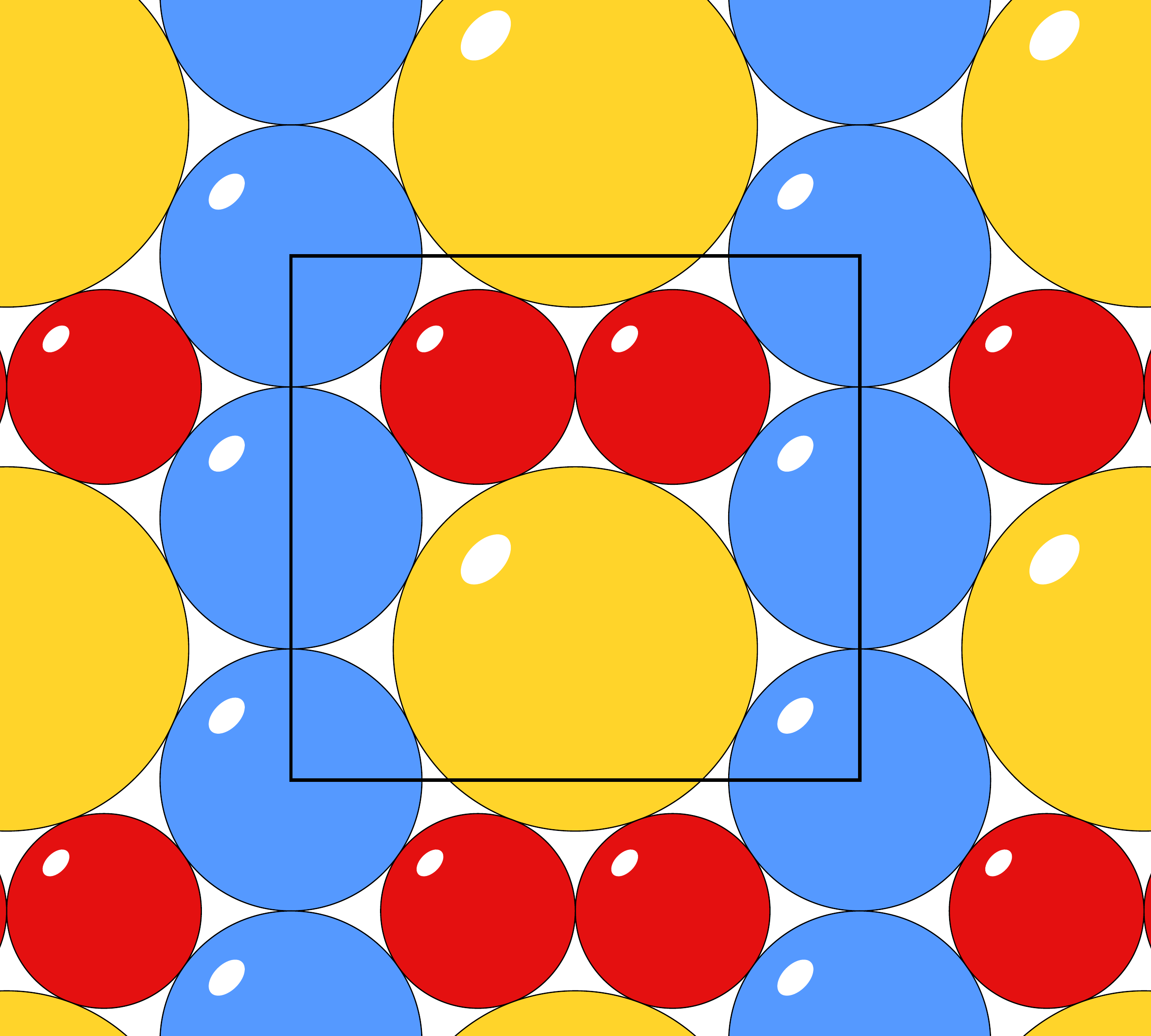} &
  \includegraphics[width=0.3\textwidth]{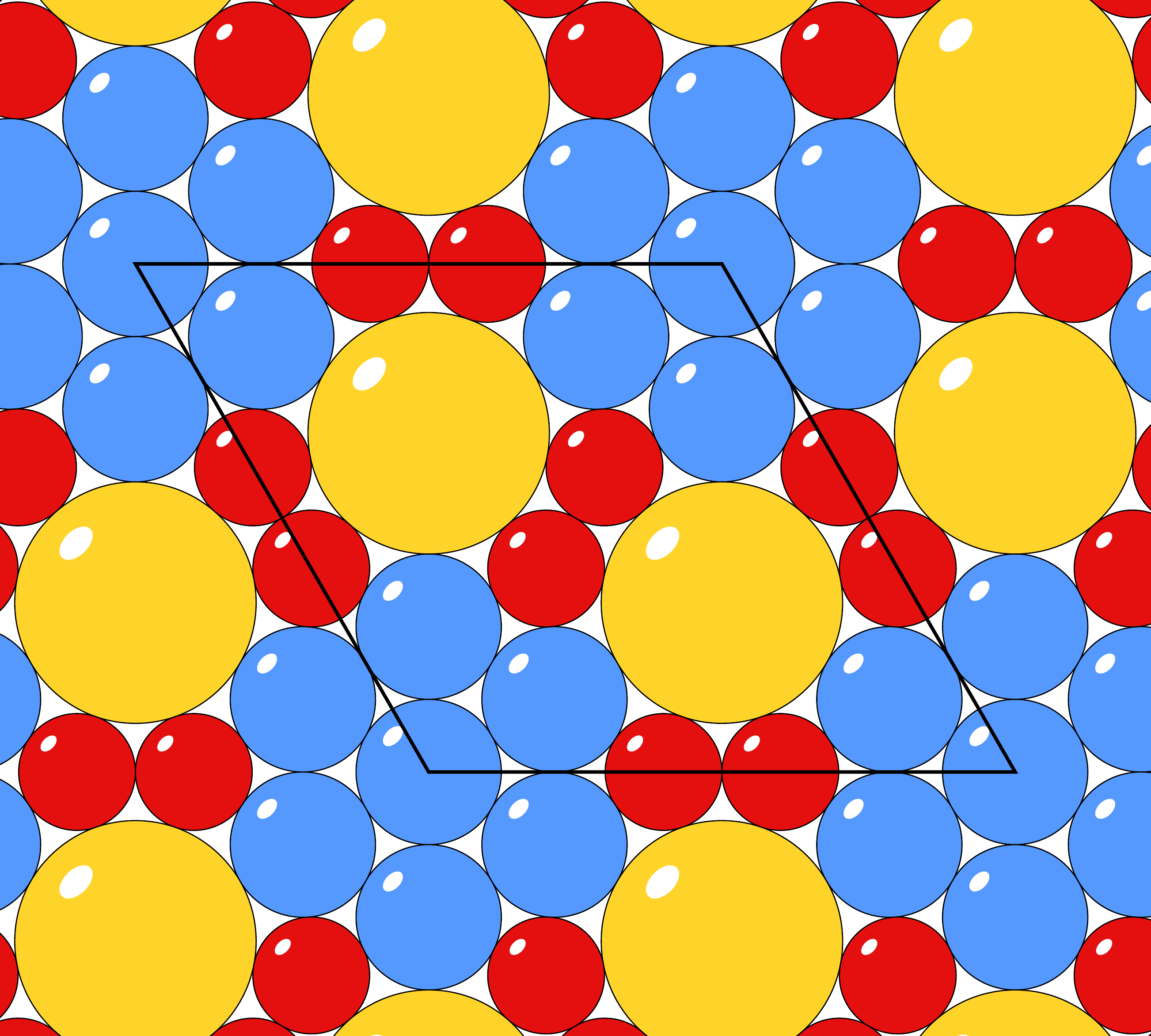}
\end{tabular}
\noindent
\begin{tabular}{lll}
  121\hfill 1rrr / 11srsrs & 122\hfill 1rrr / 1srrsrs & 123\hfill 1rrrr / 11rsrs\\
  \includegraphics[width=0.3\textwidth]{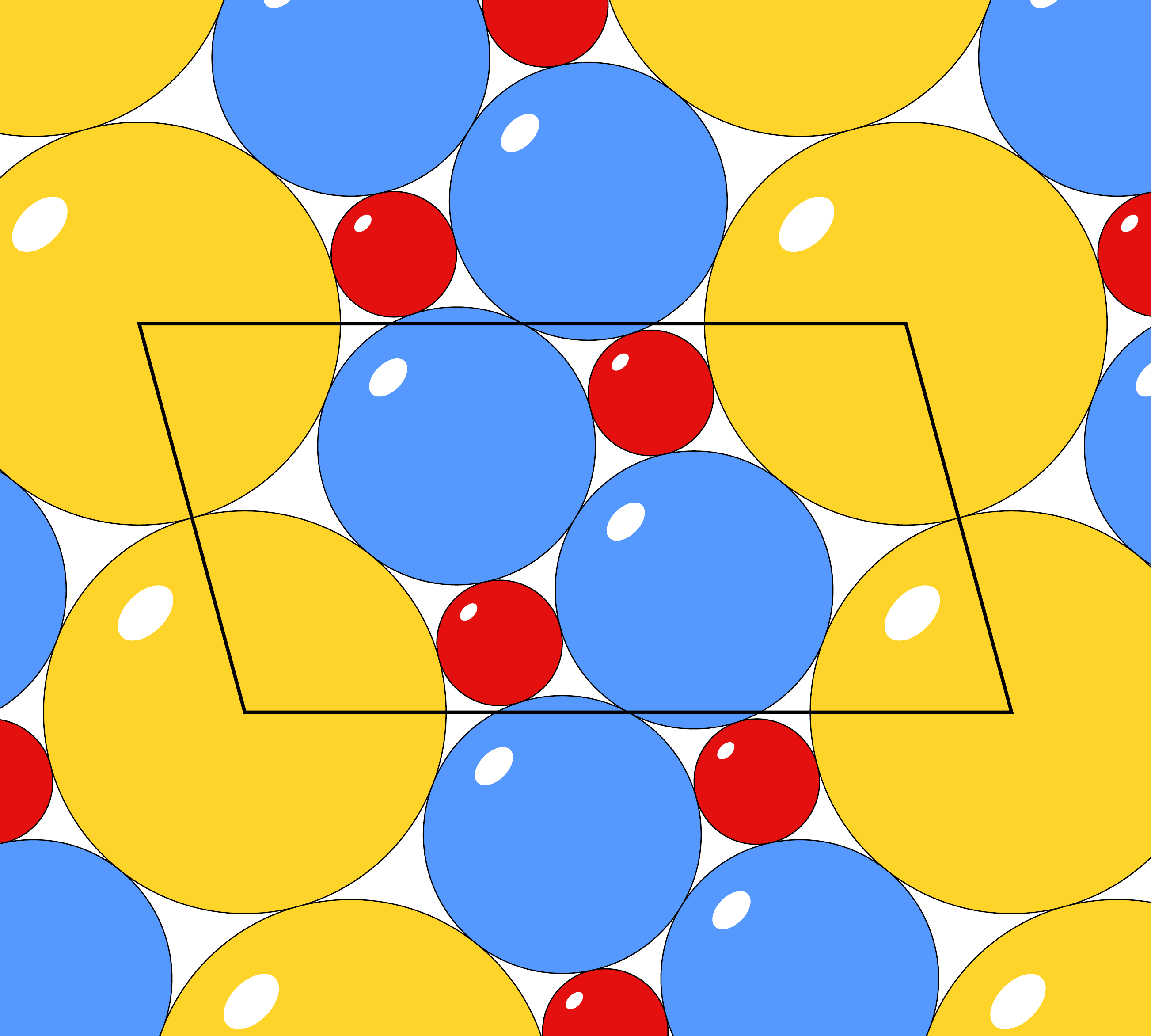} &
  \includegraphics[width=0.3\textwidth]{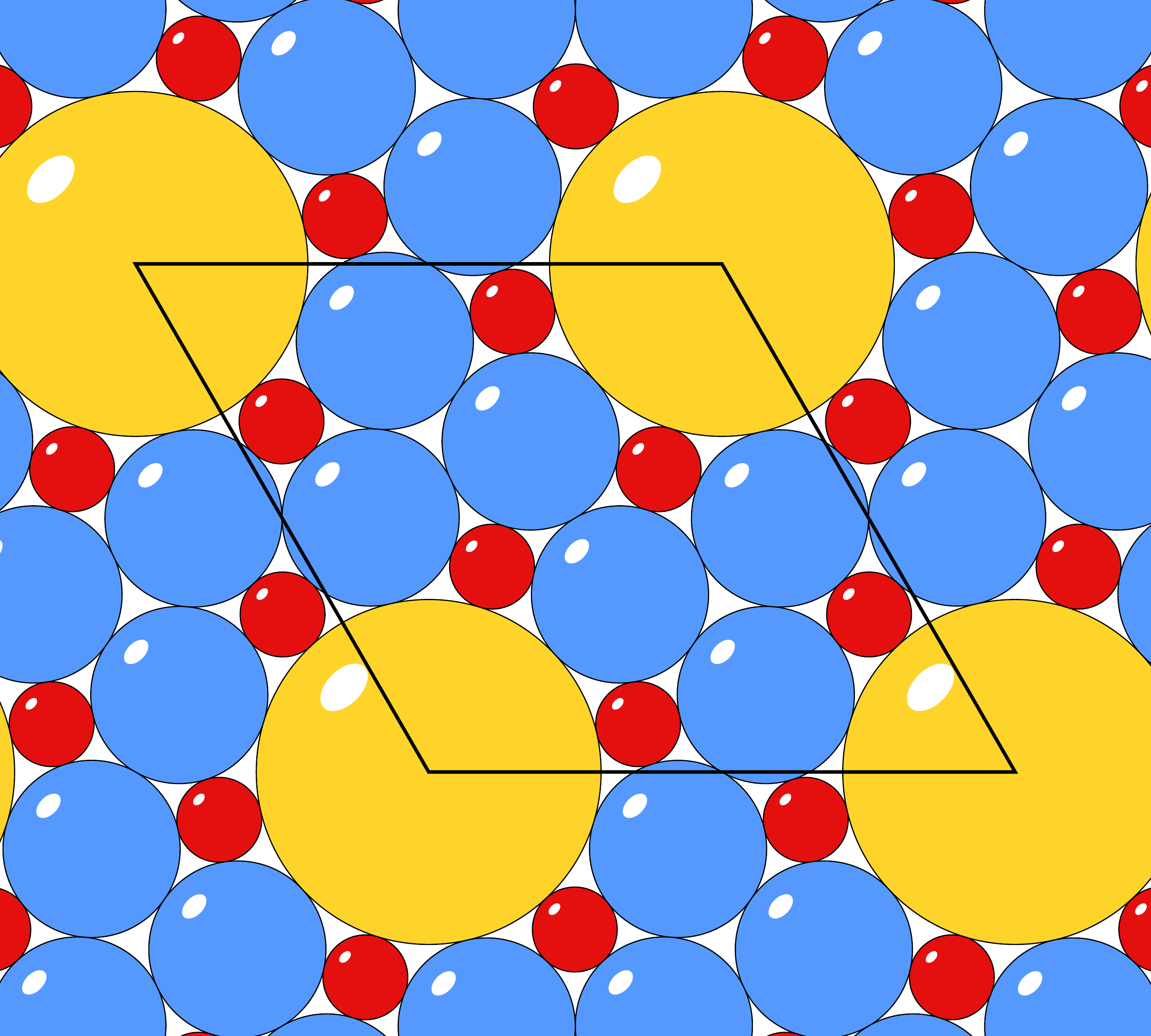} &
  \includegraphics[width=0.3\textwidth]{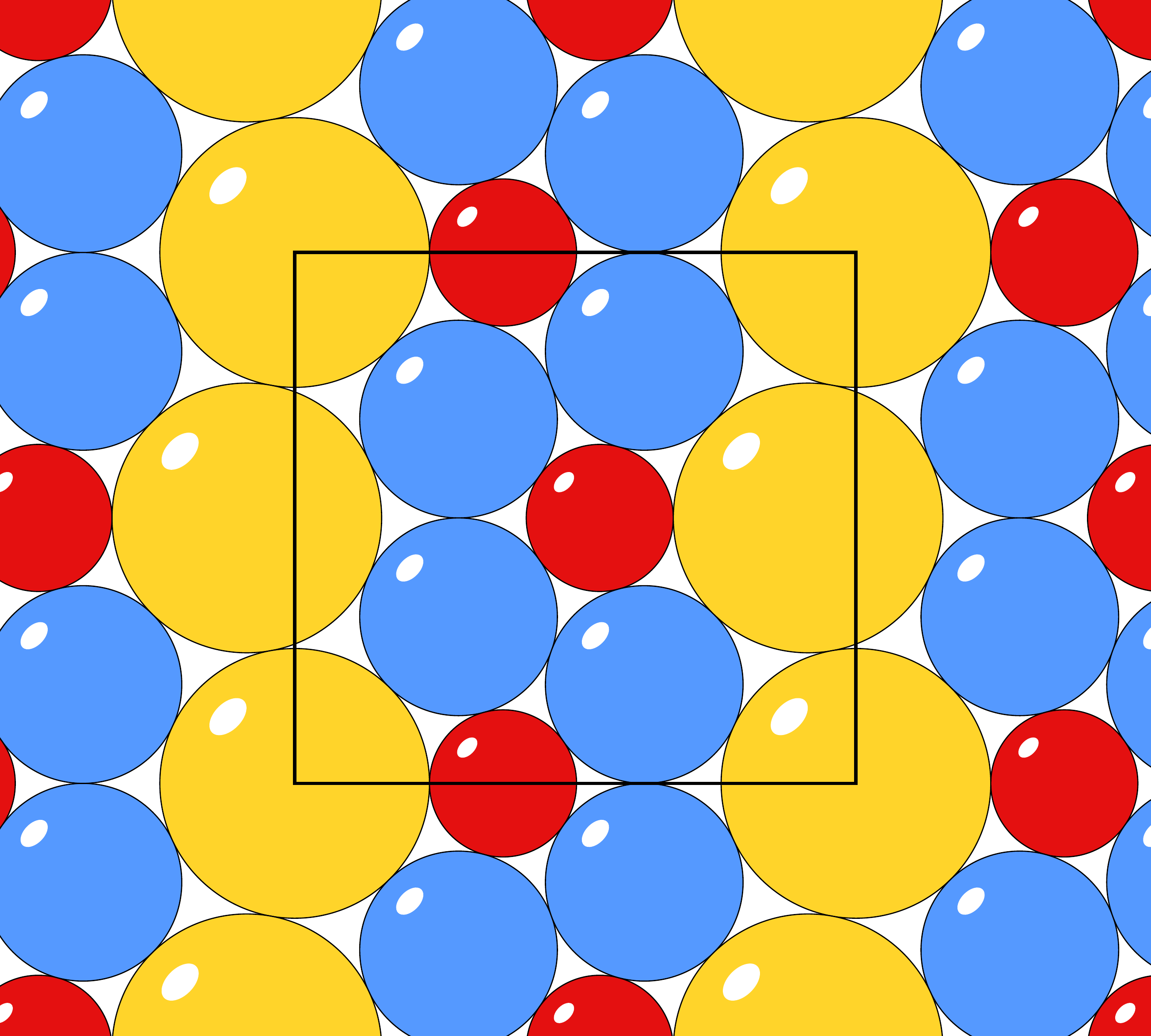}
\end{tabular}
\noindent
\begin{tabular}{lll}
  124\hfill 1rrrr / 1rrsrs & 125\hfill 1rrs / 11srsrss & 126\hfill 1rrs / 1srsrrss\\
  \includegraphics[width=0.3\textwidth]{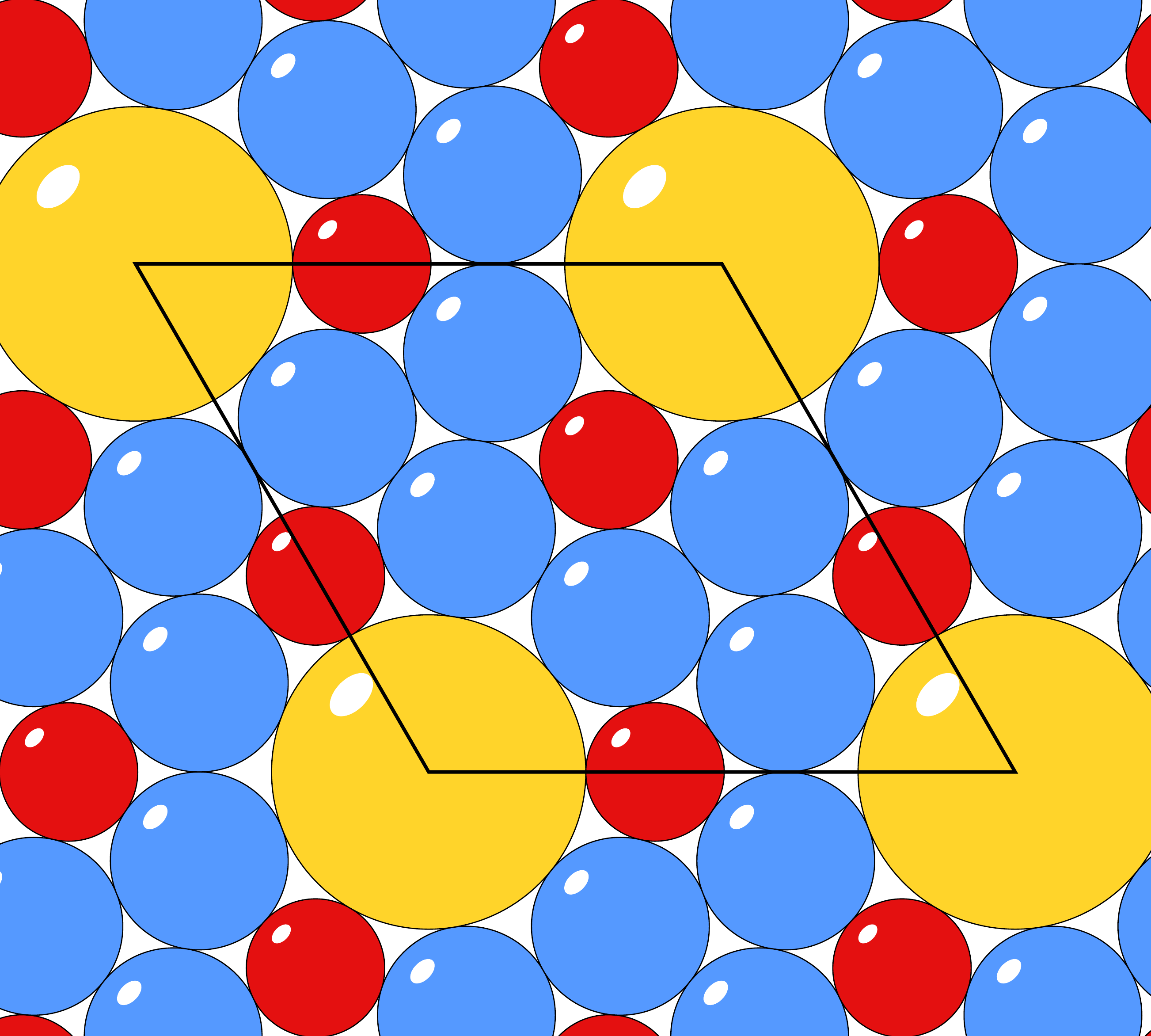} &
  \includegraphics[width=0.3\textwidth]{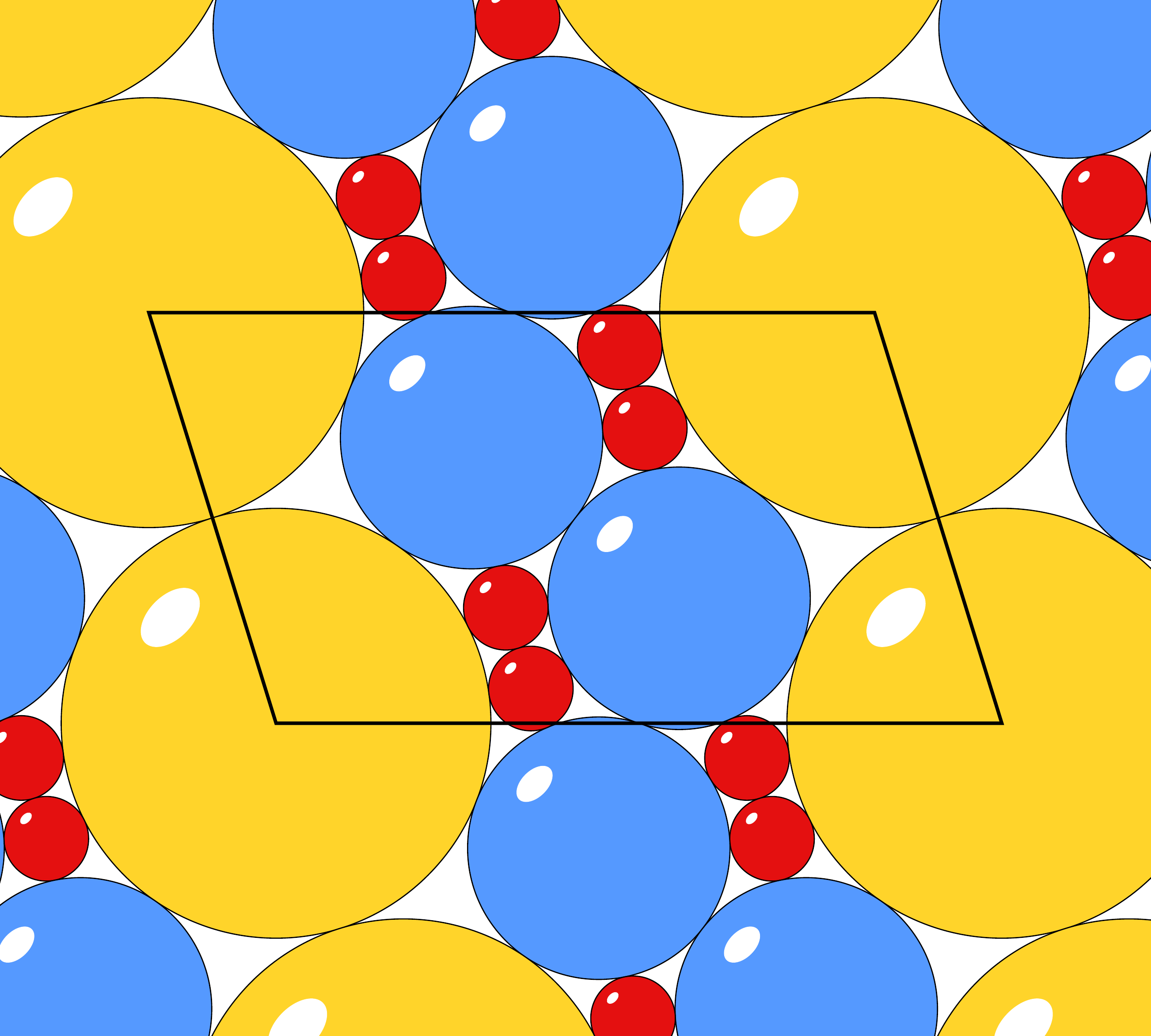} &
  \includegraphics[width=0.3\textwidth]{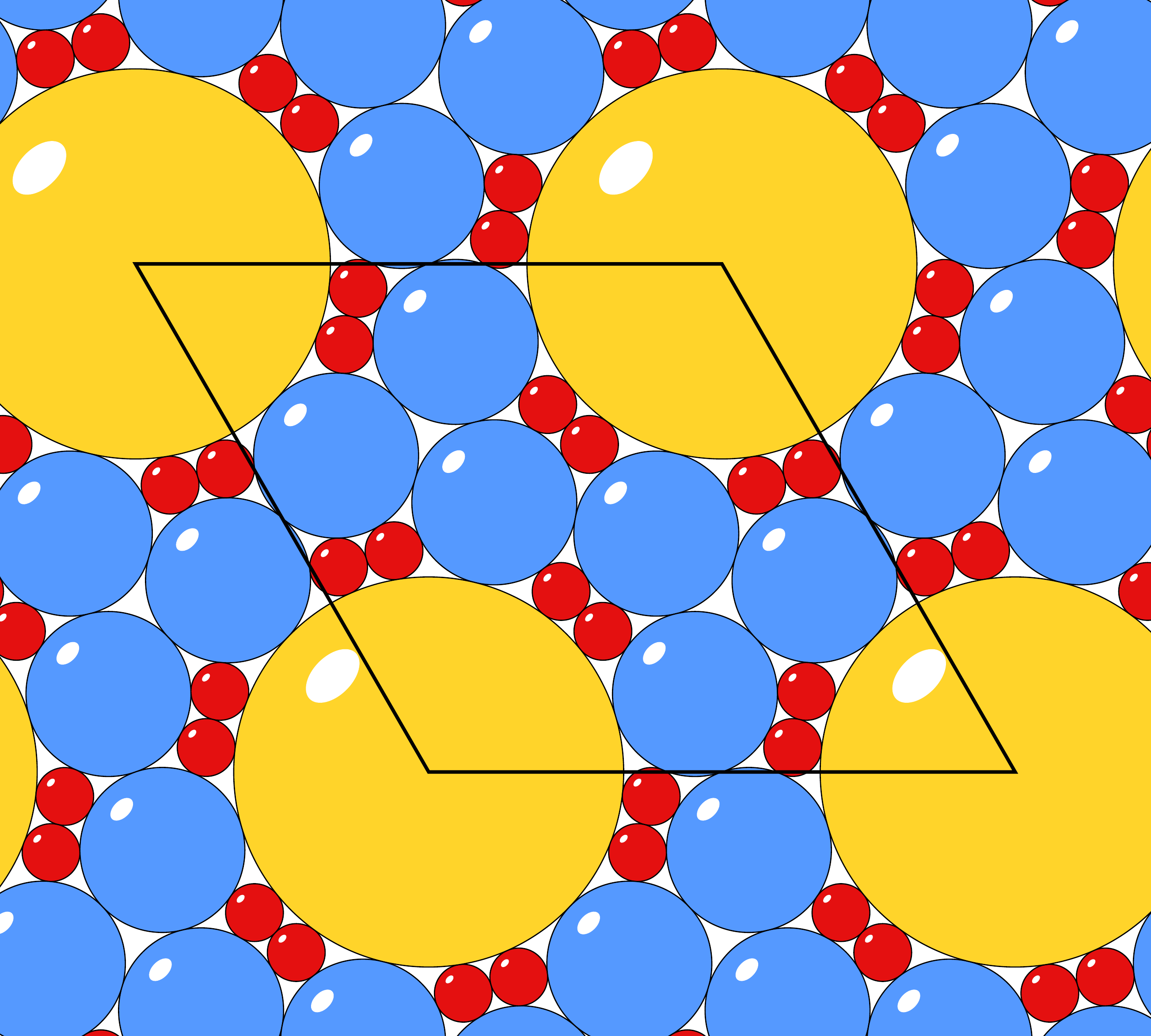}
\end{tabular}
\noindent
\begin{tabular}{lll}
  127\hfill 1rrsr / 11srss & 128\hfill 1rrsr / 1srrss & 129\hfill 1rs1s / 111ss\\
  \includegraphics[width=0.3\textwidth]{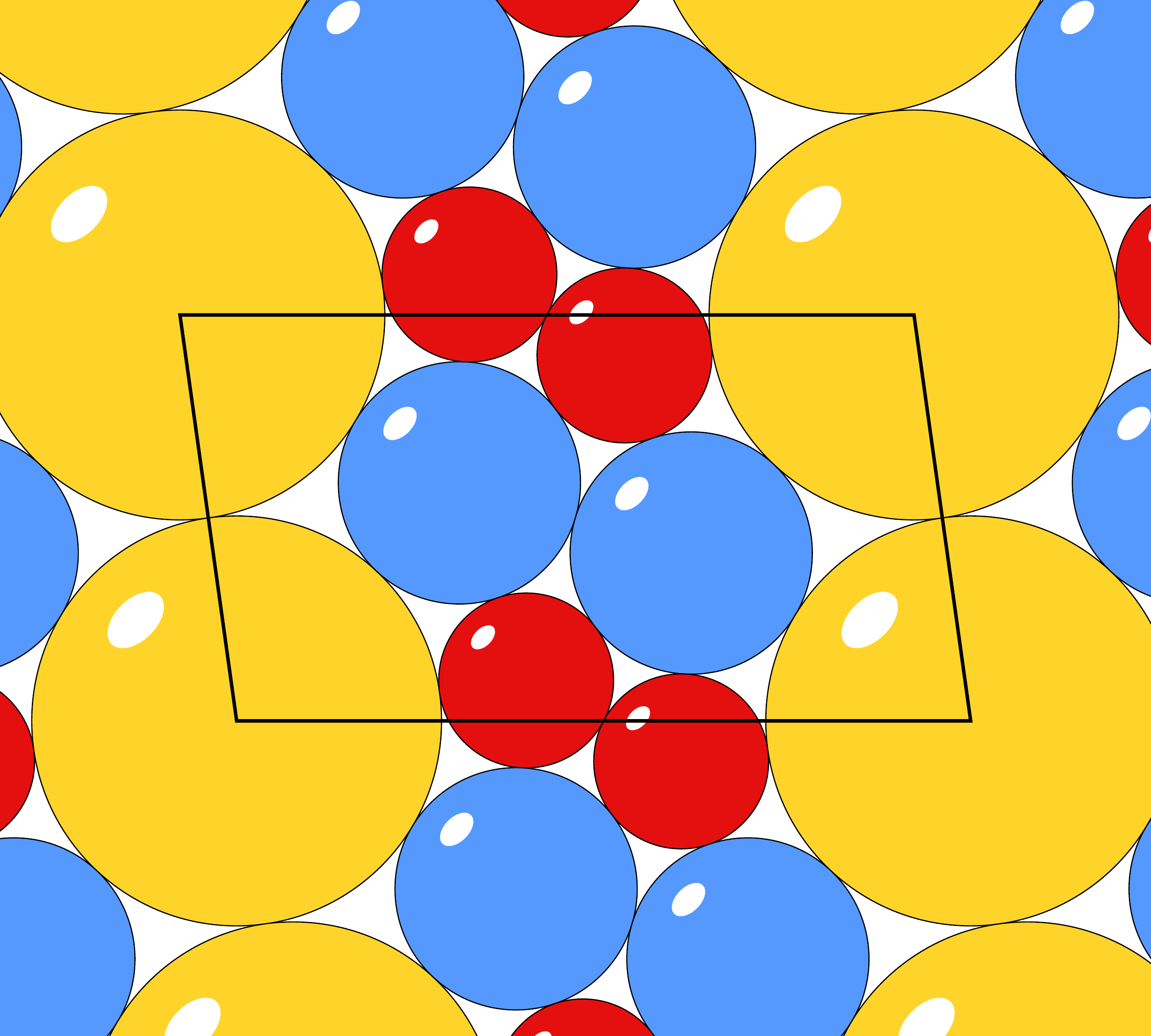} &
  \includegraphics[width=0.3\textwidth]{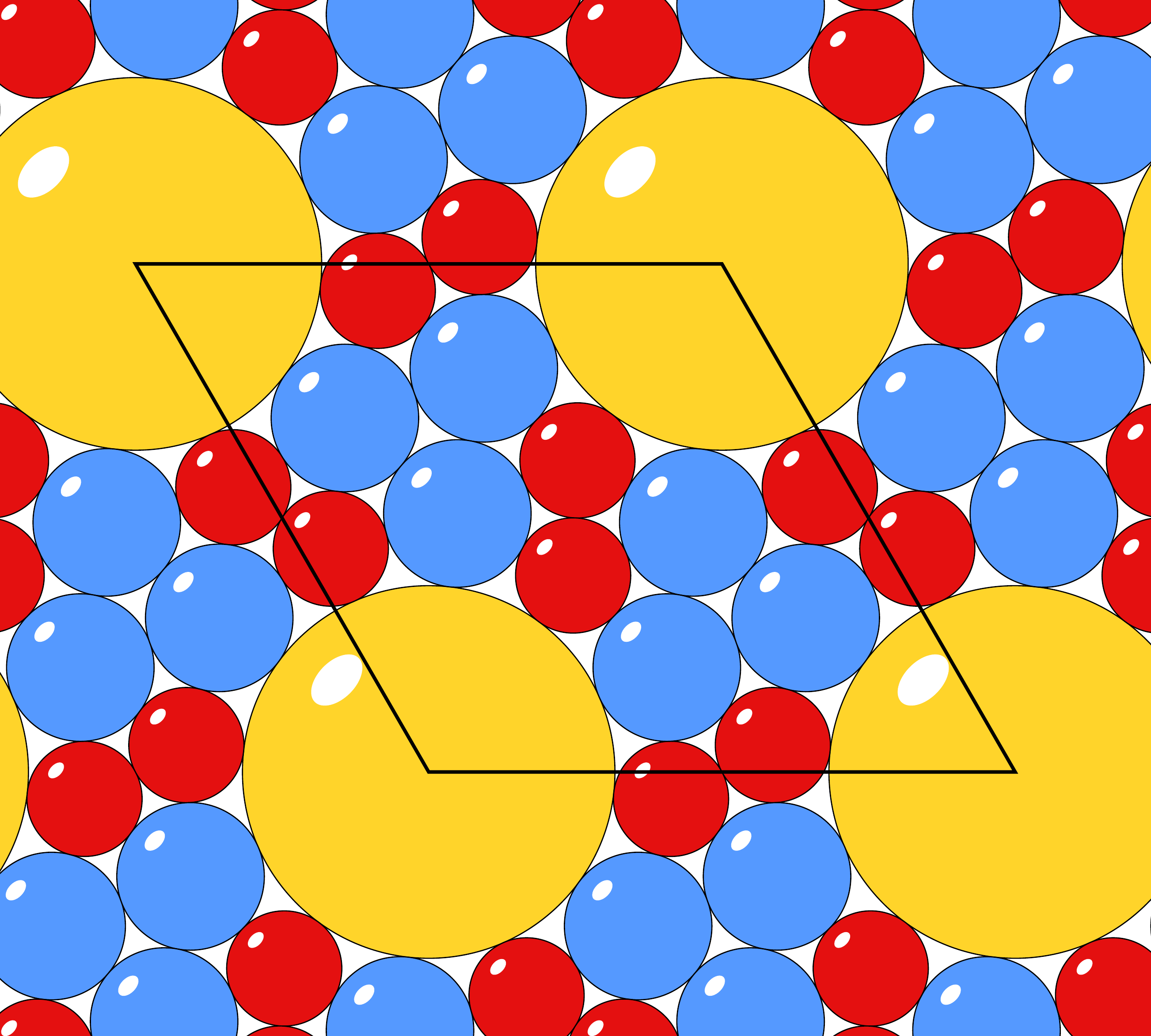} &
  \includegraphics[width=0.3\textwidth]{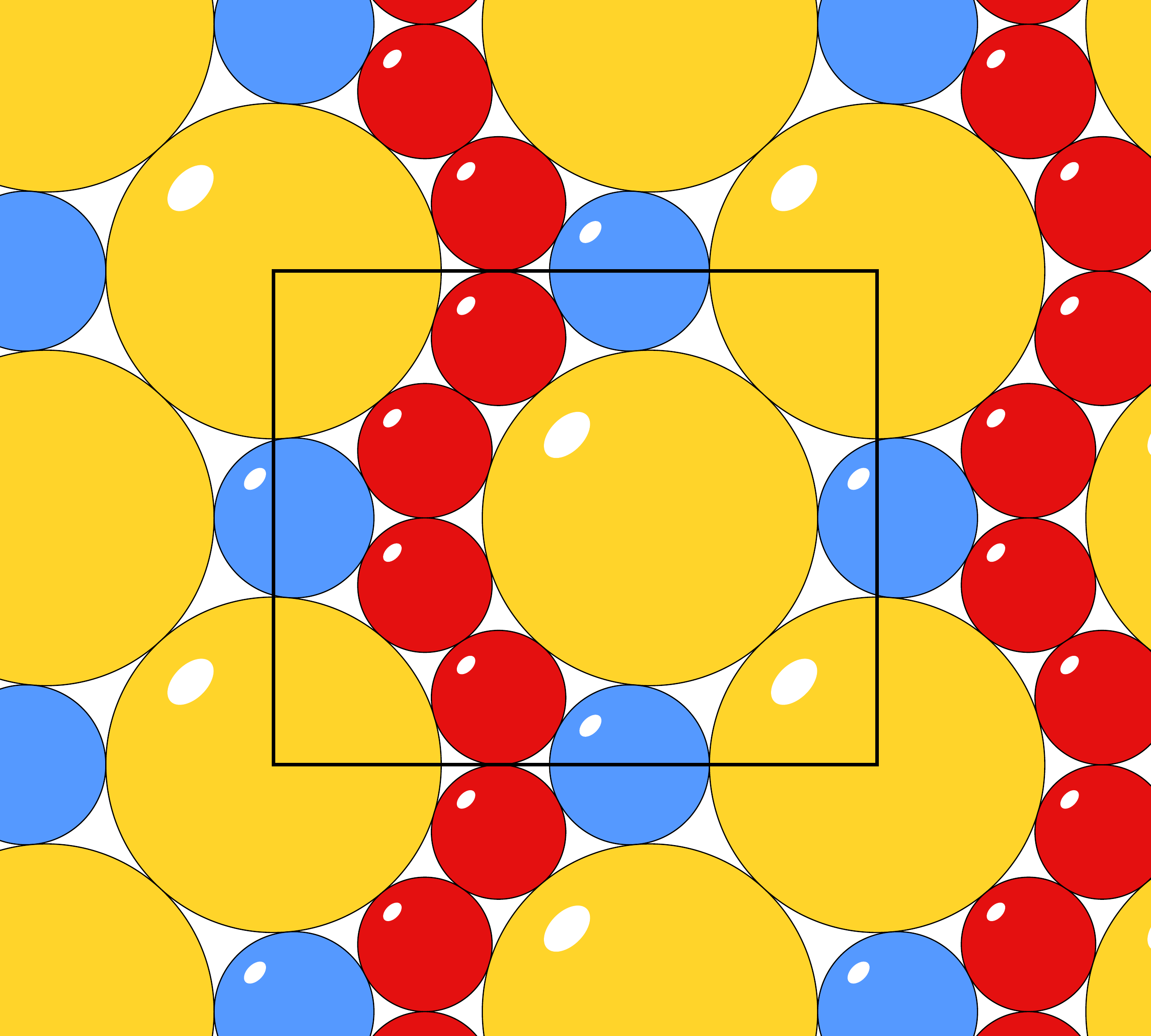}
\end{tabular}
\noindent
\begin{tabular}{lll}
  130\hfill 1rs1s / 11r1ss & 131\hfill 1rs1s / 1s1sss & 132\hfill 1rsr / 1111ss\\
  \includegraphics[width=0.3\textwidth]{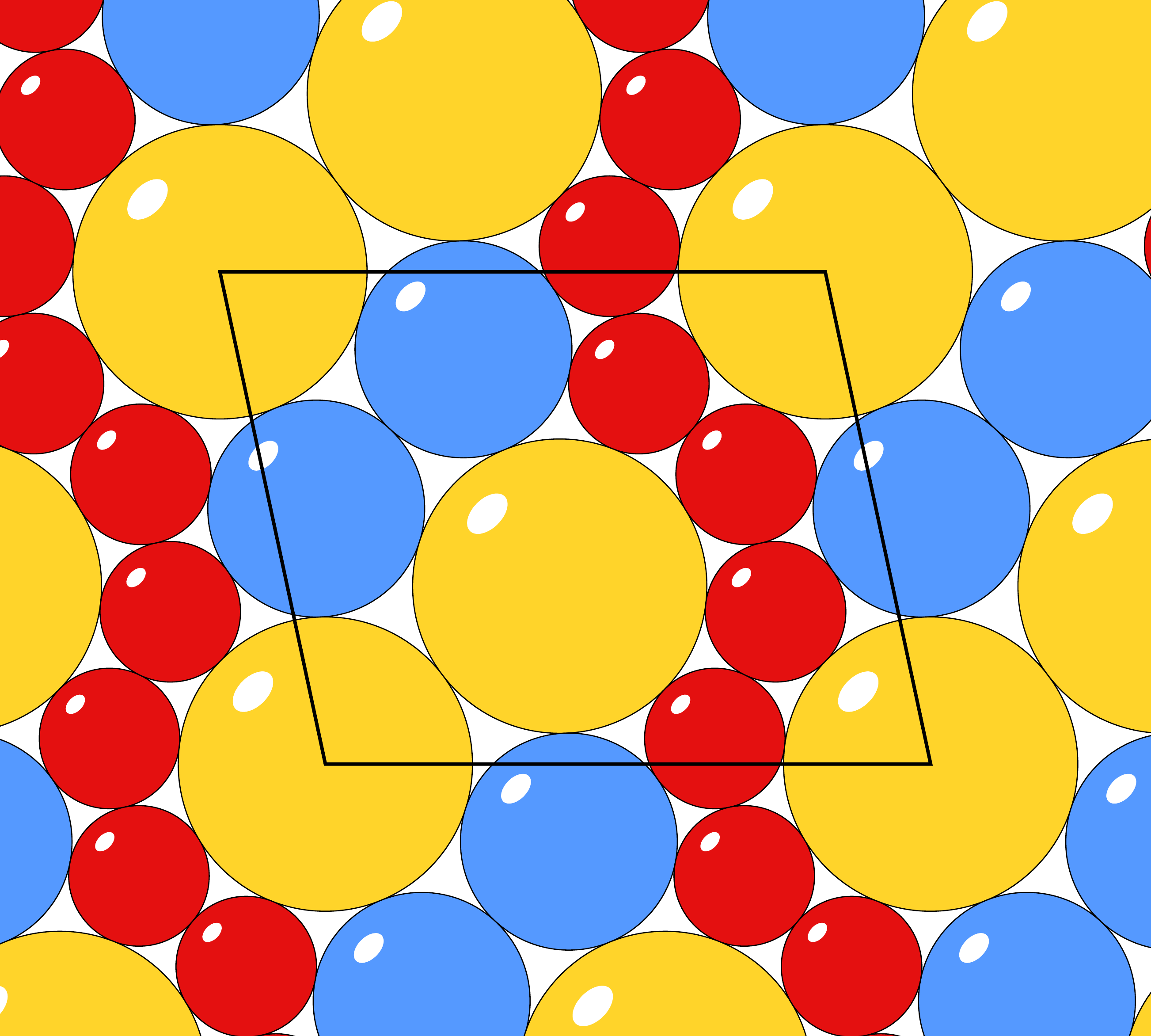} &
  \includegraphics[width=0.3\textwidth]{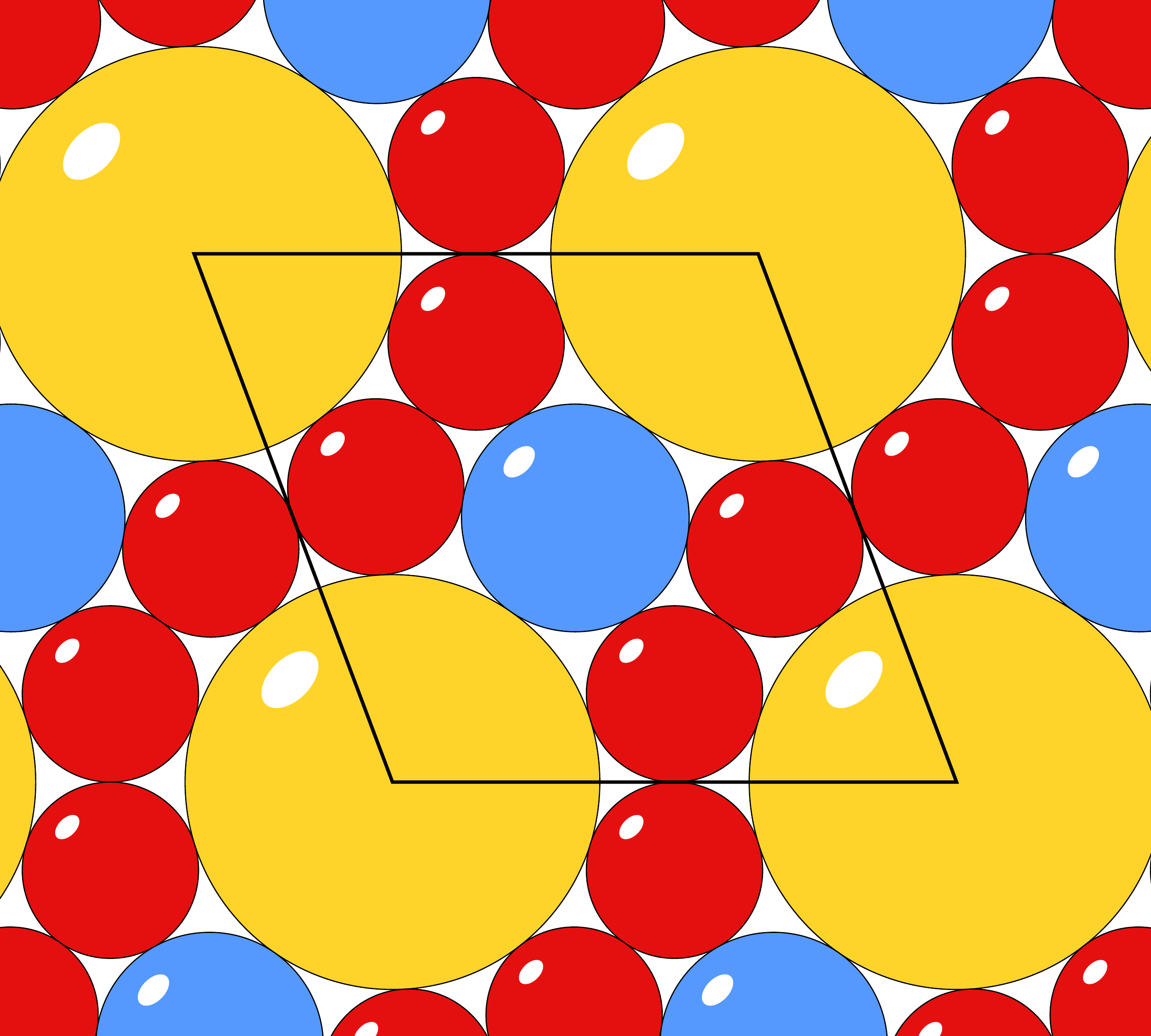} &
  \includegraphics[width=0.3\textwidth]{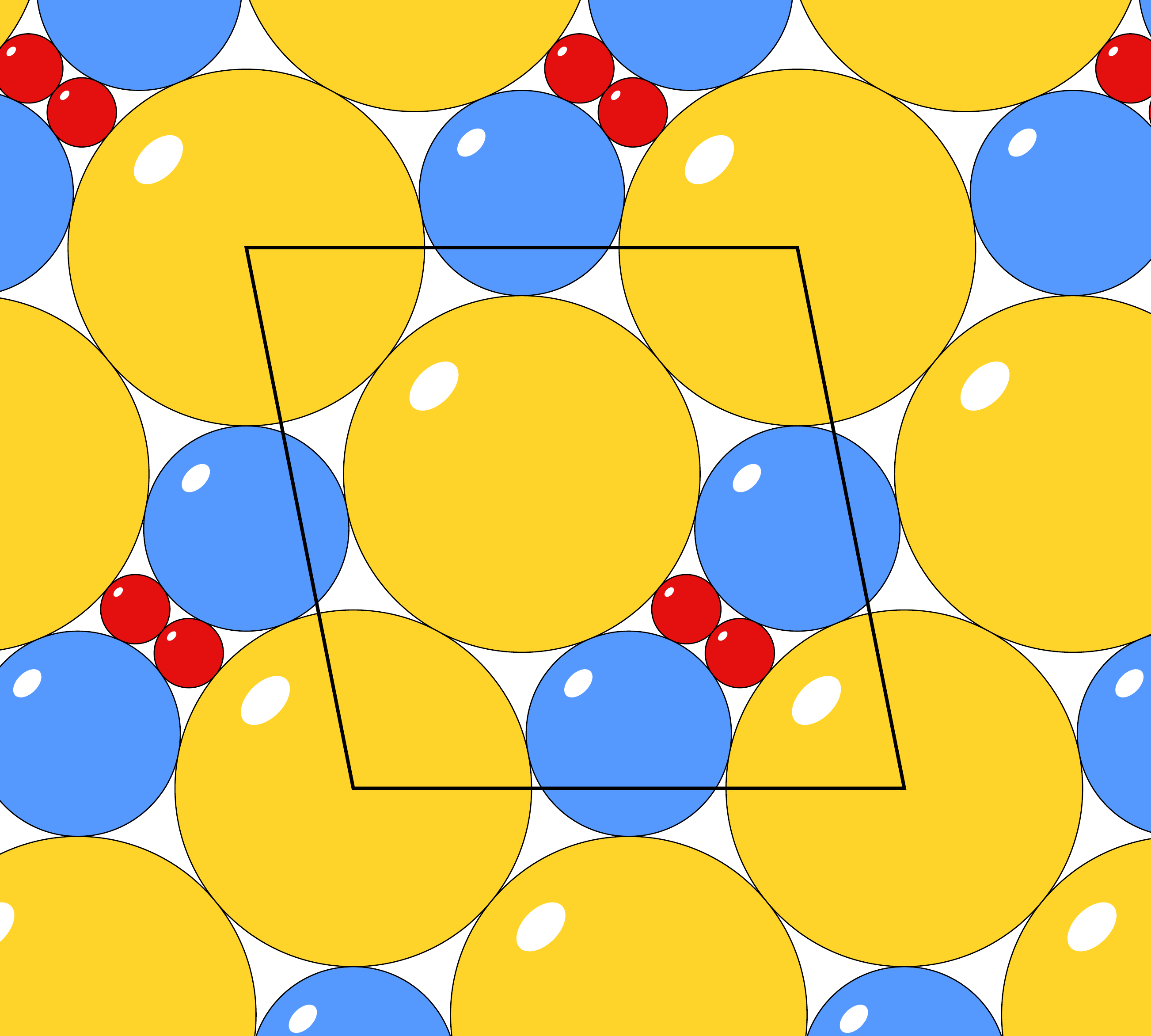}
\end{tabular}
\noindent
\begin{tabular}{lll}
  133\hfill 1rsr / 111r1ss & 134\hfill 1rsr / 111s1sss & 135\hfill 1rsr / 111ss\\
  \includegraphics[width=0.3\textwidth]{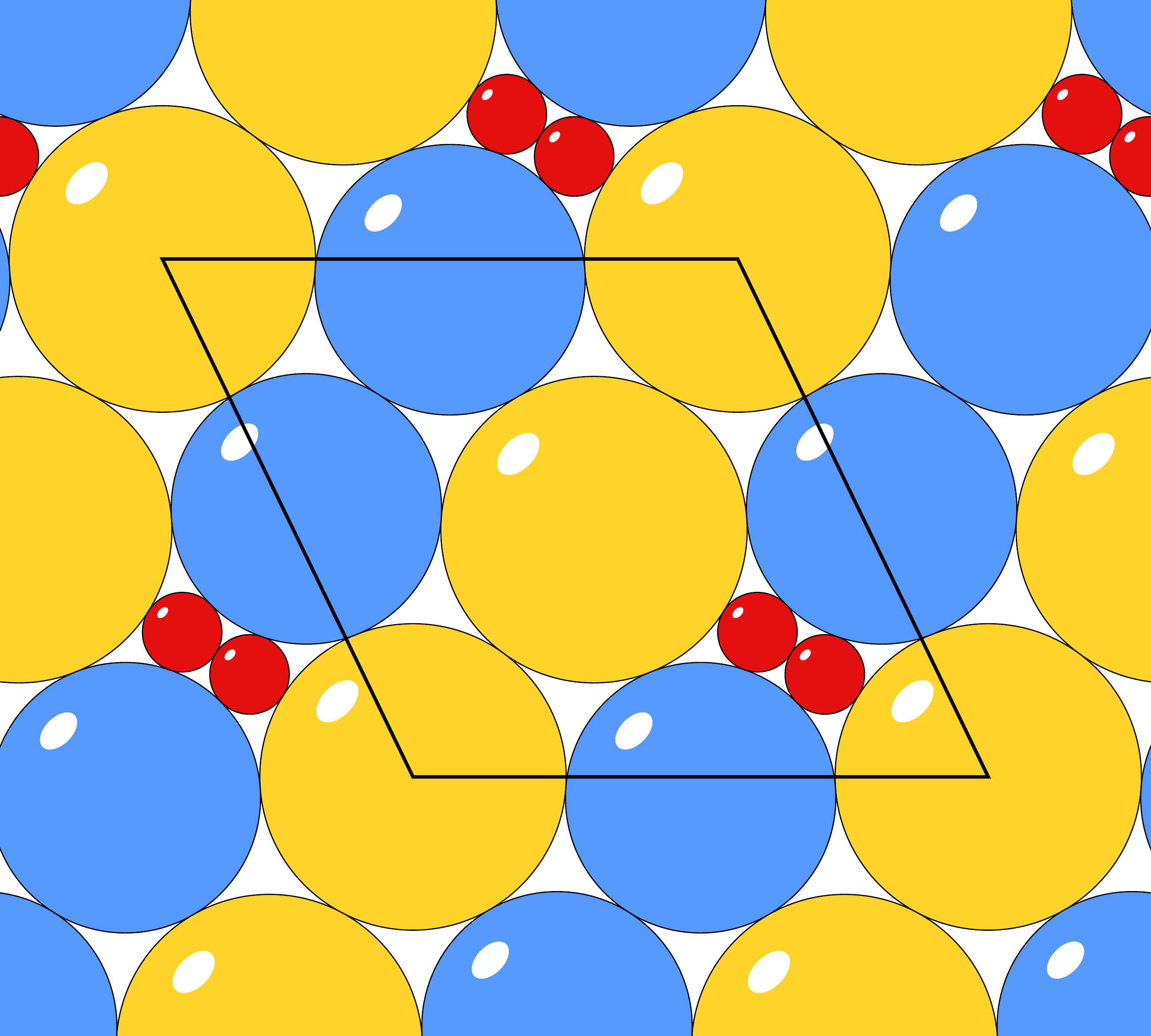} &
  \includegraphics[width=0.3\textwidth]{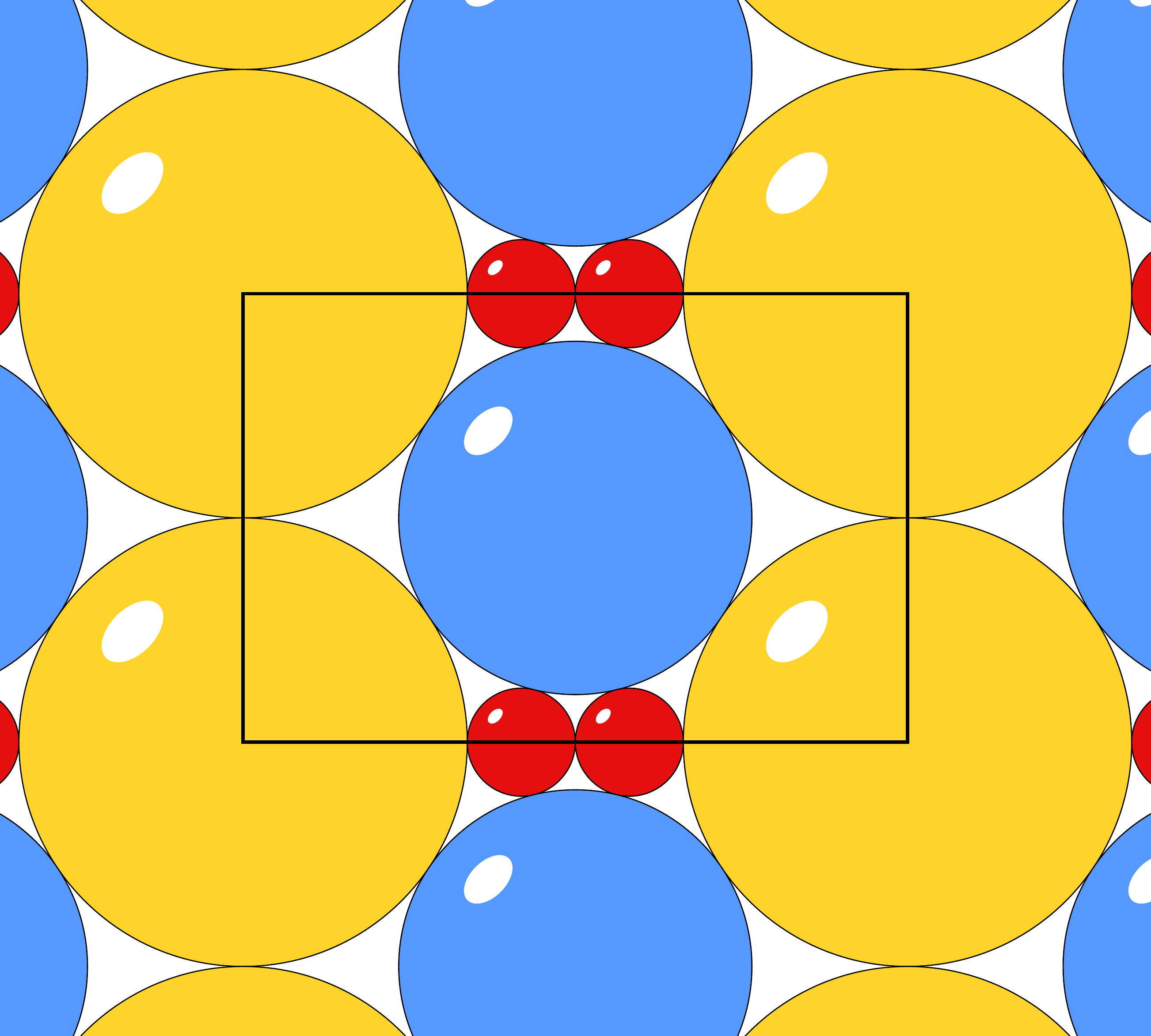} &
  \includegraphics[width=0.3\textwidth]{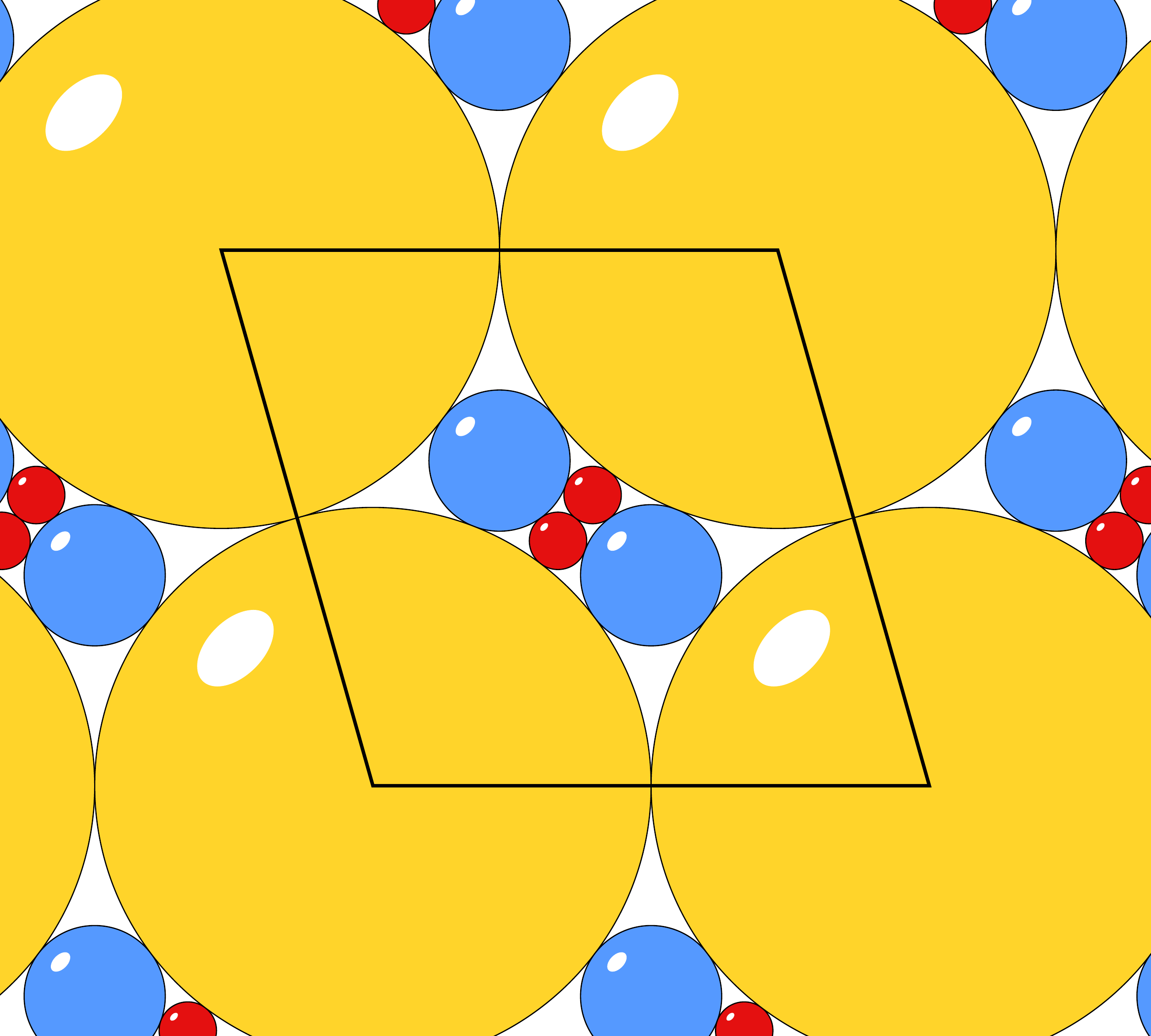}
\end{tabular}
\noindent
\begin{tabular}{lll}
  136\hfill 1rsr / 11r1ss & 137\hfill 1rsr / 11rr1ss & 138\hfill 1rsr / 11s1sss\\
  \includegraphics[width=0.3\textwidth]{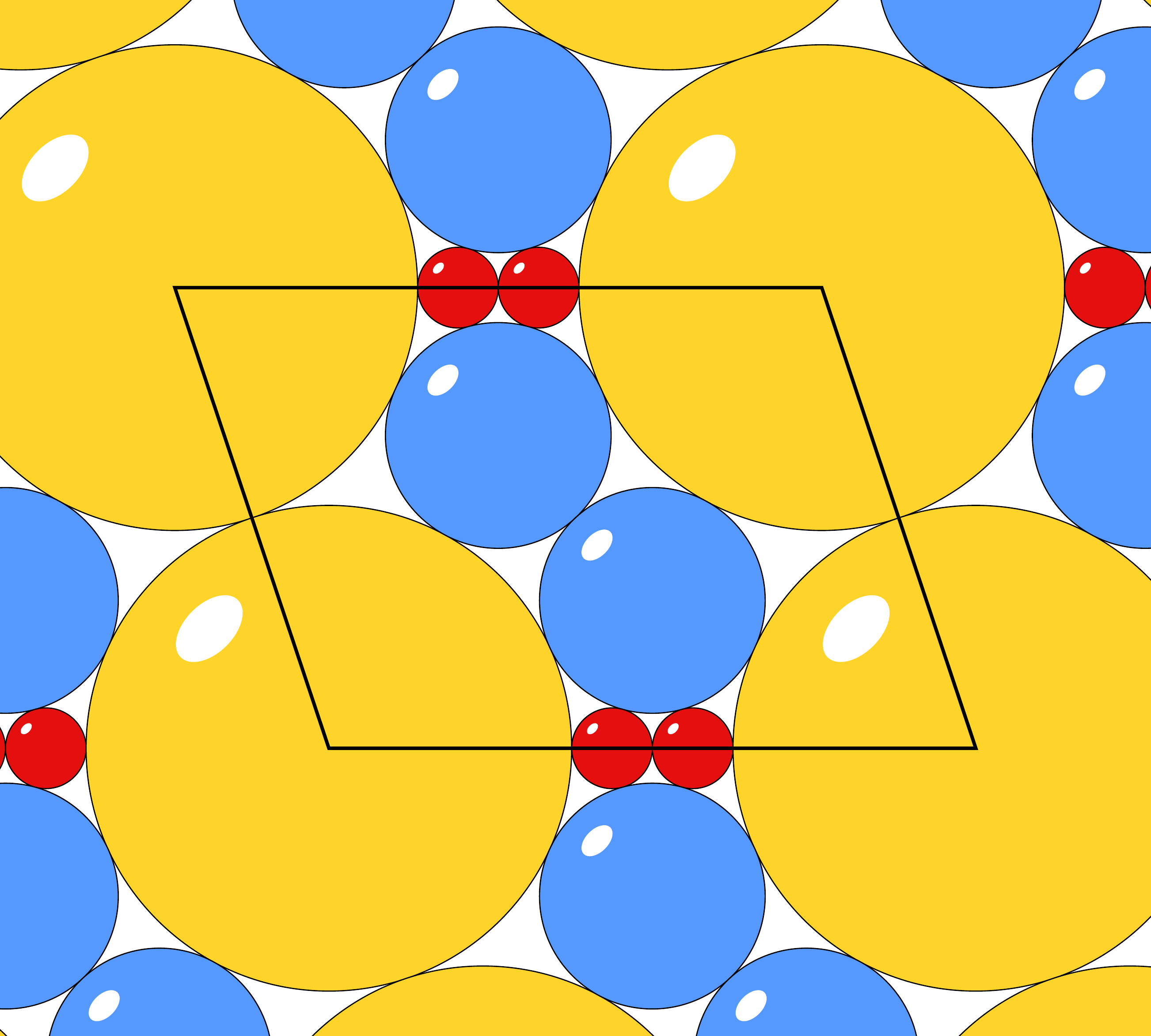} &
  \includegraphics[width=0.3\textwidth]{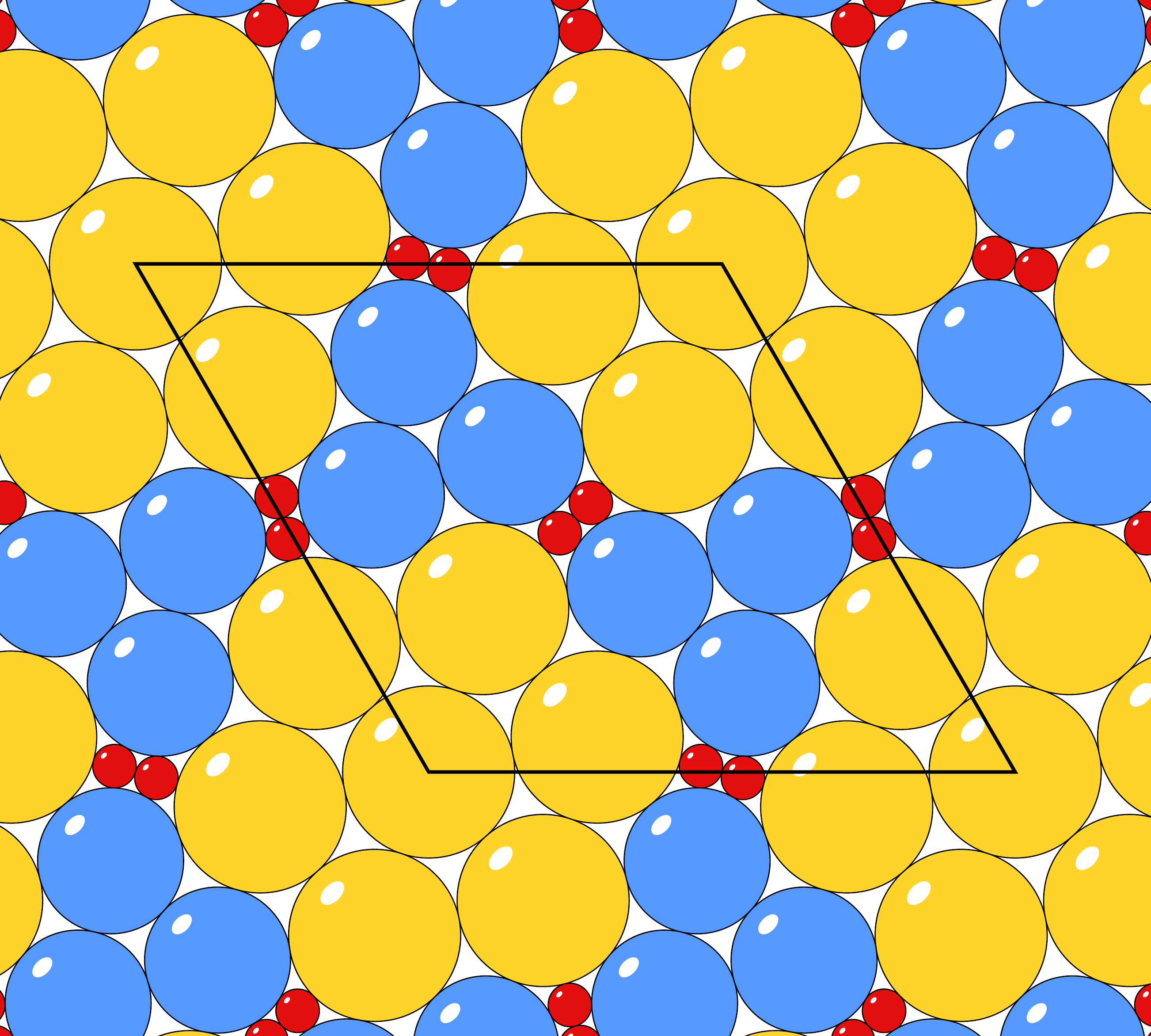} &
  \includegraphics[width=0.3\textwidth]{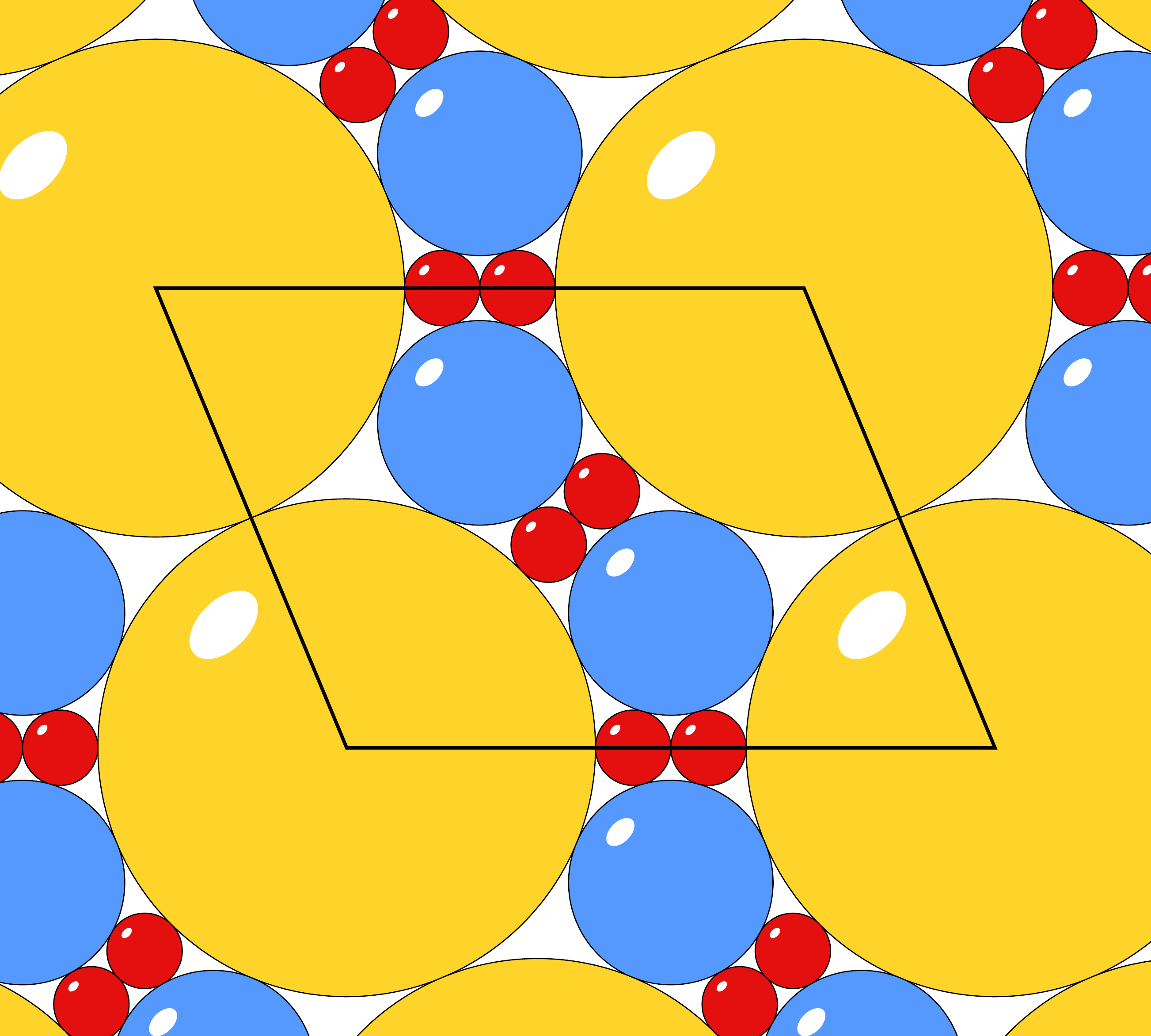}
\end{tabular}
\noindent
\begin{tabular}{lll}
  139\hfill 1rsr / 1r1r1ss & 140\hfill 1rsr / 1r1s1sss & 141\hfill 1rsr / 1rr1ss\\
  \includegraphics[width=0.3\textwidth]{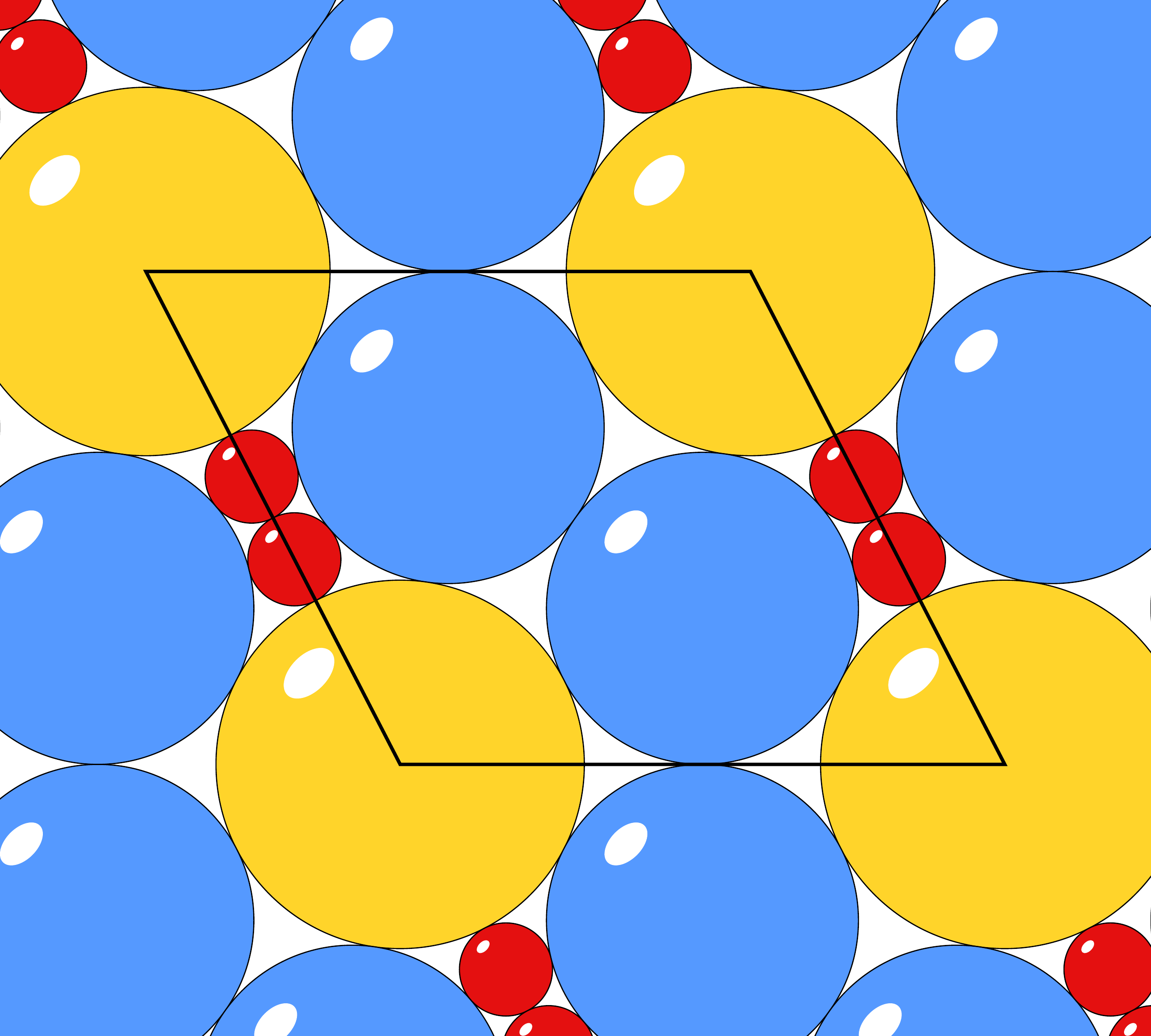} &
  \includegraphics[width=0.3\textwidth]{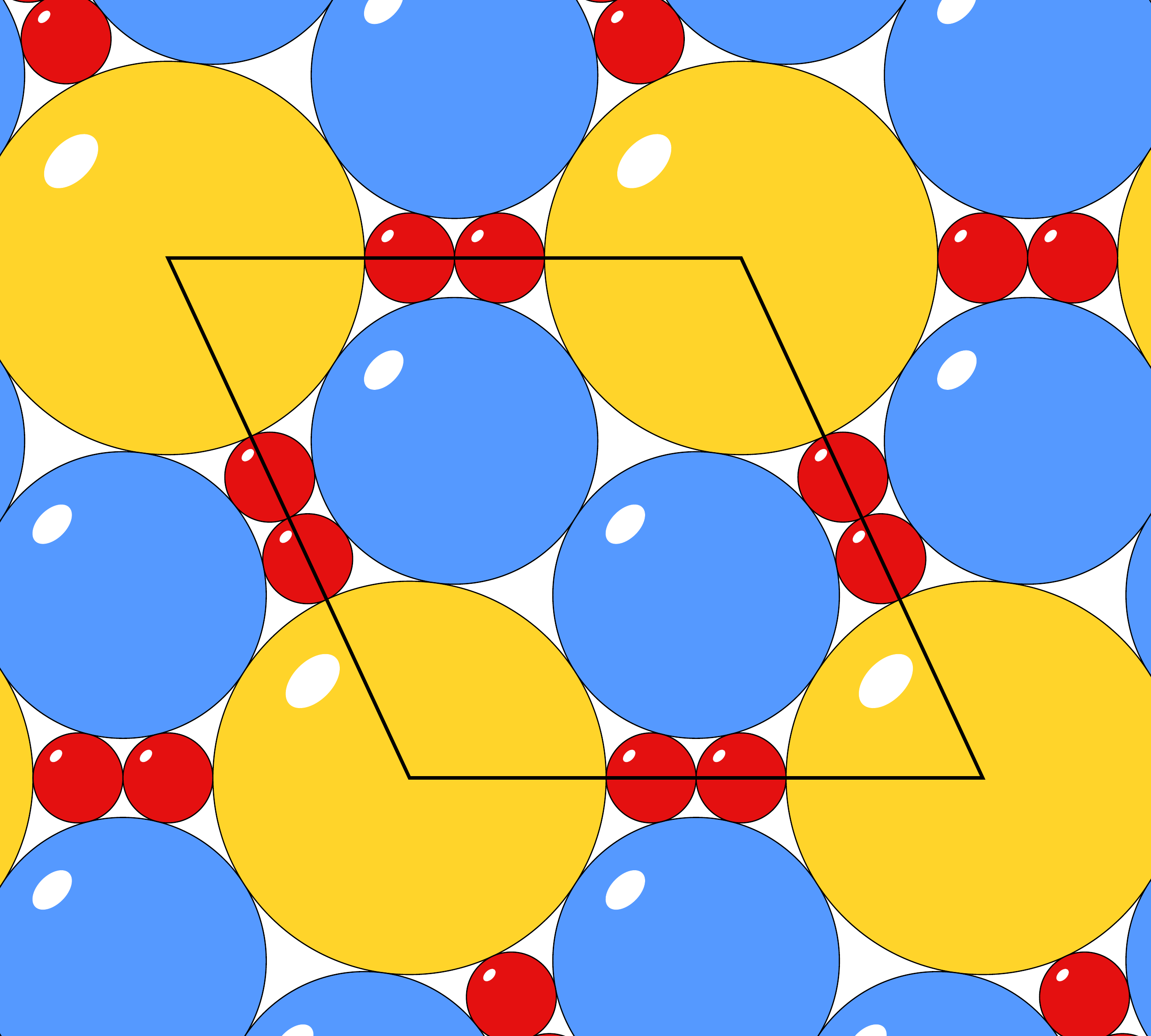} &
  \includegraphics[width=0.3\textwidth]{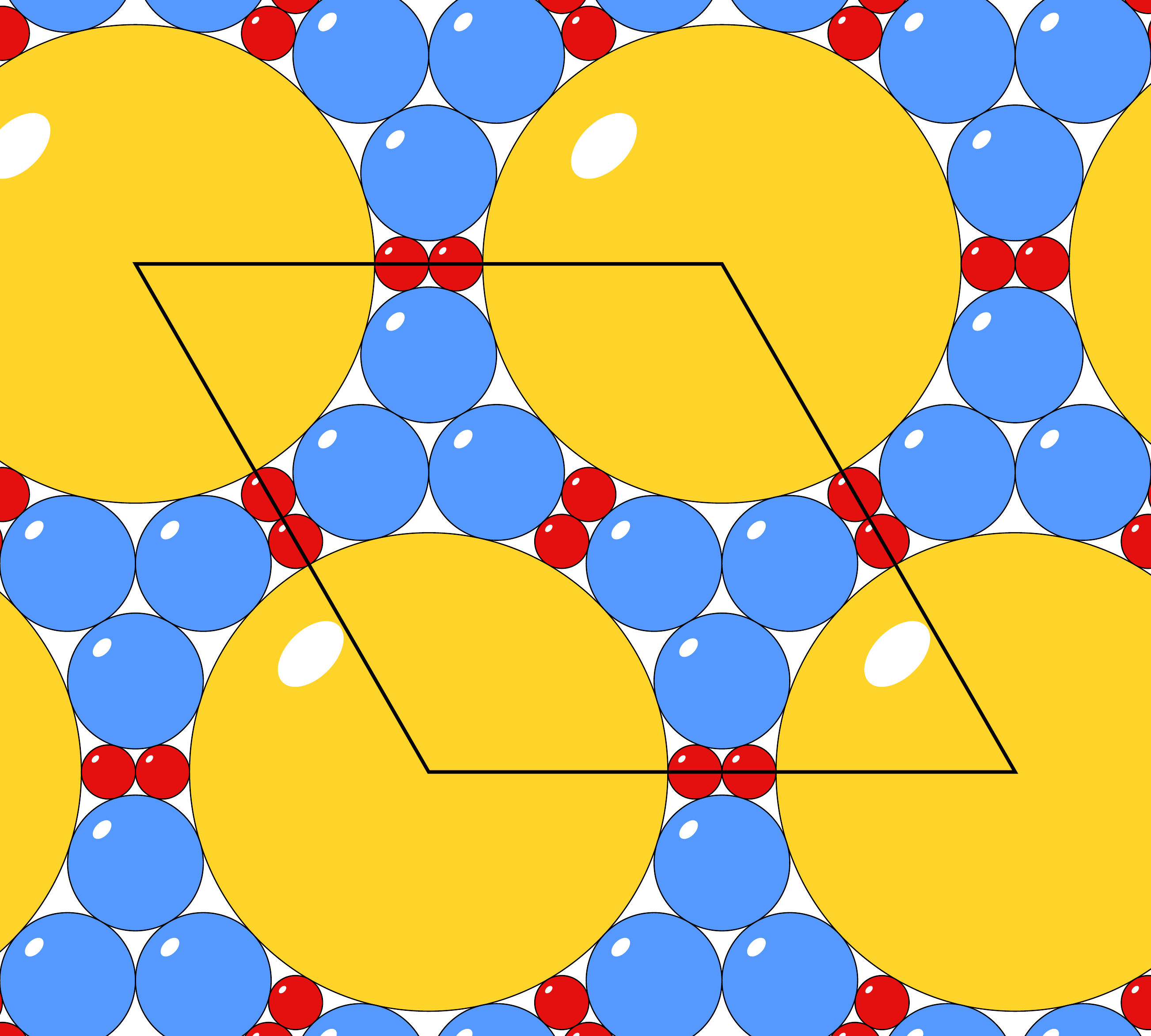}
\end{tabular}
\noindent
\begin{tabular}{lll}
  142\hfill 1rsr / 1rrr1ss & 143\hfill 1rsr / 1s1s1ssss & 144\hfill 1rsrs / 111ssss\\
  \includegraphics[width=0.3\textwidth]{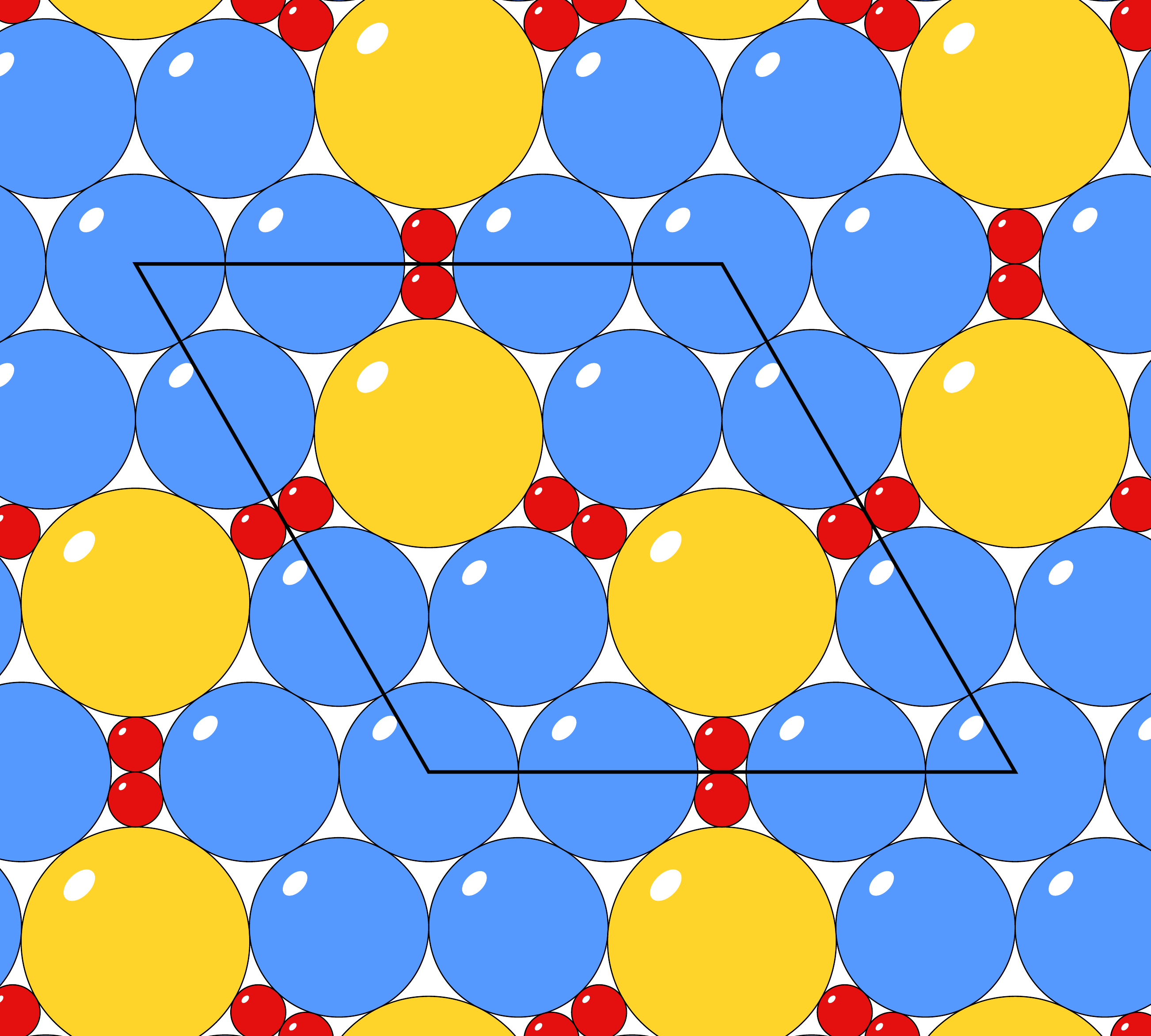} &
  \includegraphics[width=0.3\textwidth]{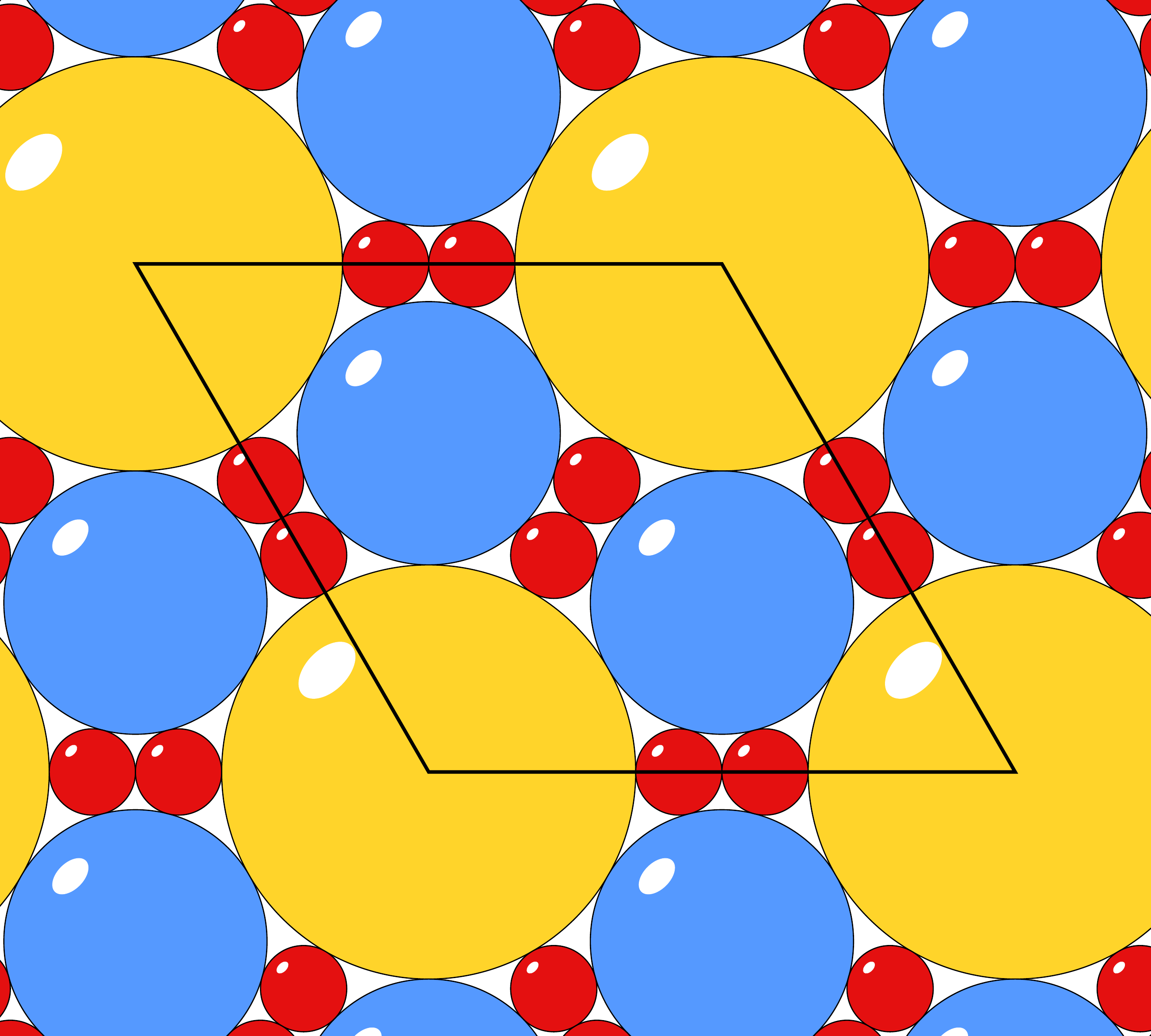} &
  \includegraphics[width=0.3\textwidth]{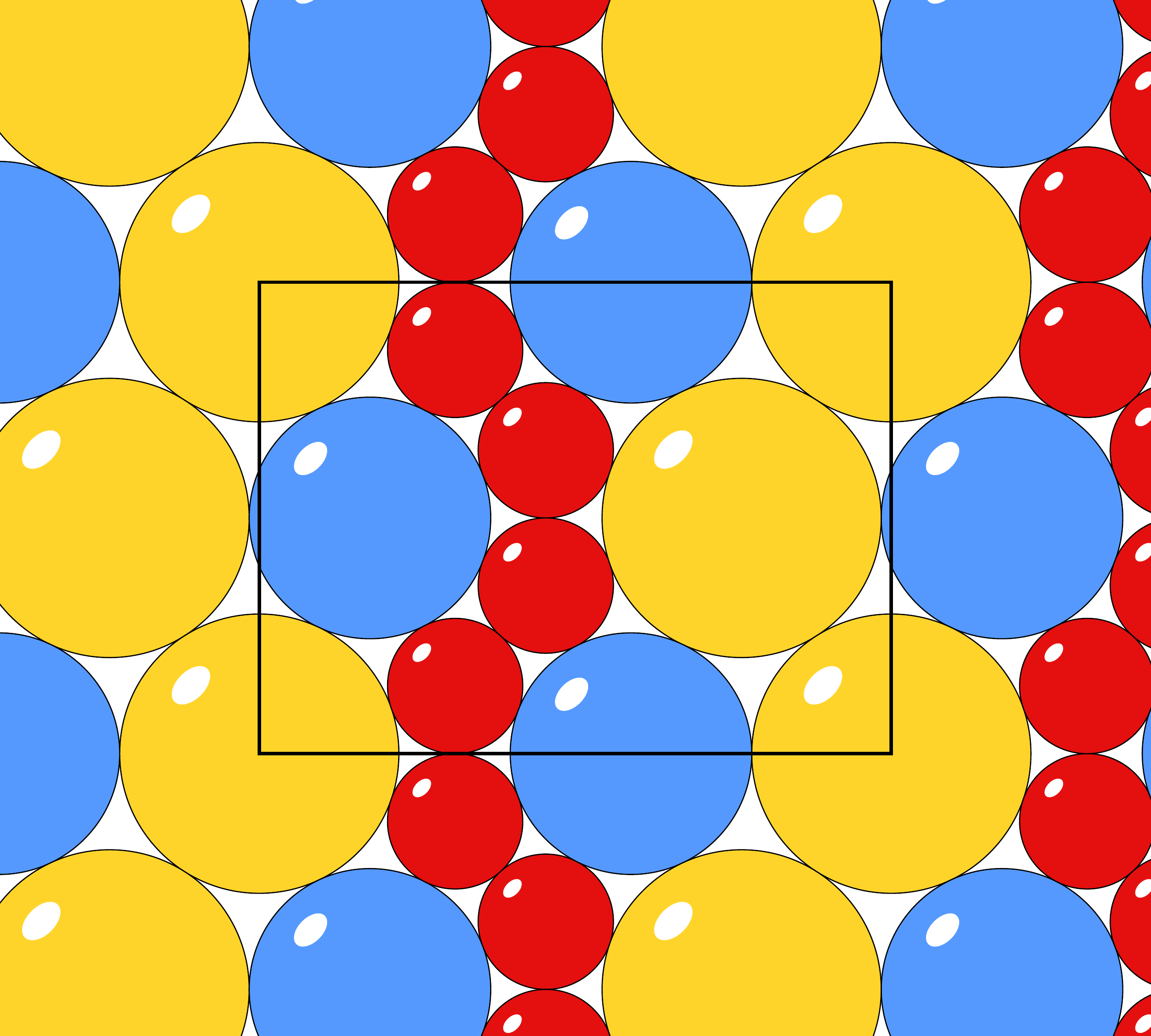}
\end{tabular}
\noindent
\begin{tabular}{lll}
  145\hfill 1rsrs / 11ssss & 146\hfill 1rsrs / 1r1ssss & 147\hfill 1rssr / 111ss\\
  \includegraphics[width=0.3\textwidth]{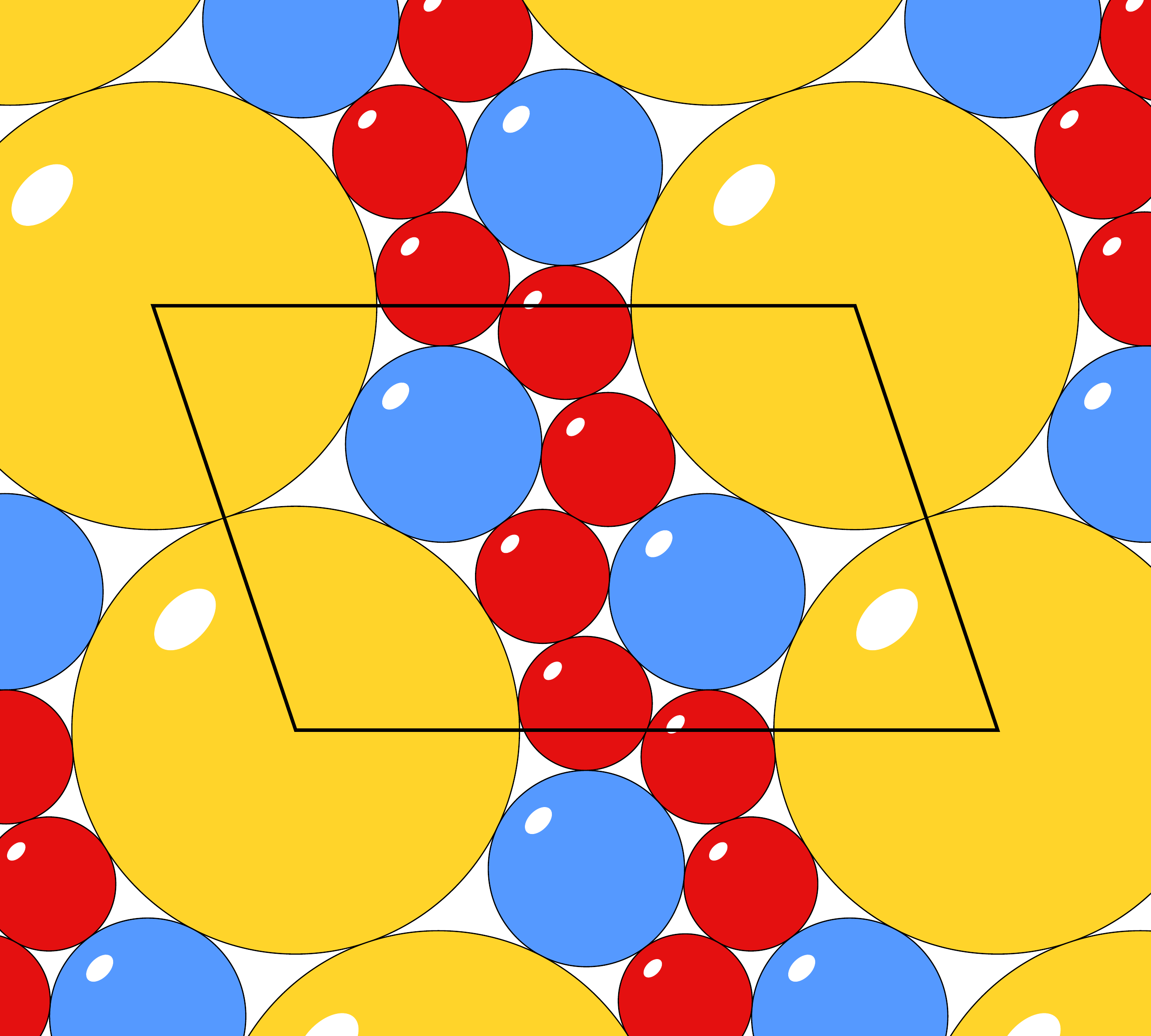} &
  \includegraphics[width=0.3\textwidth]{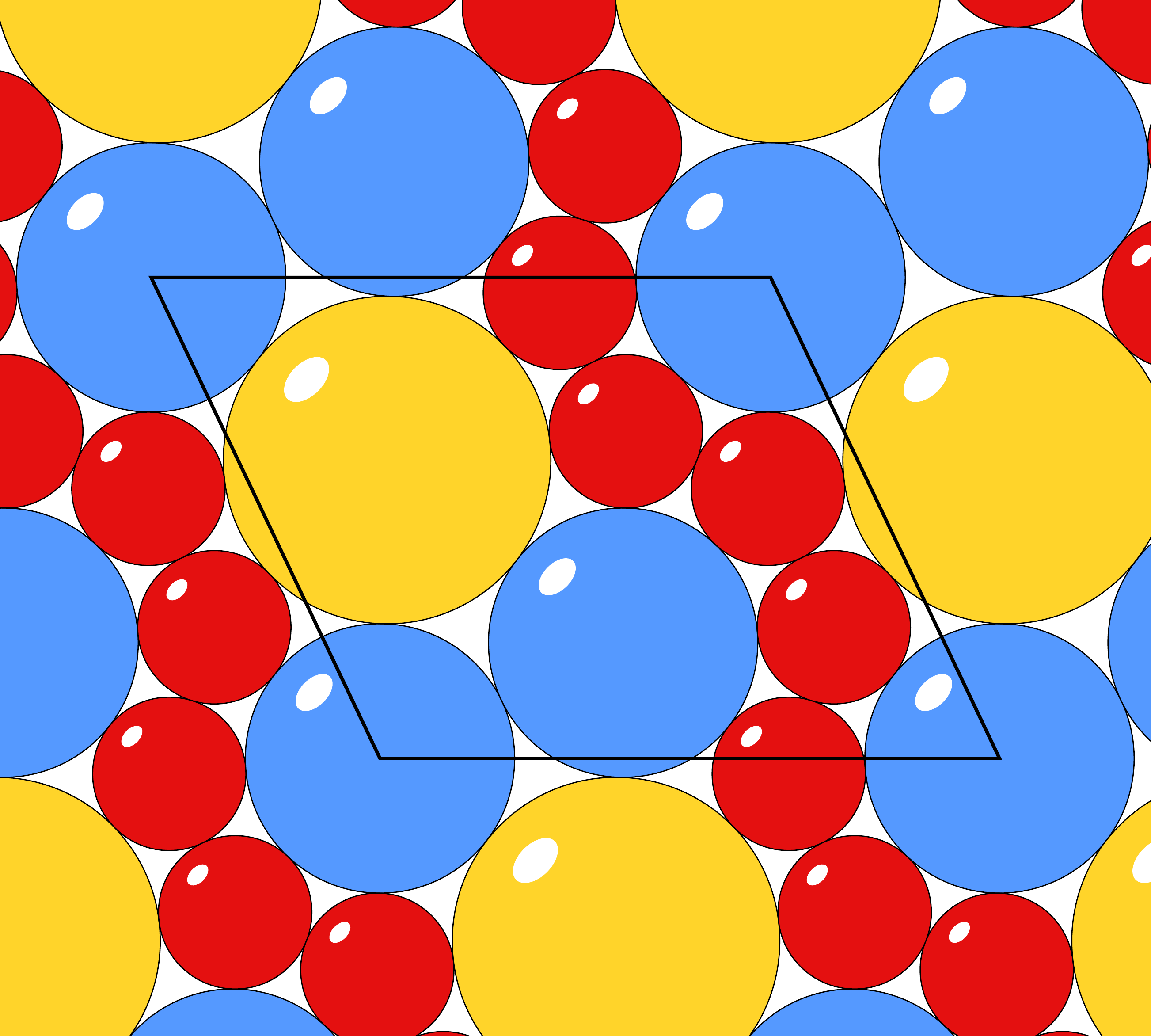} &
  \includegraphics[width=0.3\textwidth]{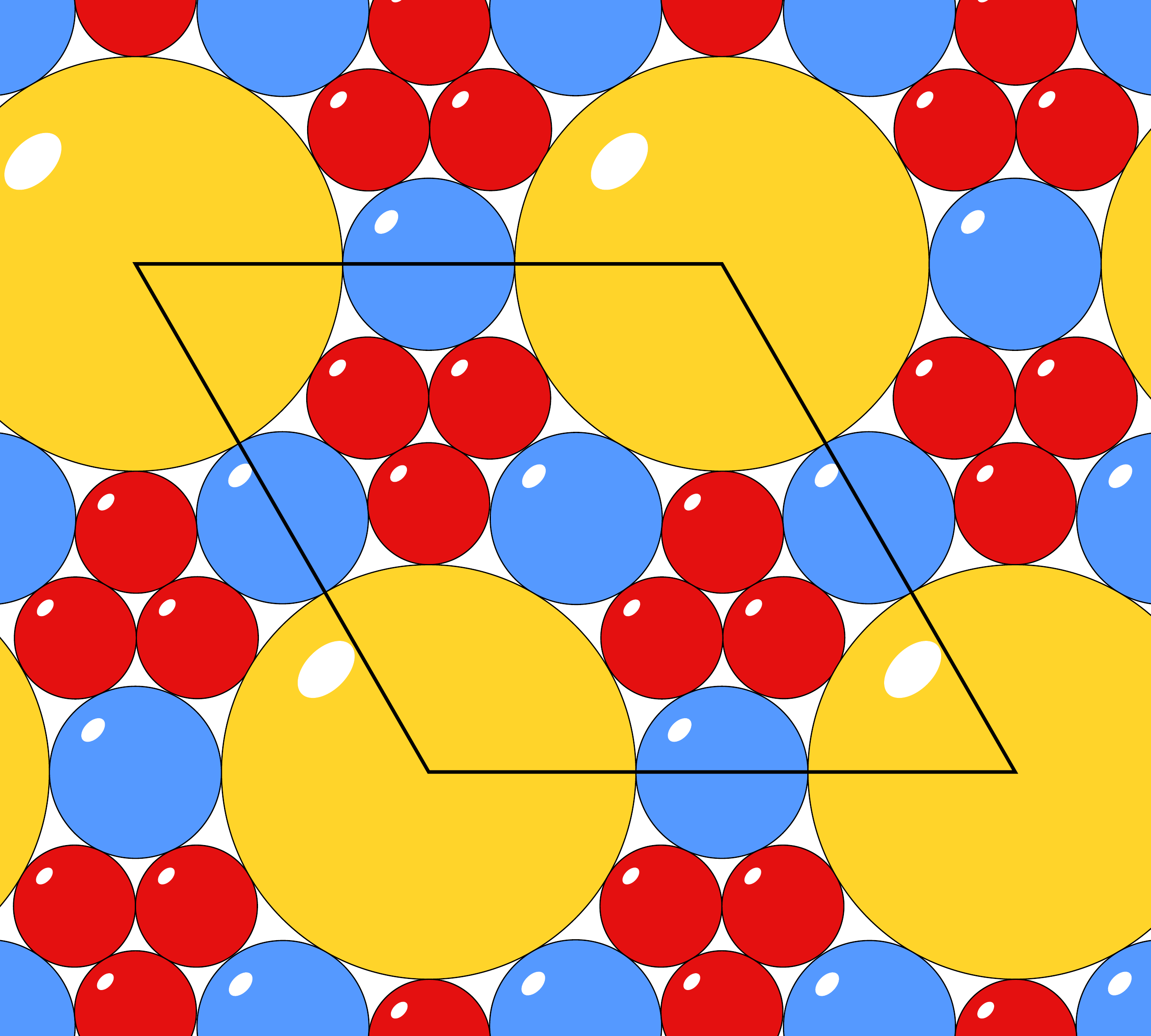}
\end{tabular}
\noindent
\begin{tabular}{lll}
  148\hfill 1rssr / 11r1ss & 149\hfill 1rssr / 11s1sss & 150\hfill 1rssr / 1r1ss\\
  \includegraphics[width=0.3\textwidth]{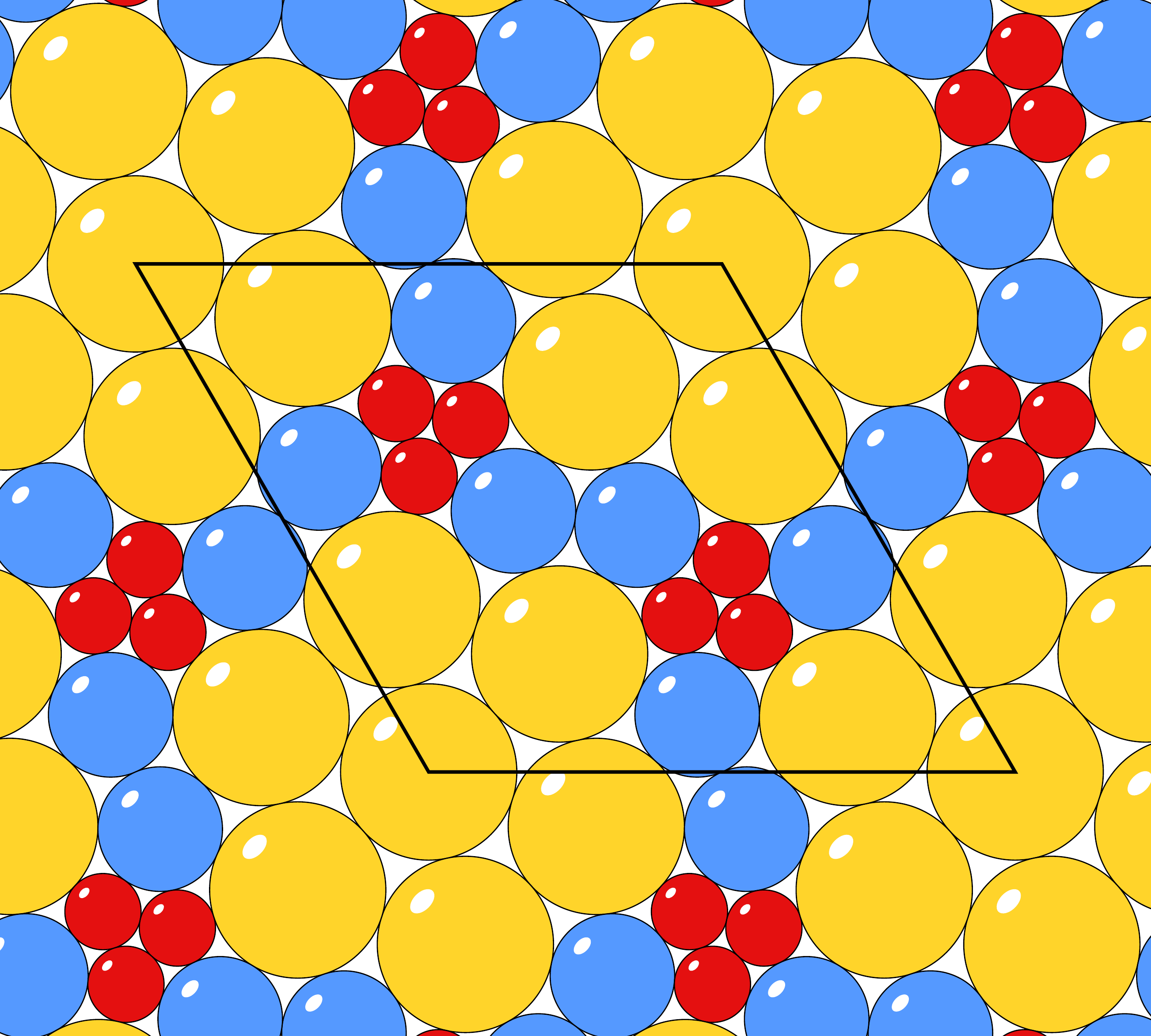} &
  \includegraphics[width=0.3\textwidth]{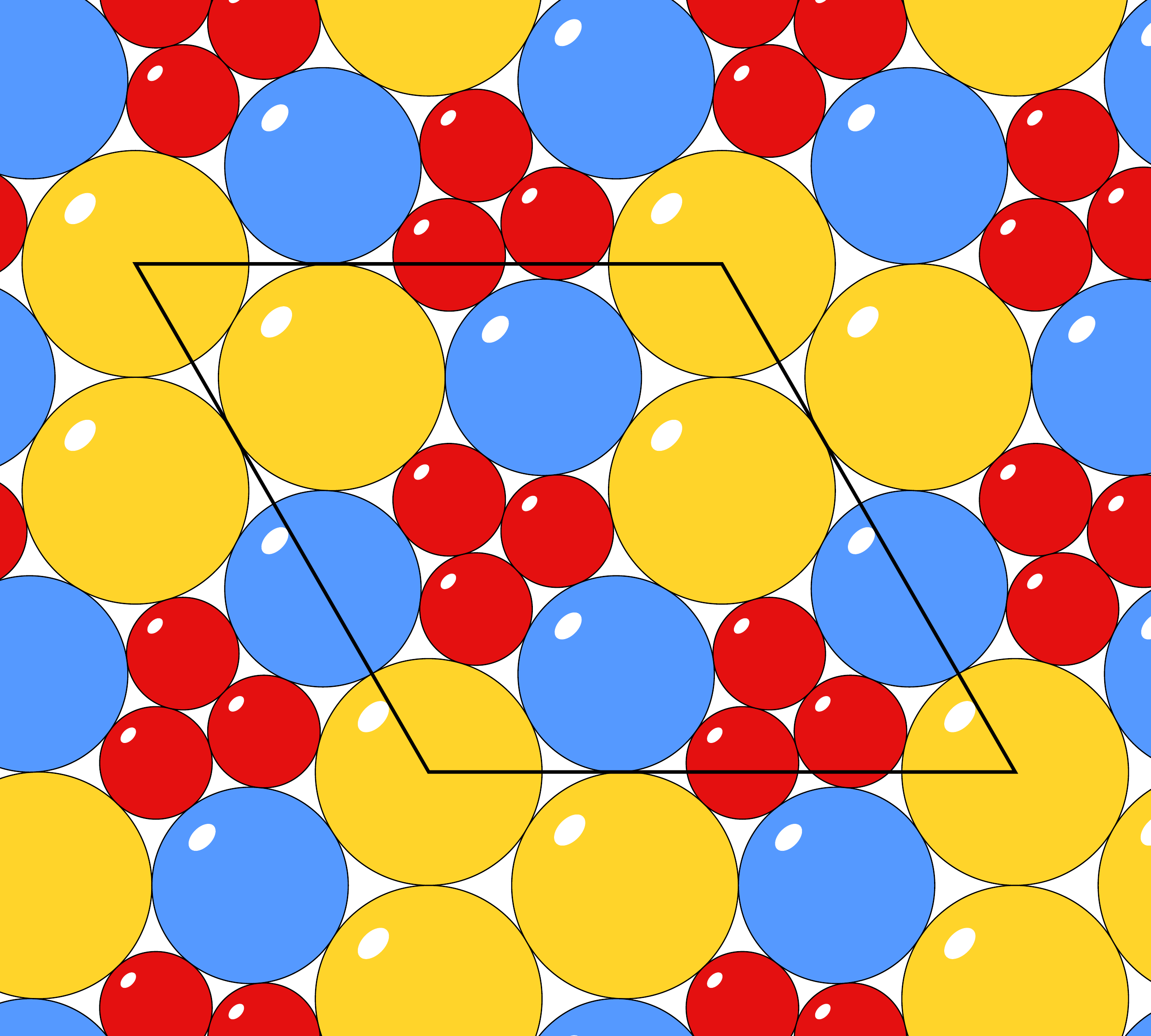} &
  \includegraphics[width=0.3\textwidth]{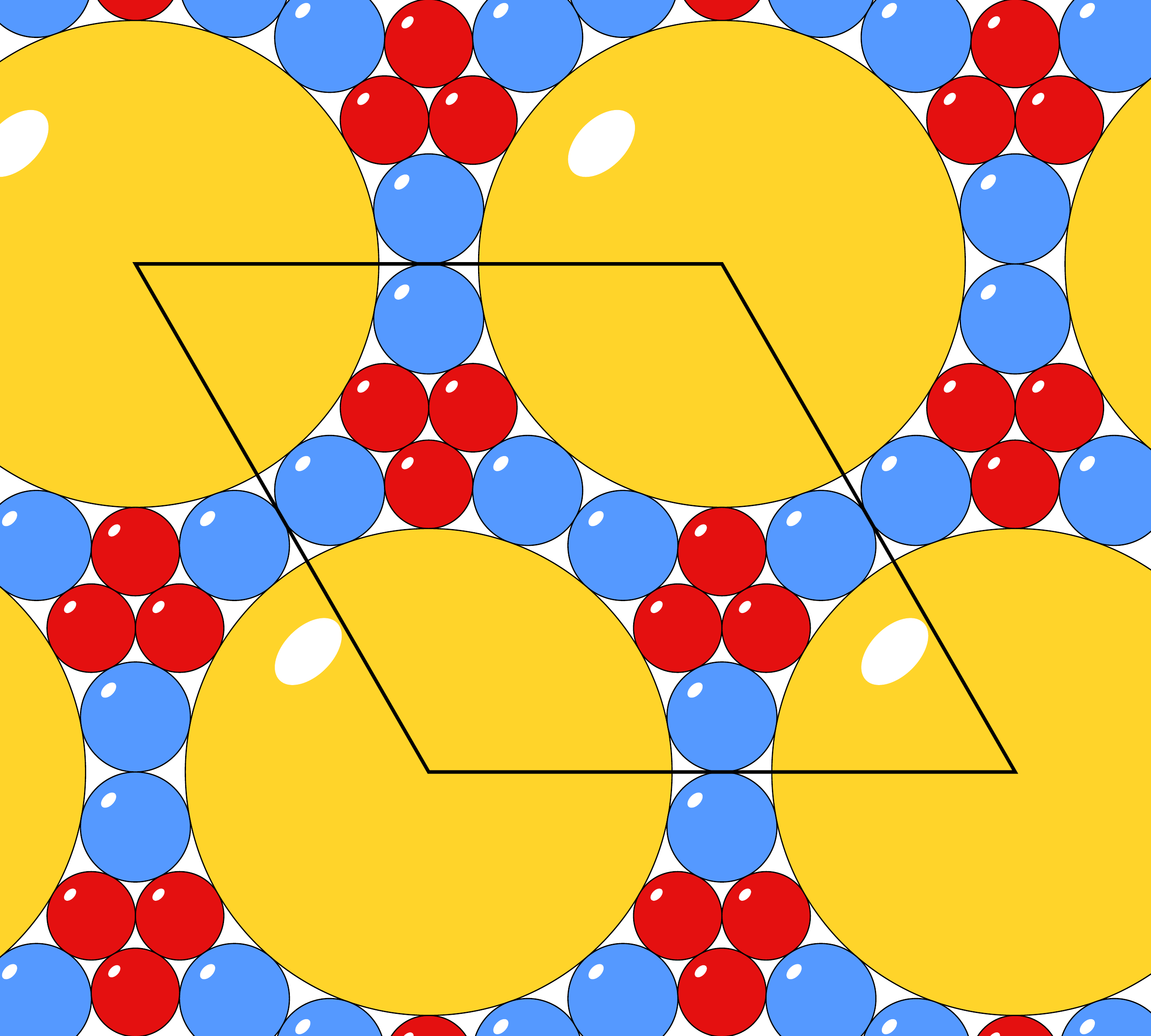}
\end{tabular}
\noindent
\begin{tabular}{lll}
  151\hfill 1rssr / 1rr1ss & 152\hfill 1rssr / 1s1sss & 153\hfill 1rsss / 111ss\\
  \includegraphics[width=0.3\textwidth]{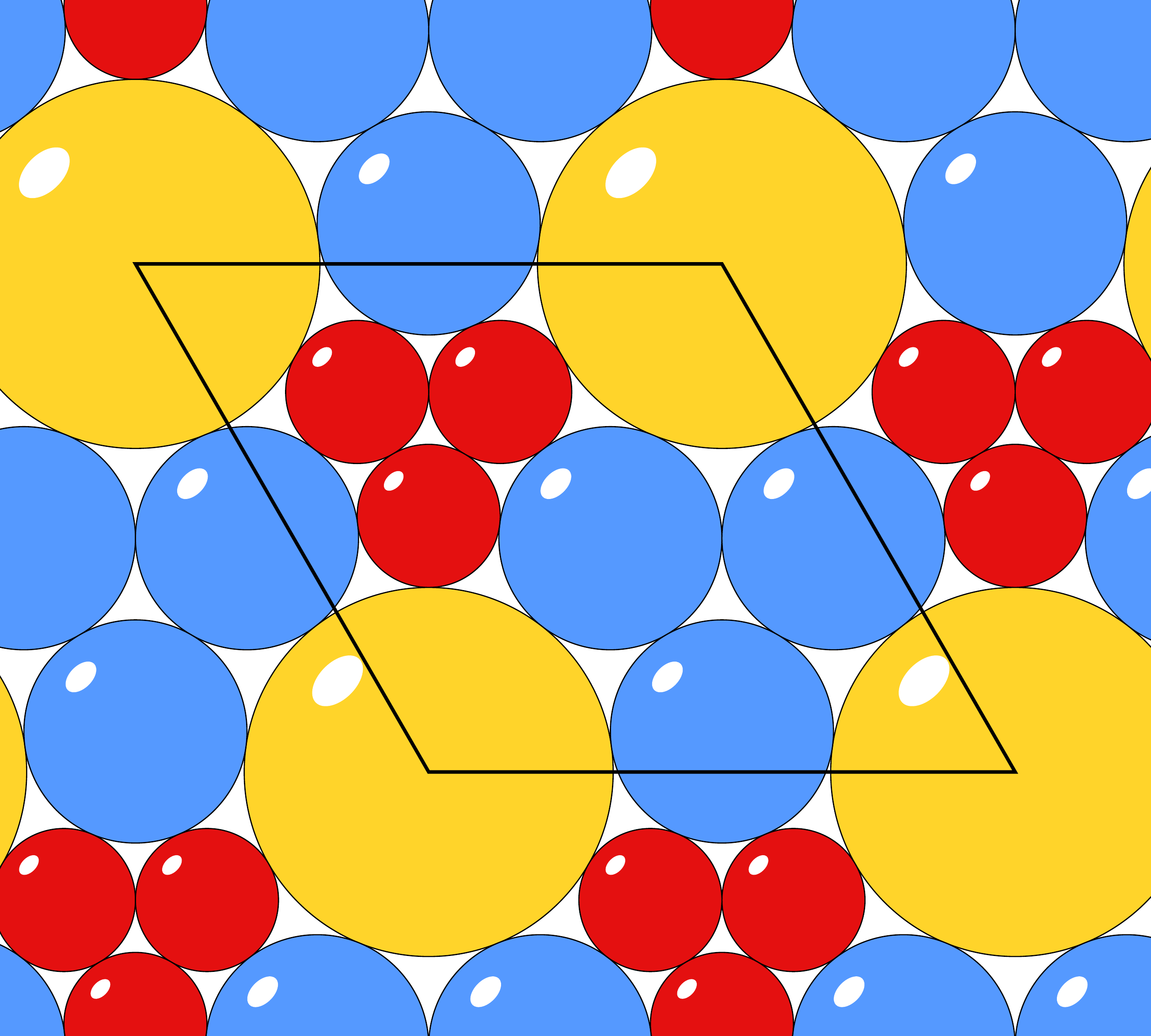} &
  \includegraphics[width=0.3\textwidth]{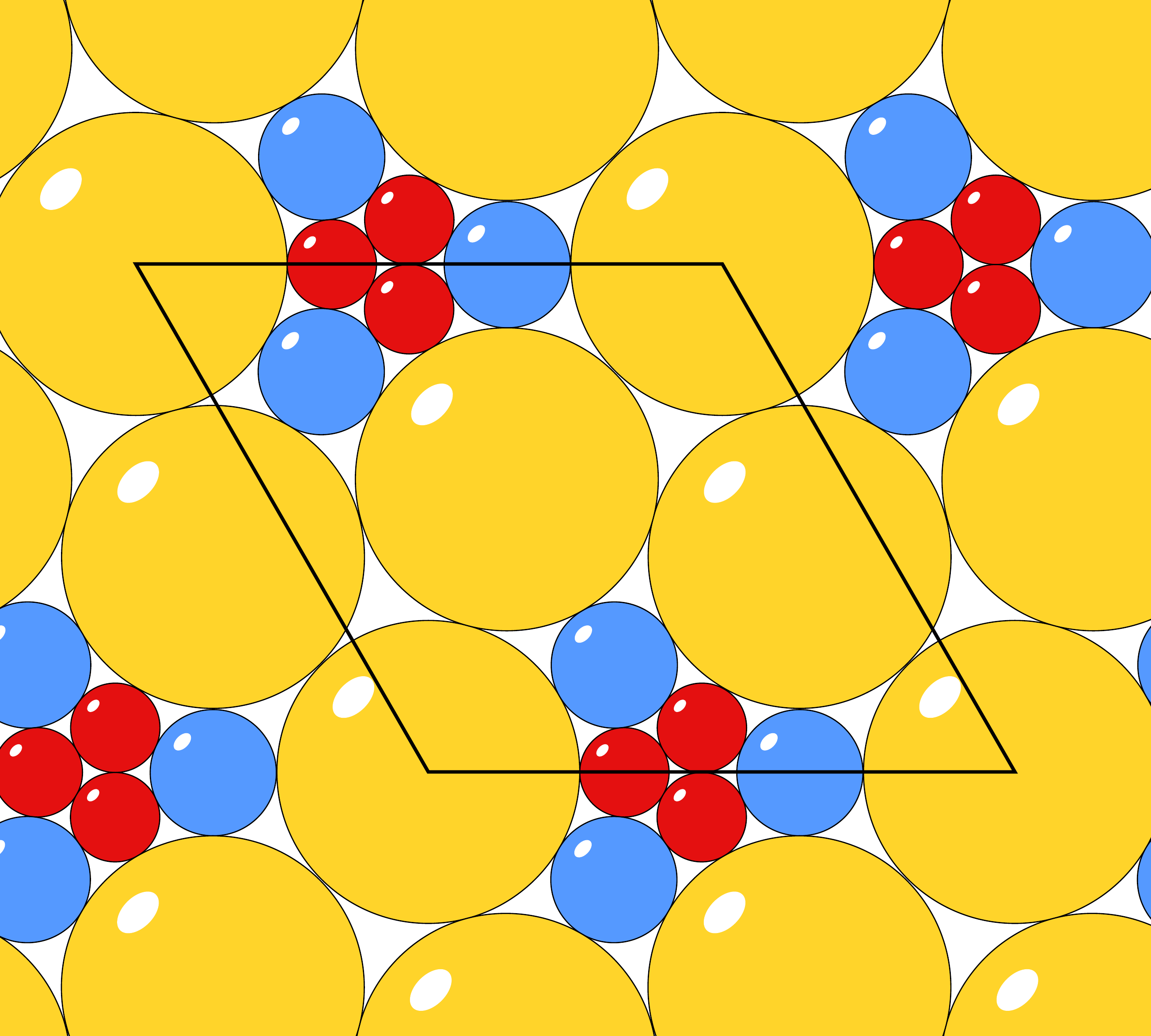} &
  \includegraphics[width=0.3\textwidth]{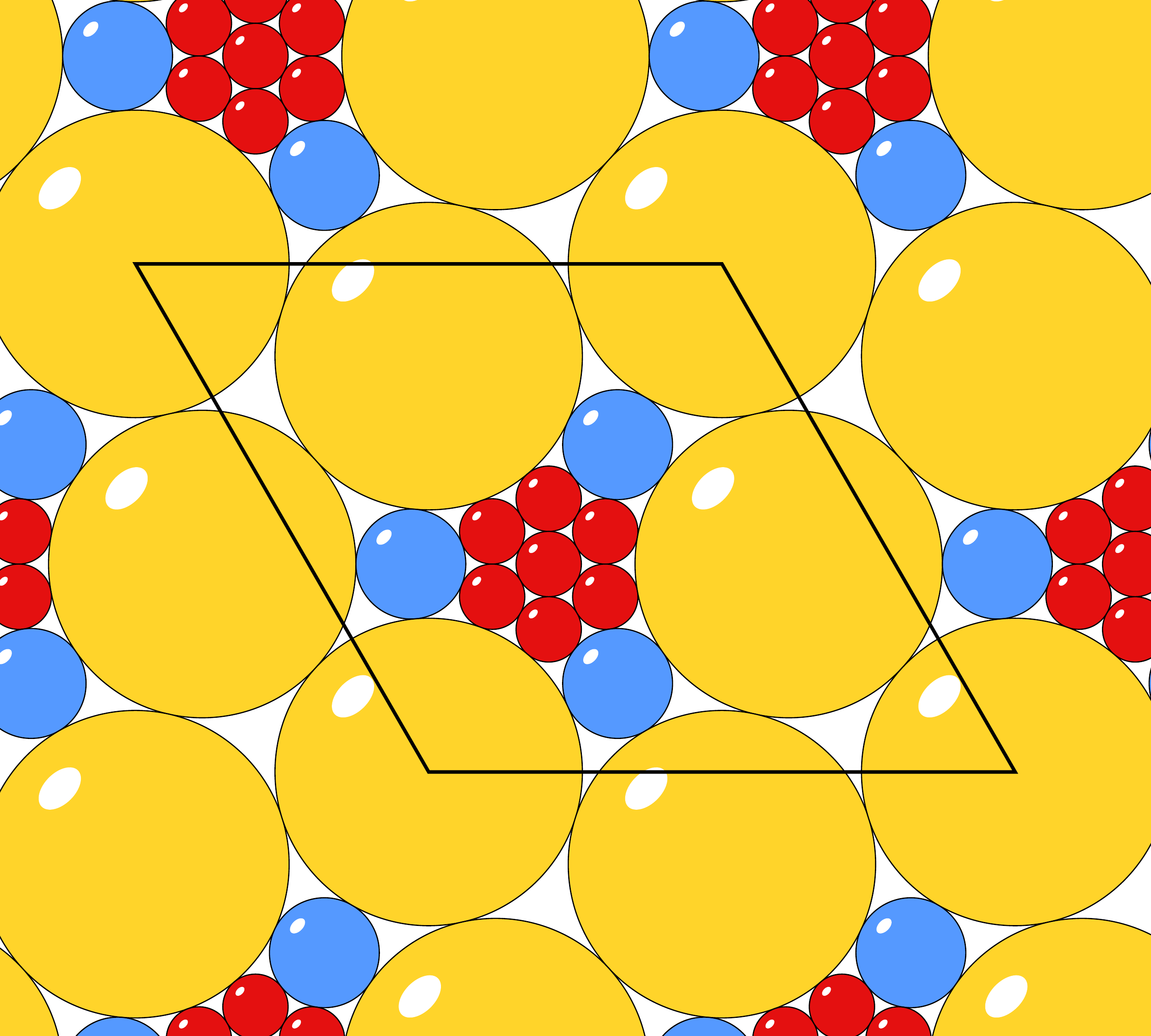}
\end{tabular}
\noindent
\begin{tabular}{lll}
  154\hfill 1rsss / 11r1ss & 155\hfill 1rsss / 11s1sss & 156\hfill 1rsss / 1r1ss\\
  \includegraphics[width=0.3\textwidth]{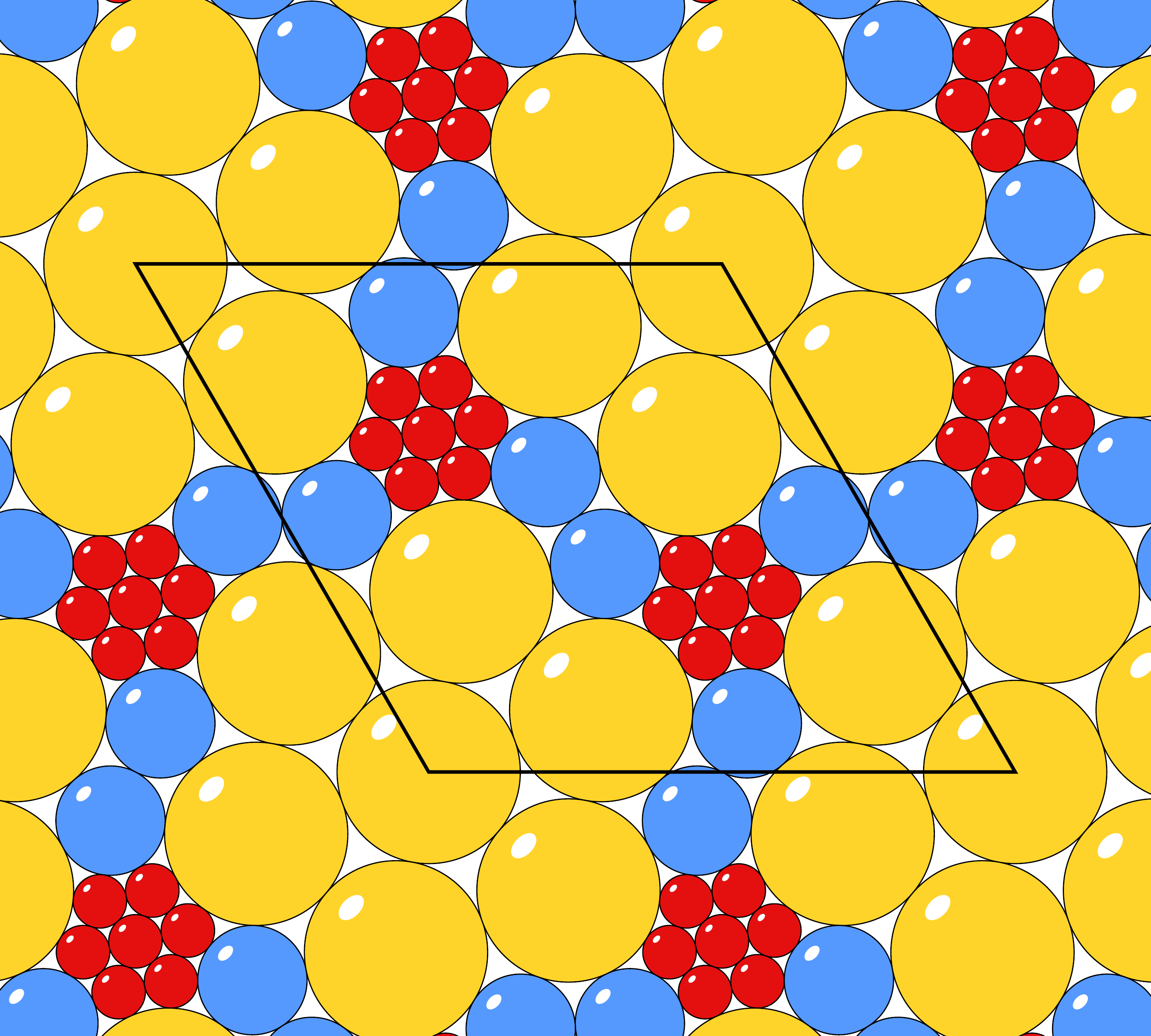} &
  \includegraphics[width=0.3\textwidth]{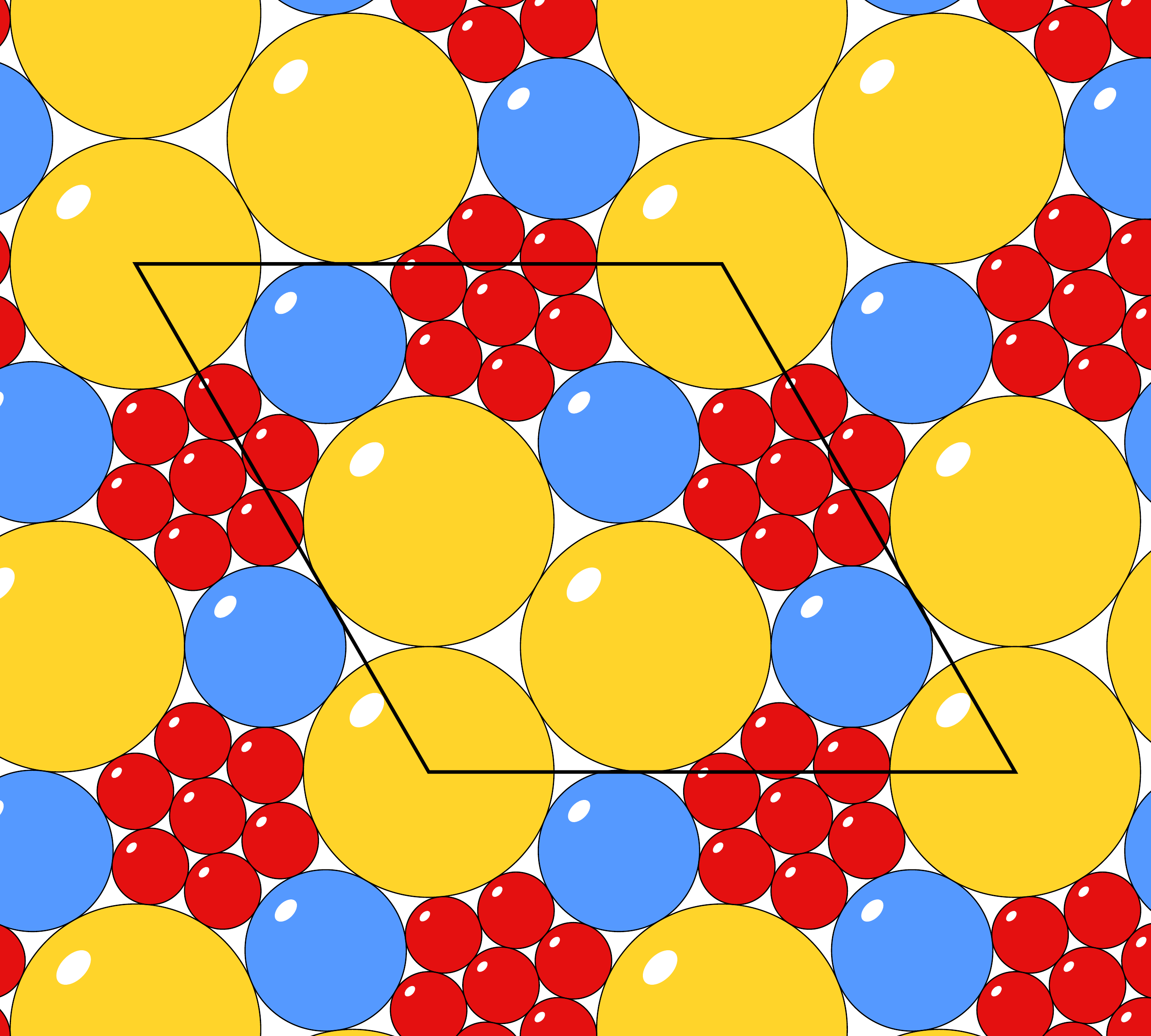} &
  \includegraphics[width=0.3\textwidth]{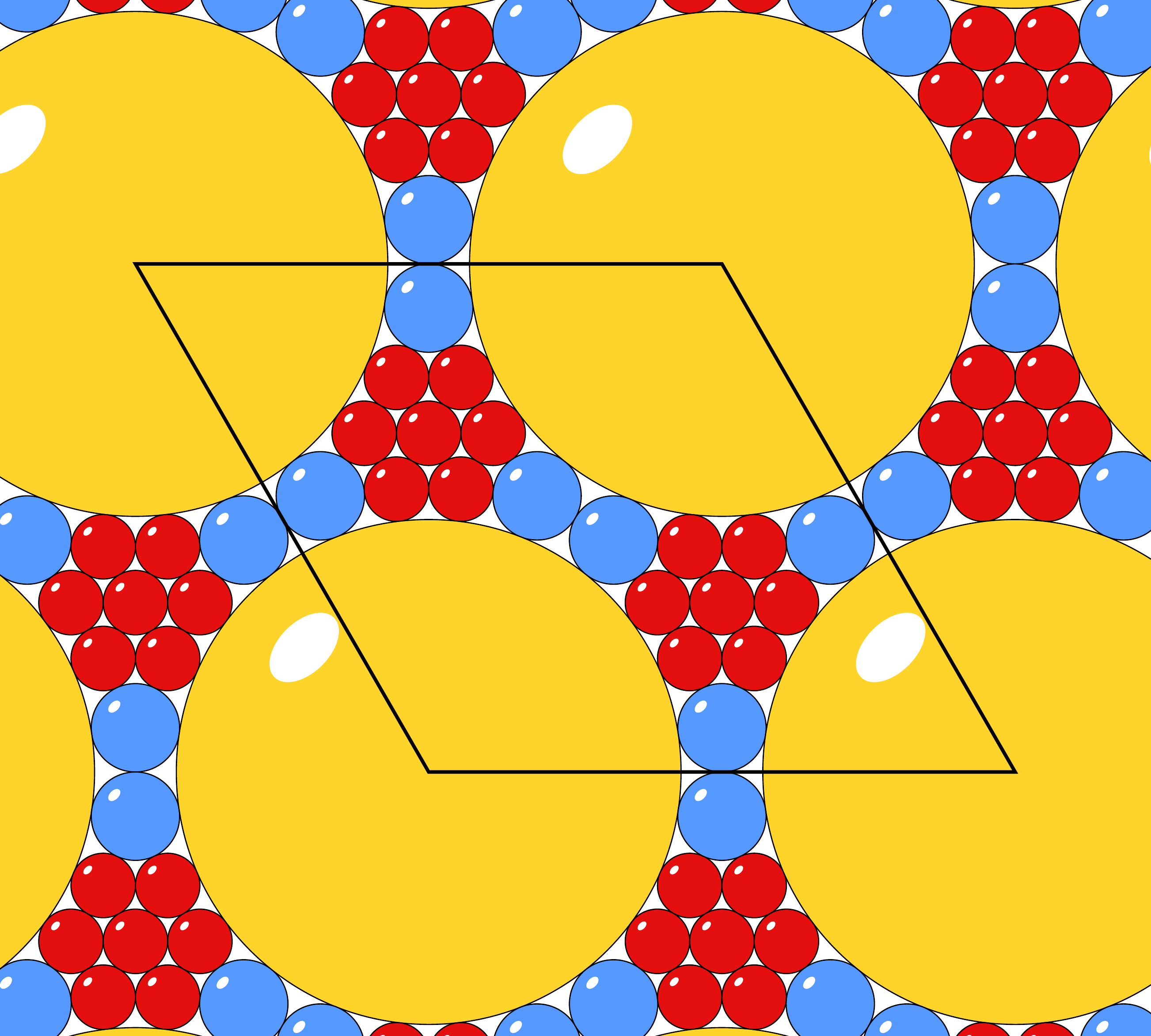}
\end{tabular}
\noindent
\begin{tabular}{lll}
  157\hfill 1rsss / 1rr1ss & 158\hfill 1rsss / 1s1sss & 159\hfill rrrr / 11rrsr\\
  \includegraphics[width=0.3\textwidth]{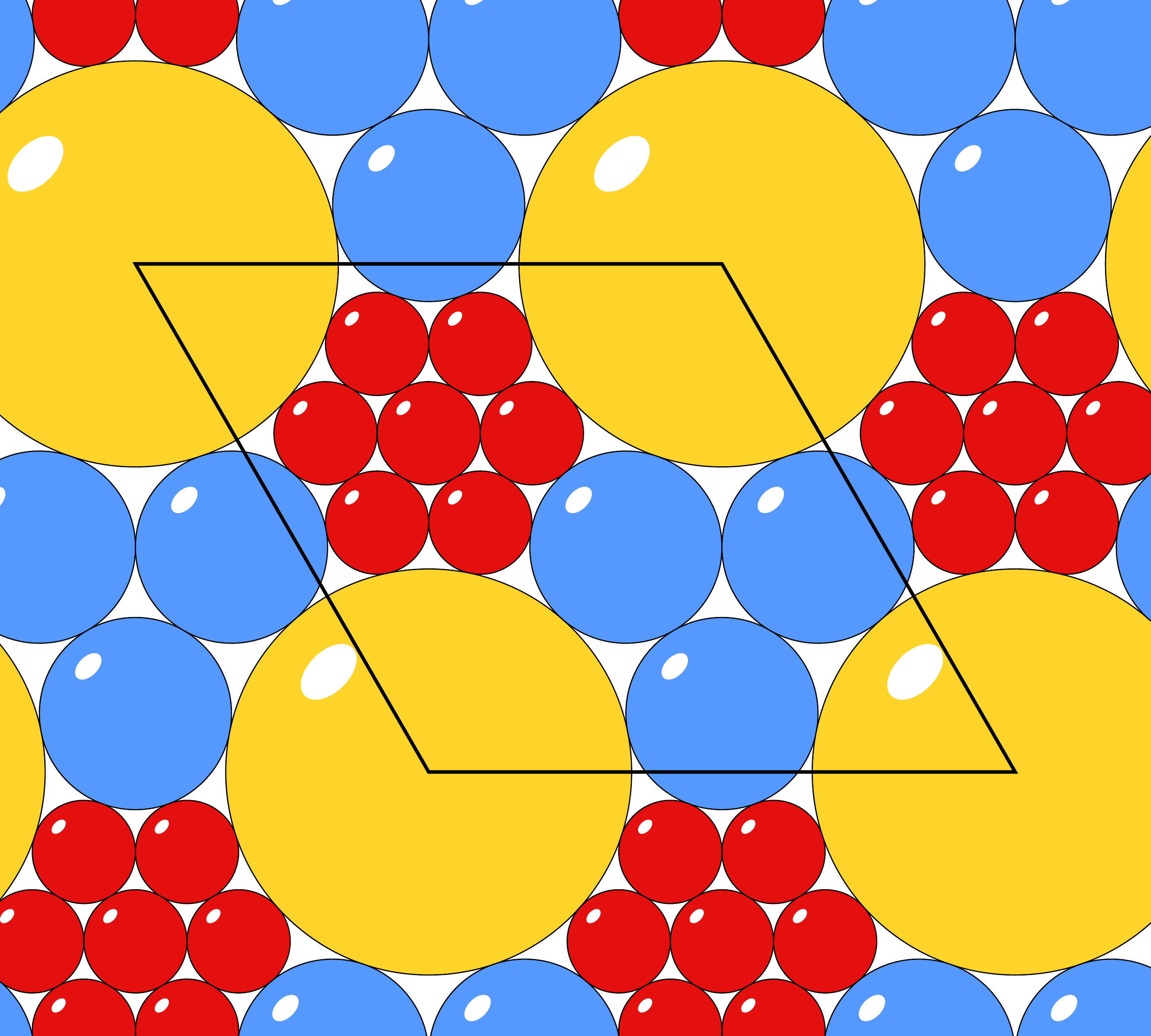} &
  \includegraphics[width=0.3\textwidth]{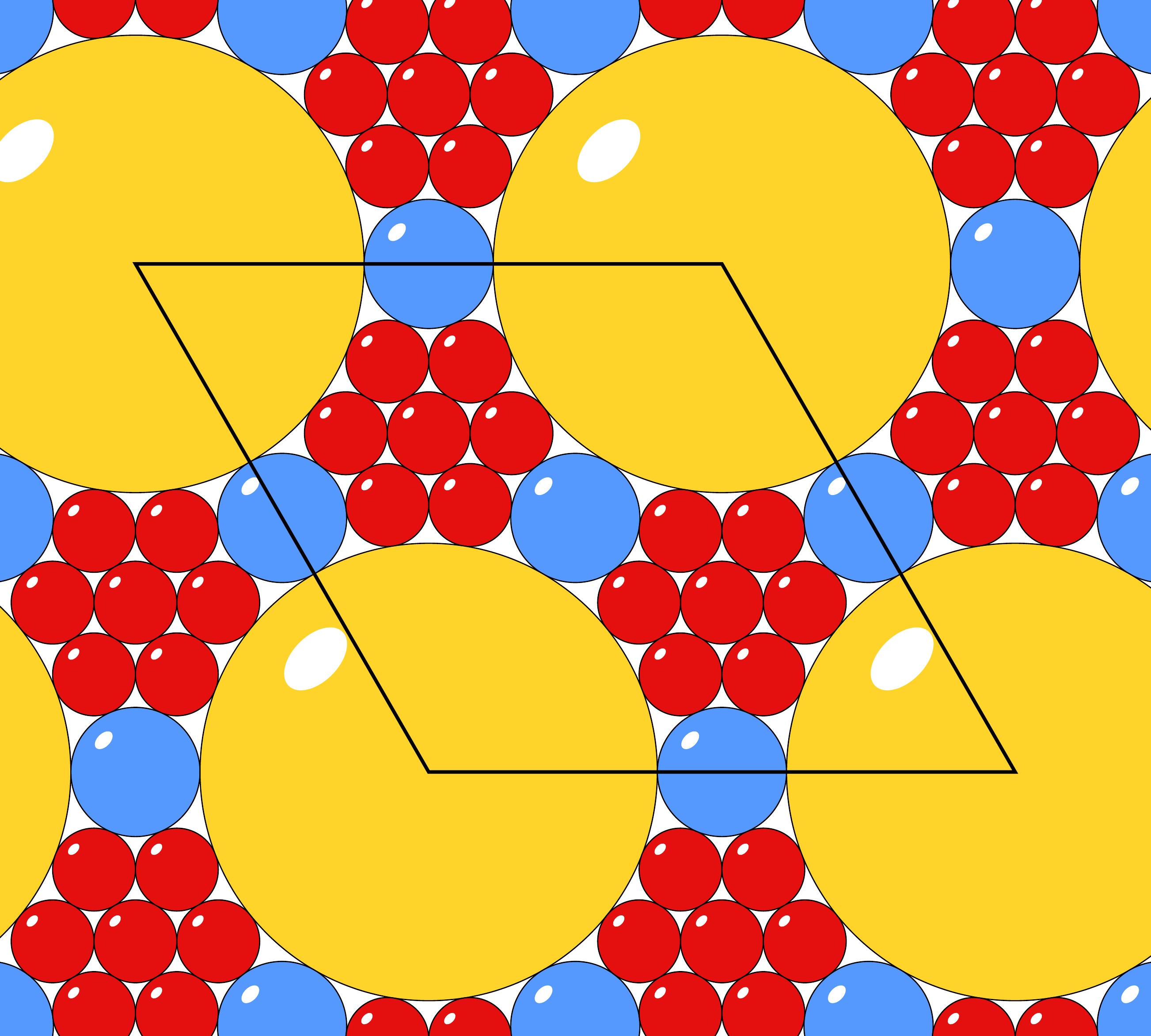} &
  \includegraphics[width=0.3\textwidth]{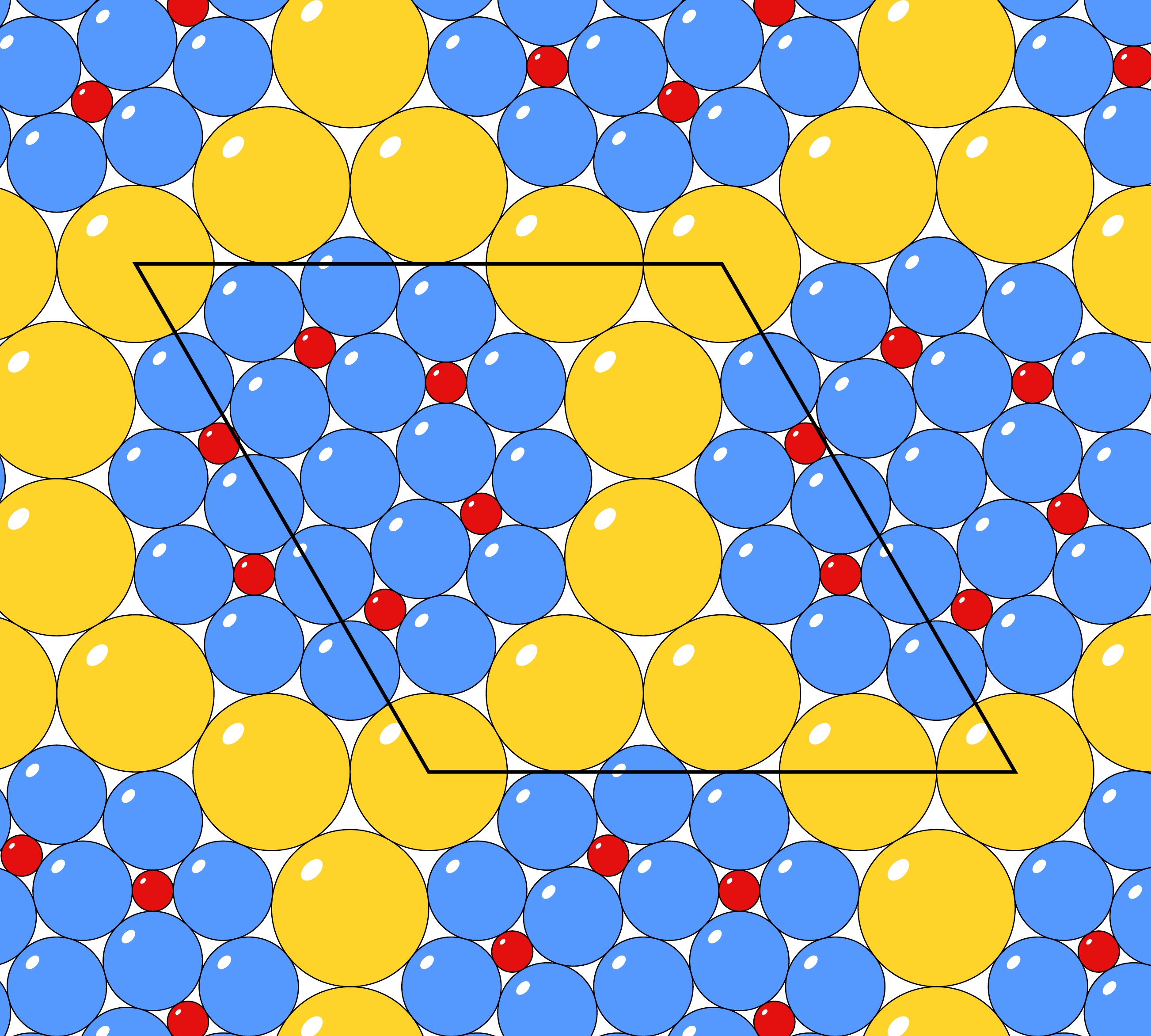}
\end{tabular}
\noindent
\begin{tabular}{lll}
  160\hfill rrrr / 11rsr & 161\hfill rrrr / 1r1rsr & 162\hfill rrrr / 1rrrsr\\
  \includegraphics[width=0.3\textwidth]{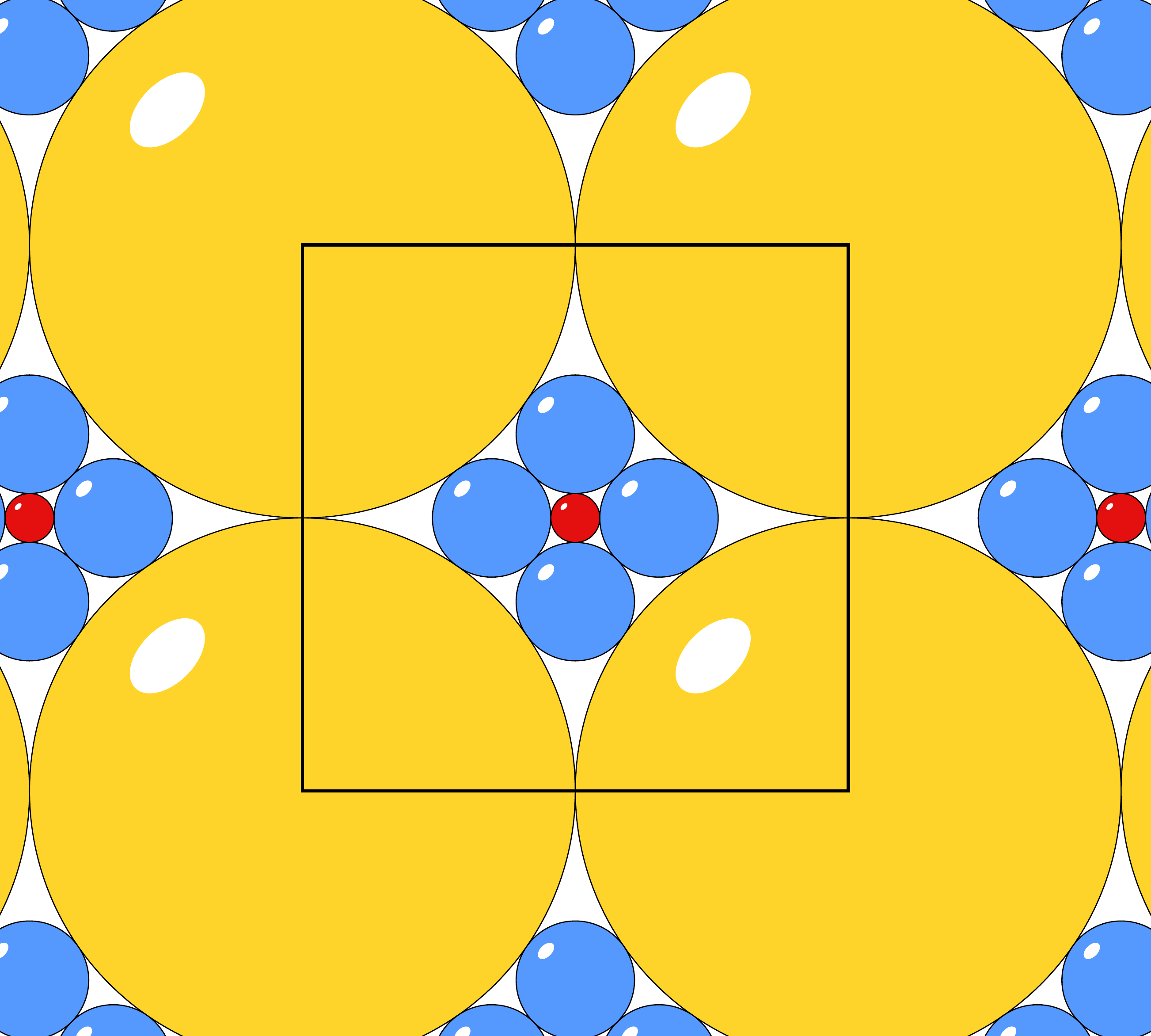} &
  \includegraphics[width=0.3\textwidth]{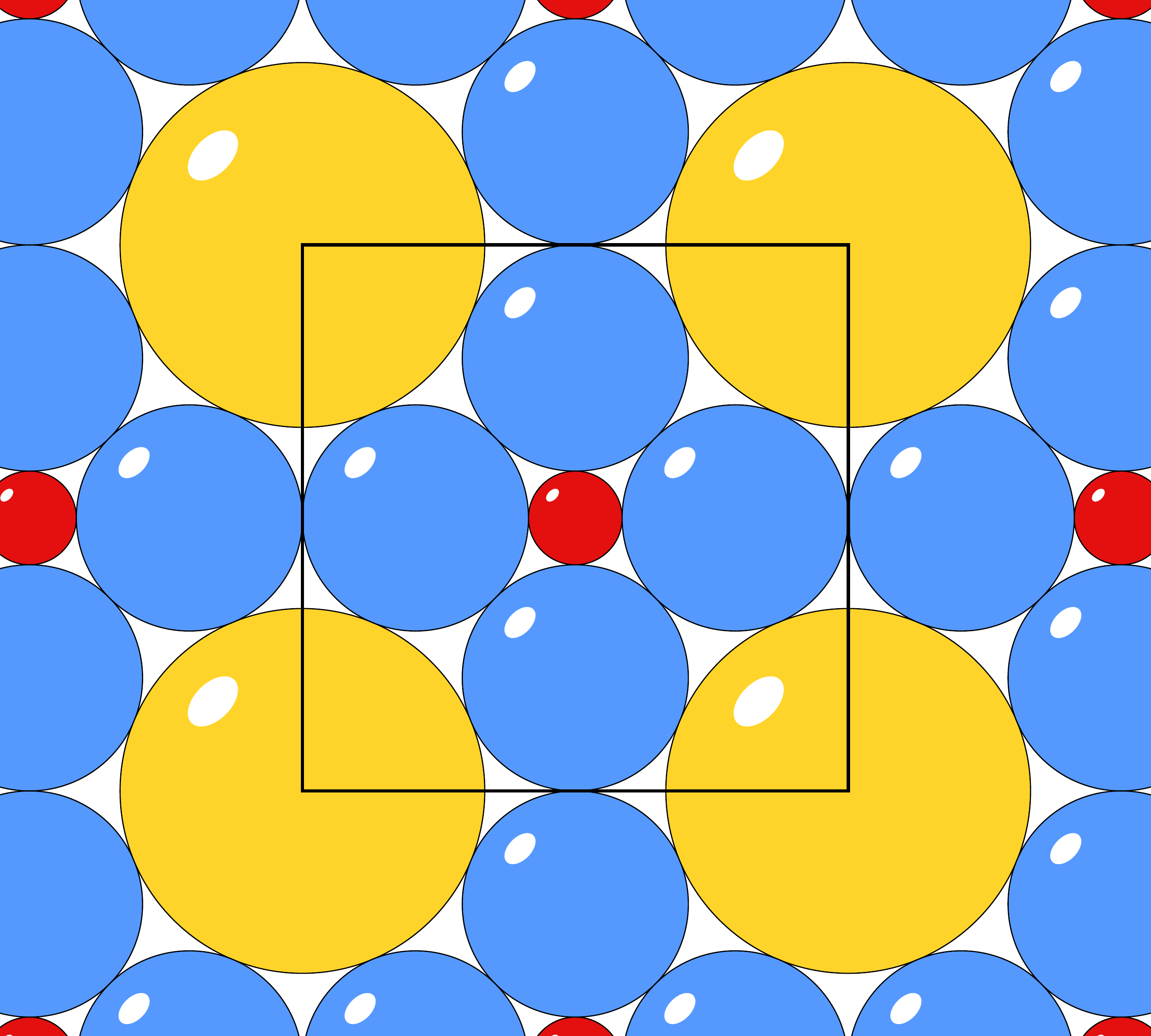} &
  \includegraphics[width=0.3\textwidth]{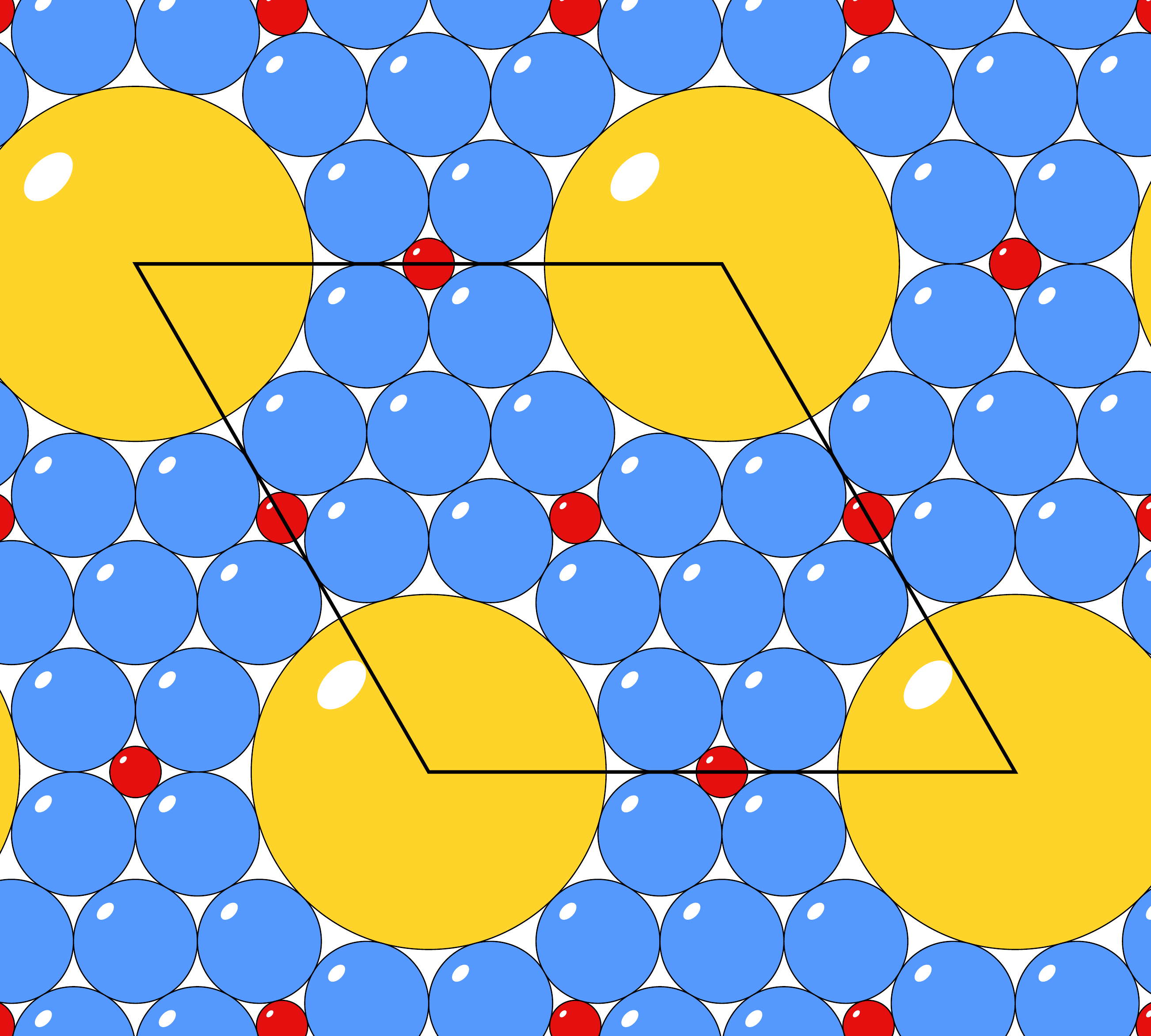}
\end{tabular}
\noindent
\begin{tabular}{lll}
  163\hfill rrsrs / 1rssssr & 164\hfill rrsss / 1rrsssr & \\
  \includegraphics[width=0.3\textwidth]{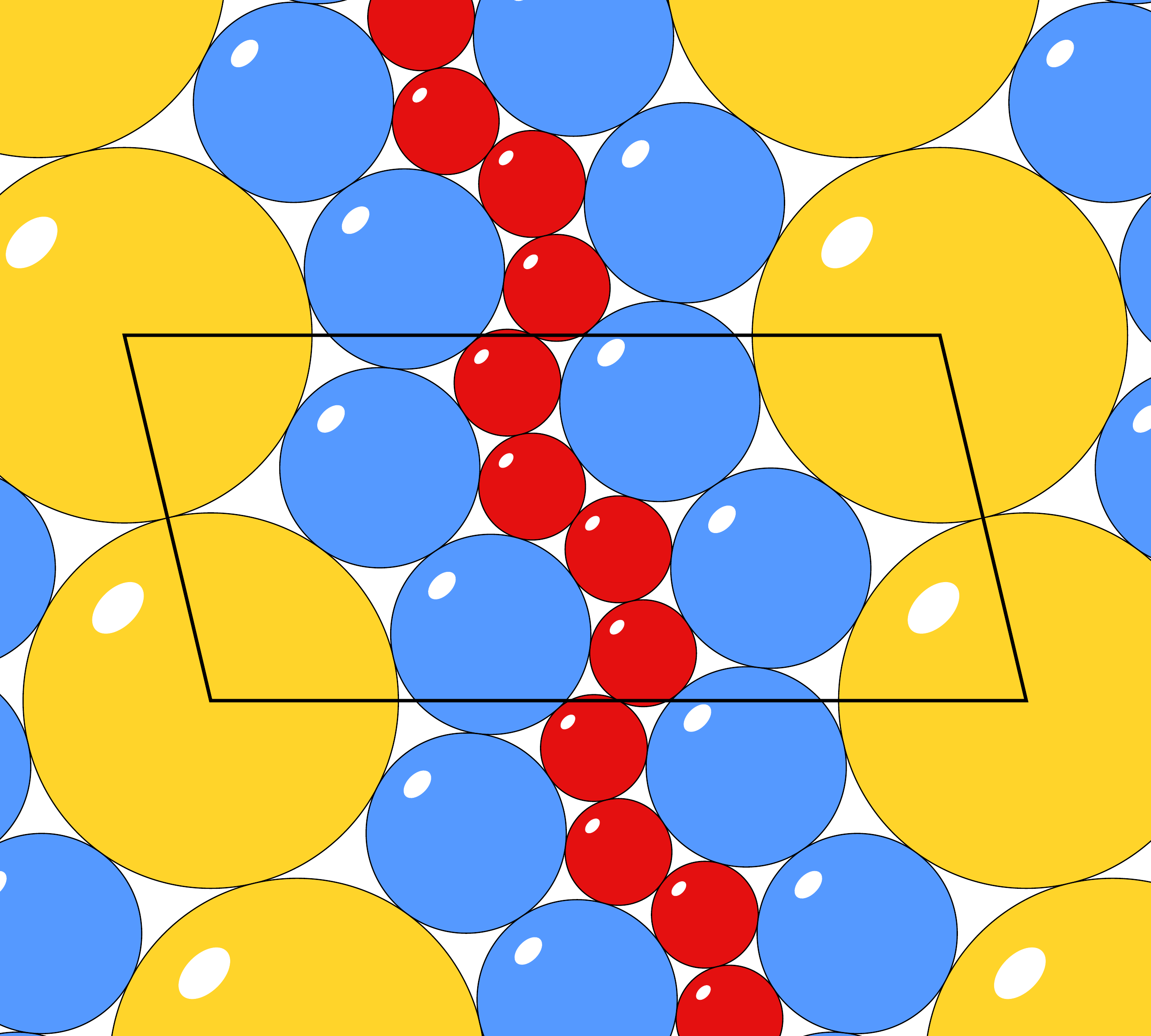} &
  \includegraphics[width=0.3\textwidth]{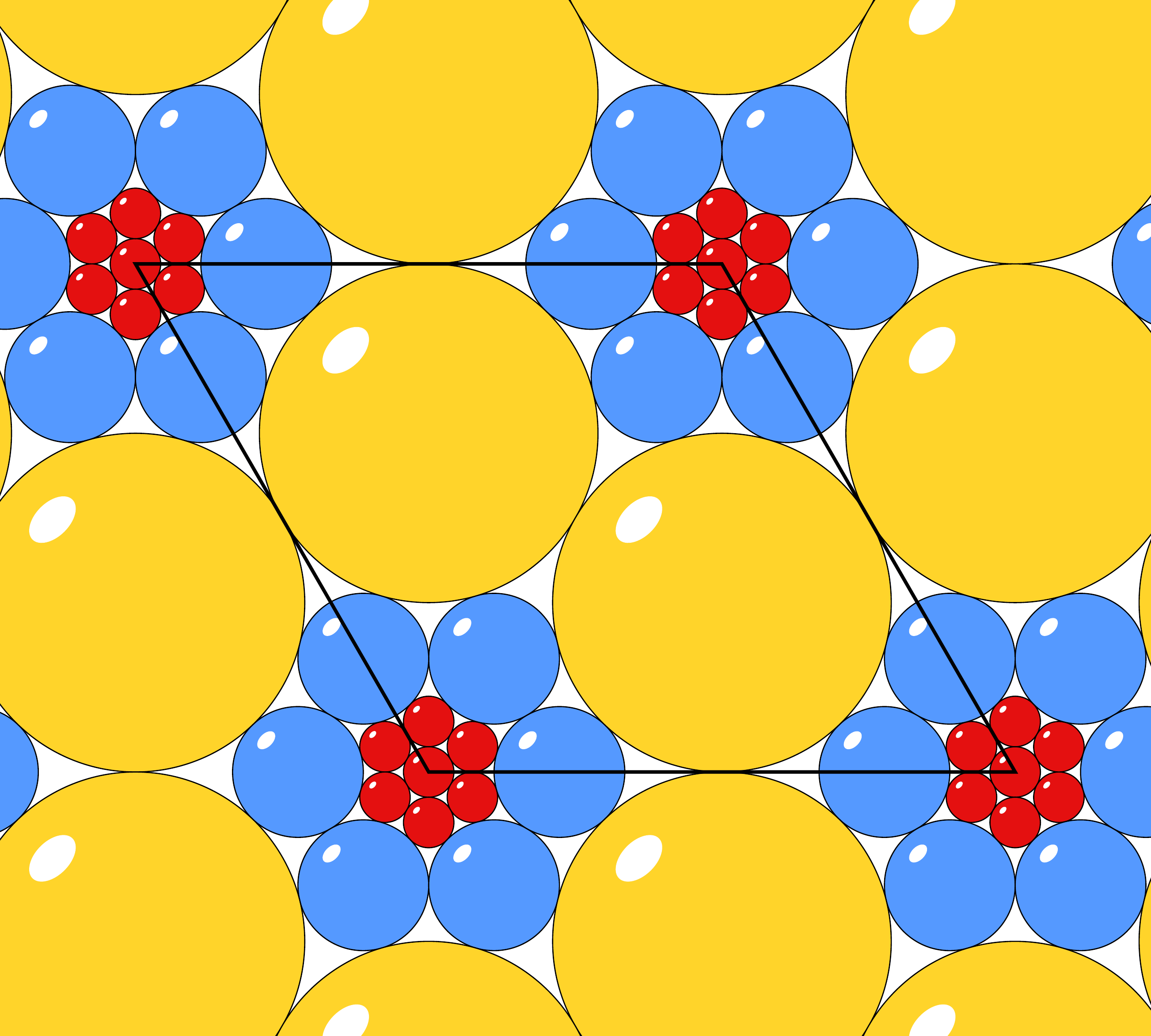} &
\end{tabular}

\vfill
\section{Code}
\label{sec:code}

L'utilisation de l'ordinateur a été cruciale pour démontrer les résultats de cet article.
Le code a été écrit pour SageMath (donc en Python).
Il peut être trouvé à l'adresse :
\begin{center}
\url{https://lipn.univ-paris13.fr/~fernique/info/code3disques.tgz}
\end{center}
Il est articulé en sept programmes :
\begin{enumerate}
  \item \verb+couronnes.sage+ : calcul des vecteurs d'angles et de leurs codages (ou inversement), calcul des couronnes compatibles avec des valeurs de $r$ et $s$ données ;
  \item \verb+equations.sage+ : passage des couronnes aux équations et vérification exacte des candidats $(r,s)$ ;
  \item \verb+biphases.sage+ : calcul et vérification des candidats $(r,s)$ dans le cas des empilements à deux phases ;
  \item \verb+2pc.sage+ : idem dans le cas des empilements avec deux petites couronnes ;
  \item \verb+1pc.sage+ : idem dans le cas des empilements avec une seule petite couronne.
  \item \verb+dessin.sage+ : fonctions pour représenter couronnes ou empilements ;
  \item \verb+empilements.sage+ : codage des $164$ empilements périodiques donné en Annexe~\ref{sec:empilements}.
\end{enumerate}
Les deux premiers, \verb+couronnes.sage+ et \verb+equations.sage+, contiennent des fonctions utilisées dans les autres.
Les trois suivants, \verb+biphases.sage+, \verb+2pc.sage+ et \verb+1pc.sage+ sont indépendants (et relativement similaires).
Les deux derniers, \verb+dessin.sage+ et \verb+empilements.sage+, peuvent être utilisés pour dessiner couronnes ou empilements.
\end{document}